\documentclass[12pt,journal,a4paper,onecolumn]{IEEEtran}
\usepackage{cite}
\usepackage{amsthm}
\usepackage{amsmath}
\usepackage{amssymb} 
\usepackage{amsmath}
\usepackage{amsfonts}
\usepackage{algorithmic}
\usepackage{graphicx}
\usepackage{subcaption}
\usepackage{accents}
\usepackage{multido}
\usepackage{enumitem}
\usepackage{enumerate}
\usepackage{etoolbox}
\usepackage{multicol}
\usepackage[printonlyused]{acronym}
\usepackage{flushend}
\usepackage{newpxtext}
\usepackage[vmargin=1cm,hmargin=2cm,head=18pt,includeheadfoot]{geometry}  
\usepackage[numbers,square,sort&compress]{natbib}
\usepackage{setspace}
\usepackage[hyphens]{url}
\usepackage{hyperref}
\usepackage[hyphenbreaks]{breakurl}
\hypersetup{
    unicode=false,          
    pdftoolbar=true,        
    pdfmenubar=true,        
    pdffitwindow=false,     
    pdfstartview={FitH},    
    pdftitle={McLeish Distribution: Performance of Digital Communications over Additive White McLeish Noise (AWMN) Channels},
    pdfauthor={Dr. Ferkan Yilmaz},
    pdfsubject={Subject},   
    pdfcreator={Dr. Ferkan Yilmaz},   
    pdfproducer={latex.exe + dvips.exe + ps2pdf.exe}, 
    pdfkeywords={Additive white McLeish noise channels, Coherent/non-coherent signaling, Conditional bit error rate, Conditional symbol error rate, McLeish distribution, McLeish Q-function, Multivariate McLeish distribution, Non-Gaussian noise.}, 
    pdfnewwindow=true,      
    linkcolor=red,          
    citecolor=green,        
    filecolor=white,        
    urlcolor=white          
}

\theoremstyle{plain}
\newtheorem{thmcounter}{Theorem}
\newtheorem{theorem}[thmcounter]{Theorem}
\newcommand{\theoremref}[1]{Theorem~\protect\ref{#1}}

\theoremstyle{plain}

\theoremstyle{plain}

\theoremstyle{plain}
\newtheorem{definitioncounter}{Theorem}
\newtheorem{definition}[definitioncounter]{Definition}
\newcommand{\defref}[1]{Definition~\protect\ref{#1}}

\newcommand{\QEDsymbol}{$\hfill\square$}

\newcommand\eq{{=}}
\newcounter{doublecolumnequation}
\preto\subequations{\ifhmode\unskip\fi}
\newcommand{\figref}[1]{Fig.~\protect\ref{#1}}

\newcommand{\secref}[1]{Section~\protect\ref{#1}}

\def\suscript(#1,#2,#3){{#1}^{#2}_{#3}}

\newcommand{\fracparams}[2]{\genfrac{}{}{0pt}{}{{#1}}{{#2}}}
\newcommand\scalemath[3]{\scalebox{#2}[#1]{\mbox{\ensuremath{\displaystyle{#3}}}}}
\newcommand{\emptycoefficient}{\raisebox{0.75mm}{\rule{16pt}{0.4pt}}}
\newcommand{\emptycoefficientSHORT}{\raisebox{0.75mm}{\rule{10pt}{0.4pt}}}
\newcommand{\FoxHDefinition}[3]{\suscript(\rm{H},{#1},{#2}){\left[#3\right]}}
\newcommand{\FoxH}[6][right]{
	\ifthenelse{\equal{#1}{right}}{\suscript(\rm{H},{#2},{#3}){\left[{#4}\left|\fracparams{#5}{#6}\right.\right]}}{
		\ifthenelse{\equal{#1}{left}}{\suscript(\rm{H},{#2},{#3}){\left[\left.{#4}\right|\fracparams{#5}{#6}\right]}}{
			\suscript(\rm{H},{#2},{#3}){\left[{#4}\left|\fracparams{#5}{#6}\right.\right]}
		}
	}
}
\newcommand{\MeijerGDefinition}[3]{\suscript(\rm{G},{#1},{#2}){\left[#3\right]}}
\newcommand{\MeijerG}[6][right]{
	\ifthenelse{\equal{#1}{right}}{\suscript(\rm{G},{#2},{#3}){\left[{#4}\left|\fracparams{#5}{#6}\right.\right]}}{
		\ifthenelse{\equal{#1}{left}}{\suscript(\rm{G},{#2},{#3}){\left[\left.{#4}\right|\fracparams{#5}{#6}\right]}}{
			\suscript(\rm{G},{#2},{#3}){\left[{#4}\left|\fracparams{#5}{#6}\right.\right]}
		}
	}
}
\newcommand{\FoxIDefinition}[3]{\suscript(\rm{I},{#1},{#2}){\left[#3\right]}}
\newcommand{\FoxI}[6][right]{
	\ifthenelse{\equal{#1}{right}}{\suscript(\rm{I},{#2},{#3}){\left[{#4}\left|\fracparams{#5}{#6}\right.\right]}}{
		\ifthenelse{\equal{#1}{left}}{\suscript(\rm{I},{#2},{#3}){\left[\left.{#4}\right|\fracparams{#5}{#6}\right]}}{
			\suscript(\rm{I},{#2},{#3}){\left[{#4}\left|\fracparams{#5}{#6}\right.\right]}
		}
	}
}

\newcommand{\BesselI}[2][0]{
	\suscript({I},{},{#1})\left({#2}\right)
}
\newcommand{\BesselK}[2][0]{
	\suscript({K},{},{#1})\!\left({#2}\right)
}

\newcommand{\HeavisideTheta}[1]{
	{\mathord{\theta}\!}\left({#1}\right)
}
\newcommand{\DiracDelta}[1]{
	{\mathord{\delta}\!\left({#1}\right)}
}
\newcommand{\KroneckerDelta}[2]{
	{\delta_{{#1},{#2}}}
}
\newcommand{\Binomial}[2]{
 	\genfrac{(}{)}{0pt}{}{#1}{#2}
}
\newcommand{\RealPart}[1]{
	\Re\!\left\{{#1}\right\}
}
\newcommand{\ImagPart}[1]{
	\Im\!\left\{{#1}\right\}
}
\DeclareMathOperator*{\argmax}{
	{\operatorname{arg}\operatorname{max}}
}
\DeclareMathOperator*{\argmin}{
	{\operatorname{arg}\operatorname{min}}
}
\newcommand{\trace}{
	{\operatorname{Tr}}
}
\newcommand{\rank}{
	{\operatorname{rank}}
}

\newcommand{\diag}{
	{\,\operatorname{diag}}
}
\newcommand{\abs}[1]{
	\left|{#1}\right|
}

\newcommand{\defmat}[1]{
	{\boldsymbol{\mathrm{#1}}}
}
\newcommand{\defvec}[1]{
	{\boldsymbol{#1}}
}
\newcommand{\defrmat}[1]{
	{\boldsymbol{\MakeUppercase{#1}}}
}
\newcommand{\defrvec}[1]{
	{\boldsymbol{\MakeUppercase{#1}}}
}

\newcommand{\Erfc}[1]{
	{{\rm{erfc}}\!\left({#1}\right)}
}

\newcommand{\Pochhammer}[2]{
	{\left({#1}\right)_{#2}}
}
\newcommand{\Kurtosis}[1]{
	{{\mathrm{Kurt}}\!\left[{#1}\right]}
}
\newcommand{\Skewness}[1]{
	{{\mathrm{Skew}}\!\left[{#1}\right]}
}
\newcommand{\Expected}[1]{
	{{\mathbb{E}}\!\left[{#1}\right]}
}
\newcommand{\Correlation}[2]{
	{\,{\mathbb{R}}\!\left[{#1};{#2}\right]}
}

\newcommand{\iseven}[1]{
	{{\mathrm{en}}\!\left({#1}\right)}
}

\newcommand{\Variance}[1]{
	{{\mathrm{Var}}\!\left[{#1}\right]}
}
\newcommand{\Covariance}[2]{
	{{\mathrm{Cov}}\!\left[{#1},{#2}\right]}
}
\newcommand{\PseudoVariance}[1]{
	{{\mathrm{PVar}}\!\left[{#1}\right]}
}
\newcommand{\imaginary}{\jmath}

\newcommand{\mathsym}[1]{{}}

\usepackage{trimspaces}
\makeatletter
\newcounter{acroplace}
\AtBeginEnvironment{acronym}{\setcounter{acroplace}{1}}
\AtEndEnvironment{acronym}{\setcounter{acroplace}{0}}
\newrobustcmd{\acrochoice}[2]{\trim@pre@space{
    \ifthenelse{\equal{\value{acroplace}}{0}}
        {{#2}}
        {\ifthenelse{\equal{\value{acroplace}}{1}}
            {{#1}}
            {\!{\MakeUppercase#2}}
        }
    }} 
\makeatother

\AtBeginDocument{ %
    \lefthyphenmin=2
    \righthyphenmin=3
    \hyphenation{fad-ing}
    \hyphenation{branch-es}
    \hyphenation{chan-nels}
    \hyphenation{per-for-mance}
    \hyphenation{gene-ra-li-zed}
    \hyphenation{theo-re-ti-cally}
    \hyphenation{base-band}
    \hyphenation{si-mu-la-tion}
    \hyphenation{contel-la-tion}
    \hyphenation{func-tion}
    \hyphenation{mu-tually}
    \hyphenation{com-bi-na-tion}
    \hyphenation{Gaus-sian}
    \hyphenation{un-cor-re-lated}
    \hyphenation{in-ves-ti-ga-tion}
    \hyphenation{in-trin-si-cal-ly}
    \hyphenation{gen-er-al-iza-tion}
 }

\begin{document}
\markboth{F. YILMAZ, McLeish Distribution: Performance of Digital Communications over ~\ldots~ (AWMN) Channels}{IEEE Access}


\title{McLeish Distribution: Performance of Digital Communications over Additive White McLeish Noise (AWMN) Channels}
\author{
    \IEEEauthorblockN{{\large Ferkan Yilmaz}}\\
    \IEEEauthorblockA{{\normalsize Y{\i}ld{\i}z Technical University, Faculty of Electrical \& Electronics Engineering, Department of Computer Engineering, Istanbul, Turkey (e-mail: {\fontfamily{qcr}\selectfont ferkan@yildiz.edu.tr}).}}
    \vspace{-15mm}
}

\maketitle

\begin{abstract}
The objective of this article is to propose and statistically validate a more general additive non-Gaussian noise distribution, which we term McLeish distribution, whose random nature can model different impulsive noise environments commonly encountered in practice and provides a robust alternative to Gaussian noise distribution. In particular, for the first time in the literature, we establish the laws of McLeish distribution and therefrom derive the laws of the sum of McLeish distributions by obtaining closed-form expressions for their probability density function (PDF), cumulative distribution function (CDF), complementary CDF (C\textsuperscript{2}DF), moment-generating function (MGF) and higher-order moments. Further, for certain problems related to the envelope of complex random signals, we extend McLeish distribution to complex McLeish distribution and thereby propose circularly\,/\,elliptically symmetric (CS\,/\,ES) complex McLeish distributions with closed-form PDF, CDF, MGF and higher-order moments. For generalization of one-dimensional distribution to multi-dimensional distribution, we develop and propose both multivariate McLeish distribution and multivariate complex CS\,/\,ES (CCS\,/\,CES) McLeish distribution with analytically tractable and closed-form PDF, CDF, C\textsuperscript{2}DF and MGF.
In addition to the proposed McLeish distribution framework and for its practical illustration, we theoretically investigate and prove the existence of McLeish distribution as additive noise in communication systems. Accordingly, we introduce additive white McLeish noise (AWMN) channels. For coherent\,/\,non-coherent signaling over AWMN channels, we propose novel expressions for maximum \textit{a priori} (MAP) and maximum likelihood (ML) symbol decisions and thereby obtain closed-form expressions for both bit error rate (BER) of binary modulation schemes and symbol error rate (SER) of various M-ary modulation schemes. Further, we verify the validity and accuracy of our novel BER\,/\,SER expressions with some selected numerical examples and some computer-based simulations.
\end{abstract}

\begin{IEEEkeywords}
Additive white McLeish noise channels,
Coherent\,/\,non-coherent signaling,
Conditional bit error rate, 
Conditional symbol error rate, 
McLeish distribution,
McLeish Q-function, 
Multivariate McLeish distribution,
Non-Gaussian noise.
\end{IEEEkeywords}

\bstctlcite{IEEEexample:BSTcontrol}

\section*{List of Acronyms}
\acresetall
\begin{acronym}[Q-function]
\newacro{i.i.d.}[\emph{i.i.d.}]{independent and identically distributed}
\newacro{i.n.i.d.}[\emph{i.n.i.d.}]{independent and non-identically distributed}
\newacro{c.i.d.}[\emph{c.i.d.}]{correlated and identically distributed}
\newacro{c.n.i.d.}[\emph{c.n.i.d.}]{correlated and non-identically distributed}
\newacro{u.i.d.}[\emph{u.i.d.}]{uncorrelated and identically distributed}
\newacro{u.n.i.d.}[\emph{u.n.i.d.}]{uncorrelated and non-identically distributed}
\newacro{RMS}[RMS]{root mean square}
\newacro{UWB}[UWB]{ultra-wide band communications}
\newacro{WCC}[WCC]{wireless chip-to-chip communications}
\newacro{WPC}[WPC]{wireless-powered communications}
\newacro{SIRP}[SIRP]{spherically invariant random process}
\acro{ASE}[ASE]{\acrochoice{Amplified Spontaneous Emission}{amplified spontaneous emission}}
\acro{ASK}[ASK]{\acrochoice{Amplitude Shift Keying}{amplitude shift keying}}
\acro{AWGN}[AWGN]{\acrochoice{Additive White Gaussian Noise}{additive white Gaussian noise}}
\acro{AWLN}[AWLN]{\acrochoice{Additive White Laplacian Noise}{additive white Laplacian noise}}
\acro{AWMN}[AWMN]{\acrochoice{Additive White McLeish Noise}{additive white McLeish noise}}
\acro{BER}[BER]{\acrochoice{Bit Error Rate}{bit error rate}}
\acro{BDPSK}[BDPSK]{\acrochoice{Binary Differential Phase Shift Keying}{binary differential phase shift keying}}
\acro{BFSK}[BFSK]{\acrochoice{Binary Frequency Shift Keying}{binary frequency shift keying}}
\acro{BNCFSK}[BNCFSK]{\acrochoice{Binary Non-Coherent Frequency Shift Keying}{binary non-coherent frequency shift keying}}
\acro{BPSK}[BPSK]{\acrochoice{Binary Phase Shift Keying}{binary phase shift keying}}
\acro{CDMA}[CDMA]{\acrochoice{Code Division Multiple Access}{code division multiple access}}
\acro{CR}[CR]{\acrochoice{Cognitive Radio}{cognitive radio}}
\acro{CS}[CS]{\acrochoice{Circularly Symmetric}{circularly symmetric}}
\acro{CSI}[CSI]{\acrochoice{Channel-Side Information}{channel-side information}}
\acro{CCDF}[C\textsuperscript{2}DF]{\acrochoice{Complementary CDF}{complementary CDF}}
\acro{CCS}[CCS]{\acrochoice{Complex and Circularly Symmetric}{complex and circularly symmetric}}
\acro{CDF}[CDF]{\acrochoice{Cumulative Distribution Function}{cumulative distribution function}}
\acro{CES}[CES]{\acrochoice{Complex and Elliptically Symmetric}{complex and elliptically symmetric}}
\acro{CLT}[CLT]{\acrochoice{Central Limit Theorem}{central limit theorem}}
\acro{DPSK}[DPSK]{\acrochoice{Differential Phase Shift Keying}{differential phase shift keying}}
\acro{DS}[DS]{\acrochoice{Direct Sequence}{direct sequence}}
\acro{DSL}[DSL]{\acrochoice{Digital Subscriber Line}{digital subscriber line}}
\acro{ES}[ES]{\acrochoice{Elliptically Symmetric}{elliptically symmetric}}
\acro{FSO}[FSO]{\acrochoice{Free-Space Optical Communications}{free-space optical communications}}
\acro{ILT}[ILT]{\acrochoice{Inverse Laplace Transform}{inverse Laplace transform}}
\acro{IQR}[IQR]{\acrochoice{Inphase-to-Quadrature Ratio}{inphase-to-quadrature ratio}}
\acro{LT}[LT]{\acrochoice{Laplace Transform}{Laplace transform}}
\acro{M-ASK}[M-ASK]{\acrochoice{M-ary Amplitude Shift Keying}{M-ary amplitude shift keying}}
\acro{M-DPSK}[M-DPSK]{\acrochoice{M-ary Differential Phase Shift Keying}{M-ary differential phase shift keying}}
\acro{M-PSK}[M-PSK]{\acrochoice{M-ary Phase Shift Keying}{M-ary phase shift keying}}
\acro{M-QAM}[M-QAM]{\acrochoice{M-ary Quadrature Amplitude Modulation}{M-ary quadrature amplitude modulation}}
\acro{MAI}[MAI]{\acrochoice{Multiple Access Interference}{multiple access interference}}
\acro{MAP}[MAP]{\acrochoice{Maximum A Posteriori Decision}{maximum a posteriori decision}}
\acro{MGF}[MGF]{\acrochoice{Moment-Generating Function}{moment-generating function}}
\acro{ML}[ML]{\acrochoice{Maximum Likelihood Decision}{maximum likelihood decision}}
\acro{MUI}[MUI]{\acrochoice{Multiple User Interference}{multiple user interference}}
\acro{MOM}[MOM]{\acrochoice{Method of Moments Estimation}{method of moments estimation}}
\acro{OOK}[OOK]{\acrochoice{On-Off Keying}{on-off keying}}
\acro{PDF}[PDF]{\acrochoice{Probability Density Function}{probability density function}}
\acro{PLC}[PLC]{\acrochoice{Power-Line Communications}{power-line communication}}
\acro{PMF}[PMF]{\acrochoice{Probability Mass Function}{probability mass function}}
\acro{Q-function}[Q-function]{\acrochoice{Quantile-function}{quantile-function}}
\acro{QAM}[QAM]{\acrochoice{Quadrature Amplitude Modulation}{quadrature amplitude modulation}}
\acro{QPSK}[QPSK]{\acrochoice{Quadrature Phase Shift Keying}{quadrature phase shift keying}}
\acro{RF}[RF]{\acrochoice{Radio Frequency}{radio frequency}}
\acro{SER}[SER]{\acrochoice{Symbol Error Rate}{symbol error rate}}
\acro{SNR}[SNR]{\acrochoice{Signal-to-Noise Ratio}{signal-to-noise ratio}}
\acro{WSS}[WSS]{\acrochoice{Wide Sense Stationary}{wide sense stationary}}
\end{acronym}

\section{Introduction}
\label{Section:Introduction} 
\IEEEPARstart{T}{he additive white noise} in communication systems\cite[and references therein]{BibAlouiniBook,BibGoldsmithBook,BibProakisBook,BibRappaportBook,BibLapidothBook2017} is commonly defined as an arbitrarily varying undesired signal that additively corrupts signal transmission over communication channels. In the last few decades, many modern techniques have been developed or improved to overcome and eliminate the problem of reliable transmission over noisy communication channels. Such techniques constitute both theoretical information and experimental results on source\,/\,channel coding and modulation schemes. In order to bring the theory and practice together concerning reliable transmission, scientists and researchers have evaluated the analyses of most of the techniques for various noisy channels where it is widely agreed to have signal transmission corrupted additively by thermal noise. The most significant property of thermal noise is that it is abstracted by a complex Gaussian distribution in consequence of the application of \acf{CLT} on the sum of infinitely small noise sources  \cite{BibProakisBook,BibRappaportBook,BibLapidothBook2017,BibFellerBook1968,BibBillingsleyBook1979}. The Gaussian abstraction of additive noise, usually termed as additive Gaussian noise, provides an insight into the underlying behavior of communication channels, while it ignores some other impairments that are common in the nature of various communication channels. For instance, rather than the thermal noise, the presence of undesirable interference signals, which arise in the form of random bursts for a short period of time, induces random fluctuations in the power~of additive noise. Such additive noise with random power fluctuations is called additive non-Gaussian noise, sometimes termed as impulsive additive noise and is of particular concern in many communication systems.   

From the experimental point of view, there~are~many~communication systems in which signal transmission is exposed to additive non-Gaussian noise. For example, in \ac{DSL} communication system, the~random~noise-power fluctuations, predominantly caused by electromagnetic interference due to electrical switches~and~home appliances, are an example source of \mbox{additive non-Gaussian noise} \cite{BibKerpezGottliebTCOM1995,BibNedevThesis2003,BibToumpakarisCioffiGardanTCOM2004,BibAlNaffouriQuadeerCaireISIT2011}. Also, \ac{PLC} is another communication system suffering from additive non-Gaussian noise. As such, in \ac{PLC} system, the impulsive nature of additive non-Gaussian noise inherently forms due to switching transients among different appliances and devices  \cite{BibFerreiraLampeNewburySwartBook,BibHanBook2017,BibGotzRappDostertTCM2004,BibKatayamaYamazatoOkadaJSAC2006,BibHanStoicaKaiserOtterbachDostertDSP2017,BibMengGuanChenTPD2005}. Even if signal transmission over \ac{PLC} networks has been verified as a good technique, the impulsive~nature of non-Gaussian noise is often observed as a hindrance for more efficient \ac{PLC}-based transmission \cite{BibZimmermannDostertTEC2002,BibMengGuanChenTPD2005,BibLinNassarEvansJSAC2013}. 
Additive non-Gaussian noise is also experienced in underwater acoustics channels, which results from interference and malicious jamming \cite{BibWenzASA1962,BibBrockettHinichWilsonASA1987,BibMiddletonOE1987,BibPowellWilsonBook1989,BibSteinASA1995,BibMachellPenrodEllis1989,BibEtterBook2018,BibAbrahamCollection2019}. Other types of communication channels, where signal transmission is subjected to additive \mbox{non-Gaussian} noise, typically include wireless fading channels such as urban and indoor radio channels \cite{BibMiddletonTEC1972,BibMiddletonTCOM1973,BibMiddletonTEC1977,BibMiddletonSpauldingBook1986,BibBlackardRappaportBostianJSAC1993,BibBlankenshipKriztmanRappaportVTC1997,BibBlankenshipRappaportTAP1998}, \ac{UWB} \cite{BibFluryBoudecICUW2006,BibDhibiKaiserMNA2006,BibBaturKocaDundar2008}, frequency\,/\,time-hopping with jamming \cite{BibMoonWongSheaTCOM2006,BibHuBeaulieuRWS2008}, millimeter-wave (around 60 GHz or higher) radio channels \cite{BibCheffenaCM2016,BibShhabRizanerUlusoyAmcaUCMMT2017,BibIqbalLuoMullerSteinbockSchneiderDupleichHafnerThoma2019}, \ac{WCC} \cite{BibMatolakTWC2012,BibNossekEtAlCollection2013,BibSwarbrickPatent2016,BibDrostHopkinsHoSutherlandJSSC2004,BibSuzukiHayakawaPRB2017}, and wireless transmissions under strong interference conditions \cite{BibThompsonChang1994,BibGowdaAnnampeduViswanathanICC1998,BibDhibiTSPKaiser2006,BibDhibiKaiserMNA2006,BibFiorinaGLOBECOM2006,BibHuBeaulieuIEEERWS2008,BibBeaulieuNiranjayan2010}. Further, some impulsive scenarios such as engine ignition, rotating machinery, lighting, as well as some impulsive multi-user interference and multi-path propagation can also produce additive non-Gaussian noise in wireless channels \cite{BibMiddletonTCOM1973,BibMiddletonTEC1977,BibMiddletonSpauldingBook1986,BibSousaIT1992,BibBlackardRappaportBostianJSAC1993,BibGonzalezPhDThesis1997,BibIlowHatzinakosTSP1998,BibWangPoorTSP1999,BibMiddletonTIT1999}. Impulsive effects that introduce additive non-Gaussian noise can also be found in \ac{CR} \cite{BibMitolaMitolaPhDThesis2000,BibHaykinJSAC2005,BibHuWillkommAbusubaihGrossVlantisGerlaWoliszICM2007,BibGoldsmithJafarMaricSrinivasa2009,BibHossainNiyatoHan2009,BibWangLiuJSTSP2011,BibAkyildizLeeVuranMohanty2006} due to the simultaneous spectrum access under miss-detection events  \cite{BibWillkommGrossWolisz2005,BibLiangZengPehHoangICC2007,BibGeirhoferTongSadlerICM2007,BibJafarSrinivasaJSAC2007}. The miss-detection event occurs when a cognitive user fails to detect an active primary user. In this event, collisions happens and generates additive non-Gaussian noise in the signal transmission, which considerably strikes the performance of cognitive links. In addition, in recent years, the impulsive nature of additive noise in \ac{FSO} has received much attention\cite[and\!~references\!~therein]{BibJakemanIOP1980,BibDorfBook2006,BibKambojMallikAgrawalSchoberSPCOM2012}. An essential aspect in optical communications is the existence of the \ac{ASE} noise \cite{BibVaninJacobsenBerntsonOSA2007,BibWitzensMullerMoscosoMartirIPJ2018}. It has been experimentally shown in  \cite{BibChanConradiJLT1997} and theoretically predicted in  \cite{BibMarcuseJLT1990,BibHumbletAzizogluJLT1991,BibMarcuseJLT1991,BibLeeShimJLT1994} that the \ac{ASE} noise follows a non-Gaussian distribution. It is also worth mentioning that, in \ac{WPC} \cite{BibBiHoZhangICM2015,BibLuWangNiyatoKimHanICAST2015,BibKrikidisTimotheouNikolaouZhengNgSchober2014,BibMohjaziMuhaidatDianatiAlQutayriAlDhahir2018,BibMohjaziPhDThesis2018}, we typically observe that wireless power transmission causes some random fluctuations in the power supply voltage of wireless powered radio circuits, which arbitrarily shifts the optimum circuit operating point. Thus, additive noise in \ac{WPC} typically follows a non-Gaussian distribution. 

Consequently, due to the facts and observations mentioned above, we can undoubtedly notice and easily deduce that additive non-Gaussian noise is extremely common in communication channels. Thus, to design different communication techniques and protocols properly, this ubiquitous presence also makes the performance evaluation  more challenging for different coherent\,/\,non-coherent signalling over additive non-Gaussian noise channels. However, to the best of our knowledge, there is no statistical framework in the literature to investigate the performance evaluation  for additive non-Gaussian noise channels. 
\subsection{Non-Gaussian Noise Distributions}
\label{Section:Introduction:NonGaussianDistributions}
From the theoretical point of view, it is worth noting that additive noise following Gaussian distribution has been shown in  \cite{BibShomoronyAvestimehrISIT2012,BibShomoronyAvestimehrTIT2013} 
and operationally justified in \cite{BibLapidothTIT1996} to be the worst noise distribution for communication channels while minimizing the capacity of signal transmission with respect to a noise variance constraint. Hence, the nature of additive non-Gaussian noise in communication channels has impulsive effects that can be properly characterized by its excess-Kurtosis \cite{BibKayBook1998}, where the excess-Kurtosis is zero for Gaussian noise distribution. A noise distribution with a positive excess-Kurtosis has a heavier tail than the Gaussian distribution and hence is identified \textit{(strictly considered)} as a non-Gaussian distribution. In order to adequately capture different impulsive noise effects, many non-Gaussian distributions such as Bernoulli-Gaussian, Middleton Class-A, Class-B and Class-C, Laplacian, symmetric $\alpha$-stable (${S}\alpha{S}$) and generalized Gaussian distributions are proposed in literature. The fact that non-Gaussian distribution may or may not provide mathematically tractable and analytically closed-form statistical results has attracted less attention from research community.

In literature, Bernoulli-Gaussian distribution has been used as an approximation of impulsive noise in communication channels  \cite{BibGhoshTCOM1996,BibShongweFerreiraHanVinck1999,BibMaSoGunawanTPD2005,BibPighiFranceschiniFerrariRaheliTCOM2009,BibAlNaffouriQuadeerCaireISIT2011,BibHerathTranLeNgocICC2012,BibShongweyVinckFerreiraISPLCA2014}. Also, Middleton Class-A, Class-B and Class-C distributions \cite{BibMiddletonTIT1999} distinguish impulsive noise according to the frequency range occupied by~the impulsive effects compared to the receiver bandwidth and have been extensively studied in the literature \cite{BibTepedelenliogluGLOBECOM2004,BibTepedelenliogluGaoTVT2005}. Laplacian distribution is another non-Gaussian distribution~used~to~model the additive impulsive noise effects in signal processing\,/\,detection and communication studies \cite{BibBernsteinIEEE1974,BibDurisiUWST2002,BibHu2004Accurate,BibFiorinaGLOBECOM2006,BibDhibiMNA2006,BibDhibiTSP2006Interference,BibBeaulieuTVT2008,BibHu2008Characterizing,BibBeaulieuIEEE2009,BibChianiIEEE2009Coexistence,BibJiang2010BER,BibKambojMallikAgrawalSchoberSPCOM2012}. Another popular non-Gaussian distribution is ${S}\alpha{S}$ distribution providing a considerably accurate model for impulsive noise \cite{BibIlowHatzinakosTSP1998,BibYangPetropuluTSP2003,BibHaenggiAndrewsBaccelliDousseFranceschettiJSAC2009,BibBrownZoubirTSP2000,BibTsihrintzisMILCOM1996,BibKuruogluRaynerFitzgeraldSSAP1998,BibRajanTepedelenliogluTWC2010}. On the top of Laplacian and ${S}\alpha{S}$ distributions, the generalized Gaussian distribution is one of the most versatile non-Gaussian distributions in the literature. It is commonly used to model noises in several digital communication systems  \cite{BibViswanathanAnsariASSP1989,BibZahabiTadaionICT2010,BibSouryYilmazAlouiniCOML2012,BibSouryYilmazAlouiniCOML2013,BibSouryYilmazAlouiniISIT2013,BibBiglieriYaoYangWCOM2015}. Each non-Gaussian noise distribution mentioned above can be considered as an alternative \textit{(but feeble alternative)} to Gaussian noise distribution and cannot be appropriately interpreted as the sum of large number of independent and identically distributed impulsive noise sources with small power. From the experimental point of view, unlike Gaussian distribution, each non-Gaussian distribution has heavy-tail behavior modeled by positive excess-Kurtosis as mentioned previously. The presence of positive excess-Kurtosis makes some statistical moments infinite and therefrom makes it impossible to fit into many real-world phenomena. For instance, the variance of ${S}\alpha{S}$ distribution is infinite for all $\alpha\!<\!2$. The lack of characterizing the real-world phenomena of impulsive noise sources from Gaussian distribution to non-Gaussian distribution is the crucial weakness of non-Gaussian noise distributions mentioned above. In this context, we propose that \textit{the distribution proposed by McLeish in}  \cite{BibMcLeishCJS1982,BibMcLeishTechReport1982} \textit{can be used as non-Gaussian distribution as \textit{a robust alternative} to Gaussian distribution}. As such, this distribution closely resembles that of the Gaussian distribution; it is symmetric and unimodal and not only has support the whole real line but also has tails that are at least as heavy as those of Gaussian distribution. More importantly, it has all moments finite, and its excess-Kurtosis is always positive (i.e., its Kurtosis is greater than or equal to that of the Gaussian distribution). \textit{We readily deduce from these features that, possessing the important features of Gaussian distribution, this distribution is very useful in modelling impulsive noise phenomena with somewhat heavier tails than the Gaussian distribution has}. However, to the best of our knowledge, the laws of this distribution have so far \textit{not attracted} the attention of theoreticians, practitioners and researchers not only in the field of wireless communications but also in other fields of engineering.

\subsection{McLeish Noise Distribution} 
\label{Section:Introduction:McLeishDistribution}
Suggested in \cite{BibMcLeishCJS1982,BibMcLeishTechReport1982} as a robust alternative to Gaussian distribution is the generalization of Laplace distribution and therefore inherently called generalized Laplacian distribution \cite{BibKozubowskiMeerschaertPodgorski2006,BibMolzKozubowskiPodgorskiCastleHJ2007,BibWuPhDThesis2008,BibKozubowskiBCP2010,BibPodgorskiWegenerCSTM2011,BibKotzKozubowskiPodgorskiBook2012,BibSchluterTredeJAP2016}. However, in literature, generalized Gaussian distribution is also called generalized Laplacian distribution
\cite[Sec.\!~4.4.2]{BibKotzKozubowskiPodgorskiBook2012},
\cite[Sec.\!~6]{BibJohnsonBalakrishnanKotzBookVol2}, \cite{BibSubbotinMS1923,BibZeckhauserThompsonRES1970,BibJakuszenkowDM1979,BibSharma1984,BibTaylorCSDA1992,BibAgroCSSC1995,BibNakamuraMatsuiECSCT1999,BibNakamuraECIJ2002,BibAyeboKozubowskiJPSS2003,BibAissaAbedMeraimESPC2008,BibDubeauElMashoubiIJPAM2011,BibSelimPJSOR2015,BibLassamiAissaAbedMeraimASILOMAR2018}. In order to avoid this confusion and in honor of D.\,\,L. McLeish for his excellent paper \cite{BibMcLeishCJS1982} and his technical report \cite{BibMcLeishTechReport1982}, the distribution suggested in  \cite{BibMcLeishCJS1982,BibMcLeishTechReport1982} has been recently renamed by us as McLeish distribution in  \cite{BibYilmazAlouiniSIU2018} and by Marichev and Trott in Wolfram's blog posts \cite{BibMcLeishDistributionWolfram2018}. Particularly, we propose McLeish distribution as a versatile additive non-Gaussian noise distribution whose statistical description is typically defined on two main observations,~one~of~which~is that~additive noise spontaneously emerges~as the~sum of~many impulsive noise sources with small power, where each impulsive noise source is found to be properly characterized~by a~Laplacian~distribution. The other main observation~is~that, according to the \ac{CLT} \cite{BibPapoulisBook}, the summation of many impulsive noise sources converges to Gaussian distribution as their number infinitely increases. Thus, we conclude that McLeish distribution provides an excellent fit not only for~Gaussian~distribution but also heavy-tailed non-Gaussian distribution and thereby captures different impulse noise environments (i.e., different impulsive noise distributions are of all special cases or approximations of McLeish distribution) \cite{BibYilmazAlouiniSIU2018}. The evolution of its impulsive nature from Gaussian distribution to non-Gaussian distribution is explicitly parameterized in a more nature-inspired way, especially than those of Laplacian, ${S}\alpha{S}$ and generalized Gaussian distributions.


\subsection{Our Motivations and Contributions}
\label{Section:Introduction:OurMotivationAndContributions}
In this article, before explaining the motivation behind our contributions, it is worth mentioning that we propose novel contributions starting from \theoremref{Theorem:McLeishMoments} to \theoremref{Theorem:MLDecisionErrorProbabilityForBDPSKCoherentSignallingOverAWMNChannels} with exact and closed-form (analytical) expressions.

Gaussian distribution has indeed emerged in almost all scientific problems. It is therefore fundamental to all branches~of science and engineering and has been well studied within the literature of probability and statistics. Since it provides closed-form expressions, it allows us to better understand the technical and conceptual problems inherent in science and engineering. On the other hand, although it has been used over and over to solve the scientific problems, it cannot provide solutions for the problems where impulsive statistics (effects) leading to heavy-tailed non-Gaussian distribution are well observed. While paying attention to non-Gaussian distributions mentioned in \secref{Section:Introduction:NonGaussianDistributions} for compatible analysis and synthesis of impulsive effects, we subsequently provide evidence that a non-Gaussian distribution in which the properties of Gaussian distribution are desirable is mostly needed in the literature. This strong piece of evidence motivates us to propose McLeish distribution \cite{BibMcLeishCJS1982,BibMcLeishTechReport1982} in \secref{Section:StatisticalBackground:McLeishDistribution} as a non-Gaussian distribution which has the well-known desirable properties of Gaussian distribution yielding closed-form results \cite{BibYilmazAlouiniSIU2018}. After showing in \secref{Section:StatisticalBackground:McLeishDistribution} that some special cases of McLeish distribution are Dirac's distribution, Laplacian distribution and Gaussian distribution, for the~first~time~in~the~literature, we present the principles behind the laws of univariate McLeish distribution. Accordingly, we propose closed-form expressions for the moments in \theoremref{Theorem:McLeishMoments}. After introducing McLeish's \ac{Q-function} and in \theoremref{Theorem:McLeishQFunctionUsingFoxHAndMeijerG} deriving its closed-from expression using Meijer's G and Fox's H functions \cite{BibPrudnikovBookVol3,BibKilbasSaigoBook,BibMathaiSaxenaHauboldBook}, we propose~the \ac{CDF} in \theoremref{Theorem:McLeishCDF}~and \ac{CCDF} in \theoremref{Theorem:McLeishCCDF}. Moreover, we obtain the lower- and upper-bound approximations for McLeish's \ac{Q-function}. As our other contributions, we propose a closed-form expression for the \ac{MGF} and~compare~its special cases with the results in the literature. To the best of our knowledge, there is no statistical framework in the literature for comparative analysis and synthesis of a univariate non-Gaussian distribution. Accordingly, we bridge the gap by proposing the framework for the laws of univariate McLeish distribution in \secref{Section:StatisticalBackground:McLeishDistribution}.

It is worth noting that many situations arise~in~all~branches of science and engineering, where the sum of distributions is inevitable. For instance, for a reliable signal transmission through additive noise channels, the additive noise can be typically explored to be the sum of noise distributions. The two most important of these situations are diversity combining and cooperative communications \cite{BibAlouiniBook,BibProakisBook,BibGoldsmithBook}. In case of impulsive effects which yields heavy-tailed non-Gaussian noise distribution, there is a demanding need to investigate the statistical laws of the sum of non-Gaussian distributions. This fact highly motivates us to propose in \secref{Section:StatisticalBackground:McLeishSumDistribution} closed-form expressions for the laws of the sums of mutually independent McLeish distributions,~each~of which is typically derived for arbitrary parameters for statistical~characterization purpose.~In particular,~we~propose the \ac{MGF} in \theoremref{Theorem:McLeishSumMGF} and thereby propose the \ac{PDF} and \ac{CDF} in \theoremref{Theorem:McLeishSumPDF} and \theoremref{Theorem:McLeishSumCDF}, respectively. Moreover, we propose the moments in \theoremref{Theorem:McLeishSumMoments} using \theoremref{Theorem:McLeishMoments}. As our other contributions, we derive the special cases of the novel expressions and compare them with the ones available in the literature.    

For our motivation behind the novel contributions~in~both  \secref{Section:StatisticalBackground:CCSMcLeishDistribution} and \secref{Section:StatisticalBackground:CESMcLeishDistribution}, it should be~mentioned~that, for the first time, complex Gaussian distribution was introduced by It\^{o} in \cite{BibItoJJM1952}. Later, the trend in the design~and~analysis of future concepts, novel ideas and new applications have led~to widespread use~of complex~Gaussian distribution in almost all branches of science and engineering. For example,~in~the branch of electrical engineering, the received signal in both \ac{RF} communications\cite[and references therein]{BibAlouiniBook,BibGoldsmithBook,BibProakisBook,BibRappaportBook,BibLapidothBook2017} and~optical communications\cite[and references therein]{BibEssiambreKramerWinzerFoschiniGoebel2010,BibVacondioRivalSimonneauGrellierBononiLorcyAntonaBigoOE2012,BibHagerPhDThesis2014} is represented by a complex signal whose inphase and quadrature parts are jointly subject to bivariate Gaussian distribution\cite[Eq. (2.3-78)]{BibProakisBook} with a simple linear correlation structure that is either \ac{CS} with zero correlation or \ac{ES}~with non-zero correlation between their real and imaginary parts. The \ac{CS} and \ac{ES} features have attracted the attention of many theoreticians, practitioners and researchers and led~to~an~active research~area~for reliable transmission over additive noise channels. However, to the best of our knowledge, for~such problems associated with~a complex signal whose inphase and quadrature parts are subject to non-Gaussian noise, there is a demand in the literature for a complex non-Gaussian distribution yielding closed-form distribution laws. This fact motivates us to propose complex (bivariate) McLeish distribution. According to the correlation structure between its inphase and quadrature parts, our other contributions in~both \secref{Section:StatisticalBackground:CCSMcLeishDistribution} and \secref{Section:StatisticalBackground:CESMcLeishDistribution} can be particularly summarized as follows.
\begin{itemize}
\setlength\itemsep{1mm}
    \item In \secref{Section:StatisticalBackground:CCSMcLeishDistribution}, we introduce in \theoremref{Theorem:CCSMcLeishDefinition} a complex McLeish distribution, similar to complex Gaussian distribution, whose inphase and quadrature~parts~are~jointly uncorrelated while its envelope and phase are mutually independent. A complex distribution is called~\ac{CS}~or~circular if rotating the complex distribution by any angle does not change its \ac{PDF}\cite[P. 64-66]{BibProakisBook}. In accordance~with~that, we propose \ac{CCS} McLeish distribution and further obtain the laws of \ac{CCS} McLeish distribution with closed-form expressions for the \ac{PDF}, \ac{CDF}, \ac{MGF} and~joint~moments in particular from \theoremref{Theorem:CCSMcLeishPDF} to \theoremref{Theorem:CCSMcLeishMoments} as our other contributions.
    \item In \secref{Section:StatisticalBackground:CESMcLeishDistribution}, we extend \ac{CCS} McLeish distribution to a complex McLeish distribution whose inphase and quadrature parts are jointly correlated with a simple linear correlation structure similar to the one found in complex (bivariate) Gaussian distribution\cite[and references therein]{BibAlouiniBook,BibGoldsmithBook,BibProakisBook,BibRappaportBook,BibLapidothBook2017}. Accordingly,~we~introduce and define in \theoremref{Theorem:CESMcLeishDefinition} \ac{CES} McLeish distribution. Thereon, as our other contributions from \theoremref{Theorem:CESMcLeishPDF} to \theoremref{Theorem:CESMcLeishMGF}, we propose the laws of \ac{CES} McLeish distribution with closed-form~\ac{PDF}, \ac{CDF} and \ac{MGF} expressions, respectively.
\end{itemize}

For our motivation in \secref{Section:StatisticalBackground:MultivariateMcLeishDistribution},~it~is~worth~noting that multivariate Gaussian~distribution, which is a generalization of one-dimensional (univariate) Gaussian distribution to the higher dimensions, plays~an~essential~role in all branches of science and engineering. For example, in the field of wireless communications, the usage of multidimensional signaling makes multivariate Gaussian distribution attractive for modeling additive noise in communication channels. On the other hand, to the best of our knowledge, there~is no multivariate non-Gaussian distribution in the literature, which is mathematically tractable and possesses the desirable properties of multivariate Gaussian distribution yielding closed-form results. To bridge this gap, we present in \secref{Section:StatisticalBackground:MultivariateMcLeishDistribution} our following novel contributions.
\begin{itemize}
\setlength\itemsep{1mm}
    \item As a robust alternative to standard multivariate Gaussian distribution \cite{BibAndersenBook1995,BibTongBook1989,BibWangKotzNgBook1989,BibBilodeauBrennerBook1999,BibKotzBalakrishnanBook2004}, we introduce and propose \textit{standard multivariate McLeish distribution} by generalizing univariate \textit{(one-dimensional)} McLeish distribution to the higher dimensions in such a~way that~we define it in \theoremref{Theorem:StandardMultivariateMcLeishDefinition} as the vector (collection) of mutually uncorrelated and identically distributed McLeish distributions with zero mean, unit variance and the same normality. We show that, similar to standard multivariate Gaussian distribution, standard multivariate McLeish distribution maintains its shape under orthogonal transformations since its covariance matrix is a unit matrix. For the first time in the literature, from \theoremref{Theorem:StandardMultivariateMcLeishPDF} to \theoremref{Theorem:StandardMultivariateMcLeishMGF}, we \textit{establish the laws of standard multivariate McLeish distribution} and check their special cases for consistency and completeness. 
    \item Further, for the vector of mutually uncorrelated~and~non-identically distributed McLeish distributions with distinct variances, we find out how the covariance matrix turns from a unit matrix into a positive definite diagonal~one. As our other contribution, we propose in \theoremref{Theorem:INIDMultivariateMcLeishDefinition} \textit{multivariate McLeish distribution with a positive definite diagonal covariance matrix}. For the first time in the literature, from \theoremref{Theorem:INIDMultivariateMcLeishPDF} to \theoremref{Theorem:INIDMultivariateMcLeishMGF}, we \textit{establish the laws of multivariate McLeish distribution with a diagonal covariance matrix}.  
    \item It is also worth noting that a measure of how multivariate Gaussian distribution varies randomly is the correlation structure among marginal Gaussian distributions, known as the covariance matrix, which allows obtaining closed-form and \textit{unique} expressions that facilitate the solutions of many problems in science and engineering. This undeniable fact motivates us to generalize~the~correlation structure of multivariate McLeish distribution from one diagonal matrix to a full-rank positive definite matrix. Accordingly, our other contribution in \secref{Section:StatisticalBackground:MultivariateMcLeishDistribution} is to discuss~the properties of covariance matrix and in \theoremref{Theorem:MultivariateMcLeishDecomposition} to propose \textit{multivariate McLeish distribution with a positive definite covariance matrix} whose distribution laws are established by closed-form expressions from \theoremref{Theorem:MultivariateMcLeishPDF} to \theoremref{Theorem:MultivariateMcLeishMGF}. Furthermore, not only in \theoremref{Theorem:MultivariateMcLeishLinearityProperty}, where we show  that multivariate McLeish distribution is closed under any non-degenerate affine transformation but also in  \theoremref{Theorem:MultivariateMcLeishPartioningProperty}, we show that its conditional and marginal distributions are also jointly multivariate McLeish distribution.
\end{itemize}

Besides, for our motivation behind the novel contributions in \secref{Section:StatisticalBackground:MultivariateComplexMcLeishDistribution}, it should be mentioned that multivariate complex distributions are predominantly used. For instance, in electrical engineering, the theory of wireless transmission mostly deals with complex distributions. In \cite{BibWoodingBiometrika1956}, Wooding proposed multivariate complex Gaussian distribution and studied its correlation structure. Later, Goodman \cite{BibGoodmanAMS1963} discussed its statistical properties with the analogue of the Wishart distribution by considering multiple and partial correlations. After Goodman \cite{BibGoodmanAMS1963} and others  \cite{BibGiriAMS1965,BibKabeAJS1966A,BibKabeAJS1966B,BibBrillinger1969,BibGoodmanDubmanISMA1969,BibCaponGoodmanProcIEEE1970,BibYoungTechReport1971,BibEatonBook1983,BibvandenBosIT1995,BibAndersenHojbjerreSorensenEriksenBook1995,BibvandenBosIT1998}, research and studies on multivariate complex Gaussian statistical analysis got an impetus. As such, the trend in the design and analysis of transmission technologies provoked the widespread use~of multivariate complex Gaussian distribution to model random fluctuations in \ac{RF} communications\cite[and references therein]{BibAlouiniBook,BibGoldsmithBook,BibProakisBook,BibRappaportBook,BibLapidothBook2017} and~optical communications\cite[and references therein]{BibEssiambreKramerWinzerFoschiniGoebel2010,BibVacondioRivalSimonneauGrellierBononiLorcyAntonaBigoOE2012,BibHagerPhDThesis2014,BibMarcuseJLT1990}.
It was later either explicitly or implicitly often assumed and experimentally verified that multivariate additive noise in wireless transmissions follows a multivariate \ac{CCS}\,/\,\ac{CES} Gaussian distribution with a \ac{CS}\,/\,\ac{ES} correlation structure\cite[and references therein]{BibAlouiniBook,BibGoldsmithBook,BibProakisBook,BibRappaportBook,BibLapidothBook2017}. On the other hand, due to the purposes mentioned previously, there exists a demand for multivariate \ac{CCS}\,/\,\ac{CES} non-Gaussian distribution that yields closed-form distribution laws. This fact motivates us to propose in \secref{Section:StatisticalBackground:MultivariateComplexMcLeishDistribution} the extension of multivariate McLeish distribution to \textit{multivariate \ac{CCS}\,/\,\ac{CES} McLeish distribution} with the following novel contributions.  
\begin{itemize}
\setlength\itemsep{1mm}
    \item For the vector of \textit{uncorrelated and identically} distributed CCS McLeish distributions with zero mean, unit variance, and the same normality, we introduce in \theoremref{Theorem:StandardMultivariateCCSMcLeishDefinition} \textit{standard multivariate CCS McLeish distribution} and establish its distribution laws by obtaining closed-form \ac{PDF}, \ac{CDF}, \ac{CCDF}, and \ac{MGF} expressions from \theoremref{Theorem:StandardMultivariateCCSMcLeishPDF} to \theoremref{Theorem:StandardMultivariateCCSMcLeishMGF}, respectively. Further, we show that standard multivariate CCS McLeish distribution is closed under unitary transformation. 
    \item For the vector of \textit{uncorrelated and non-identically}~distri\-buted \ac{CCS} McLeish distributions with different~variances, we introduce in \theoremref{Theorem:INIDMultivariateCESMcLeishDefinition} \textit{multivariate \ac{CES} McLeish distribution with a diagonal covariance matrix} whose distribution laws are established by closed-form \ac{PDF}, \ac{CDF}, \ac{CCDF}, and \ac{MGF} expressions from \theoremref{Theorem:INIDMultivariateCESMcLeishPDF} to \theoremref{Theorem:INIDMultivariateCESMcLeishMGF}, respectively. 

    \item As our other contributions in \secref{Section:StatisticalBackground:MultivariateComplexMcLeishDistribution}, for~the~vector of \textit{correlated and non-identically} distributed \ac{CES} McLeish distributions with different variances, we introduce in \theoremref{Theorem:MultivariateCESMcLeishDecomposition} \textit{multivariate \ac{CES} McLeish distribution with a complex covariance matrix}. After investigating the circular symmetry and positive definite properties of complex covariance~matrices,~we~obtain closed-form \ac{PDF}, \ac{CDF}, \ac{CCDF}, and \ac{MGF} expressions from \theoremref{Theorem:MultivariateCESMcLeishPDF} to \theoremref{Theorem:MultivariateCESMcLeishMGF}, respectively, establishing the distribution laws of \textit{multivariate \ac{CES} McLeish distribution} in general.   
\end{itemize}

In consequence with our above-mentioned contributions, to study impulsive statistics using McLeish distribution, for the first time, we propose in \secref{Section:StatisticalBackground} a general framework both in scalar as well as in vector version. With the aid of this framework, we propose \textit{\ac{AWMN} channels} in \secref{Section:AWMNChannels} and investigate the impulsive effects within wireless communication systems. We explain the motivation behind our contributions as follows.  
\begin{itemize}
\setlength\itemsep{1mm}
    \item In the literature, it is widely assumed that noise variance (i.e., noise power) is constant and precisely known to the receiver \cite{BibProakisBook,BibAlouiniBook,BibGoldsmithBook,BibRappaportBook,BibLapidothBook2017}. However, this is practically impossible since noise variance in any wireless communication indeed fluctuates randomly over time due to temperature change, ambient interference, and filtering \cite{BibGardnerTCOM1988,BibShellhammerTandraIEEEStd2006,BibCabricBrodersenTechReport2007,BibYesteOjedaGrajalWSSP2007,BibLarssonThobabenWangWCNC2010,BibMarianiGiorgettiChiani2018}. The noise variance fluctuations are known as impulsive effects. Depending on the presence of impulsive effects in the communication channel, the variance of noise variance fluctuations changes from one communication system to the other communication system; sometimes, it can be very severe to be considered and sometimes very weak to be ignored. In \secref{Section:AWMNChannels:RandomNoiseVarianceFluctuations}, we investigate noise variance fluctuations. In \theoremref{Theorem:VarianceCorrelation}, we propose the usage of Allan's variance to determine the correlation within noise variance fluctuations and obtain in \theoremref{Theorem:VarianceCorrelationCoefficient} the corresponding  auto-correlation coefficient. From these results, we propose in \theoremref{Theorem:VarianceCoherenceWindow} the coherence time for noise variance fluctuations. According to the uncertainty of noise variance and the comparison of this coherence time both with the coherence time of fading conditions and the symbol duration, we introduce the classification of additive noise channels as (i) constant variance, (ii) slow-variance uncertainty, (iii) fast-variance uncertainty. Subsequently, we emphasize that the McLeish distribution can model these three classes of noise variance uncertainties with the aid of its normality parameter.
    
    \item In \secref{Section:AWMNChannels:McLeishNoiseExistence}, we investigate~the~existence~of~McLeish noise distribution in wireless communications. In more detail, for the first time in the literature, we show in \secref{Section:AWMNChannels:McLeishNoiseExistence:JohnsonNoise} that the thermal noise in electronic materials follows McLeish distribution rather than Gaussian distribution. Further, we show in 
    \secref{Section:AWMNChannels:McLeishNoiseExistence:MAIAndMUIInterference} that \ac{MAI}\,/\,\ac{MUI} also follow McLeish distribution rather than Laplacian distribution. To represent how McLeish noise distribution can model wide range of realistic impulsive effects (uncertainty of noise variance), we emphasize in \secref{Section:AWMNChannels:McLeishNoiseExistence:Versatility} that the McLeish distribution demonstrates a superior fit to the different impulsive noise from non-Gaussian to Gaussian distribution.
\end{itemize}

From the important findings highlighted in \secref{Section:AWMNChannels},~it~is obviously more appropriate to model additive white noise in communication channels by McLeish distribution rather~than Gaussian distribution. The impulsive effects follow a non-Gaussian distribution, requiring the use of McLeish distribution as a convenient non-Gaussian model for additive white noise in communication channels. Therefore, the performance analysis of communication systems is critical if they are exposed to additive non-Gaussian noise that follows McLeish distribution. Thanks to our contributions as mentioned earlier, and thanks to the statistical framework that we propose for McLeish distribution from \secref{Section:StatisticalBackground} to \secref{Section:AWMNChannels}, we propose in \secref{Section:SignallingOverAWMNChannels} \textit{complex correlated \ac{AWMN} vector channels}. For coherent and non-coherent signaling over complex correlated \ac{AWMN} vector channels, we have the following contribution sets about the performance of binary and M-ary modulation schemes. 
\begin{itemize}
\setlength\itemsep{1mm}
    \item In \secref{Section:SignallingOverAWMNChannels:CoherentSignalling}, we investigate the \ac{CSI} requirements for coherent signaling over complex correlated \ac{AWMN} vector channels and thereby propose closed-form \ac{MAP} and \ac{ML} decision rules for M-ary modulation schemes from \theoremref{Theorem:MAPDecisionRuleForComplexAWMNVectorChannel} to \theoremref{Theorem:MLDecisionRuleForPrecodedComplexAWMNVectorChannel}. Thanks to the closed-form \ac{MAP} and \ac{ML} decision rules, we analyze in \secref{Section:SignallingOverAWMNChannels:CoherentSignalling:SymbolErrorProbability} the \ac{BER}\,/\,\ac{SER} performance of coherent signaling. As such, in \theoremref{Theorem:UnionUpperBoundSEPForPrecodedComplexAWMNVectorChannel}, we obtain \textit{closed-form upper-bound expressions} for \ac{SER} of M-ary modulation schemes. From \theoremref{Theorem:MAPDecisionRuleForBinaryCoherentSignallingOverAWMNChannels} to
    \theoremref{Theorem:MLDecisionRuleForBinaryCoherentSignallingOverAWMNChannels}, we obtain the \ac{MAP} and \ac{ML} decision rules for binary modulation schemes. Furthermore, from \theoremref{Theorem:MLDecisionErrorProbabilityForBinaryCoherentSignallingOverAWMNChannels} to \theoremref{Theorem:ConditionalBEPForOOKSignaling}, we analyze the \ac{BER} performance~of~binary modulation schemes and therein propose \textit{exact closed-form \ac{BER}~performance~expressions}~of~\ac{BPSK}, \ac{BFSK}, and \ac{OOK} modulation schemes. As our other contributions, from \theoremref{Theorem:MLDecisionErrorProbabilityForASKCoherentSignallingOverAWMNChannels} to \theoremref{Theorem:MLDecisionErrorProbabilityForMPSKCoherentSignallingOverAWMNChannelsII}, we obtain \textit{exact closed-form \ac{SER} expressions of M-ary modulation schemes} such as \ac{M-ASK}, \ac{M-QAM}, \ac{M-PSK} and \ac{QPSK}, 
    \item For our motivation in \secref{Section:SignallingOverAWMNChannels:NonCoherentSignalling},~it~is worth mentioning that, in wireless communications, when~the~phase~of the received signal cannot be accurately recovered at the receiver,~coherent signaling cannot be performed. In such scenarios, communication systems must rely upon non-coherent or differentially coherent signal reception. Accordingly, we investigate the \ac{CSI} requirements for non-coherent signaling over complex correlated \ac{AWMN} vector channels and propose the \ac{MAP} and \ac{ML} decision rules from \theoremref{Theorem:NoncoherentMAPDecisionRuleForComplexAWMNVectorChannel} to
    \theoremref{Theorem:NoncoherentMLDecisionRuleForPrecodedComplexAWMNVectorChannel}. After deriving in \theoremref{Theorem:NoncoherentInaccurateComponentPDF} and \theoremref{Theorem:NoncoherentAccurateComponentPDF} the \acp{PDF} for the inphase and quadrature projections of the received complex signal on possible modulation symbols, we analyze the non-coherent signalling over complex correlated \ac{AWMN} vector channels from \theoremref{Theorem:MAPDecisionErrorProbabilityForNoncoherentOrthogonalSignallingOverAWMNChannels} to
    \theoremref{Theorem:MLDecisionErrorProbabilityForBDPSKCoherentSignallingOverAWMNChannels}, where we propose closed-form expressions for non-coherent orthogonal signaling, \ac{BNCFSK}, \ac{M-DPSK}, and  \ac{BDPSK}.
\end{itemize}

As a result, with the extensive aid of the contributions~mentioned above, we can conclude that multivariate \ac{CCS} and \ac{CES} McLeish distribution is more general additive noise distribution that can be readily used in all branches of science and engineering. 

\subsection{Article Organization} 
We organize the remainder of the article as follows. In \secref{Section:Preliminaries}, we introduce the notation and statistical definitions. In \secref{Section:StatisticalBackground}, we establish the laws of McLeish distribution, where we start from its univariate case and continue through to its multivariate case both in real domains and complex domains. In \secref{Section:AWMNChannels}, we investigate the variance-uncertainty of additive noise and then introduce \ac{AWMN} channels with existence examples in the communication technologies. After presenting the complex \ac{AWMN} vector channels in \secref{Section:SignallingOverAWMNChannels}, we study the \ac{BER}\,/\,\ac{SER} performance of modulation schemes in \secref{Section:SignallingOverAWMNChannels:CoherentSignalling} for coherent signaling and \secref{Section:SignallingOverAWMNChannels:NonCoherentSignalling} for non-coherent signaling over \ac{AWMN} channels. Finally, we offer some concluding results in the last section. 

\section{Preliminaries}\label{Section:Preliminaries}
In this section, we introduce the notations used in this article and present some special functions and statistical definitions. 

\subsection{Notations}\label{Section:Preliminaries:Notation}
In general, scalar numbers such as integer, real and complex numbers are denoted by lowercase letters, e.g. $n$, $x$,
$z$. Let $\mathbb{N}$ denote the set of natural numbers, $\mathbb{R}$ the set of real numbers. As such, $\mathbb{R}_{+}$ and $\mathbb{R}_{-}$ denote the sets of positive and negative real numbers, respectively.  Appropriately, the set of complex numbers, denoted by $\mathbb{C}$, is the plane
$\mathbb{R}\!\times\!\mathbb{R}\!=\!\mathbb{R}^2$ equipped with complex addition, complex multiplication, yielding
complex space. The complex conjugate of $z\!=\!(x,y)\!=\!{x}+\imaginary{y}\in\mathbb{C}$ is denoted by
$z^{*}\!=\!(x,-y)\!=\!{x}-\imaginary{y}$, where $x,y\!\in\!\mathbb{R}$ and $\imaginary\!=\!\sqrt{-1}$
denotes the imaginary number. Furthermore, the inphase $x\!=\!\RealPart{z}$ and the quadrature $y\!=\!\ImagPart{z}$, where $\Re\{\cdot\}$ and $\Im\{\cdot\}$ give the real part and
imaginary part of a given complex number, respectively. Any
non-zero complex number has a polar representation $z\!=\!|z|\exp(\imaginary\theta)$, where
$\theta\!=\!\arg(z)\!\in\![-\pi,\pi)$ is called the \emph{argument} of $z$, and $|z|\!=\!d(z,0)$ denotes the
($L_2$-norm) \emph{modulus} of $z$, where $d^2(\cdot,\cdot)\colon\mathbb{C}\times\mathbb{C}\!\rightarrow\!\mathbb{R}$ denotes the Euclidean squared-distance between $z_{k}\!=\!{x}_{k}+\imaginary{y}_{k}\!\in\!\mathbb{C}$ and
$z_{\ell}\!=\!{x}_{\ell}+\imaginary{y}_{\ell}\in\mathbb{C}$, defined as
\begin{equation}
	d^2(z_{k},z_{\ell})=
		{\langle{z_{k}-{z_\ell}},{z_{k}-{z_\ell}}\rangle},
\end{equation}
where $\langle\cdot,\cdot\rangle\colon\mathbb{C}\!\times\!\mathbb{C}\!\rightarrow\!\mathbb{R}$ denotes the Euclidean inner product in complex space, defined as
\begin{equation}\label{Eq:EuclideanInnerProduct}
	\langle{z_{k}},{z_\ell}\rangle=\RealPart{z^{*}_k{z}_\ell}=
		\frac{1}{2}z^{*}_k{z}_\ell+\frac{1}{2}z_k{z}^{*}_\ell=
			x_{k}x_{\ell}+y_{k}y_{\ell}.
\end{equation}
Further, the inphase and quadrature of any $z\!\in\!\mathbb{C}$ are given by $\RealPart{z}\!=\!{\langle{1},{z}\rangle}$, and $\ImagPart{z}\!=\!{\langle{\imaginary},{z}\rangle}$, 
respectively. Also, the modulus is given by $|z|\!=\!\sqrt{\langle{z},{z}\rangle}$.
When the inphase and quadrature numbers of a complex space are correlated, the distance between $z_{k}\!=\!{x}_{k}+\imaginary{y}_{k}\!\in\!\mathbb{C}$ and $z_{\ell}\!=\!{x}_{\ell}+\imaginary{y}_{\ell}\!\in\!\mathbb{C}$ is obtained by Mahalanobis squared-distance, that is
\begin{equation}
	d^2(z_{k},z_{\ell})=
		{\langle{z_{k}-{z_\ell}},{z_{k}-{z_\ell}}\rangle}_{\rho},
\end{equation}
where $\rho\!\in\![-1,1]$ denotes the correlation coefficient between the inphase and quadrature numbers, and $\langle\cdot,\cdot\rangle_{\rho}\colon\mathbb{C}\!\times\!\mathbb{C}\!\rightarrow\!\mathbb{R}$ denotes the  
Mahalanobis inner product in complex space, defined as
\begin{equation}\label{Eq:MahalanobisInnerProduct}
\!\!{\langle{z_{k}},{z_\ell}\rangle}_{\rho}
	={(x_{k}x_{\ell}+y_{k}y_{\ell}-\rho{x_{k}y_{\ell}}-\rho{y_{k}x_{\ell}})}/{(1-\rho^2)},\!\!
\end{equation}
in correlated complex space (i.e, $\rho\!\neq\!0$).
The modulus of $z$ is given by $|z|_{\rho}\!=\!\sqrt{\langle{z},{z}\rangle_{\rho}}$.
Setting $\rho\!=\!{0}$ in \eqref{Eq:MahalanobisInnerProduct} yields \eqref{Eq:EuclideanInnerProduct}, i.e.,
${\langle{z_{k}},{z_\ell}\rangle}_{0}\!=\!{\langle{z_{k}},{z_\ell}\rangle}$. Thus, $|z|_{0}\eq|z|$.

For simplicity in multi-dimensional space,~column~vectors are denoted by boldfaced lowercase letters, e.g. $\defvec{z}\!=\!\defvec{x}+\allowbreak\imaginary\defvec{y}\!\in\!\mathbb{C}^{m}$, where $\defvec{x}\eq[x_1,x_2,\ldots,x_m]\!\in\!\mathbb{R}^{m}$ and $\defvec{y}\eq[y_1,\allowbreak{y}_2,\ldots,\allowbreak y_m]\!\in\!\mathbb{R}^{m}$.
Similarly, matrices are denoted by boldfaced uppercase non-italic letters, e.g. $\defmat{Z}\!=\!\defmat{X}+\imaginary\defmat{Y}\!\in\!\mathbb{C}^{{m}\times{n}}$, where
$\defmat{X},\defmat{Y}\!\in\!\mathbb{R}^{{m}\times{n}}$. Moreover, the identity matrix of size ${m}\!\times\!{m}$ is fixedly
denoted by $\defmat{I}_{m}$, and both zero vector of size ${m}$ and zero matrix of size ${m}\!\times\!{m}$ 
are also fixedly denoted by $\defmat{0}_{m}$.
Further, transpose and hermitian (conjugate) transpose are denoted by $(\cdot)^T$ and $(\cdot)^H$, respectively.
$\det(\cdot)$, $(\cdot)^{-1}$ and $\trace(\cdot)$ denote the determinant, inverse and trace matrix operations,
respectively. $\diag(\cdot)$ yields a square diagonal matrix whose diagonal is formed from an 
vector. Furthermore, in multi-dimensional space whose dimensions are correlated, the Mahalanobis squared-distance 
between $\defvec{x}\!\in\!\mathbb{R}^{m}$ and $\defvec{y}\!\in\!\mathbb{R}^{m}$ is given by 
\begin{equation}\label{Eq:MahalanobisSquaredDistanceInHigherDimensions}
	d^2(\defvec{x},\defvec{y})={\langle{\defvec{x}-\defvec{y},\defvec{x}-\defvec{y}}\rangle}_{\defmat{P}}
\end{equation}
with the correlation coefficient matrix $\defmat{P}\!=\!\begin{bmatrix}\rho_{jk}\end{bmatrix}_{{m}\times{m}}$, where 
$\rho_{jj}\!=\!{1}$, $\rho_{jk}\!=\!\rho_{kj}$ and $-1\!\leq\!\rho_{j,k}\!\leq\!1$ for all $1\!\leq\!{j,k}\!\leq\!m$. Note that  
$\defmat{P}$ must be symmetric and  
positive definite (i.e., $\defvec{x}^T\defmat{P}\defvec{x}\!>\!{0}$ 
for all $\defvec{x}\!\in\!\mathbb{R}^{m}$). Moreover, in \eqref{Eq:MahalanobisSquaredDistanceInHigherDimensions}, 
$\langle\cdot,\cdot\rangle_{\defmat{P}}\colon\mathbb{R}^{m}\!\times\!\mathbb{R}^{m}\!\rightarrow\!\mathbb{R}$ denotes the  
Mahalanobis inner product in higher dimensional space, and is typically defined as   
\begin{equation}\label{Eq:MahalanobisInnerProductInHigherDimensions}
	{\langle{\defvec{x},\defvec{y}}\rangle}_{\defmat{P}}=\defvec{x}^{T}\defmat{P}^{-1}\defvec{y}.
\end{equation}
Herewith, the norm of $\defvec{x}$, defined as $\lVert{\defvec{x}}\rVert_{\defmat{P}}\!=\!d(\defvec{x},\defvec{0})$, is written as 
$\lVert{\defvec{x}}\rVert_{\defmat{P}}\!=\!\sqrt{{\langle{\defvec{x},\defvec{x}}\rangle}_{\defmat{P}}}\!=\!\lVert{\defmat{P}^{-{1}/{2}}\defvec{x}}\rVert$. In case of no correlation, we have $\defmat{P}\!=\!\defmat{I}$, and hence
reduce \eqref{Eq:MahalanobisInnerProductInHigherDimensions} to the well-known Euclidean inner product in higher dimensional space, that is given by
\begin{equation}\label{Eq:EuclideanInnerProductInHigherDimensions}
	{\langle{\defvec{x},\defvec{y}}\rangle}=\defvec{x}^{T}\defvec{y},
\end{equation}  
and the norm of $\defvec{x}$ to 
$\lVert{\defvec{x}}\rVert\!=\!\sqrt{{\langle{\defvec{x},\defvec{x}}\rangle}}$. In multi-dimensional complex spaces, similar notations also exist but treat Hermitian instead of transpose operation. For example, for $\defvec{z},\defvec{w}\!\in\!\mathbb{C}^{m}$ and $\defmat{\Sigma}\!\in\!\mathbb{C}^{{m}\times{m}}$, we have
\begin{equation}
    {\langle{\defvec{z},\defvec{w}}\rangle}_{\defmat{\Sigma}}=
        \defvec{z}^{H}\defmat{\Sigma}^{-1}\defvec{w},
\end{equation}
Moreover, when $\defmat{\Sigma}\!=\!\defmat{I}$, it simplifies more to ${\langle{\defvec{z},\defvec{w}}\rangle}\!=\!\defvec{z}^{H}\defvec{w}$. Appropriately, the norm of $\defvec{z}$ is written as $\lVert{\defvec{z}}\rVert_{\defmat{\Sigma}}\!=\!\sqrt{{\langle{\defvec{z},\defvec{z}}\rangle}_{\defmat{\Sigma}}}\!=\!\lVert{\defmat{\Sigma}^{-{1}/{2}}\defvec{z}}\rVert$. Further, in case of $\defmat{\Sigma}\!=\!\defmat{I}$, it reduces more to $\lVert{\defvec{z}}\rVert\!=\!\sqrt{{\langle{\defvec{z},\defvec{z}}\rangle}}$ as expected.  

In order to make the accomplishments of probability and statistics concise and comprehensible, $\Pr\{A\}$ and $\Pr\{A|B\}$ will denote the probability of event $A$ and the probability of event $A$ given event
$B$, respectively. Random distributions will be denoted by uppercase letters, e.g. $X$, $Y$, $Z$. Random vectors and random matrices will be denoted by calligraphic boldfaced uppercase letters, e.g. $\defrmat{X}$, $\defrmat{Y}$, $\defrmat{Z}$.
Let $X$ be a random distribution, then its \ac{PDF} is defined by
\begin{equation}
	f_X(x)\!=\!\Expected{\DiracDelta{x-X}},
\end{equation}
where $\Expected{\cdot}$ denotes the expectation operator, and 
$\DiracDelta{\cdot}$ denotes the Dirac's delta function\cite[Eq.\!~(1.8.1)]{BibZwillingerBook}.
Besides, its \ac{CDF} is defined by
\begin{equation}
	F_X(x)\!=\!\Expected{\HeavisideTheta{x-X}},
\end{equation}
where $\HeavisideTheta{\cdot}$ is the Heaviside's theta function\cite[Eq.\!~(1.8.3)]{BibZwillingerBook}.
Furthermore, the conditional \ac{PDF} and \ac{CDF} of $X$ given $G$ will also be  denoted by $f_{X|G}(x|g)$ and $F_{X|G}(x|g)$, respectively. Denoted by $\defrmat{Z}\!=\![X,Y]^T$ is a real random vector formed of the real and imaginary parts of complex random distribution $Z\!=\!{X}+\imaginary{Y}$, where $X$ and $Y$ are two real random distributions whose joint \ac{PDF} $f_{\defrmat{Z}}(x,y)$ is 
\begin{equation}
	f_{\defrmat{Z}}(x,y)=\Expected{\DiracDelta{x-X}\DiracDelta{y-Y}}.
\end{equation}
Since $Z\!=\!{X}+\imaginary{Y}$ as a linear combination of $X$ and 
$Y$ \cite{BibHandersenHojbjerreSorensenEriksenBook1995}, the \ac{PDF} of $Z$ is given by
$f_{Z}(z)\!=\!{f}_{\defrmat{Z}}(\Re\{z\},\Im\{z\})$.
Similarly, the joint \ac{CDF} of $X$ and $Y$ is
\begin{equation}
	F_{\defrmat{Z}}(x,y)=\Expected{\HeavisideTheta{x-X}\HeavisideTheta{y-Y}}.
\end{equation}
The \ac{CDF} of $Z$ is readily given by $F_{Z}(z)\!=\!{F}_{\defrmat{Z}}(\Re\{z\},\Im\{z\})$.
In addition, upon considering $Z$ as a linear combination of $X$ and $Y$, the \ac{MGF} is useful
for finding the \ac{PDF} and \ac{CDF} of $Z$.
The \ac{MGF} of $Z$, defined as $M_{Z}(s)\!=\!\Expected{\exp(-\langle{s},{Z}\rangle)}$ for
$s\!=\!{s}_{X}+\imaginary{s}_{Y}\in\mathbb{C}$ and ${s}_{X},{s}_{Y}\in\mathbb{R}$, is equivalent to the 
joint \ac{MGF} of $\defrmat{Z}$, that is
\begin{equation}\label{Eq:JointMGFDefinition}
	M_{\defrmat{Z}}({s}_{X},{s}_{Y})=\Expected{\exp(-{s}_{X}X-{s}_{Y}Y)},
\end{equation}
which is finite in $s\!\in\!\mathbb{D}\!\subset\!\mathbb{C}^2$. Thus, we rewrite 
$M_{Z}(s)\!=\!{M}_{\defrmat{Z}}(\Re\{s\},\Im\{s\})$ exploiting complex notations. 
Similarly, the \acp{MGF} of $X$ and $Y$ are respectively denoted by 
$M_{X}(s)\!=\!\Expected{\exp(-sX)}$ and $M_{Y}(s)\!=\!\Expected{\exp(-sY)}$. 
In statistical analysis, $\Variance{\cdot}$, $\PseudoVariance{\cdot}$,
$\Covariance{\cdot}{\cdot}$, $\Skewness{\cdot}$ and $\Kurtosis{\cdot}$ will represent variance,
pseudovariance, covariance, skewness and Kurtosis operators, respectively. 
Consequently, $\mathbb{E}[Z]$ is written as 
$\mathbb{E}[Z]=\mathbb{E}[X]+\imaginary\mathbb{E}[Y]$. 
Besides, $\Variance{Z}\!=\!\mathbb{E}[|Z-\mathbb{E}[Z]|^2]$ is written as 
\begin{equation}
	\Variance{Z}=\Variance{X}+\Variance{Y} 
\end{equation}
which does not possess any information about  
$\Covariance{X}{Y}\!=\!\mathbb{E}[(X-\mathbb{E}[X])(Y-\mathbb{E}[Y])]$.
However, the pseudovariance of $Z$, defined as  
 $\PseudoVariance{Z}\!=\!\mathbb{E}[(Z-\mathbb{E}[Z])^2]$, contains it, that is 
\begin{equation}
	\PseudoVariance{Z}=\Variance{X}-\Variance{Y}+\imaginary\,2\Covariance{X}{Y}. 
\end{equation}
   
In addition, for shorthand notations of random distributions,
$\mathcal{N}(\mu,\sigma^2)$, $\mathcal{L}(\mu,\sigma^2)$,
and $\mathcal{M}_{\nu}(\mu,\sigma^2)$ denote Gaussian distribution, Laplacian distribution, and McLeish
distribution, respectively, with $\nu$ normality, $\mu$ mean and $\sigma^2$ variance. Their \ac{CCS} distributions are denoted by $\mathcal{CN}(\mu,\sigma^2)$, $\mathcal{CL}(\mu,\sigma^2)$, and $\mathcal{CM}_{\nu}(\mu,\sigma^2)$, respectively. Similarly, their \ac{CES} distributions for a  correlation coefficient $\rho\!\in\![-1,1]$ are similarly denoted by
$\mathcal{EN}(\mu,\sigma^2,\rho)$, $\mathcal{EL}(\mu,\sigma^2,\rho)$, and $\mathcal{EM}_{\nu}(\mu,\sigma^2,\rho)$,
respectively. Further, $\mathcal{E}(\Omega)$ and $\mathcal{G}(m,\Omega)$ denote an exponential distribution and a Gamma
distribution, where $\Omega\in\mathbb{R}_{+}$ denotes the average power and  $m\!\in\!\mathbb{R}_{+}$ denotes the fading figure \textit{(shape parameter)} describing the amount of spread from the average power $\Omega$. In addition, the symbol $\sim$ stands for \textit{``distributed as''}, e.g., $X\!\sim\!\mathcal{M}_{\nu}(\mu,\sigma^2)$.

\begin{figure*}[tp] 
\centering
\begin{subfigure}{0.7\columnwidth}
    \centering
    \includegraphics[clip=true, trim=0mm 0mm 0mm 0mm, width=1.0\columnwidth,height=0.85\columnwidth]{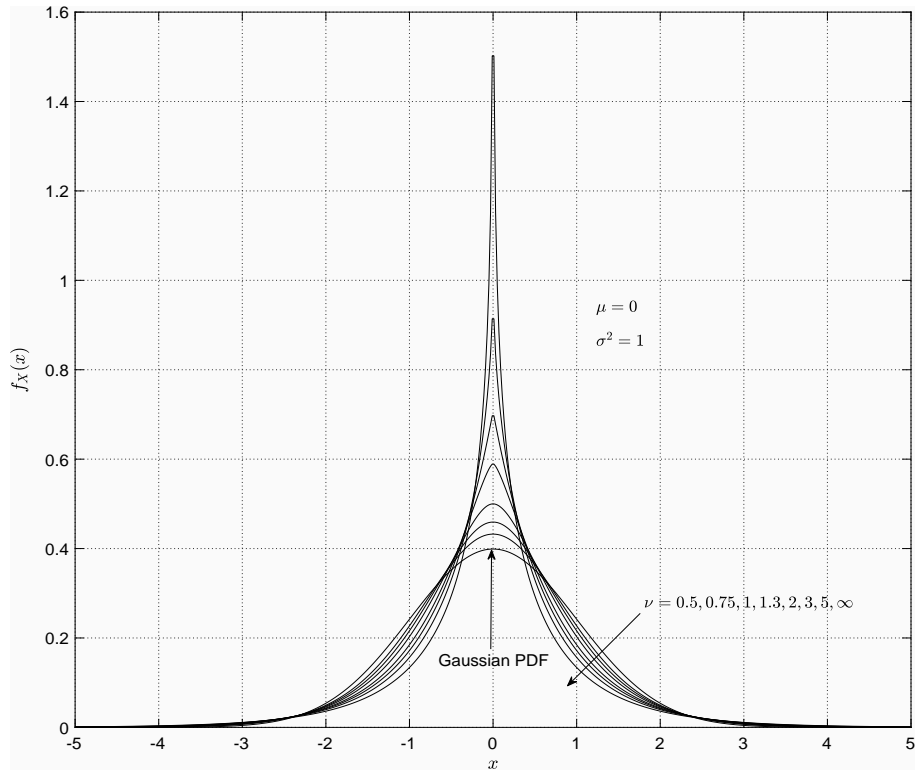}
    \caption{With respect to normality.}
    \label{Figure:McLeishPDFA}
\end{subfigure}
~~~
\begin{subfigure}{0.7\columnwidth}
    \centering
    \includegraphics[clip=true, trim=0mm 0mm 0mm 0mm, width=1.0\columnwidth,height=0.85\columnwidth]{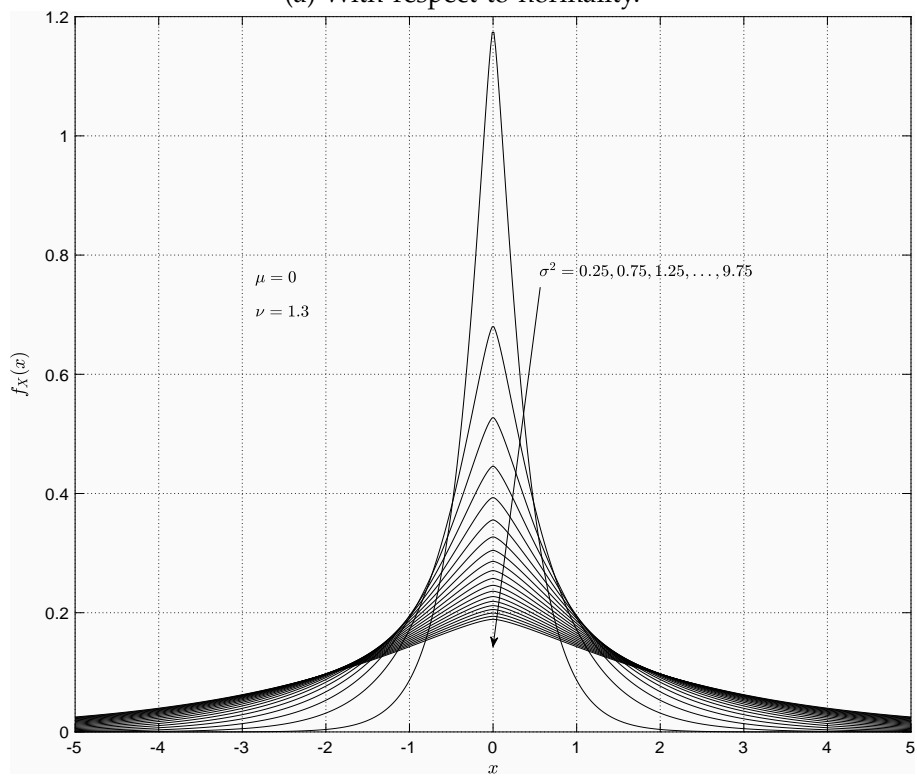}
    \caption{With respect to variance.}
    \label{Figure:McLeishPDFB}
\end{subfigure}
\caption{The \ac{PDF} of $\mathcal{M}_{\nu}(0,\sigma^2)$ with zero mean (i.e., the illustration of \eqref{Eq:McLeishPDF} for $\mu\!=\!0$).}
\label{Figure:McLeishPDF}
\vspace{-2mm} 
\end{figure*} 

In accordance with previously described notation of random matrices, the joint \ac{PDF} and \ac{CDF} of the real random vector $\defrmat{X}\!\in\!\mathbb{R}^{m}$ are respectively expressed by $f_{\defrmat{X}}\colon\mathbb{R}^{m}\rightarrow\mathbb{R}_{+}$ and  $F_{\defrmat{X}}\colon\mathbb{R}^{m}\rightarrow[0,1]$, and are respectively defined by 
\setlength\arraycolsep{1.4pt}   
\begin{eqnarray}
    f_{\defrmat{X}}(\defvec{x})&=&\mathbb{E}\bigl[\DiracDelta{\defvec{x}-\defrmat{X}}\bigr],\\
    F_{\defrmat{X}}(\defvec{x})&=&\mathbb{E}\bigl[\HeavisideTheta{\defvec{x}-\defrmat{X}}\bigr],
\end{eqnarray}
for $\defvec{x}\!\in\!\mathbb{R}^{m}$, where $\forall\defvec{y}\!\in\!\mathbb{R}^{m}$, we have $\DiracDelta{\defvec{y}}\!=\!\prod_{k=1}^{m}\DiracDelta{{y}_{k}}$ and $\HeavisideTheta{\defvec{y}}\!=\!\prod_{k=1}^{m}\HeavisideTheta{{y}_{k}}$. Moreover, the \ac{MGF} of $\defrmat{X}$ is expressed as $M_{\defrmat{X}}\colon\mathbb{R}^{m}\rightarrow[0,1]$ and defined by 
\begin{equation}
    M_{\defrmat{X}}(\defvec{s})=\mathbb{E}\bigl[\exp(-\langle\defvec{s},\defrmat{X}\rangle)\bigr]=\mathbb{E}\bigl[\exp(-\defvec{s}^{T}\defrmat{X}\bigr],
\end{equation}
where $\defvec{s}\!\in\!\mathbb{R}^{m}$. For simplicity, the mean vector of $\defrmat{X}\!\in\!\mathbb{R}^{m}$ is defined by
\begin{equation}
    \defvec{\mu}=\mathbb{E}[\defrmat{X}]=[\mu_1,\mu_2,\ldots,\mu_m]^T,
\end{equation}
where $\mu_{i}\!=\!\mathbb{E}[X_{i}]$, $1\!\leq\!i\!\leq\!m$. In multi-dimensional real space, the covariance matrix of
$\defrmat{X}$ is defined by $\defmat{\Sigma}\!\in\!\mathbb{R}^{{m}\times{m}}$,~that can be rewritten as 
\begin{subequations}\label{Eq:CovarianceMatrix}
\setlength\arraycolsep{1.4pt}   
\begin{eqnarray}
    \label{Eq:CovarianceMatrixA}
	\defmat{\Sigma}&=&
	    \mathbb{E}[(\defrmat{X}-\defvec{\mu})(\defrmat{X}-\defvec{\mu})^T],\\
    \label{Eq:CovarianceMatrixB}	    
	&=&\mathbb{E}[\defrmat{X}\defrmat{X}^T]-\defvec{\mu}\defvec{\mu}^{T},\\
	\label{Eq:CovarianceMatrixC}
	&=&\begin{bmatrix}
    	\sigma_{ij},
	\end{bmatrix}_{{1}\leq{i,j}\leq{m}},
\end{eqnarray}
\end{subequations}
where $\sigma_{ij}\!=\!\Covariance{X_{i}}{X_{j}}$, $1\!\leq\!i,j\!\leq\!m$. There is obviously no 
restriction on $\defvec{\mu}$, but $\defmat{\Sigma}$ must be real, symmetric, full rank,
invertible, and hence positive definite (i.e., $\defvec{x}^T\defmat{\Sigma}\defvec{x}\!>\!{0}$ 
for all $\defvec{x}\!\in\!\mathbb{R}^{m}$). For the shorthand notations of random vectors, let $\mathcal{N}^{m}(\defvec{\mu},\defmat{\Sigma})$, $\mathcal{L}^{m}(\defvec{\mu},\defmat{\Sigma})$, and $\mathcal{M}_{\nu}^{m}(\defvec{\mu},\defmat{\Sigma})$ denote an $m$-dimensional Gaussian random vector, an $m$-dimensional Laplacian random vector, and an $m$-dimensional McLeish random vector, respectively, with $\nu$ normality, $\defvec{\mu}$ mean vector and $\defmat{\Sigma}$ covariance matrix.

In multi-dimensional complex space, the joint \ac{PDF} and \ac{CDF} of the complex random vector $\defrmat{Z}\!\in\!\mathbb{C}^{m}$ are respectively expressed by $f_{\defrmat{Z}}\colon\mathbb{C}^{m}\rightarrow\mathbb{R}_{+}$ and  $F_{\defrmat{Z}}\colon\mathbb{C}^{m}\rightarrow[0,1]$, and are respectively defined by 
\setlength\arraycolsep{1.4pt}   
\begin{eqnarray}
    f_{\defrmat{Z}}(\defvec{z})&=&\mathbb{E}\bigl[
        \DiracDelta{\defvec{z}-\defrmat{Z}}
        \bigr],\\
    F_{\defrmat{Z}}(\defvec{z})&=&\mathbb{E}\bigl[
        \HeavisideTheta{\defvec{z}-\defrmat{Z}}
        \bigr],
\end{eqnarray}
for $\defvec{z}\!\in\!\mathbb{C}^{m}$ and $\defvec{s}\!\in\!\mathbb{C}^{m}$, where, for all $\defvec{z}\!=\!\defvec{x}+\imaginary\defvec{y}\!\in\!\mathbb{C}^{m}$ with $\defvec{x},\defvec{y}\!\in\!\mathbb{R}^{m}$, we have $\DiracDelta{\defvec{z}}\!=\!\DiracDelta{\defvec{x}}\DiracDelta{\defvec{y}}$ and $\HeavisideTheta{\defvec{z}}\!=\!\HeavisideTheta{\defvec{x}}\HeavisideTheta{\defvec{y}}$. Further, the \ac{MGF} of $\defrvec{Z}$ is expressed as $M_{\defrmat{X}}\colon\mathbb{C}^{m}\rightarrow[0,1]$ and defined by 
\begin{equation}
    M_{\defrmat{Z}}(\defvec{s})=\mathbb{E}\bigl[ 
        \exp(-\langle\defvec{s},\defrmat{Z}\rangle)\bigr]=\mathbb{E}\bigl[
        \exp(-\defvec{s}^{H}\defrmat{Z})\bigr],
\end{equation} 
where $\defvec{s}\!\in\!\mathbb{C}^{m}$. The mean vector of $\defrmat{Z}$ is given by $\defvec{\mu}\!=\!\mathbb{E}\bigl[\defrmat{Z}\bigr]$. In distinction from \eqref{Eq:CovarianceMatrix}, the covariance matrix of $\defrvec{Z}$ is defined in multi-dimensional complex space $\defmat{\Sigma}\!\in\!\mathbb{C}^{{m}\times{m}}$, that is
\begin{subequations}
\setlength\arraycolsep{1.4pt}   
\begin{eqnarray}
	\defmat{\Sigma}&=&
	    \mathbb{E}[(\defrmat{Z}-\defvec{\mu})(\defrmat{Z}-\defvec{\mu})^{H}],\\
	&=&\mathbb{E}[\defrmat{Z}\defrmat{Z}^{H}]-\defvec{\mu}\defvec{\mu}^{H},\\
	&=&\begin{bmatrix}
    	\sigma_{ij}
	\end{bmatrix},
\end{eqnarray}
\end{subequations}
where $\sigma_{ij}\!=\!\Covariance{Z_{i}}{Z_{j}}$, $1\!\leq\!i,j\!\leq\!m$. For the shorthand notations of random vectors, $\mathcal{CN}^{m}(\defvec{\mu},\defmat{\Sigma})$, $\mathcal{CL}^{m}(\defvec{\mu},\defmat{\Sigma})$, and $\mathcal{CM}_{\nu}^{m}(\defvec{\mu},\defmat{\Sigma})$ denote an $m$-dimensional \ac{CCS} Gaussian random vector, an $m$-dimensional \ac{CCS} Laplacian random vector, and an $m$-dimensional \ac{CCS} McLeish random vector, respectively, with $\nu$ normality, $\defvec{\mu}$ mean vector and $\defmat{\Sigma}$ covariance matrix. Further, $\mathcal{EN}^{m}(\defvec{\mu},\defmat{\Sigma})$, $\mathcal{EL}^{m}(\defvec{\mu},\defmat{\Sigma})$, and $\mathcal{EM}_{\nu}^{m}(\defvec{\mu},\defmat{\Sigma})$ denote an $m$-dimensional \ac{CES} Gaussian random vector, an $m$-dimensional \ac{CES} Laplacian random vector, and an $m$-dimensional \ac{CES} McLeish random vector. 

\section{Statistical Background}\label{Section:StatisticalBackground}
In this section, as an alternative to the well-known framework for the laws of Gaussian distribution, we develop and propose a conceptually novel framework for the laws of McLeish distribution, both in scalar and vector versions, contributing to the literature of probability and statistics necessary to all branches of science and engineering.

\subsection{McLeish Distribution}
\label{Section:StatisticalBackground:McLeishDistribution}
Let $X$ be $\mathcal{M}_{\nu}(\mu,\sigma^2)$ whose \ac{PDF} is given by \cite[Eq.\!~(3)]{BibMcLeishCJS1982}
\begin{equation}\label{Eq:McLeishPDF}
	f_{X}(x)=\frac{2}{\sqrt{\pi}}
		\frac{\abs{x-\mu}^{\nu-\frac{1}{2}}}{\Gamma(\nu)\,\lambda^{\nu+\frac{1}{2}}}
			\BesselK[\nu-\frac{1}{2}]{\frac{2\abs{x-\mu}}{\lambda}},
\end{equation}
defined over $x\!\in\!\mathbb{R}$, where $\nu\!\in\!\mathbb{R}_{+}$ and $\sigma^2\!\in\!\mathbb{R}_{+}$ denote the normality and variance, respectively, and $\lambda\!=\!\sigma\lambda_{0}\!=\!\sqrt{{2\sigma^2}/{\nu}}$ denotes the component deviation (power normalizing) factor. Further, $\Gamma(x)\!=\!\int_{0}^{\infty}\!u^{x-1}\allowbreak\exp(-u)\,du$ is the Gamma function\cite[Eq.\!~(6.1.1)]{BibAbramowitzStegunBook}, and  $\BesselK[n]{x}\!=\!\int_{0}^{\infty}\!{e}^{-x\cosh(u)}\allowbreak\cosh(nu)\,du$ is the modified Bessel function of the second kind\cite[Eq.\!~(9.6.2)]{BibAbramowitzStegunBook}. In order to illustrate the versatility and heavy-tail behaviour of $\mathcal{M}_{\nu}(\mu,\sigma^2)$, the \ac{PDF}, given in \eqref{Eq:McLeishPDF}, is aptly illustrated with respect to $\nu\!\in\!\mathbb{R}_{+}$ and $\sigma^2\!\in\!\mathbb{R}_{+}$ for a certain $\mu\!\in\!\mathbb{R}$ in \figref{Figure:McLeishPDF} on the top of the this page. 

The special cases of $\mathcal{M}_{\nu}(\mu,\sigma^2)$
consist of Dirac, Laplacian and Gaussian distributions. In more detail, as $\nu\!\rightarrow\!0$,
\eqref{Eq:McLeishPDF} reduces to 
\begin{equation}\label{Eq:DiracPDF}
 	f_{X}(x)=\DiracDelta{x-\mu},
\end{equation}
which is the \ac{PDF} of Dirac's distribution, where $\DiracDelta{\cdot}$ denotes the Dirac's delta
function\cite[Eq.\!~(1.8.1)]{BibZwillingerBook}. Further, substituting $\nu\!=\!{1}$ into \eqref{Eq:McLeishPDF} and then
utilizing \cite[Eq.\!~(9.7.8)]{BibAbramowitzStegunBook} yields the \ac{PDF} of $\mathcal{L}(\mu, \sigma^2)$, that is
\begin{equation}\label{Eq:LaplacianPDF}
	f_{X}(x)=\frac{1}{\sqrt{2\sigma^2}}\exp\bigl(-\scalemath{0.9}{0.9}{\sqrt{{2}/{\sigma^2}}}\abs{x-\mu}\bigr),
\end{equation}
Besides, limiting $\nu\!\rightarrow\!\infty$ in \eqref{Eq:McLeishPDF} and 
using \cite[Eq.\!~(9.7.8)]{BibAbramowitzStegunBook} yields 
\begin{equation}\label{Eq:SumTwoLaplacianPDF}
	f_{X}(x)=\frac{1}{\sqrt{2\pi\sigma^2}}\exp\Bigl(-\frac{(x-\mu)^2}{2\sigma^2}\Bigr),
\end{equation}
which is the \ac{PDF} of $X\!\sim\!\mathcal{N}(\mu,\sigma^2)$. In addition,   
$\mathcal{M}_{\nu}(\mu,\sigma^2)$ demonstrates a superior fit to different impulsive noise characteristics
with respect to $\nu\!\in\!\mathbb{R}_{+}$, and therefore
it is reasonably fit to any noise distribution, especially by estimating $\nu$, $\mu$, and $\sigma^2$ with the aid of 
\ac{MOM} in which sample moments are 
equated with theoretical moments of $\mathcal{M}_{\nu}(\mu,\sigma^2)$, that is  
\begin{equation}\label{Eq:VarianceAndGaussianityMOMEstimation}
	\hat{\mu}=\Expected{X},{~}\hat{\sigma}^2=\Variance{X},{~}\text{and}{~}\hat\nu=\frac{3}{\Kurtosis{X}-3}. 
\end{equation}
For that purpose, the higher-order moments of $\mathcal{M}_{\nu}(\mu,\sigma^2)$ are given in the following theorem.

\begin{theorem}\label{Theorem:McLeishMoments}
The moments of $X\!\sim\!\mathcal{M}_{\nu}(\mu,\sigma^2)$ is given by
\begin{equation}\label{Eq:McLeishMoments}
\!\!\!\!\!\mathbb{E}\bigl[X^n\bigr]\eq
	\mu^{n}\!\sum_{k=0}^{n}\!
		\Binomial{n}{k}
		\frac{\Gamma(\nu+{k}/{2})\Gamma({1}/{2}+{k}/{2})}{\Gamma(\nu)\Gamma({1}/{2})}
				\Bigl(\!\frac{\lambda}{\mu}\!\Bigr)^{\!k}\!\iseven{k}\!\!\!\!
\end{equation}
defined for $n\!\in\!\mathbb{N}$, where $\iseven{k}$ returns $1$ if $k$ is an even number, otherwise returns $0$.
\end{theorem}
\begin{figure*}[tp] 
\centering
\begin{subfigure}{0.7\columnwidth}
    \centering
    \includegraphics[clip=true, trim=0mm 0mm 0mm 0mm,width=1.0\columnwidth,height=0.85\columnwidth]{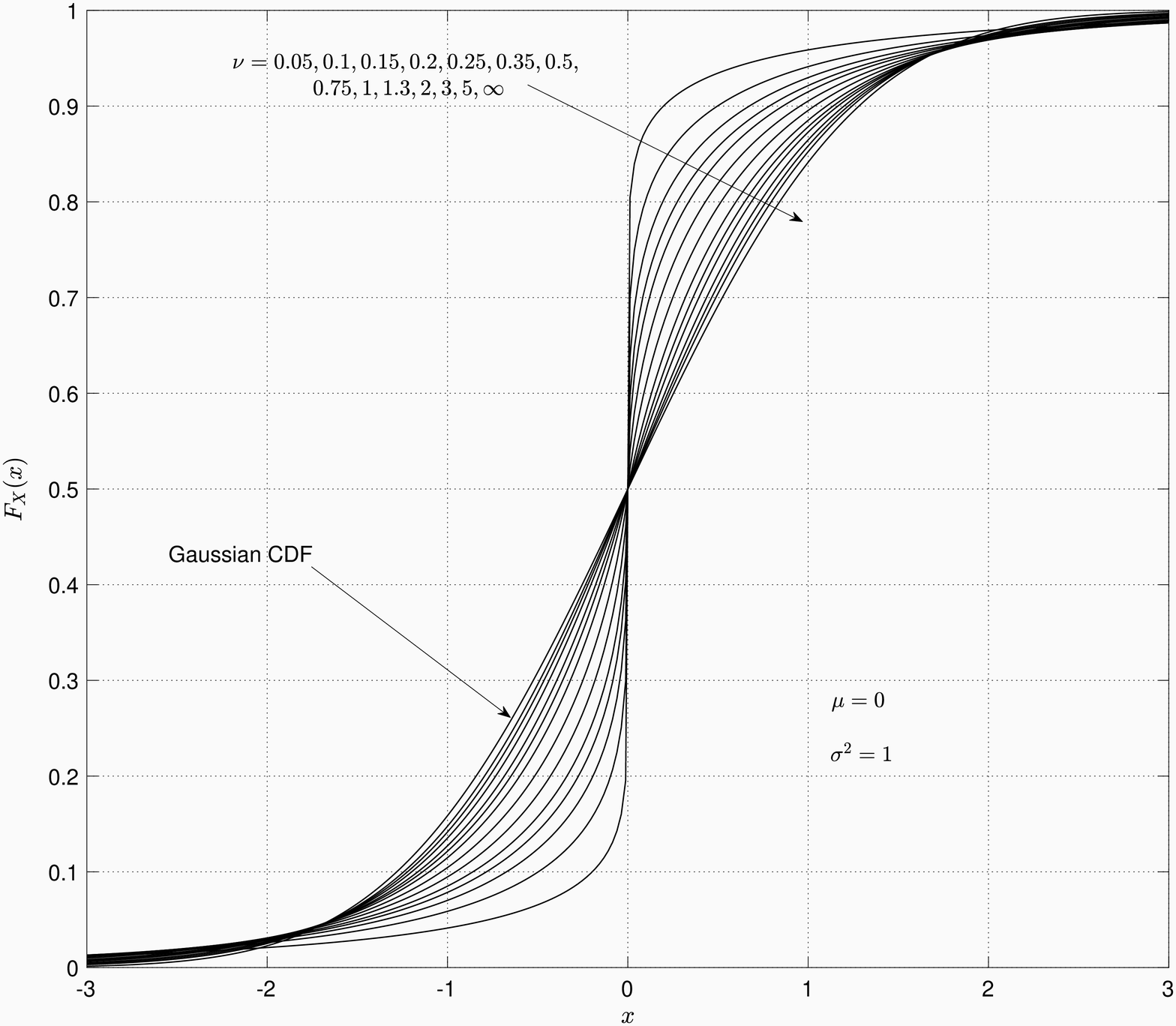}
    \caption{With respect to normality.}
    \label{Figure:McLeishCDFA}
\end{subfigure}
~~~
\begin{subfigure}{0.7\columnwidth}
    \centering
    \includegraphics[clip=true, trim=0mm 0mm 0mm 0mm,width=1.0\columnwidth,height=0.85\columnwidth]{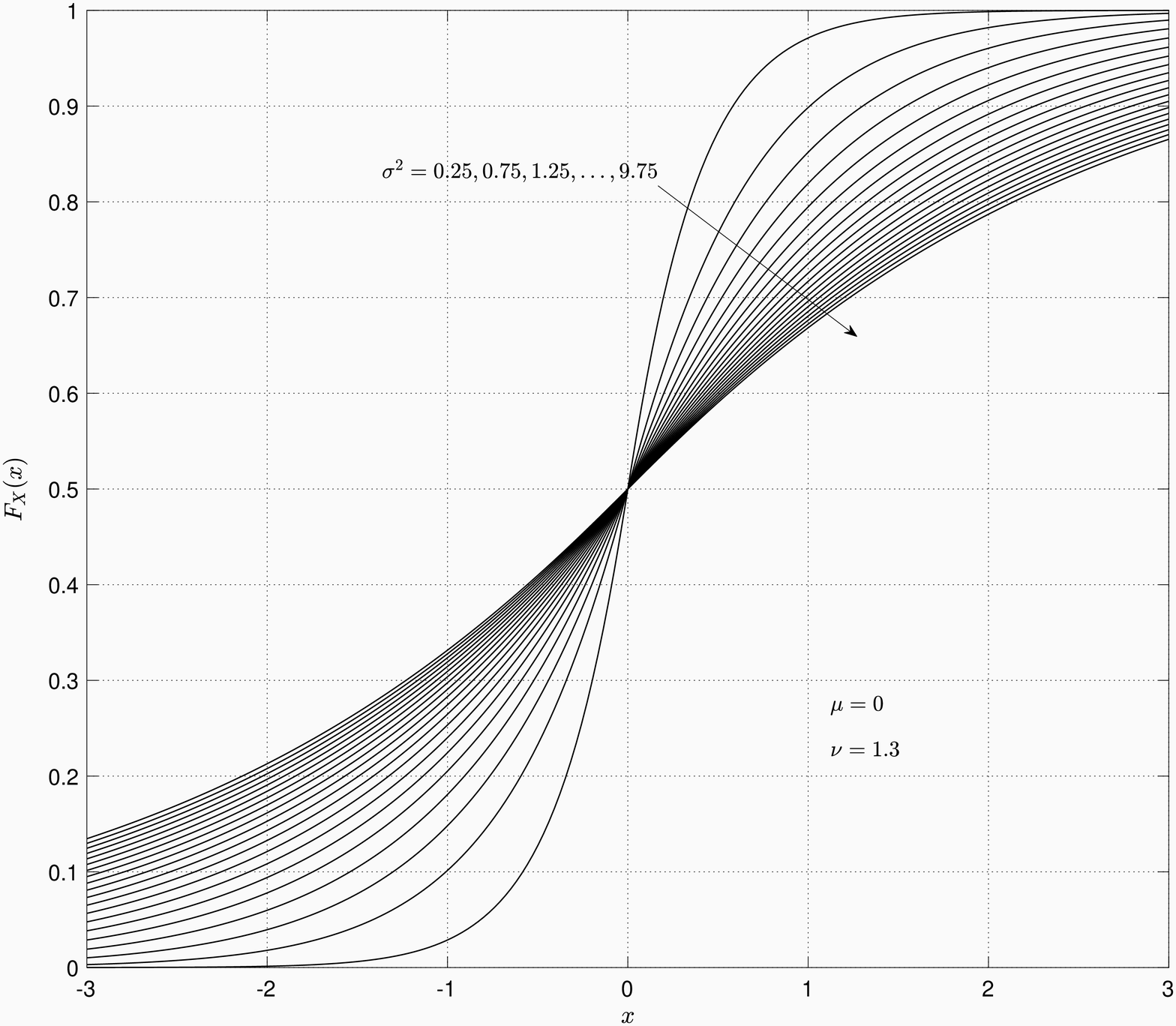}
    \caption{With respect to variance.}
    \label{Figure:McLeishCDFB}
\end{subfigure}
\caption{The \ac{CDF} of $\mathcal{M}_{\nu}(0,\sigma^2)$ with zero mean (i.e., the illustration of \eqref{Eq:McLeishCDF} for $\mu\!=\!0$).}
\label{Figure:McLeishCDF}
\vspace{-2mm} 
\end{figure*} 
\begin{proof}
Note that $X$ is readily expressed as $X\!=\!\mu+W$, where $W\!\sim\!\mathcal{M}_{\nu}(0,\sigma^2)$. Thus,
$\mathbb{E}[X^n]=\mathbb{E}[{(\mu+W)^{n}}]$ can be written using binomial expansion as follows
\begin{equation}
    \mathbb{E}[X^n]=\mu^{n}\sum_{k=0}^{n}\Binomial{n}{k}\frac{\mathbb{E}[{W}^k]}{\mu^k},
\end{equation}
where the binomial coefficient\cite[Eq.\!~(1.1.1)]{BibGradshteynRyzhikBook} is defined as
\begin{equation}
    \Binomial{n}{k}=\frac{n!}{(n-k)!\,k!}=
        \frac{(n+1)^{k}}{k!}\prod_{j=1}^{k}\Bigl(1-\frac{j}{n+1}\Bigr).
\end{equation}
With the aid of utilizing 
$\BesselK[n]{x}\!=\!\MeijerG[right]{2,0}{0,2}{{x^2}/{2}}{\emptycoefficientSHORT}{{n}/{2},-{n}/{2}}$
\cite[Eq.~(03.04.26.0008.01)]{BibWolfram2010Book}, where $\MeijerGDefinition{m,n}{p,q}{\cdot}$ denotes the Meijer's~G function\cite[Eq.\!~(8.2.1/1)]{BibPrudnikovBookVol3}, the \ac{PDF} of $W$ can be given in terms of the Meijer's G function. After endorsing $\mu\!=\!{0}$ and applying \cite[Eqs.\!~(2.9.1)\!~and\!~(2.9.19)]{BibKilbasSaigoBook} on \eqref{Eq:McLeishPDF},
$\mathbb{E}\bigl[W^n\bigr]$ is then expressed for $k\!\in\!\mathbb{N}$ as follows
\begin{equation}\label{Eq:McLeishMomentsUsingCentralMoments}
	\mathbb{E}[W^k]=\int_{-\infty}^{\infty}w^{k}
		\frac{1}{\sqrt{\pi}\lambda\Gamma(\nu)} 
			\scalemath{0.95}{0.95}{
				\MeijerG[right]{2,0}{0,2}{\frac{w^2}{\lambda}}{\emptycoefficient}{0,\nu-\frac{1}{2}}}dw,
\end{equation}
where $\emptycoefficient$ denotes the empty coefficient set. Immediately afterwards, in \eqref{Eq:McLeishMomentsUsingCentralMoments}, changing the variable $x^2\!\rightarrow\!{y}$ and employing \cite[Eqs.~(2.5.1)\!~and\!~(2.9.1)]{BibKilbasSaigoBook} results in   
\begin{equation}\label{Eq:McLeishCentralMoments}
	\mathbb{E}[W^k]=\frac{\Gamma(\nu+{k}/{2})}{\Gamma(\nu)} 
		\frac{\Gamma({1}/{2}+{k}/{2})}{\Gamma({1}/{2})}
			{\lambda}^{k}\iseven{k},
\end{equation}
where $\iseven{k}$ returns $1$ if $k$ is an even number, otherwise returns $0$. Finally, substituting \eqref{Eq:McLeishCentralMoments} into \eqref{Eq:McLeishMomentsUsingCentralMoments} readily results
in \eqref{Eq:McLeishMoments}, which completes the proof of \theoremref{Theorem:McLeishMoments}.
\end{proof}

\begin{definition}[McLeish's Quantile]\label{Definition:McLeishQFunction}
The~McLeish's~\ac{Q-function} is defined by 
\begin{equation}\label{Eq:McLeishQFunctionIntegral}
Q_{\nu}(x)=\int^{\infty}_{x}
	\!\!\frac{2}{\sqrt{\pi}}
	\frac{\abs{w}^{\nu-{1}/{2}}}
		{\Gamma(\nu)\lambda_{0}^{\nu+{1}/{2}}}
			 {K}_{\nu-{1}/{2}}
				 \Bigl(
					\frac{2\abs{w}}{\lambda_{0}}
				 \Bigr)\,{dw},
\end{equation}
for $x\in\mathbb{R}$. Alternatively, it is given for $x\geq{0}$ by 
\begin{subequations}\label{Eq:McLeishQFunction}
\begin{equation}\label{Eq:McLeishQFunctionForPositiveArguments}
\!\!\!Q_{\nu}\bigl(x\bigr)=\frac{2^{1-\nu}}{\pi\Gamma(\nu)}
  		\int_{0}^{\frac{\pi}{2}}\!{\Bigl(\frac{2x}{\lambda_{0}\sin(\theta)}\Bigr)}^{\nu}
  			\!\BesselK[\nu]{\frac{2x}{\lambda_{0}\sin(\theta)}}d\theta,\!
\end{equation}
and given for $x<{0}$ by
\begin{equation}\label{Eq:McLeishQFunctionFornegativeArguments}
	Q_{\nu}\bigl(x\bigr)={1}-Q_{\nu}(\abs{x}).
\end{equation}
\end{subequations}
\end{definition}

In wireless communications\cite[and~references~therein]{BibProakisBook,BibGoldsmithBook,BibAlouiniBook}, the \ac{CDF} of the additive noise is used as a quantile function to compare different systems in the context of~channel~reliability. In this connection, the \ac{CDF} of $X\!\sim\!\mathcal{M}_{\nu}(\mu,\sigma^2)$, i.e., $F_{X}(x)\!=\!\Pr\{X\leq{x}\}$
for $x\!\in\!\mathbb{R}$ is obtained in the following.

\begin{theorem}\label{Theorem:McLeishCDF}
The \ac{CDF} of $X\!\sim\!\mathcal{M}_{\nu}(\mu,\sigma^2)$, which is defined as $F_{X}(x)=\Pr\{X\leq{x}\}$, is given by
\begin{equation}\label{Eq:McLeishCDF}
	F_{X}(x)=1-Q_{\nu}\biggl({\frac{x-\mu}{\sigma}}\biggr),
\end{equation}
where $Q_{\nu}(\cdot)$ is the McLeish's \ac{Q-function} defined in \eqref{Eq:McLeishQFunction}. 
\end{theorem} 

\begin{proof}
Let us define a random variable, $W\!=\!{(X-\mu)}/{\sigma}$, where $W\!\sim\!\mathcal{M}_{\nu}(0,1)$, whose 
\ac{PDF} is given, using \eqref{Eq:McLeishPDF}, by
\begin{equation}\label{Eq:StandardMcLeishPDF}
	f_{W}(w)=\frac{2}{\sqrt{\pi}}
		\frac{\abs{w}^{\nu-{1}/{2}}}{\Gamma(\nu)\,\lambda_{0}^{\nu+{1}/{2}}}
			\BesselK[\nu-{1}/{2}]{\frac{2\abs{w}}{\lambda_{0}}}.
\end{equation}
whose distributional symmetry around $0$ consequences that 
the \ac{CDF} $F_{W}(w)\!=\!\Pr\{W\leq{w}\}\!=\!\int_{-\infty}^{w}f_{W}(w)dw$ can be rewritten 
as $F_{W}(w)=1-F_{W}(|w|)$ for $w\in\mathbb{R}^{-}$.
But for $w\in\mathbb{R}_{+}$, $F_{W}(w)$ is written as
$F_{W}(w)=1-\int_{w^2}^{\infty}\frac{1}{\sqrt{2w}}f_{W}(\sqrt{w})dw$. After some algebraic manipulations, it is rewritten as 
\begin{equation}\nonumber
	F_{W}(w)=1-\frac{2^{1-\nu}}{\pi\Gamma(\nu)}\int_{0}^{1}\!\frac{1}{\sqrt{1-w^2}}
		{\Bigl(\frac{2}{w\lambda_{0}}\Bigr)}^{\nu}
			\BesselK[\nu]{\frac{2}{w\lambda_{0}}}dw,
\end{equation}
where changing the variable as $w\!\rightarrow\!\sin(\theta)$ and utilizing \eqref{Eq:McLeishQFunction} results in $F_{W}(w)=1-Q_{\nu}(w)$. Accordingly, the \ac{CDF} of $X$ can be readily given as in
\eqref{Eq:McLeishCDF}, which proves \theoremref{Theorem:McLeishCDF}.
\end{proof}

The \ac{CDF} of $X\!\sim\!\mathcal{M}_{\nu}(\mu,\sigma^2)$ is described in \figref{Figure:McLeishCDF}~in~detail using \eqref{Eq:McLeishCDF}. It is therefore worth for the consistency and validity of the McLeish's \ac{Q-function} to mention that \eqref{Eq:McLeishQFunction} reduces for $\nu\!\rightarrow\!\infty$ to the well-known result, that is
\begin{equation}\label{Eq:McLeishQFunctionandGaussianQFunctionRelation}
	\lim_{\nu\rightarrow\infty}Q_{\nu}(x)=Q(x)
\end{equation}
where $Q(x)\!=\!\frac{1}{\sqrt{2\pi}}\int_{x}^{\infty}e^{-\frac{1}{2}u^2}du$ denotes the standard Gaussian \ac{Q-function}\cite[Eq.\!~(2.3-10)]{BibProakisBook}. Further, following are some of the fundamental properties
of McLeish's \ac{Q-function}:
\begin{subequations}\label{Eq:McLeishQFunctionProperties}
\setlength\arraycolsep{1.4pt}
\begin{eqnarray}
\label{Eq:McLeishQFunctionPropertiesA}
	Q_{\nu}(-x)=1-Q_{\nu}(x)
		&\text{ and }&
			Q_{\nu}(\pm\infty)=\scalemath{0.85}{0.85}{\frac{1}{2}}(1\mp{1}),~~~\\
\label{Eq:McLeishQFunctionPropertiesB}
	Q_{\nu}(0)=\scalemath{0.85}{0.85}{\frac{1}{2}}
		&\text{ and }&
			Q_{0}(x)\rightarrow{0}^{+},
\end{eqnarray}
\end{subequations}
In addition, It is worth examining not only the special cases of McLeish's \ac{Q-function} for the special non-extreme finite values of the normality $\nu$, but also for lower and upper bounds. Accordingly, setting $\nu\!=\!1$ reduces McLeish's \ac{Q-function} to the Laplacian \ac{Q-function}, that is
\begin{equation}\label{Eq:LaplacianQFunction}
\!\!\!LQ(x)=\begin{cases}
	\displaystyle\frac{1}{2}
 		\exp(-2\sqrt{2}x), & \!\!\text{if}\!~x\geq{0},\\[2mm]
 	\displaystyle{1}-LQ(\abs{x}), & \!\!\text{if}\!~x<{0}.
\end{cases}
\end{equation}
As seen in the following sections, the McLeish's \ac{Q-function} is often used in the \ac{BER}\,/\,\ac{SER} analysis of the signaling using modulation schemes over AWMN channels. The McLeish's \ac{Q-function} can be tabulated, or implemented as a built-in functions in mathematical software tools. However, in many cases it is useful to have closed-form bounds or approximations instead of the exact expression. In fact, these approximations are particularly useful in evaluating the \ac{BER}\,/\,\ac{SER} in many problems of the communication theory. For that purpose, the lower and upper bounds of the McLeish's \ac{Q-function} are found to be obtained for $x\!>\!0$ using Taylor series expansion under some simplification, that is
\begin{equation}
    {Q}^{\text{LB}}_{\nu}(x)\leq{Q}_{\nu}(x)\leq{Q}^{\text{UB}}_{\nu}(x),~\text{for}~x>0,
\end{equation}
where the lower bound approximation ${Q}^{\text{LB}}_{\nu}(x)$ is given by
\begin{equation}\label{Eq:McLeishQFunctionLowerBound}
    {Q}^{\text{LB}}_{\nu}(x)=\frac{1}{\sqrt{\pi}\Gamma(\nu)}
            \Bigl(\frac{x}{\lambda_{0}}\Bigr)^{\nu-\frac{1}{2}}
            \Bigl(
                K_{\nu+\frac{1}{2}}\Bigl(\frac{2x}{\lambda_{0}}\Bigl)
                -
                \frac{\lambda_{0}}{2x}
                    K_{\nu+\frac{3}{2}}\Bigl(\frac{2x}{\lambda_{0}}\Bigl)
            \Bigr),
\end{equation}
and the upper bound approximation ${Q}^{\text{UB}}_{\nu}(x)$ is given by 
\begin{equation}\label{Eq:McLeishQFunctionUpperBound}
    {Q}^{\text{UB}}_{\nu}(x)=\frac{1}{\sqrt{\pi}\Gamma(\nu)}
            \Bigl(\frac{x}{\lambda_{0}}\Bigr)^{\nu-\frac{1}{2}}
                K_{\nu+\frac{1}{2}}\Bigl(\frac{2x}{\lambda_{0}}\Bigl).
\end{equation}
Then, the gap between ${Q}^{\text{LB}}_{\nu}(x)$ and ${Q}^{\text{UB}}_{\nu}(x)$ is given by 
\begin{equation}\label{Eq:McLeishQFunctionGAPBetweenUpperBoundAndLowerBound}
{Q}^{\text{UB}}_{\nu}(x)-{Q}^{\text{LB}}_{\nu}(x)=
    \frac{1}{2\sqrt{\pi}\Gamma(\nu)}
            \Bigl(\frac{x}{\lambda_{0}}\Bigr)^{\nu-\frac{3}{2}}
                K_{\nu+\frac{3}{2}}\Bigl(\frac{2x}{\lambda_{0}}\Bigl).
\end{equation}

Note that H-transforms, also known as Mellin-Barnes integrals\footnote{For further details about both H-transforms and
Fox's H functions, readers are referred to \cite[and\!~references\!~therein]{BibKilbasSaigoBook}.} are the
integral kernels involving Meijer's G and Fox's H functions that have found many applications in such fields
as physics, statistics, and engineering \cite{BibKilbasSaigoBook}.
In the literature of wireless communications, H-transforms have been gained some attention to find closed-form expressions for averaged performance analysis, and also Fox's H function has recently started to be used 
as a possible fading distribution, commonly referred as the Fox's H distribution \cite{BibYilmazAlouiniTCOM2012}. It 
is thus useful to express McLeish's \ac{Q-function} in terms of Meijer's G and Fox's H functions. Such expressions allow the use of Mellin-Barnes integrals to obtain new closed-form expressions.

\begin{theorem}\label{Theorem:McLeishQFunctionUsingFoxHAndMeijerG}
McLeish's \ac{Q-function} can be alternatively expressed in terms of Fox's H function as follows
\begin{equation}\label{Eq:McLeishQFunctionUsingFoxH}
\!\!\!Q_{\nu}(x)=\begin{cases}
	\displaystyle\frac{1}{\Gamma(\nu)}
 		\FoxH[right]{2,0}{1,2}{2\nu{x}^2}{(1,1)}{(0,2),(\nu,1)}, & x\geq{0},\\[2mm]
 	\displaystyle{1}-Q_{\nu}(\abs{x}), & x<{0},
\end{cases}
\end{equation}
where $\FoxHDefinition{m,n}{p,q}{\cdot}$ is the Fox's H function\emph{\cite[Eq.\!~(8.3.1/1)]{BibPrudnikovBookVol3},\cite[Eq.\!~(1.1.1)]{BibKilbasSaigoBook}}; or~in~terms~of~Meijer's~G~function as follows 
\begin{equation}\label{Eq:McLeishQFunctionUsingMeijerG}
\!\!Q_{\nu}(x)=\begin{cases}
	\displaystyle\frac{1}{2\sqrt{\pi}\Gamma(\nu)}
 		\MeijerG[right]{3,0}{1,3}{2\nu{x}^2}{1}{0,\frac{1}{2},\nu}, & x\geq{0},\\[2mm]
 	\displaystyle{1}-Q_{\nu}(\abs{x}), & x<{0}.
\end{cases}
\end{equation}
\end{theorem}

\begin{proof}
Note that in \eqref{Eq:McLeishQFunctionForPositiveArguments}, taking place of the modified Bessel function of the second by
\cite[Eq.\!~(2.9.19)]{BibKilbasSaigoBook} and then performing some algebraic manipulations yields 
\begin{equation}\nonumber
\!\!Q_{\nu}(x)=\frac{1}{\pi\Gamma(\nu)}\int_{0}^{\frac{\pi}{2}}
	\scalemath{0.95}{0.95}{
		\FoxH[right]{2,0}{0,2}{\frac{x^2}{\lambda_{0}^2\sin^{2}(\theta)}}{\emptycoefficient}{(0,1),(\nu,1)}}d\theta,\!
\end{equation}
where employing \cite[Eq.\!~(1.1.1)]{BibKilbasSaigoBook} results in Mellin-Barnes contour integral in which changing the
order of integrals and using \cite[Eq. (3.621/1)]{BibGradshteynRyzhikBook}
\begin{equation}
    \int_{0}^{{\pi}/{2}}\!\sin^{2s}(\theta)d\theta=
        \frac{\sqrt{\pi}\Gamma(\frac{1}{2}+s)}{2\Gamma(1+s)}    
\end{equation}
for $\RealPart{s}\!>\!-\frac{1}{2}$ yields \eqref{Eq:McLeishQFunctionUsingFoxH}, which readily proves the first step
of \theoremref{Eq:McLeishQFunctionUsingFoxH}. In the second step, after using \cite[Eq.~(8.3.2/22)]{BibPrudnikovBookVol3},
\eqref{Eq:McLeishQFunctionUsingFoxH} reduces to \eqref{Eq:McLeishQFunctionUsingMeijerG}, which completes the
proof of \theoremref{Theorem:McLeishQFunctionUsingFoxHAndMeijerG}.
\end{proof}

Immediately after we examine the results provided in  \cite{BibProakisBook,BibAlouiniBook,BibGoldsmithBook,BibCraigMILCOM1991}, we readily recognize that Craig's partial \ac{Q-function}, defined as $Q\bigl(x,\phi\bigr)\!=\!\frac{1}{2\pi}\int_{0}^{\phi}\exp\bigl(-{x^2}/{\sin^2(\theta)}\bigr)d\theta$, is widely exploited in the \ac{SER} analysis of M-ary modulation and 2-dimensional modulation schemes, for example in  \cite{BibCraigMILCOM1991}, and \cite[Eqs. (4.9), (4.16), (4.17), (4.18), (4.19) and (5.77)]{BibAlouiniBook}. Analogously, we can define the McLeish's partial \ac{Q-function} as it is shown in the following. 

\begin{definition}[McLeish's Partial Quantile]\label{Definition:McLeishPartialQFunction}
For a certain $\phi\in[0,{\pi}/{2}]$, McLeish's partial \ac{Q-function} is defined as 
\begin{subequations}\label{Eq:McLeishPartialQFunction}
\begin{equation}\label{Eq:McLeishPartialQFunctionA}
\!\!\!\!\!Q_{\nu}\bigl(x,\phi\bigr)=\frac{2^{1-\nu}}{\pi\Gamma(\nu)}
  		\int_{0}^{\phi}\!\!{\Bigl(\frac{2x}{\lambda_{0}\sin(\theta)}\Bigr)}^{\nu}
  			\!\BesselK[\nu]{\frac{2x}{\lambda_{0}\sin(\theta)}}\!d\theta\!\!\!
\end{equation}
for $x\geq{0}$; 
\begin{equation}\label{Eq:McLeishPartialQFunctionB}
	Q_{\nu}\bigl(x,\phi\bigr)={1}-Q_{\nu}(\abs{x},\phi),
\end{equation}
for $x<{0}$; such that $Q_{\nu}\bigl(x\bigr)=Q_{\nu}\bigl(x,{\pi}/{2}\bigr)$.  
\end{subequations}
\end{definition}

In wireless communications\cite[and\!~references\!~therein]{BibProakisBook,BibGoldsmithBook,BibAlouiniBook}, the \ac{CCDF} of the additive noise is used as a quantile function to compare different systems
in the context of \ac{BER} or \ac{SER}. In this connection, the \ac{CCDF} of $X\!\sim\!\mathcal{M}_{\nu}(\mu,\sigma^2)$ is 
obtained in the following.

\begin{theorem}\label{Theorem:McLeishCCDF}
The \ac{CCDF} of $X\!\sim\!\mathcal{M}_{\nu}(\mu,\sigma^2)$, which is defined as $\widehat{F}_{X}(x)=\Pr\{X>{x}\}$, is given by
\begin{equation}\label{Eq:McLeishCCDF}
	\widehat{F}_{X}(x)=Q_{\nu}\biggl({\frac{x-\mu}{\sigma}}\biggr).
\end{equation}
\end{theorem} 

\begin{proof}
Note that $\widehat{F}_{X}(x)\!=\!{1}-{F}_{X}(x)$ since $\Pr\{X>{x}\}=1-\Pr\{X\leq{x}\}$. 
The proof is thus obvious using \theoremref{Theorem:McLeishCDF}. 
\end{proof}

As mentioned in \cite{BibSimonAlouiniProcIEEE1998,BibAlouiniSimonEL1998,BibAlouiniGoldsmithTCOM1999,BibAlouiniBook}, 
the \ac{MGF} is an efficient mathematical instrument not only to derive inequalities on tail probabilities of distributions but to achieve their statistical characterisations, and therefore is extremely common in performance results for communication problems related to partially coherent, differentially coherent, and non-coherent communications and is very useful in statistics. We derive the \ac{MGF} of McLeish distribution as it is given in the following.

\begin{theorem}\label{Theorem:McLeishMGF}
The \ac{MGF} of $X\!\sim\!\mathcal{M}_{\nu}(\mu,\sigma^2)$ is given by
\begin{equation}\label{Eq:McLeishMGF}
	M_{X}(s)={e}^{-s\mu}\Bigl(1-\scalemath{0.9}{0.9}{\frac{\lambda^2}{4}}s^2\Bigr)^{-\nu}
\end{equation}
with the existence region ${-S_0}\!<\!{\RealPart{s}}\!<\!{S_0}$, where ${S_0}\!\in\!\mathbb{R}_{+}$ is given by  
${S_0}=2/{\lambda}$. 
\end{theorem}

\begin{proof}
Note that $M_{X}(s)\!=\!\mathbb{E}[\exp(-sX)]$ can be expressed as 
$M_{X}(s)\!=\!{s}\int_{-\infty}^{\infty}\exp(-sx)F_{X}(x)dx$, where susing \eqref{Eq:McLeishCDF} yields  
\begin{equation} 
	M_{X}(s)=s\int_{-\infty}^{\infty}\exp(-sx)\,Q_{\nu}\bigl(\frac{x-\mu}{\sigma}\bigr)dx.
\end{equation}
which can be divided two integration, i.e., $M_{X}(s)=sI_{+}(s)+sI_{-}(s)$, where $I_{\pm}(s)$ is written as  
\begin{equation}\label{Eq:McLeishMGFSemiIntegration}
	I_{\pm}(s)=\pm\int_{0}^{\infty}\exp(\mp{s}x)\,Q_{\nu}\bigl(\frac{\pm{x}-\mu}{\sigma}\bigr)dx,
\end{equation} 
Subsequently, substituting \eqref{Eq:McLeishQFunctionUsingFoxH} in \eqref{Eq:McLeishMGFSemiIntegration} and then using both
$\exp(-x)\!=\!\MeijerG[right]{1,0}{0,1}{x}{\emptycoefficientSHORT}{0}$\cite[Eq.\!~(8.4.3/1)]{BibPrudnikovBookVol3}, and
$\exp(x)\!=\!\frac{\pi}{\sin(\pi{c})}\MeijerG[right]{1,0}{1,2}{x}{1-c}{0,1-c}$\cite[Eq.\!~(8.4.3/5)]{BibPrudnikovBookVol3} 
results in a Mellin-Barnes integration\cite[Theorem\!~2.9]{BibKilbasSaigoBook} that readily reduces to
\begin{equation}
	I_{\pm}(s)={e}^{-s\mu}\Bigl(1-\scalemath{0.9}{0.9}{\frac{\lambda^2}{4}}s^2\Bigr)^{-\nu}
	\biggl(\frac{1}{2s}
	\pm
	\frac{\lambda}{4\pi}\sin(\pi\nu)\Pochhammer{\scalemath{0.7}{0.7}{\frac{1}{2}}}{\!\nu}
		\MeijerG[right]{1,2}{2,2}{-\scalemath{0.8}{0.8}{\frac{\lambda^2}{4}}s^2}{{1}/{2},\nu}{0,-{1}/{2}}
		\biggr)
\end{equation} 
within the convergence region $-{2}/{\lambda}\!\leq\!\RealPart{s}\!\leq\!{2}/{\lambda}$, where
$\Pochhammer{a}{n}\!=\!{\Gamma(a+n)}/{\Gamma(a)}$ denotes Pochhammer's
symbol\cite[Eq.\!~(1.2.6)]{BibWolfram2010Book}. Consequently, $M_{X}(s)=sI_{+}(s)+sI_{-}(s)$ simplifies to
\eqref{Eq:McLeishMGF}, which completes the proof of \theoremref{Theorem:McLeishMGF}.
\end{proof}

For consistency, letting $\nu\!\rightarrow\!{0}$ in \eqref{Eq:McLeishMGF} results in $\exp(-s\mu)$, which
is the \ac{MGF} of the Dirac's distribution with mean $\mu$. For $\nu\!=\!{1}$, \eqref{Eq:McLeishMGF} simplifies to the \ac{MGF} of $\mathcal{L}(\mu,\sigma^2)$, that is
$M_{X}(s)\!=\!{e}^{-s\mu}(1-\sigma^2{s}^2/2)^{-1}$ \cite{BibForbesEvansHastingsPeacockBook,BibZwillingerKokoskaBook,BibKrishnamoorthyBook,BibKotzKozubowskiPodgorskiBook2012}.
In addition, when letting $\nu\!\rightarrow\!\infty$ and then using
$\lim_{n\rightarrow\infty}(1+\frac{x}{n})^{n}\!=\!\exp(x)$ \cite[Eq.\!~(01.03.09.0001.01)]{BibWolfram2010Book},
\eqref{Eq:McLeishMGF} simplifies to
$M_{X}(s)\!=\!\exp\bigl(-s\mu+\sigma^2{s}^2/2\bigr)$ \cite{BibProakisBook,BibKrishnamoorthyBook,BibForbesEvansHastingsPeacockBook,BibZwillingerKokoskaBook} which is the well-known \ac{MGF} of $\mathcal{N}(\mu,\sigma^2)$. Notice that the MGF is also used to derive the moments \cite{BibPapoulisBook}. Hence, the analytical correctness of \eqref{Eq:McLeishMGF} can also be checked using \eqref{Eq:McLeishMoments}. Using \cite[Eq. (8.4.2/5)]{BibPrudnikovBookVol3}, we can express \eqref{Eq:McLeishMGF} in terms of Meijer's G function as
\begin{equation}\label{Eq:McLeishMGFUsingMeijerG}
    M_{X}(s)=\frac{{e}^{-s\mu}}{\Gamma(\nu)}
        \MeijerG[right]{1,1}{1,1}{\frac{\lambda^2}{4}s^2}{1-\nu}{0},
\end{equation}
whose $n$th derivation with respect to $s$, i.e. $\bigl({\partial}/{\partial{s}}\bigr)^{n}M_{X}(s)$ can be attained using Leibniz's rule\cite[Eq. (0.42)]{BibGradshteynRyzhikBook} and \cite[Eqs. (8.3.2/21) and (8.3.2/21)]{BibPrudnikovBookVol3}, and therein setting $s\!\rightarrow\!{0}$ yields \eqref{Eq:McLeishMoments} as expected. It is also worth mentioning that the \acp{MGF} are very useful for the analysis of sums of the McLeish distributions as exemplified in the following.

\subsection{Sum of McLeish Distributions}
\label{Section:StatisticalBackground:McLeishSumDistribution}
Let $X_{\ell}\!\sim\!\mathcal{M}_{\nu_\ell}(\mu_{\ell},\sigma^2_{\ell})$, $\ell\!=\!{1,2,\ldots,L}$ be $L$ \ac{i.n.i.d.} distributions. Then, their sum is written as
\begin{equation}\label{Eq:McLeishSumDistribution}
	X_{\Sigma}=\textstyle\sum_{\ell=1}^{L}X_{\ell},
\end{equation} 
whose statistically characterization is given in the following. 

\begin{theorem}\label{Theorem:McLeishSumMGF}
The \ac{MGF} of \eqref{Eq:McLeishSumDistribution} is given by
\begin{equation}\label{Eq:McLeishSumMGF}
	M_{X_{\Sigma}}(s)=
		{e}^{-s{\sum_{\ell=1}^{L}\mu_{\ell}}}
			\prod_{\ell=1}^{L}
				\Bigl(1-{\frac{\lambda_{\ell}^2}{4}}s^2\Bigr)^{-\nu_{\ell}}
\end{equation}
with the existence region ${-S_0}\!<\!{\RealPart{s}}\!<\!{S_0}$, where ${S_0}\!\in\!\mathbb{R}_{+}$ is
given by ${S_0}=2/\max_{\ell\in\{1..L\}}{\lambda_{\ell}}$.
\end{theorem}

\begin{proof}
Since $\{X_{\ell}\}_{\ell=1}^{L}$ are mutually independent, the \ac{MGF} of $X_{\Sigma}$ is defined as the
product of their \acp{MGF}, that is 
$M_{X_{\Sigma}}(s)\!=\!\mathbb{E}[\exp(-s\sum_{\ell=1}^{L}X_{\ell})]\!=\!\prod_{\ell=1}^{L}M_{X_{\ell}}(s)$,
where using \eqref{Eq:McLeishMGF} yields \eqref{Eq:McLeishSumMGF}, which proves \theoremref{Theorem:McLeishSumMGF}. 
\end{proof}

Let us now consider some special cases of \eqref{Eq:McLeishSumMGF}. In case of $\nu_{\ell}\in\mathbb{Z}^{+}$ and
$\lambda_{\ell}\!\neq\!\lambda_{m}$ for all $\ell\!\neq\!m$, $X_{\Sigma}$ follows a hyper McLeish distribution, which is also called a mixture McLeish distribution. Simplifying \eqref{Eq:McLeishSumMGF} using pole factorization (partial fraction decomposition) of rational polynomials\cite[Sec. 2.2.4]{BibZwillingerBook}, we obtain the \ac{MGF} as
\begin{equation}\label{Eq:McLeishSumMGFSpecialCaseI}
	\!\!M_{X_{\Sigma}}(s)=
		{e}^{-s{\sum_{\ell=1}^{L}\mu_{\ell}}}
		\sum_{\ell=1}^{L}
		\sum_{m=0}^{\nu_{\ell}-1}
			w_{\ell{m}}
				\Bigl(1-{\frac{\lambda_{\ell}^2}{4}}s^2\Bigr)^{m-\nu_{\ell}}\!\!\!\!,
\end{equation}
where the weight coefficients $\{w_{\ell{m}}\}$, which certainly support that
$\sum_{\ell=1}^{L}\sum_{m=0}^{\nu_{\ell}-1}w_{\ell{m}}\!=\!{1}$, are defined as
\begin{equation}\label{Eq:McLeishSumMGFCoefficient}
	w_{\ell{m}}=
		\frac{{4}^{\!m}}{\lambda^{2m}_{\ell}m!}
			\biggl.
			\left(\frac{\partial}{\partial{s}}\right)^{\!m}\!\!\!\!
				\prod_{{j=1,j\neq\ell}}^{L}
					\!\!\scalemath{0.9}{0.9}{
					\biggl(1-\frac{\lambda_{j}^2}{\lambda_{\ell}^2}+{\frac{\lambda_{j}^2}{4}}s\biggr)^{-\nu_{j}}}
					\biggr|_{s\rightarrow{0}},
\end{equation}
where the $m$th order derivative can be mathematically defined in several ways
\cite[and\!~references\!~therein]{BibCafagnaMIE2007,BibSabatierAgrawalMachadoBook2007,BibMeerschaertBook2012}. We find the Gr\"{u}nwald-Letnikov derivative to be convenient for its numerical computation. In addition, the other special case of \eqref{Eq:McLeishSumMGF} is obtained when $\lambda_{\ell}\!=\!\lambda_{\Sigma}$
with distinct $\sigma^2_{\ell}$ for $\ell\!=\!{1,2,\ldots,n}$; $X_{\Sigma}$ follows a McLeish
distribution, i.e., $X_{\Sigma}\!\sim\!\mathcal{M}_{\nu_\Sigma}(\mu_{\Sigma},\sigma^2_{\Sigma})$, whose \ac{MGF} is readily
deduced similar to \eqref{Eq:McLeishMGF}, that is 
\begin{equation}\label{Eq:McLeishSumMGFSpecialCaseII}
	M_{X_{\Sigma}}(s)=
		{e}^{-s\mu_{\Sigma}}
			\Bigl(1-{\frac{\lambda^2_{\Sigma}}{4}}s^2\Bigr)^{-\nu_{\Sigma}},
\end{equation}
where the normality $\nu_{\Sigma}\!=\!\sum_{\ell=1}^{n}\nu_{\ell}$, the mean
$\mu_{\Sigma}\!=\!\sum_{\ell=1}^{n}\mu_{\ell}$ and the variance
$\sigma^2_\Sigma\!=\!{\nu_{\Sigma}\lambda^2_{\Sigma}}/{2}$. In addition, the other special cases can be deduced
for certain normalities $\nu_{\ell}\!\rightarrow\!{0}$, $\nu_{\ell}\!\rightarrow\!{1}$ and
$\nu_{\ell}\!\rightarrow\!\infty$ in \eqref{Eq:McLeishSumMGF}. Specifically, when $\forall\nu_{\ell}\!\rightarrow\!{0}$, \eqref{Eq:McLeishSumMGF} and \eqref{Eq:McLeishSumMGFSpecialCaseI} reduces to $M_{X_{\Sigma}}(s)={e}^{-s{\mu_{\Sigma}}}$, which is the \ac{MGF} of the Dirac's distribution. Further, when $\forall\nu_{\ell}\!\rightarrow\!{1}$, \eqref{Eq:McLeishSumMGF} turns to the \ac{MGF} of sum of independent and not identically distributed Laplace distributions, that is \cite[Sec.~10.4]{BibWalckReport1996}
\begin{equation}
    M_{X_{\Sigma}}(s)=
        {e}^{-s{\sum_{\ell=1}^{L}\mu_{\ell}}}
            \prod_{\ell=1}^{L}
                \Bigl(1-{\frac{\sigma_{\ell}^2}{2}}s^2\Bigr)^{-1}.
\end{equation} 
In addition, when $\forall\nu_{\ell}\!\rightarrow\!{\infty}$, \eqref{Eq:McLeishSumMGF} turns to the \ac{MGF} of sum of \emph{i.n.i.d} Gaussian distributions, that is \cite[Sec.~34.5]{BibWalckReport1996}
\begin{equation}
    M_{X_{\Sigma}}(s)=
        \exp\Bigl(-s{\mu_{\Sigma}}+\frac{s^2}{2}{\sigma^2_{\Sigma}}\Bigr).
\end{equation} 
Speaking of statistically characterization, we efficiently exploit the \ac{MGF} to find the \ac{PDF} of the sums of independent random distributions \cite{BibPapoulisBook}. Accordingly, the \ac{PDF} of $X_{\Sigma}$ is obtained in the following.

\begin{theorem}\label{Theorem:McLeishSumPDF}
The \ac{PDF} of \eqref{Eq:McLeishSumDistribution} is given by 
\begin{equation}\label{Eq:McLeishSumPDF}
	f_{X_{\Sigma}}(x)=\FoxI[right]{L,L}{2L,2L}
		{\frac{\exp(-x)}{\exp(-\mu_{\Sigma})}}
			{\boldsymbol{\Xi}_{L}^{(1)},\boldsymbol{\Xi}_{L}^{(3)}}
				{\boldsymbol{\Xi}_{L}^{(2)},\boldsymbol{\Xi}_{L}^{(0)}}
\end{equation}
with mean $\mu_\Sigma\!=\!\mu_{1}+\mu_{2}+\ldots+\mu_{L}$, where the coefficient set
$\boldsymbol{\Xi}_{n}^{(\alpha)}$, consisting of 3-tuples of size $n$, is defined as
\begin{equation}
	\boldsymbol{\Xi}_{n}^{(\alpha)}=\textstyle
		{\bigl({\alpha-1,{\frac{\lambda_{1}}{2}},\nu_{1}}\bigr)},
		\cdots,
		{\bigl({\alpha-1,{\frac{\lambda_{n}}{2}},\nu_{n}}\bigr)},
\end{equation}
for $n\!\in\!\mathbb{N}$ and $\alpha\!\in\!\mathbb{R}$. Moreover in \eqref{Eq:McLeishSumPDF},
$\FoxIDefinition{m,n}{p,q}{\cdot}$ denotes Fox's I function\emph{\cite[Eq.\!~(3.1)]{BibRrathieLeMath1997}}.
\end{theorem}

\begin{proof}
For $\ell\!\in\!\{1,2,\ldots,n\}$, the \ac{MGF} of $X_{\ell}$, i.e.,
$M_{X_{\ell}}(s)\!=\!\mathbb{E}[\exp(-sX_{\ell})]$ can be rewritten as
\begin{equation}\nonumber
	M_{X_{\ell}}(s)=
		{e}^{-s\mu_{\ell}}
			\Bigl(1-{\frac{\lambda_{\ell}}{2}}s\Bigr)^{-\nu_{\ell}}
				\Bigl(1+{\frac{\lambda_{\ell}}{2}}s\Bigr)^{-\nu_{\ell}}
\end{equation}
by utilizing $1-x^2\eq(1-x)(1+x)$ on \eqref{Eq:McLeishMGF}.
Then, exploiting the relation $\Gamma(1+{x})\!=\!{x}\Gamma({x})$
 \cite{BibAbramowitzStegunBook,BibZwillingerBook,BibKilbasSaigoBook}, $M_{X_{\Sigma}}(s)$ has been already obtained in
\eqref{Eq:McLeishSumMGF} and can be rewritten as
\begin{equation}\label{Eq:McLeishSumFactoizedMGF}
	\!\!\!\!M_{X_{\Sigma}}(s)=
		{e}^{-s{\sum_{\ell=1}^{L}\mu_{\ell}}}
			\prod_{\ell=1}^{L} 
				\frac{\Gamma^{\nu_\ell}\bigl(1+\frac{\lambda_{\ell}}{2}s\bigr)}
					 {\Gamma^{\nu_\ell}\bigl(2+\frac{\lambda_{\ell}}{2}s\bigr)}
				\frac{\Gamma^{\nu_\ell}\bigl(1-\frac{\lambda_{\ell}}{2}s\bigr)}
					 {\Gamma^{\nu_\ell}\bigl(2-\frac{\lambda_{\ell}}{2}s\bigr)}.
\end{equation}
Note that by means of \eqref{Eq:McLeishSumFactoizedMGF}, we express the \ac{PDF} of $X_{\Sigma}$ via the \ac{ILT} \cite{BibDishonWeissJCP1978},\cite[Chap.\!~3]{BibDebnathBhattaBook} as
\begin{equation}\label{Eq:McLeishSumPDFUsingILT}
	f_{X_{\Sigma}}(x)=\frac{1}{2\pi\imaginary}\int_{c-\imaginary\infty}^{c+\imaginary\infty}M_{X_{\Sigma}}(s)\exp(sx)ds
\end{equation}
within the existence region ${-S_0}\!<\!{\RealPart{s}}\!<\!{S_0}$, where ${S_0}\!\in\!\mathbb{R}_{+}$ is defined by  
${S_0}\!=\!{2}/\max_{\ell\in\{1..n\}}{\lambda_{\ell}}$.
Finally, substituting \eqref{Eq:McLeishSumFactoizedMGF} into \eqref{Eq:McLeishSumPDFUsingILT} and then using the
mathematical formalism given in \cite[Eq.\!~(3.1)]{BibRrathieLeMath1997} results in \eqref{Eq:McLeishSumPDF}, which
proves \theoremref{Theorem:McLeishSumPDF}.
\end{proof}
\begin{figure}[tp] 
\centering
\begin{subfigure}{0.7\columnwidth}
    \centering
    \includegraphics[clip=true, trim=0mm 0mm 0mm 0mm,width=1.0\columnwidth,height=0.85\columnwidth]{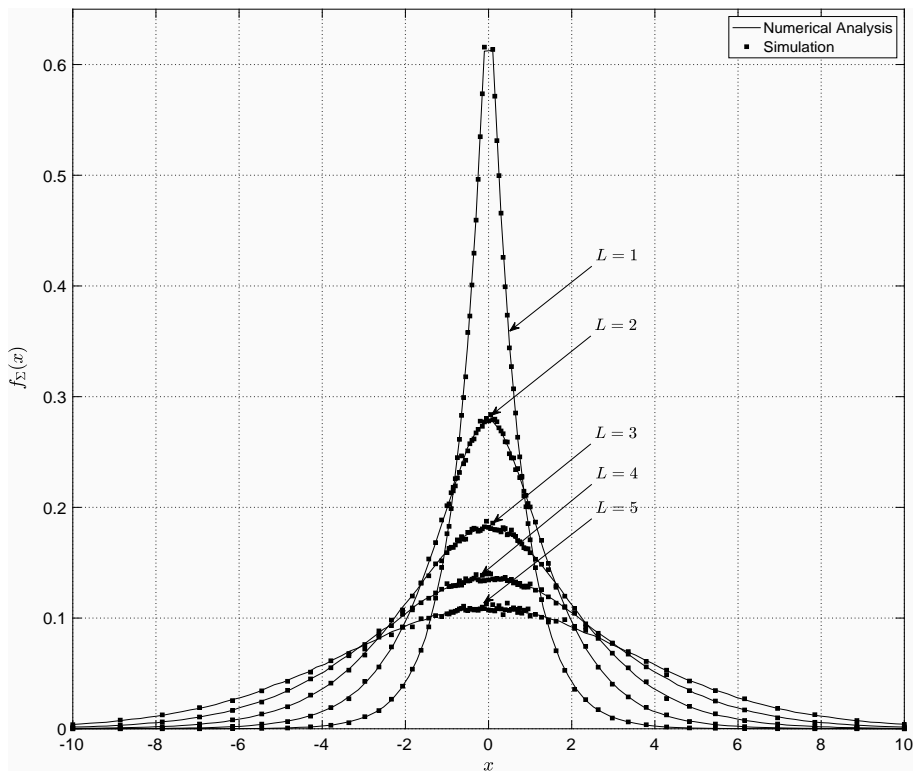}
    \caption{PDF (the number of samples for simulation is chosen as $10^6$).}
    \vspace{5mm}
    \label{Figure:McLeishSumPDF}
\end{subfigure}
\begin{subfigure}{0.7\columnwidth}
    \centering
    \includegraphics[clip=true, trim=0mm 0mm 0mm 0mm, width=1.0\columnwidth,height=0.85\columnwidth]{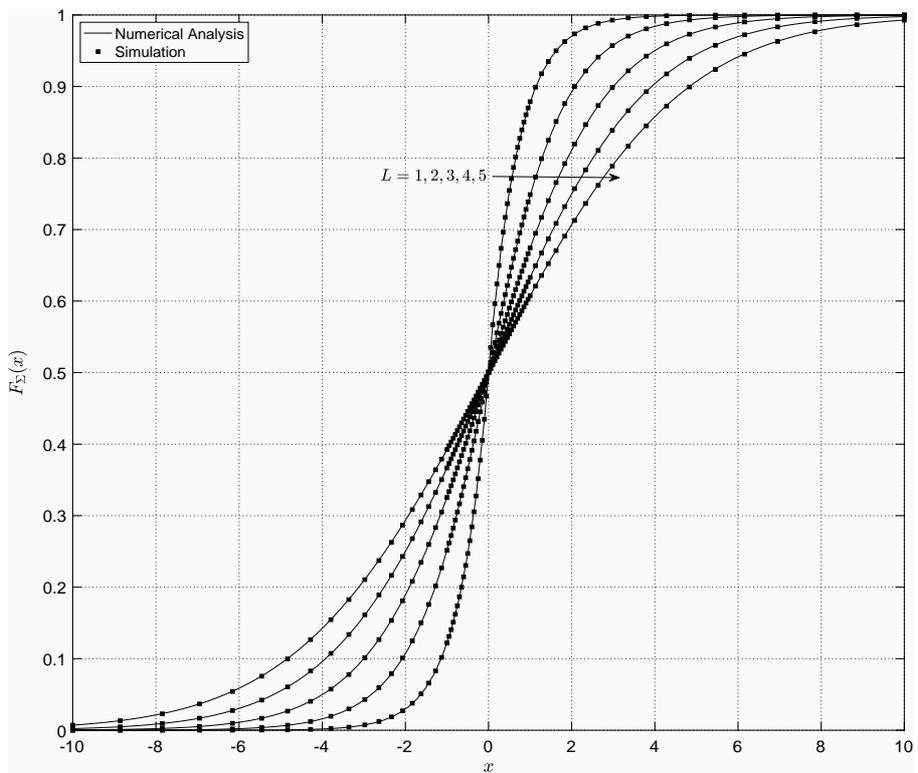}
    \caption{CDF (the number of samples for simulation is chosen as $10^6$).}
    \label{Figure:McLeishSumCDF}
\end{subfigure}
\caption{The \ac{PDF} and \ac{CDF} of sum of $L$ McLeish distributions with means $\mu_\ell\!=\!0$, and normalities $\nu_\ell\!=\!\ell$, and variances $\sigma^2_\ell\!=\!L-\ell+1$ for all ${1}\!\leq\!{\ell}\!\leq\!{L}$.}
\label{Figure:McLeishSumPDFCDF}
\vspace{-2mm} 
\end{figure} 
The PDF of $X_{\Sigma}$ is depicted in \figref{Figure:McLeishSumPDF} for different number of variables. Referring to \theoremref{Theorem:McLeishSumPDF}, some special cases are given for consistency in the following. In case of $\nu_{\ell}\!\in\!\mathbb{Z}^{+}$ and $\lambda_{\ell}\!\neq\!\lambda_{m}$ for all $\ell\!\neq\!m$, \eqref{Eq:McLeishSumDistribution} follows a hyper McLeish distribution whose \ac{PDF} can be deduced from \theoremref{Theorem:McLeishSumPDF} as     
\begin{equation}\label{Eq:McLeishSumPDFSpecialCaseI}
	f_{X_{\Sigma}}(x)=
		\sum_{\ell=1}^{L}
			\sum_{m=0}^{\nu_{\ell}-1}
				\frac{2w_{\ell{m}}}{\sqrt{\pi}\Gamma(\nu_{\ell}-{m})}
					\frac{\abs{x-\mu_{\Sigma}}^{\nu-\frac{1}{2}}}{\,\lambda^{\nu_{\ell}-{m}+\frac{1}{2}}}
						\BesselK[\nu_{\ell}-{m}-\frac{1}{2}]{\frac{2\abs{x-\mu_{\Sigma}}}{\lambda}},
\end{equation}
Further, when $\lambda_{\ell}\!=\!\lambda_{\Sigma}$ with distinct $\sigma^2_{\ell}$ for $\ell\!=\!{1,2,\ldots,n}$, \eqref{Eq:McLeishSumDistribution} certainly follows $\mathcal{M}_{\nu_\Sigma}(\mu_\Sigma,\sigma^2_\Sigma)$, whose
\ac{PDF} has been already given in \eqref{Eq:McLeishPDF}, that is 
\begin{equation}\label{Eq:McLeishSumPDFSpecialCaseII}
	f_{X_{\Sigma}}(x)=
		\frac{2\abs{x-\mu_{\Sigma}}^{\nu_{\Sigma}-\frac{1}{2}}}
			{\sqrt{\pi}\,\Gamma(\nu_{\Sigma})\,\lambda_{\Sigma}^{\nu_{\Sigma}+\frac{1}{2}}}
		\BesselK[\nu_{\Sigma}-\frac{1}{2}]{\frac{2\abs{x-\mu_{\Sigma}}}{\lambda_{\Sigma}}}.
\end{equation}
Additionally, the other special cases can be easily deduced for the certain normalities $\nu_{\ell}\!\rightarrow\!{0}$, $\nu_{\ell}\!\rightarrow\!{1}$ and $\nu_{\ell}\!\rightarrow\!\infty$ in \eqref{Eq:McLeishSumPDF}.
Accordingly, setting $\forall\nu_{\ell}\!\rightarrow\!{0}$ in \eqref{Eq:McLeishSumPDFUsingILT} and using  $\FoxI[right]{0,0}{0,0}{\exp(-x)}{\emptycoefficientSHORT}{\emptycoefficientSHORT}\!=\!\allowbreak\DiracDelta{x}$ with the aid of \cite[Eq.\!~(2.1)]{BibRrathieLeMath1997} and \cite[Eq.\!~(1.8.1/8)]{BibZwillingerBook}, we readily notice that \eqref{Eq:McLeishSumPDF} evolves into $f_{X_{\Sigma}}(x)\!=\!\delta\bigl(x-\mu_{\Sigma}\bigr)$. Further, setting 
$\forall\nu_{\ell}\!\rightarrow\!{1}$, \eqref{Eq:McLeishSumPDF} simplifies to the \ac{PDF} of 
the sum of \ac{i.n.i.d.} Laplace distributions, that is 
\begin{equation}\label{Eq:LaplaceSumPDF}
 	f_{X_{\Sigma}}(x)=
 	\frac{2^L}{\prod_{\ell=1}^{L}\sigma^2_{\ell}}
 	\MeijerG[right]{L,L}{2L,2L}
 		{\frac{\exp(-x)}{\exp(-\mu_{\Sigma})}}
 			{\boldsymbol{\Phi}_{L}^{(1)},\boldsymbol{\Phi}_{L}^{(3)}}				{\boldsymbol{\Phi}_{L}^{(2)},\boldsymbol{\Phi}_{L}^{(0)}},
\end{equation}
where the coefficient set $\boldsymbol{\Phi}_{n}^{(\alpha)}$ is given by
\begin{equation}
 	\boldsymbol{\Phi}_{n}^{(\alpha)}=\textstyle
 		{{\sqrt{2}(\alpha-1)}/{\sigma^2_1}},
 		\cdots,
 		{{\sqrt{2}(\alpha-1)}/{\sigma^2_n}}.
\end{equation}
In addition, when we choose all normalities to be infinity  (i.e., while having $\forall\ell\in\{1,2,\ldots,L\},\nu_{\ell}\!\rightarrow\!\infty$), we readily deduce $M_{X_{\ell}}(s)\!=\!\exp(-s{\mu_{\Sigma}}+{s^2}{\sigma^2_{\Sigma}}/{2})$ and accordingly reduce \eqref{Eq:McLeishSumPDF} to the \ac{PDF} of $\mathcal{N}(\mu_\Sigma,\sigma^2_\Sigma)$ as expected.

\begin{theorem}\label{Theorem:McLeishSumCDF}
The \ac{CDF} of \eqref{Eq:McLeishSumDistribution} is given by
\begin{equation}\label{Eq:McLeishSumCDF}
	\!\!F_{X_{\Sigma}}(x)=\FoxI[right]{n+1,n}{2n+1,2n}
		{\frac{\exp(-x)}{\exp(-\mu_{\Sigma})}}
			{{\boldsymbol{\Xi}}_{n}^{(1)},{\boldsymbol{\Xi}}_{n}^{(3)},(1,1,1)}
				{(0,1,1),{\boldsymbol{\Xi}}_{n}^{(2)},{\boldsymbol{\Xi}}_{n}^{(0)}}.\!\!
\end{equation}
\end{theorem}

\begin{proof}
Note that $F_{X_{\Sigma}}(x)\!=\!\Pr(X_{\Sigma}<x)$ is readily computed by using
$F_{X_{\Sigma}}(x)\!=\!\int_{-\infty}^{x}p_{X_{\Sigma}}(u)du$, where utilizing \eqref{Eq:McLeishSumPDFUsingILT} yields
\begin{equation}\label{Eq:McLeishSumCDFUsingILT}
	F_{X_{\Sigma}}(x)=\frac{1}{2\pi\imaginary}\int_{c-\imaginary\infty}^{c+\imaginary\infty}
		\biggl\{\int_{-\infty}^{x}{e}^{su}du\biggr\}M_{X_{\Sigma}}(s)ds
\end{equation}
within the existence region ${-S_0}\!<\!{\RealPart{s}}\!<\!{S_0}$. 
Accordingly, using $\int_{-\infty}^{x}{e}^{su}du\!=\!{e}^{sx}/{s}$ for $\RealPart{s}\!>\!0$\cite[Eq.\!~(3.310)]{BibGradshteynRyzhikBook}, 
\eqref{Eq:McLeishSumCDFUsingILT} can be easily rewritten as  
\begin{equation}\label{Eq:McLeishSumCDFUsingILTII}
	F_{X_{\Sigma}}(x)=\frac{1}{2\pi\imaginary}\int_{c-\imaginary\infty}^{c+\imaginary\infty}
		\frac{\Gamma(s)}{\Gamma(1+s)}M_{X_{\Sigma}}(s)ds
\end{equation}
within the existence region $0\!<\!{\RealPart{s}}\!<\!{S_0}$.
Finally, using the mathematical formalism given in \cite[Eq.\!~(3.1)]{BibRrathieLeMath1997} results in
\eqref{Eq:McLeishSumCDF}, which proves \theoremref{Theorem:McLeishSumCDF}.
\end{proof}

The CDF of $X_{\Sigma}$ is depicted in \figref{Figure:McLeishSumCDF} for different number of variables. 
Note that for $\nu_{\ell}\!\in\!\mathbb{Z}^{+}$ and $\lambda_{\ell}\!\neq\!\lambda_{m}$ for all $\ell\!\neq\!m$,
\eqref{Eq:McLeishSumCDF} reduces by using \eqref{Eq:McLeishSumPDFSpecialCaseI} with \theoremref{Theorem:McLeishCDF} as
follows
\begin{equation}\label{Eq:McLeishSumCDFSpecialCaseI}
F_{X_{\Sigma}}(x)=\sum_{\ell=1}^{L}\sum_{m=0}^{\nu_{\ell}-1}w_{\ell{m}}
	Q_{\nu_{\ell}-{m}}\Bigl(\frac{x-\mu_{\Sigma}}{\sigma_{\Sigma}}\Bigr).
\end{equation}
For $\lambda_{\ell}\!=\!\lambda$ with distinct $\nu_{\ell}$ and $\sigma^2_{\ell}$ for $\ell\!=\!{1,2,\ldots,n}$,
\eqref{Eq:McLeishSumDistribution} certainly follows a McLeish distribution whose \ac{PDF}  
is already obtained in \eqref{Eq:McLeishSumPDFSpecialCaseII}, and whose \ac{CDF} is then deduced as 
\begin{equation}\label{Eq:McLeishSumCDFSpecialCaseII}
	F_{X_{\Sigma}}(x)={Q}_{\nu_{\Sigma}}\Bigl(\frac{x-\mu_{\Sigma}}{\sigma_{\Sigma}}\Bigr).
\end{equation}
Further, the other special cases for $\nu_{\ell}\!\rightarrow\!{0}$, $\nu_{\ell}\!\rightarrow\!{1}$ and
$\nu_{\ell}\!\rightarrow\!\infty$ are herein ignored since being well-predicted utilizing the results that are previously obtained above. 

\begin{theorem}\label{Theorem:McLeishSumMoments}
The $n$th moment of \eqref{Eq:McLeishSumDistribution} is given by
\begin{equation}\label{Eq:McLeishSumMoments}
\Expected{X_{\Sigma}^n}=\sum_{k_1+k_2+\ldots+k_L=n}^{n}
	\frac{n!}{\prod_{\ell=1}^{L}k_{\ell}!}
		\prod_{\ell=1}^{L}\mathbb{E}\bigl[X_{\ell}^{k_{\ell}}\bigr],
\end{equation}
where $\mathbb{E}\bigl[X_{\ell}^{n}\bigr]$ is given in \eqref{Eq:McLeishMoments}
\end{theorem}

\begin{proof}
The proof is obvious by applying multinomial expansion\cite[Eq.\!~(24.1.2)]{BibAbramowitzStegunBook} on
$\Expected{X_{\Sigma}^n}\!=\!\mathbb{E}[(\sum_{\ell=1}^{n}X_{\ell})^{n}]$.
\end{proof}

For the statistical characterization of a McLeish distribution, such as its central tendency, dispersion,
skewness and Kurtosis, \eqref{Eq:McLeishSumMoments} can be easily used, and its special cases
can be obtained by setting its parameters.

\subsection{Complex and Circularly-Symmetric McLeish Distribution}
\label{Section:StatisticalBackground:CCSMcLeishDistribution}
Let $Z\!\sim\!\mathcal{CM}_{\nu}(\mu,\sigma^2)$ be a \ac{CCS} distribution, defined as 
\begin{equation}\label{Eq:ComplexMcLeishDefinition}
	Z\!=\!X_1+\imaginary{X_2},
\end{equation}
which is also, as mentioned before, deduced as a vector $\defrmat{Z}\!=\![X_1,~X_2]^T$, where $X_1\!\sim\!\mathcal{M}_{\nu_{1}}(\mu_{1},\sigma^2)$ and $X_2\!\sim\!\mathcal{M}_{\nu_{2}}(\mu_{2},\sigma^2)$
are, without loss of generality, such two mutually \ac{c.i.d.} distributions that $\mu\!=\!\mu_{1}+\imaginary\mu_{2}$ and $\nu\!=\!\nu_{1}\!=\!\nu_{2}$.

\begin{theorem}\label{Theorem:CCSMcLeishDefinition}
Under the condition of being \ac{CCS}, the definition of $Z\!\sim\!\mathcal{CM}_{\nu}(\mu,\sigma^2)$ can be decomposed as 
\begin{equation}\label{Eq:CCSMcLeishDistributionDefinition}
	Z=\sqrt{G}{Z_0}+\mu=\sqrt{G}({X_0}+\imaginary{Y_0})+\mu,
\end{equation}  
where $Z_{0}\!\sim\!\mathcal{CN}(0,\sigma^2)$, $X_{0}\!\sim\!\mathcal{N}(0,\sigma^2)$,
$Y_{0}\!\sim\!\mathcal{N}(0,\sigma^2)$, and ${G}\!\sim\!\mathcal{G}(\nu,1)$.
\end{theorem}

\begin{proof} 
By the definition of \ac{CCS} random distributions \cite{BibGallagerPRPRT2008}, both $(Z-\mu)$ and ${e}^{\imaginary\phi}(Z-\mu)$ follow the same distribution for any rotation $\phi\!\in\![-\pi,\pi)$. Accordingly, we affirm that the phase of $Z$ around its mean $\mu$ is typically given by 
\begin{equation}
	\Phi=
	    \arctan\bigl({X_1-\mu_{1}},{X_2-\mu_{2}}\bigr),
\end{equation}
where $\arctan(\cdot,\cdot)$ denotes the two-argument inverse tangent function 
\cite[Eq. (01.15.02.0001.01)]{BibWolfram2010Book}, and $\Phi$ is uniformly dis\-tri\-buted 
over $[-\pi,\pi)$ and independent of both $X$ and $Y$ (i.e., $\Covariance{\Phi}{X_1}\!=\!{0}$ and $\Covariance{\Phi}{X_2}\!=\!{0}$), Therefore, $W\!=\!\tan(\Phi)$ follows a zero-mean
Cauchy distribution whose \ac{PDF} is given by $f_{W}(w)\!=\!{\pi^{-1}(1+w^2)^{-1}}$ over
$w\!\in\!\mathbb{R}$ \cite{BibKrishnamoorthyBook,BibForbesEvansHastingsPeacockBook,BibJohnsonBalakrishnanKotzBookVol1}.
Upon $Z_{0}\!=\!{X_0}+\imaginary{Y_0}$, where $X_0\!\sim\!\mathcal{N}(0,\sigma^2_{Z}/2)$ and
$Y_0\!\sim\!\mathcal{N}(0,\sigma^2_{Z}/2)$, $Y_0/X_0$ follows a Cauchy distribution with zero mean 
and unit variance. Accordingly, $W$ is rewritten as
\begin{equation}
	W=\frac{X_2-\mu_{2}}{X_1-\mu_{1}}=\frac{\sqrt{G}Y_0}{\sqrt{G}X_0},
\end{equation} 
where without loss of generality, $G$ will follow a non-negative distribution characterized by 
\begin{equation}\label{Eq:PowerFluctuationDistribution}
	\sqrt{G}=\frac{\abs{X_1-\mu_{1}}}{\abs{X_0}}=\frac{\abs{X_2-\mu_{2}}}{\abs{Y_0}}.
\end{equation}
Utilizing \cite[Eq.\!~(2.9.19)]{BibKilbasSaigoBook} after performing absolute-value transformation on
\eqref{Eq:McLeishPDF}, we can deduce the \ac{PDF} of $\abs{X_1-\mu_{1}}$ in terms of Fox's H function as follows
\begin{equation}\label{Eq:NumeratorPDF}
	f_{\abs{X_1-\mu_{1}}}(x)=\frac{1}{\sqrt{\pi}\Gamma(\nu)}
		\scalemath{0.95}{0.95}{
		\FoxH[left]{2,0}{0,2}{\frac{2x^2}{\lambda^2}}{\emptycoefficient}{(0,1),(\nu-\frac{1}{2},1)}}
\end{equation}
defined over $x\!\in\!\mathbb{R}_{+}$. 
Similarly, using \cite[Eq.\!~(2.9.4)]{BibKilbasSaigoBook}, we can also deduce the \ac{PDF} of $\abs{X_0}$, that is 
\begin{equation}\label{Eq:DenominatorPDF}
	f_{\abs{X_{0}}}(x)=
	\sqrt{\frac{2}{\pi\sigma^2}}\FoxH[left]{1,0}{0,1}{
		\frac{x^2}{\sigma^2}
		}{\emptycoefficient}{(0,1)}
\end{equation}
defined over $x\!\in\!\mathbb{R}_{+}$. Immediately, embedding both \eqref{Eq:NumeratorPDF} and \eqref{Eq:DenominatorPDF} within
\cite[Theorem~4.3]{BibCarterSpringerSIAM1977} and thereon exercising \cite[Eqs.~(2.1.1),~(2.1.4)~and~(2.1.4)]{BibKilbasSaigoBook}, we derive the \ac{PDF} of $G$ as
\begin{equation}\label{Eq:ProportionPDF}
	f_{G}(g)=\frac{\nu^{\nu}}{\Gamma(\nu)}g^{\nu-1}\exp\left(-\nu{g}\right),
\end{equation}
defined over $g\!\in\!\mathbb{R}_{+}$. This consequence can also be reached from the ratio of $\abs{X_2-\mu_2}$ and $\abs{Y_0}$ in conformity with \eqref{Eq:PowerFluctuationDistribution}. Eventually, with the aid of \cite[Eq. (2.3-67)]{BibProakisBook} and \cite[Eqs. (2.20) and (2.21)]{BibAlouiniBook}, we notice that $G$ is a non-negative distribution following Gamma (squared Nakagami-\emph{m}) distribution. Therefore, $G\!\sim\!\mathcal{G}(\nu,1)$, where the diversity figure is given by $\nu\!=\!\Expected{G}^2/\Variance{G}$\cite[Eq.\!~(2.3-69)]{BibProakisBook} and the average power is
by $\Expected{G}\!=\!{1}$\cite[Eq.\!~(2.3-68)]{BibProakisBook}.
Consequently, the definition of \ac{CCS} McLeish distribution, given in \eqref{Eq:ComplexMcLeishDefinition}, is rewritten as
in \eqref{Eq:CCSMcLeishDistributionDefinition}, which proves \theoremref{Theorem:CCSMcLeishDefinition}.
\end{proof} 

With the aid of \theoremref{Theorem:CCSMcLeishDefinition}, the \ac{PDF} of $Z$ (i.e, the joint \ac{PDF} $f_{\defrmat{Z}}(x,y)$
of $\defrmat{Z}$) is given in the following theorem.

\begin{theorem}\label{Theorem:CCSMcLeishPDF}
Under the condition of being \ac{CCS}, the \ac{PDF} of $Z\!\sim\!\mathcal{CM}_{\nu}(\mu,\sigma^2)$ is given by
\begin{equation}\label{Eq:CCSMcLeishPDF}
	f_{Z}(z)=\frac{2}{\pi}
		\frac{\abs{z-\mu}^{\nu-1}}{\Gamma(\nu)\,\lambda^{\nu+1}}
			\BesselK[\nu-1]{\frac{2\abs{z-\mu}}{\lambda}},
\end{equation}
defined over $z\!\in\!\mathbb{C}$, where the factor $\lambda=\sqrt{2\sigma^2/\nu}$. 
\end{theorem}

\begin{proof}
Referring to \theoremref{Theorem:CCSMcLeishDefinition}, 
the \ac{PDF} of $Z\!\sim\!\mathcal{CM}_{\nu}(\mu,\sigma^2)$ conditioned on $G$ is therefore written as \cite[Eq.~(2.6-1)]{BibProakisBook}
\begin{equation}\label{Eq:ConditionedCCSMcLeishPDF}
	f_{Z|G}(z|g)=\frac{1}{\pi{g}\,\sigma^2}
		\exp\Bigl(-\frac{1}{g}\Bigl\langle{\frac{z-\mu}{\sigma},\frac{z-\mu}{\sigma}}\Bigr\rangle\Bigr),
\end{equation} 
for $g\!\in\!\mathbb{R}_{+}$. In accordance, the \ac{PDF} of $Z$ can be expressed as
$f_{Z}(z)\!=\!\int_{0}^{\infty}f_{Z|G}(z|g)f_{G}(g)dg$,  
where substituting \eqref{Eq:ProportionPDF} and \eqref{Eq:ConditionedCCSMcLeishPDF}, and subsequently employing
\cite[Eq.\!~(3.471/9)]{BibGradshteynRyzhikBook} results in \eqref{Eq:CCSMcLeishPDF}, 
which proves \theoremref{Theorem:CCSMcLeishPDF}. 
\end{proof}

The \ac{PDF} of $Z\!\sim\!\mathcal{CM}_{\nu}(\mu,\sigma^2)$ and its contour plot are well described in \figref{Figure:CCSMcLeishPDFA} and \figref{Figure:CCSMcLeishPDFB}, respectively. Further worth noting that the \ac{CS} property of $Z\!\sim\!\mathcal{CM}_{\nu}(\mu,\sigma^2)$ is observed in \figref{Figure:CCSMcLeishPDFB} such that $\forall\theta\!\in\![-\pi,\pi)$, $f_{Z}(z)\!=\!f_{Z}(z\exp(\imaginary\theta))$ for $\mu\!=\!0$. Accordingly, for a given contour value $c\!\in\!\mathbb{R}_{+}$, the contours, presented in \figref{Figure:CCSMcLeishPDFB}, can be obtained by
\begin{eqnarray}
    \nonumber
    (z|c)&=&\bigl\{z=\widehat{\xi}\exp(\imaginary\theta)\,\,
        \bigl|\,\,
        \theta\in[-\pi,\pi),\text{~and~}\\ 
        &~&{~~~~~~~~~~~~~~~~~~}\widehat{\xi}=
            \argmin_{\xi\in\mathbb{R}_{+}}
                \lVert
                    {f^{2}_{Z}(\xi)-c}
                \rVert^2
        \bigr.
    \bigr\}.{~~~}
\end{eqnarray}

For consistency, let us now consider some special cases of \theoremref{Theorem:CCSMcLeishPDF}. 
Substituting $\nu\!=\!{1}$ into \eqref{Eq:CCSMcLeishPDF} yields
the \ac{PDF} of $\mathcal{CL}(\mu,\sigma^2)$\cite[Eq.~(6)]{BibSouryElsawyAlouiniWCOM2017}, that is
\begin{equation}\label{Eq:CCSLaplacePDF}
    f_{Z}(z)=\frac{2}{\pi\lambda^2}{K}_{0}\left(\frac{2}{\lambda}\left|z-\mu\right|\right).
\end{equation}
Moreover, substituting $\nu\!\rightarrow\!\infty$ results in 
\begin{equation}\label{Eq:CCSGaussianPDF}
	f_{Z}(z)=
	    \frac{1}{2\pi\sigma^2}
	        \exp\Bigl(-\frac{1}{2\sigma^2}\left|z-\mu\right|\Bigr),
\end{equation}
which is the \ac{PDF} of $\mathcal{CN}(\mu,\sigma^2)$\cite[Eq.~(2.6-1)]{BibProakisBook}. 
\begin{figure*}[tp] 
\centering
\begin{subfigure}{0.7\columnwidth}
    \centering
    \includegraphics[clip=true, trim=0mm 0mm 0mm 0mm,width=1.0\columnwidth,height=0.85\columnwidth]{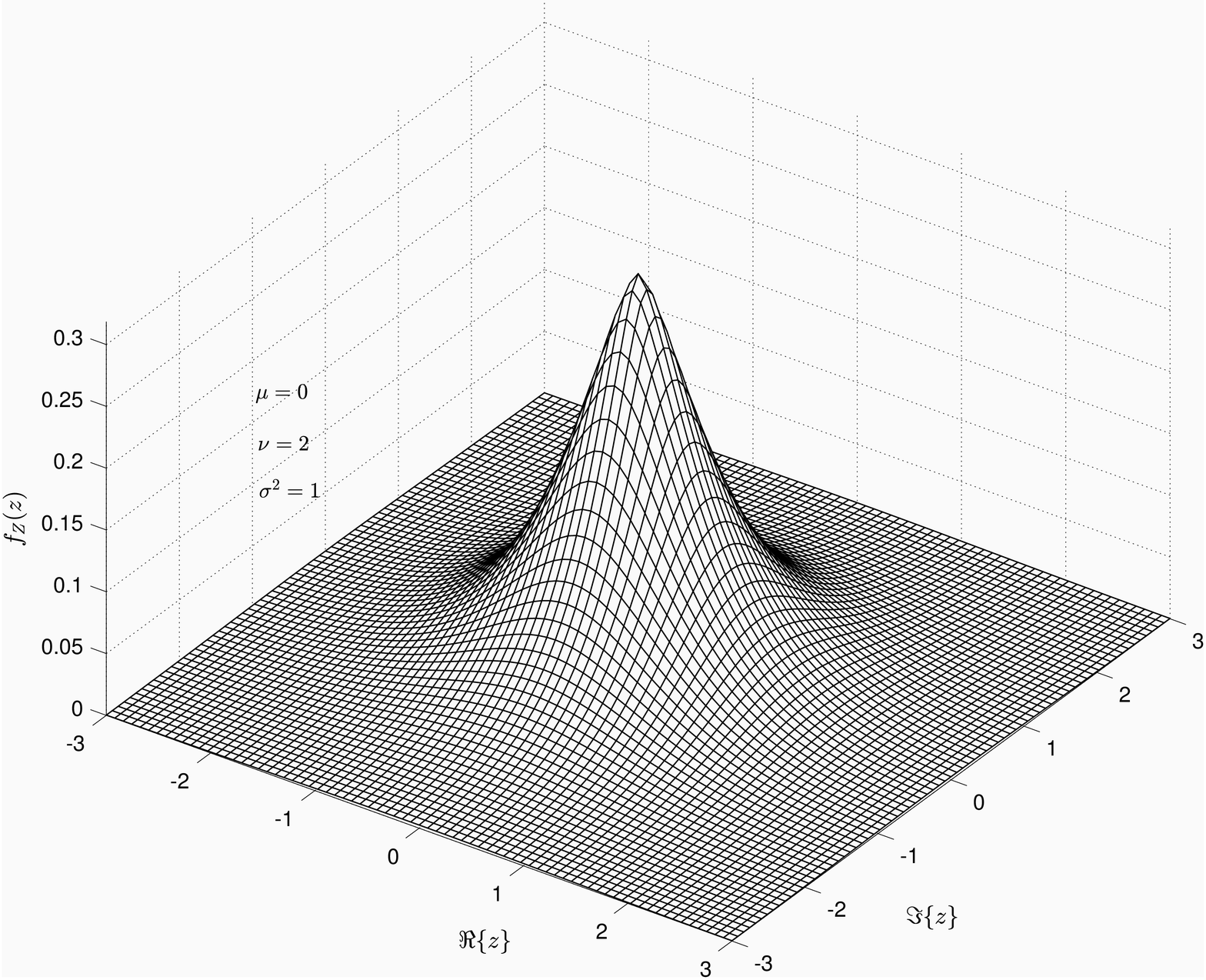}
    \caption{\ac{PDF} illustration in complex space.}
    \label{Figure:CCSMcLeishPDFA}
\end{subfigure}
~~~
\begin{subfigure}{0.7\columnwidth}
    \centering
    \includegraphics[clip=true, trim=0mm 0mm 0mm 0mm,width=1.0\columnwidth,height=0.85\columnwidth]{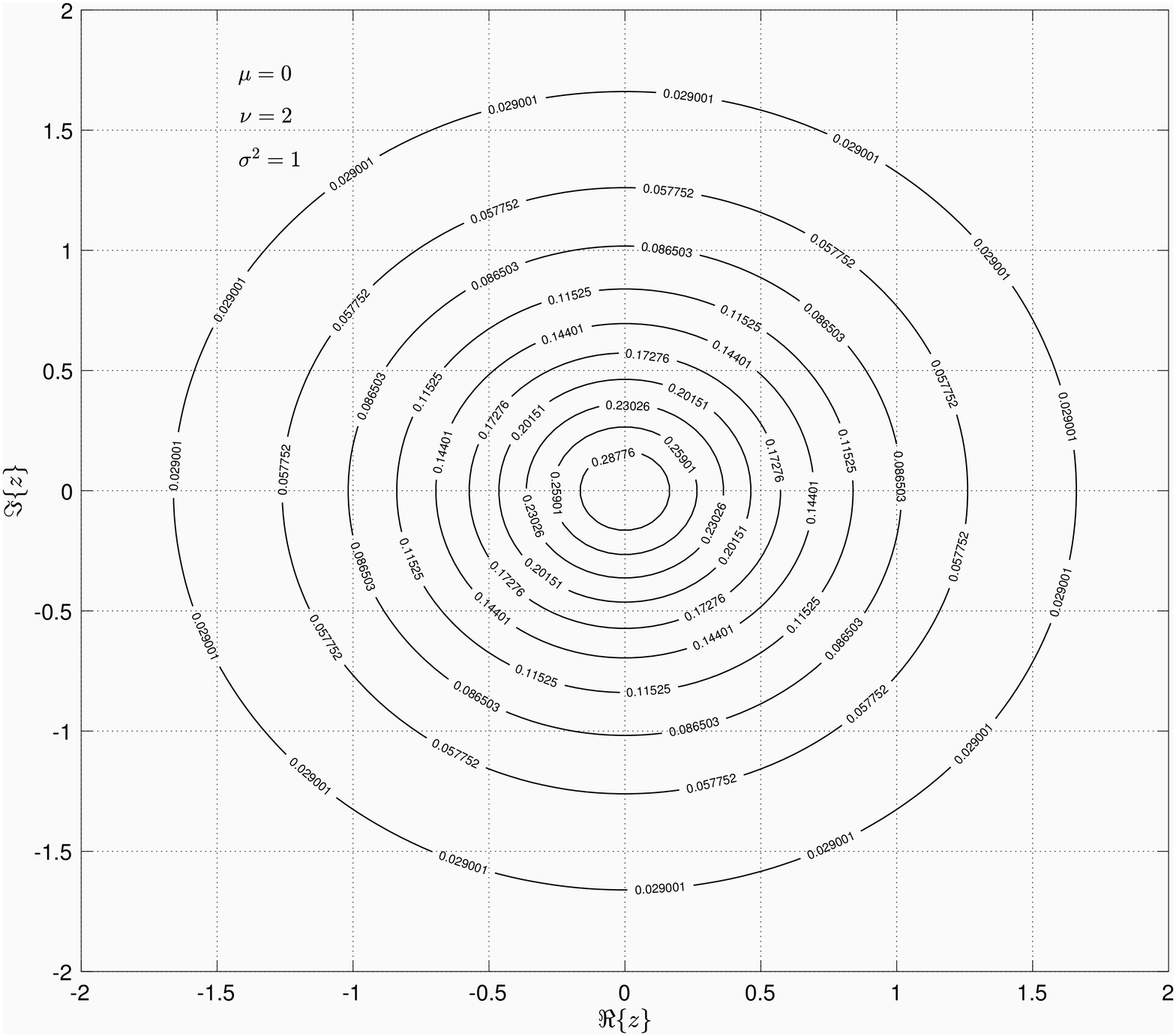}
    \caption{\ac{PDF} contour curves.}
    \label{Figure:CCSMcLeishPDFB}
\end{subfigure}
\caption{The \ac{PDF} and contour of $\mathcal{CM}_{\nu}(0,\sigma^2)$ (i.e., the illustration of \eqref{Eq:CCSMcLeishPDF} for $\mu\!=\!0$).}
\label{Figure:CCSMcLeishPDF}
\vspace{-2mm} 
\end{figure*}

\begin{theorem}\label{Theorem:CCSMcLeishCDF}
Under the condition of being \ac{CCS}, the \ac{CDF} of $Z\!\sim\!\mathcal{CM}_{\nu}(\mu,\sigma^2)$ is given for the complex quadrants, that is 
\begin{subequations}\label{Eq:CCSMcLeishCDF}
\begin{multline}\label{Eq:CCSMcLeishCDFA} 
	\!\!F_{Z}(z)=1
		-Q_{\nu}\Bigl(\sqrt{2}\Bigl\langle{1,\frac{z-\mu}{\sigma}}\Bigr\rangle\Bigr)
		-Q_{\nu}\Bigl(\sqrt{2}\Bigl\langle{\imaginary,\frac{z-\mu}{\sigma}}\Bigr\rangle\Bigr)~~~~~~~~~~~~~\\
		+\frac{1}{2}Q_{\nu}\Bigl(\sqrt{2\Bigl\langle{\frac{z-\mu}{\sigma},\frac{z-\mu}{\sigma}}\Bigr\rangle\sin^2(\phi)},\phi\Bigr)\\
		+\frac{1}{2}Q_{\nu}\Bigl(\sqrt{2\Bigl\langle{\frac{z-\mu}{\sigma},\frac{z-\mu}{\sigma}}\Bigr\rangle\cos^2(\phi)},\frac{\pi}{2}-\phi\Bigr),
\end{multline}
for the upper right quadrant (i.e., $\RealPart{z}\!\geq\!0$ and $\ImagPart{z}\!\geq\!0$);
\begin{multline}\label{Eq:CCSMcLeishCDFB} 
	\!\!F_{Z}(z)=
		Q_{\nu_Z}\Bigl(\sqrt{2}\Bigl\langle{1,\frac{\mu-z}{\sigma}}\Bigr\rangle\Bigr)
		-\frac{1}{2}Q_{\nu}\Bigl(\sqrt{2\Bigl\langle{\frac{z-\mu}{\sigma},\frac{z-\mu}{\sigma}}\Bigr\rangle\sin^2(\phi)},\phi\Bigr)\\
		-\frac{1}{2}Q_{\nu}\Bigl(\sqrt{2\Bigl\langle{\frac{z-\mu}{\sigma},\frac{z-\mu}{\sigma}}\Bigr\rangle\cos^2(\phi)},\frac{\pi}{2}-\phi\Bigr),
\end{multline}
for the upper left quadrant (i.e., $\RealPart{z}\!<\!0$ and $\ImagPart{z}\!\geq\!0$);
\begin{multline}\label{Eq:CCSMcLeishCDFC} 
	\!\!F_{Z}(z)=
		\frac{1}{2}Q_{\nu}\Bigl(\sqrt{2\Bigl\langle{\frac{z-\mu}{\sigma},\frac{z-\mu}{\sigma}}\Bigr\rangle\sin^2(\phi)},\phi\Bigr)\\
		+\frac{1}{2}Q_{\nu}\Bigl(\sqrt{2\Bigl\langle{\frac{z-\mu}{\sigma},\frac{z-\mu}{\sigma}}\Bigr\rangle\cos^2(\phi)},\frac{\pi}{2}-\phi\Bigr),
\end{multline}
for the lower left quadrant (i.e., $\RealPart{z}\!<\!0$ and $\ImagPart{z}\!<\!0$);
\begin{multline}\label{Eq:CCSMcLeishCDFD} 
	\!\!F_{Z}(z)=
		Q_{\nu_Z}\Bigl(\sqrt{2}\Bigl\langle{\imaginary,\frac{\mu-z}{\sigma}}\Bigr\rangle\Bigr)
		-\frac{1}{2}Q_{\nu}\Bigl(\sqrt{2\Bigl\langle{\frac{z-\mu}{\sigma},\frac{z-\mu}{\sigma}}\Bigr\rangle\sin^2(\phi)},\phi\Bigr)\\
		-\frac{1}{2}Q_{\nu}\Bigl(\sqrt{2\Bigl\langle{\frac{z-\mu}{\sigma},\frac{z-\mu}{\sigma}}\Bigr\rangle\cos^2(\phi)},\frac{\pi}{2}-\phi\Bigr),
\end{multline}
\end{subequations}
for the lower right quadrant (i.e., $\RealPart{z}\!\geq\!0$ and $\ImagPart{z}\!<\!0$); where $\phi\in[0,\frac{\pi}{2})$
is given by $\phi\!=\!\arctan\bigl(\abs{\RealPart{z}},\abs{\ImagPart{z}}\bigr)$. 
\end{theorem}

\begin{proof}
Note that the \ac{CDF} of
$Z_{0}\!\sim\!\mathcal{CN}(0,\sigma^2)$ is defined by
$F_{Z_0}(z_{\ell}|\sigma)\!=\!\Pr\{{X_{0}\!\leq\!\langle{1,z_{\ell}}\rangle}\,\cap\,{Y_{0}\!\leq\!\langle{\imaginary,z_{\ell}}\rangle}\,|\,\sigma\}$ conditioned on $\sigma$. Utilizing \cite[Eqs.\!~(2.3-10)\!~and\!~(2.3-11)]{BibProakisBook}
and \cite[Eqs.\!~(4.3)]{BibAlouiniBook} with $\langle{1,z}\rangle\!=\!\RealPart{z}$ and
$\langle{\imaginary,z}\rangle\!=\!\ImagPart{z}$,  $F_{Z_0}(z_{\ell}|\sigma)$ can be readily expressed for a certain $z\!=\!{x}+\imaginary{y}\in\mathbb{C}$ as follows
\begin{subequations}\label{Eq:ComplexGaussianCDF}
\ifCLASSOPTIONtwocolumn
\begin{multline}\label{Eq:ComplexGaussianCDFA} 
	\!\!F_{Z_0}(z|\sigma)=1
		-Q\bigl(\sqrt{2}\bigl\langle{1,{z}/{\sigma}}\bigr\rangle\bigr)
		-Q\bigl(\sqrt{2}\bigl\langle{\imaginary,{z}/{\sigma}}\bigr\rangle\bigr)\\
		+Q\bigl(\sqrt{2}\bigl\langle{1,{z}/{\sigma}}\bigr\rangle\bigr)\,
		 Q\bigl(\sqrt{2}\bigl\langle{\imaginary,{z}/{\sigma}}\bigr\rangle\bigr),
\end{multline}
\else
\begin{equation}\label{Eq:ComplexGaussianCDFA} 
	\!\!F_{Z_0}(z|\sigma)=1
		-Q\bigl(\sqrt{2}\bigl\langle{1,{z}/{\sigma}}\bigr\rangle\bigr)
		-Q\bigl(\sqrt{2}\bigl\langle{\imaginary,{z}/{\sigma}}\bigr\rangle\bigr)
		+Q\bigl(\sqrt{2}\bigl\langle{1,{z}/{\sigma}}\bigr\rangle\bigr)\,
		 Q\bigl(\sqrt{2}\bigl\langle{\imaginary,{z}/{\sigma}}\bigr\rangle\bigr),
\end{equation}
\fi
for the upper right quadrant (i.e., $\RealPart{z}\!\geq\!0$ and $\ImagPart{z}\!\geq\!0$);
\ifCLASSOPTIONtwocolumn
\begin{multline}\label{Eq:ComplexGaussianCDFB}
	\!\!F_{Z_0}(z|\sigma)=
		Q\bigl(-\sqrt{2}\bigl\langle{1,{z}/{\sigma}}\bigr\rangle\bigr)\\
		-Q\bigl(-\sqrt{2}\bigl\langle{1,{z}/{\sigma}}\bigr\rangle\bigr)\,
		 Q\bigl(\sqrt{2}\bigl\langle{\imaginary,{z}/{\sigma}}\bigr\rangle\bigr),
\end{multline}
\else
\begin{equation}\label{Eq:ComplexGaussianCDFB}
	\!\!F_{Z_0}(z|\sigma)=
		Q\bigl(-\sqrt{2}\bigl\langle{1,{z}/{\sigma}}\bigr\rangle\bigr)
		-Q\bigl(-\sqrt{2}\bigl\langle{1,{z}/{\sigma}}\bigr\rangle\bigr)\,
		 Q\bigl(\sqrt{2}\bigl\langle{\imaginary,{z}/{\sigma}}\bigr\rangle\bigr),
\end{equation}
\fi
for the upper left quadrant (i.e., $\RealPart{z}\!<\!0$ and $\ImagPart{z}\!\geq\!0$);
\begin{equation}\label{Eq:ComplexGaussianCDFC}
	\!\!F_{Z_0}(z|\sigma)=
		Q\bigl(-\sqrt{2}\bigl\langle{1,{z}/{\sigma}}\bigr\rangle\bigr)
		Q\bigl(-\sqrt{2}\bigl\langle{\imaginary,{z}/{\sigma}}\bigr\rangle\bigr),
\end{equation}
for the lower left quadrant (i.e., $\RealPart{z}\!<\!0$ and $\ImagPart{z}\!<\!0$);
\ifCLASSOPTIONtwocolumn
\begin{multline}\label{Eq:ComplexGaussianCDFD}
	\!\!F_{Z_0}(z|\sigma)=
		Q\bigl(-\sqrt{2}\bigl\langle{\imaginary,{z}/{\sigma}}\bigr\rangle\bigr)\\
		-Q\bigl(\sqrt{2}\bigl\langle{1,{z}/{\sigma}}\bigr\rangle\bigr)
		 Q\bigl(-\sqrt{2}\bigl\langle{\imaginary,{z}/{\sigma}}\bigr\rangle\bigr),
\end{multline}
\else
\begin{equation}\label{Eq:ComplexGaussianCDFD}
	\!\!F_{Z_0}(z|\sigma)=
		Q\bigl(-\sqrt{2}\bigl\langle{\imaginary,{z}/{\sigma}}\bigr\rangle\bigr)
		-Q\bigl(\sqrt{2}\bigl\langle{1,{z}/{\sigma}}\bigr\rangle\bigr)
		 Q\bigl(-\sqrt{2}\bigl\langle{\imaginary,{z}/{\sigma}}\bigr\rangle\bigr),
\end{equation}
\fi
for the lower right quadrant (i.e., $\RealPart{z}\!\geq\!0$ and $\ImagPart{z}\!<\!0$).
\end{subequations}
Worth noticing that \emph{the
argument of all Gaussian Q-functions in \eqref{Eq:ComplexGaussianCDF} is positive}, so the well-known Craig's representation\cite[Eq.\!~(4.2)]{BibAlouiniBook} and Simon-Divsalar's representation\cite[Eq.\!~(4.6)]{BibAlouiniBook} can be easily utilized in all equations from \eqref{Eq:ComplexGaussianCDFA} to \eqref{Eq:ComplexGaussianCDFD}. Then, referring
\eqref{Eq:CCSMcLeishDistributionDefinition}, the \ac{CDF} of $Z\!\sim\!\mathcal{CM}_{\nu}(\mu,\sigma^2)$ is
explicitly written as $F_{Z}(z)\!=\!\int_{0}^{\infty}F_{Z_0}(z-\mu\,|\sqrt{g}\sigma)f_{G}(g)\,dg$,
where substituting \eqref{Eq:ProportionPDF} yields
\begin{subequations}\label{Eq:CCSMcLeishCDFExpression}
\begin{multline}\label{Eq:CCSMcLeishCDFExpressionA} 
	\!\!F_{Z}(z)=1
		-I_{1}\bigl(\sqrt{2}\bigl\langle{1,{(z-\mu)}/{\sigma}}\bigr\rangle\bigr)
		-I_{1}\bigl(\sqrt{2}\bigl\langle{\imaginary,{(z-\mu)}/{\sigma}}\bigr\rangle\bigr)\\
		+I_{2}\bigl(\sqrt{2}\bigl\langle{1,{(z-\mu)}/{\sigma}}\bigr\rangle,\bigl\langle{\imaginary,{(z-\mu)}/{\sigma}}\bigr\rangle\bigr),
\end{multline}
for the upper right quadrant (i.e., $\RealPart{z}\!\geq\!0$ and $\ImagPart{z}\!\geq\!0$);
\ifCLASSOPTIONtwocolumn
\begin{multline}\label{Eq:CCSMcLeishCDFExpressionB} 
	\!\!F_{Z}(z)=
		 I_{1}\bigl(\sqrt{2}\bigl\langle{1,{(\mu-z)}/{\sigma}}\bigr\rangle\bigr)\\
		-I_{2}\bigl(\sqrt{2}\bigl\langle{1,{(\mu-z)}/{\sigma}}\bigr\rangle,\bigl\langle{\imaginary,{(z-\mu)}/{\sigma}}\bigr\rangle\bigr),
\end{multline}
\else
\begin{equation}\label{Eq:CCSMcLeishCDFExpressionB} 
	\!\!F_{Z}(z)=
		 I_{1}\bigl(\sqrt{2}\bigl\langle{1,{(\mu-z)}/{\sigma}}\bigr\rangle\bigr)
		-I_{2}\bigl(\sqrt{2}\bigl\langle{1,{(\mu-z)}/{\sigma}}\bigr\rangle,\bigl\langle{\imaginary,{(z-\mu)}/{\sigma}}\bigr\rangle\bigr),
\end{equation}
\fi
for the upper left quadrant (i.e., $\RealPart{z}\!<\!0$ and $\ImagPart{z}\!\geq\!0$);
\begin{equation}\label{Eq:CCSMcLeishCDFExpressionC}
	\!\!F_{Z}(z)=I_{2}\bigl(\sqrt{2}\bigl\langle{1,{(\mu-z)}/{\sigma}}\bigr\rangle,\bigl\langle{\imaginary,{(\mu-z)}/{\sigma}}\bigr\rangle\bigr),\!
\end{equation}
for the lower left quadrant (i.e., $\RealPart{z}\!<\!0$ and $\ImagPart{z}\!<\!0$);
\begin{equation}\label{Eq:CCSMcLeishCDFExpressionD}
	F_{Z}(z)=
		 I_{1}\bigl(\sqrt{2}\bigl\langle{\imaginary,{(\mu-z)}/{\sigma}}\bigr\rangle\bigr)
		-I_{2}\bigl(\sqrt{2}\bigl\langle{1,{(z-\mu)}/{\sigma}}\bigr\rangle,\bigl\langle{\imaginary,{(\mu-z)}/{\sigma}}\bigr\rangle\bigr),
\end{equation}
\end{subequations}
for the lower right quadrant (i.e., $\RealPart{z}\!\geq\!0$ and $\ImagPart{z}\!<\!0$), 
where $I_1(x)$ and $I_2(x,y)$ are respectively defined as
\setlength\arraycolsep{1.4pt}
\begin{eqnarray}
    \label{Eq:McLeishQFunctionIntegralWithoutCorrelation}
    I_1(x)&=&\int_{0}^{\infty}Q\bigl(\sqrt{x^2/g}\bigr)f_{G}(g)\,dg,\\
    \label{Eq:McLeishBivariateQFunctionIntegralWithoutCorrelation}
    I_2(x,y)&=&\int_{0}^{\infty}Q\bigl(\sqrt{x^2/g}\bigr)Q\bigl(\sqrt{y^2/g}\bigr)f_{G}(g)\,dg,
\end{eqnarray}
for $x,y\!\in\!\mathbb{R}_{+}$. Eventually, substituting 
$Q(x)\!=\!\frac{1}{2}\Erfc{x/\sqrt{2}}$\cite[Eq.\!~(2.3-18)]{BibProakisBook} and \cite[Eq.\!~(06.27.26.0006.01)]{BibWolfram2010Book} 
into \eqref{Eq:McLeishQFunctionIntegralWithoutCorrelation}, and then using 
\cite[Eqs.\!~(2.8.4)\!~and\!~(2.9.1)]{BibKilbasSaigoBook}, $I_1(x)$ results in \eqref{Eq:McLeishQFunctionUsingFoxH}.
In addition, substituting \cite[Eq.\!~(4.6)\!~and\!~(4.8)]{BibSimonDivsalarTCOM1998} into \eqref{Eq:McLeishBivariateQFunctionIntegralWithoutCorrelation} and 
using \cite[Eq. (3.471/9)]{BibGradshteynRyzhikBook}, $I_2(x,y)$ is obtained as
\ifCLASSOPTIONtwocolumn
\begin{multline}
I_2(x,y)=\frac{1}{2}Q_{\nu}\Bigl(\sqrt{(x^2+y^2)\sin(\phi)^2},\phi\Bigr)\\
	+\frac{1}{2}Q_{\nu}\Bigl(\sqrt{(x^2+y^2)\cos(\phi)^2},\frac{\pi}{2}-\phi\Bigr),
\end{multline}
\else
\begin{equation}
I_2(x,y)=\frac{1}{2}Q_{\nu}\Bigl(\sqrt{(x^2+y^2)\sin(\phi)^2},\phi\Bigr)
	+\frac{1}{2}Q_{\nu}\Bigl(\sqrt{(x^2+y^2)\cos(\phi)^2},\frac{\pi}{2}-\phi\Bigr),
\end{equation}
\fi
where $\phi\!=\!\arctan(x,y)$. Consequently, substituting $I_1(x)$ and $I_2(x,y)$ into \eqref{Eq:CCSMcLeishCDFExpression} yields \eqref{Eq:CCSMcLeishCDF}, which proves   
\theoremref{Theorem:CCSMcLeishCDF}.
\end{proof}   

The \ac{CDF} of $Z\!\sim\!\mathcal{CM}_{\nu}(\mu,\sigma^2)$ and its contour plot are well described in \figref{Figure:CCSMcLeishCDFA} and \figref{Figure:CCSMcLeishCDFB}, respectively. For a given contour value $c\!\in\![0,1)$, the contours, presented in \figref{Figure:CCSMcLeishCDFB}, can be obtained by
\begin{equation}
    \nonumber
    (z|c)=\bigl\{z=\widehat{\xi}\exp(\imaginary\theta)\,\,
        \bigl|\,\,
        \theta\in[-\pi,\pi),\text{~and~} 
        \widehat{\xi}=
            \argmin_{\xi\in\mathbb{R}_{+}}
                \lVert
                    {F^{2}_{Z}(\xi)-c}
                \rVert^2
        \bigr.
    \bigr\}.{~~~}
\end{equation}

\begin{figure}[tp] 
\centering
\begin{subfigure}{0.7\columnwidth}
    \centering
    \includegraphics[clip=true, trim=0mm 0mm 0mm 0mm,width=1.0\columnwidth,height=0.85\columnwidth]{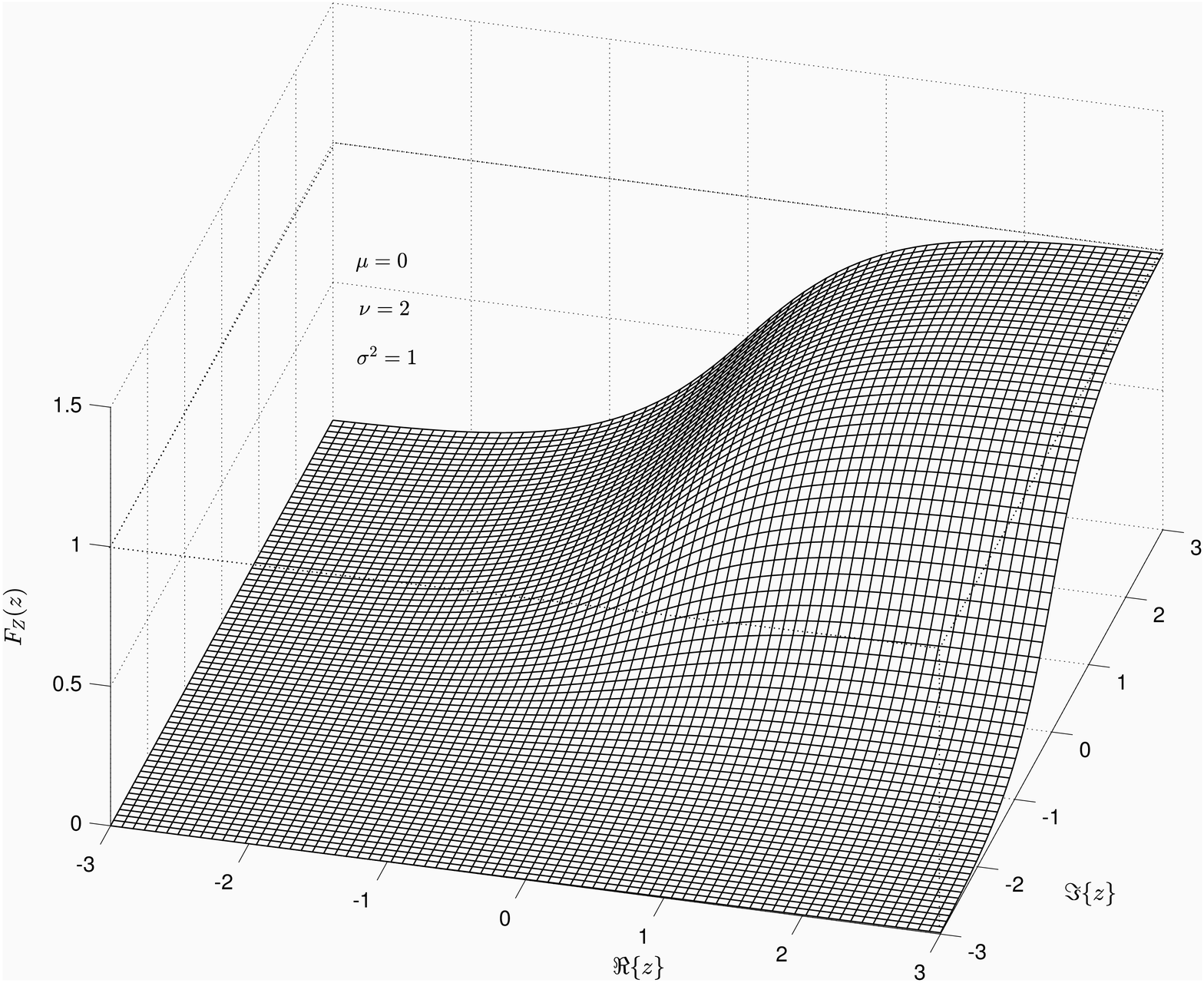}
    \caption{\ac{CDF} illustration in complex space.}
    \vspace{5mm}
    \label{Figure:CCSMcLeishCDFA}
\end{subfigure}
~~~
\begin{subfigure}{0.7\columnwidth}
    \centering
    \includegraphics[clip=true, trim=0mm 0mm 0mm 0mm,width=1.0\columnwidth,height=0.85\columnwidth]{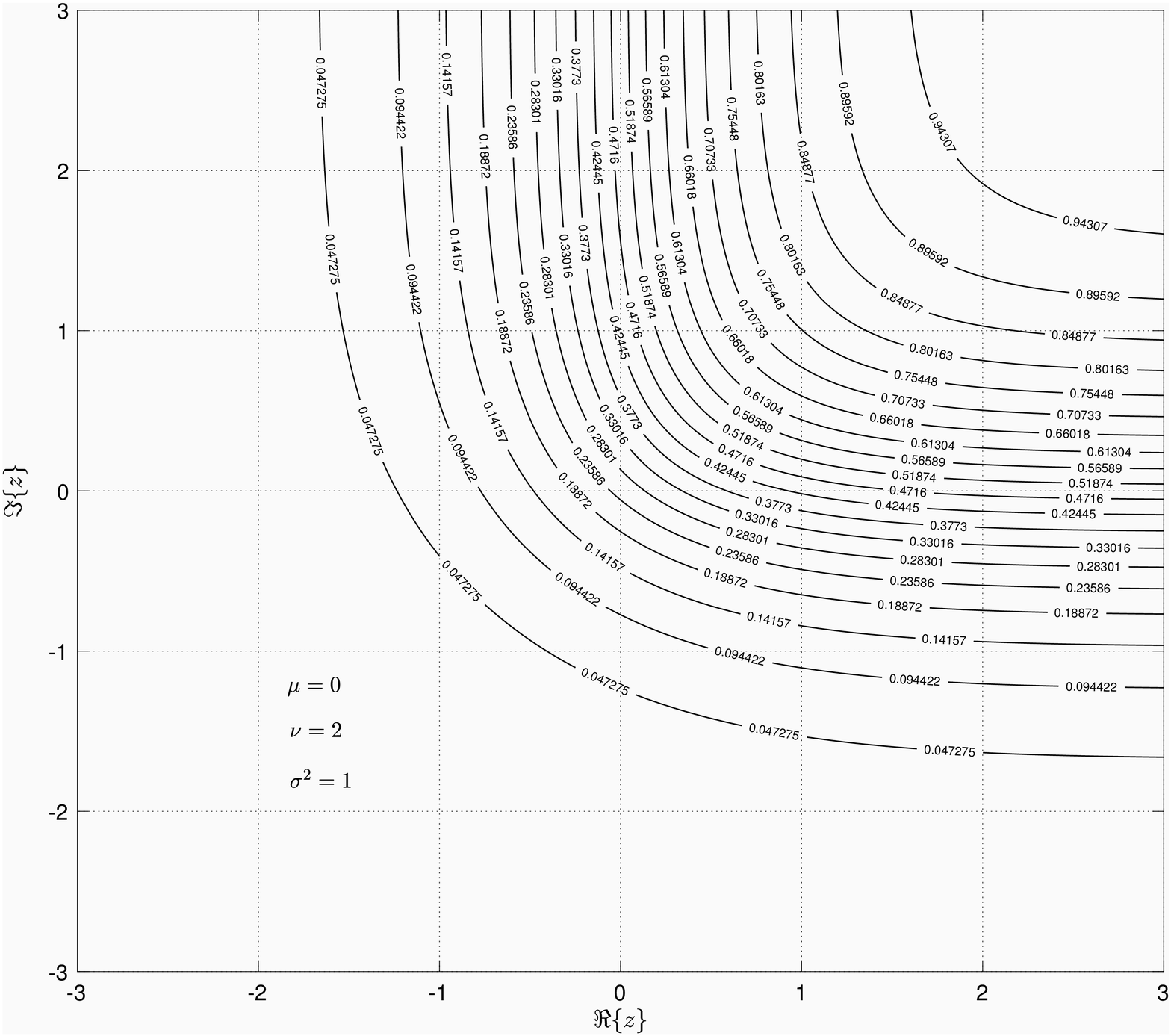}
    \caption{\ac{CDF} contour curves.}
    \label{Figure:CCSMcLeishCDFB}
\end{subfigure}
\caption{The \ac{CDF} and contour of $\mathcal{CM}_{\nu}(0,\sigma^2)$ (i.e., the illustration of \eqref{Eq:CCSMcLeishCDF} for $\mu\!=\!0$).}
\label{Figure:CCSMcLeishCDF}
\vspace{-2mm} 
\end{figure} 
\begin{theorem}\label{Theorem:CCSMcLeishMGF}
Under the condition of being \ac{CCS}, the \ac{MGF} of $Z\!\sim\!\mathcal{CM}_{\nu}(\mu,\sigma^2)$ is given by
\begin{equation}\label{Eq:CCSMcLeishMGF}
	M_{Z}(s)={e}^{-\langle{s,\mu}\rangle}\Bigl(1-\scalemath{0.9}{0.9}{\frac{\lambda^2}{8}}\langle{s,s}\rangle\Bigr)^{-\nu},
\end{equation}
where $s\!=\!{s}_{X}+\imaginary{s}_{Y}\in\mathbb{C}$ within the existence region $s\!\in\!\mathbb{C}_{0}$, 
and the region $\mathbb{C}_{0}$ is given by 
$\mathbb{C}_0=\bigl\{
        s\,    
        \bigl|\,
        \bigl\langle{s,s}\bigr\rangle
        \leq
        {8}/{\lambda^2}
        \bigr.
        \bigr\}$.
\end{theorem}

\begin{proof}
The \ac{MGF} of $Z$ (i.e., the joint \ac{MGF} $M_{\defrmat{Z}}(s_{X},s_{Y})$
of $\defrmat{Z}$) is defined as $M_{Z}(s)=\Expected{\exp(-\langle{s,Z}\rangle)}$,
where utilizing \theoremref{Theorem:CCSMcLeishDefinition} yields  
\begin{equation}\label{Eq:CCSMcLeishMGFIntegral}
	M_{Z}(s)=e^{\langle{s,\mu}\rangle}\int_{0}^{\infty}\Expected{\exp(-\langle{s,\sqrt{g}Z_0}\rangle)}\,f_{G}(g)dg,
\end{equation}
where $\Expected{\exp(-\langle{s,\sqrt{g}Z_0}\rangle)}$ is the \ac{MGF} of $\mathcal{CN}(0,g\sigma^2)$ given by
$\exp(-g\frac{\sigma^2}{4}\langle{s,s}\rangle)$
 \cite{BibHandersenHojbjerreSorensenEriksenBook1995,BibKrishnamoorthyBook,BibForbesEvansHastingsPeacockBook,BibZwillingerKokoskaBook}.
Then, substituting \eqref{Eq:ProportionPDF} into \eqref{Eq:CCSMcLeishMGFIntegral}, we have
\begin{equation}\label{Eq:CCSMcLeishMGFIntegralII}
		M_{Z}(s)=\frac{\nu^{\nu}e^{\langle{s,\mu}\rangle}}{\Gamma(\nu)}
			\int_{0}^{\infty}g^{\nu-1}{e}^{-g\nu\bigl(1-{\lambda^2}\langle{s,s}\rangle/8\bigr)}\,dg.
\end{equation}
Consequently, utilizing $\int_{0}^{\infty}x^{a-1}\exp(-bx)\!=\!{b}^{-a}\Gamma(a)$ for any
$\RealPart{a},\RealPart{b}\!>\!0$\cite[Eq.\!~(3.381/4)]{BibGradshteynRyzhikBook}, and correspondingly in a certain  existence region $1-\lambda^2\langle{s,s}\rangle/{8}\!>\!{0}$, we simplify \eqref{Eq:CCSMcLeishMGFIntegralII} into \eqref{Eq:CCSMcLeishMGF}, which proves \theoremref{Theorem:CCSMcLeishMGF}.  
\end{proof}

For consistency, setting $\nu\!\rightarrow\!{0}$ simplifies \eqref{Eq:CCSMcLeishMGF} into the \ac{MGF} of Dirac's distribution, that is $M_{Z}(s)\!=\!\exp(-\langle{s,\mu}\rangle)$. Further, setting $\nu\!=\!{1}$ simplifies \eqref{Eq:CCSMcLeishMGF} into the \ac{MGF} of $\mathcal{CL}(\mu,\sigma^2)$, that is  $M_{Z}(s)\!=\!{e}^{-\langle{s,\mu}\rangle}(1-{\sigma^2}\langle{s,s}\rangle/{4})^{-1}$.
In addition, setting the limit $\nu\!\rightarrow\!\infty$  on \eqref{Eq:CCSMcLeishMGF} and applying \cite[Eq.\!~(1.211/4)]{BibGradshteynRyzhikBook} results in $M_{Z}(s)\!=\!\exp(-\langle{s,\mu}\rangle-\frac{1}{4}\sigma^2\langle{s,s}\rangle)$, which is the \ac{MGF} of $\mathcal{CN}(\mu,\sigma^2)$ \cite{BibHandersenHojbjerreSorensenEriksenBook1995,BibKrishnamoorthyBook,BibForbesEvansHastingsPeacockBook,BibZwillingerKokoskaBook} as expected. 

For the purpose of achieving statistical characterization, the moment of $Z\!\sim\!\mathcal{CM}_{\nu}(\mu,\sigma^2)$ (i.e., the joint moment $\Expected{{X}_1^m{X}_2^n}$ of $\defrmat{Z}\!=\![{X}_1,{X}_2]^T$ as referring to \eqref{Eq:ComplexMcLeishDefinition}, where $m\!\in\!\mathbb{N}$ and $n\!\in\!\mathbb{N}$) are needed in a closed form, and for which the \ac{MGF} is a very useful instrument\cite[Eqs.\!~(3.79)\!~and\!~(3.80)]{BibZwillingerKokoskaBook} as follows
\begin{equation}\label{Eq:CCSMcLeishMomentsUsingJointMGF}
	\Expected{X_{1}^mX_{2}^n}=
	    \biggl.{(-1)}^{m+n}
	        \frac{\partial^{m+n}}
	            {\partial{s_{1}}^{m}\partial{s_{2}}^{n}}M_{Z}(s)
		            \biggr|_{\substack{s_{1}\rightarrow{0}\\s_{2}\rightarrow{0}}}
\end{equation}
where $s\!=\!{s}_{1}+\imaginary{s}_{2}\!\in\!\mathbb{C}$. Hence, replacing \eqref{Eq:CCSMcLeishMGF} into \eqref{Eq:CCSMcLeishMomentsUsingJointMGF} and thereon applying two times the Leibniz's higher order derivative rule\cite[Eq.\!~(0.42)]{BibGradshteynRyzhikBook} yields \eqref{Eq:CCSMcLeishMoments} as shown below.

\begin{theorem}\label{Theorem:CCSMcLeishMoments}
Under the condition of being \ac{CCS}, the joint moment $\Expected{{X}_1^m{X}_2^n}$, $m,n\!\in\!\mathbb{N}$, of $Z\!\sim\!\mathcal{CM}_{\nu}(\mu,\sigma^2)$ is given as referring to \eqref{Eq:ComplexMcLeishDefinition} by
\begin{equation}\label{Eq:CCSMcLeishMoments}
\!\!\!\!\Expected{{X}_1^m{X}_2^n}=\mu_{1}^m\mu_{2}^n
		\sum_{k=0}^{m}
		\sum_{l=0}^{n}\!
			\Binomial{m}{k}\!\Binomial{n}{l}\,
			\Xi_{k,l}
			     \frac{\lambda^{k+l}}{\mu_{1}^k\mu_{2}^l}\iseven{k,l},\!\!
\end{equation}
where $\iseven{k,l}\!=\!\iseven{k}\iseven{l}$, and the weight $\Xi_{k,l}$~is~defined~as
\begin{equation}
\Xi_{k,l}=\sqrt{2^{k+l}}\,{\textstyle
			\Pochhammer{\frac{1}{2}}{{k}/{2}}\,
			\Pochhammer{\frac{1}{2}}{{l}/{2}}\,
			\Pochhammer{\nu}{{(k+l)}/{2}}\,}.
\end{equation}
where $\Pochhammer{a}{n}\!=\!a(a+1)\cdots(a+n-1)$ denotes Pochhammer's symbol (or shifted factorial) \emph{\cite{BibGradshteynRyzhikBook,BibAbramowitzStegunBook}}.
\end{theorem}

\begin{proof}
Based on \theoremref{Theorem:CCSMcLeishDefinition}, the joint moment $\Expected{{X}_1^m{X}_2^n}$ can be
readily rewritten as
\begin{equation}\label{Eq:ComplexMcLeishMomentIntegral}
	\Expected{{X}_1^m{X}_2^n}=\mathbb{E}\bigl[(\sqrt{G}X_0-\mu_{1})^m(\sqrt{G}Y_0-\mu_{2})^n\bigr].
\end{equation}
Afterwards, applying binomial expansion on \eqref{Eq:ComplexMcLeishMomentIntegral}, we have
\begin{equation}
	\Expected{{X}_1^m{X}_2^n}=
		\mu_{1}^{m}\mu_{2}^{n}
		\sum_{k=0}^{m}\sum_{l=0}^{n}
		\Binomial{m}{k}
		\Binomial{n}{l}
		\frac{1}{\mu_{1}^{k}\mu_{2}^{l}}
		\mathbb{E}\bigl[G^{\frac{k+l}{2}}\bigr]
		\mathbb{E}\bigl[X_{0}^{k}\bigr]
		\mathbb{E}\bigl[Y_{0}^{l}\bigr],
\end{equation}
where substituting \cite[Eq.\!~(2.3-20)]{BibProakisBook} and \cite[Eq.\!~(2.23)]{BibAlouiniBook} 
\setlength\arraycolsep{1.4pt}
\begin{eqnarray}
\mathbb{E}\bigl[X_{0}^{n}\bigr]
    &=&\mathbb{E}\bigl[Y_{0}^{n}\bigr]
        =\frac{\Gamma({1}/{2}+n)}{2\Gamma({1}/{2})}\sigma^2\,\iseven{n}\\
\mathbb{E}\bigl[G^{n}\bigr]
    &=&\frac{\Gamma(\nu+n)}{\Gamma(\nu)\nu^{n}},
\end{eqnarray}
and then performing simple algebraic manipulations results in \eqref{Eq:CCSMcLeishMoments}, which proves
\theoremref{Theorem:CCSMcLeishMoments}.
\end{proof}

\subsection{Complex and Elliptically-Symmetric McLeish Distribution}
\label{Section:StatisticalBackground:CESMcLeishDistribution}
The bivariate Gaussian \ac{PDF} has several beneficial and elegant properties and, for this reason, it is a conventionally used model in the literature. Regarding this fact while to have more than the previous subsection, we infer many such properties, so let us consider a more generalized case, i.e., that the  mixture $Z\!=\!{X}_1+\imaginary{X}_2$ follows a \ac{CES} distribution whose inphase ${X}_1\!\sim\!\mathcal{M}_{\nu_{1}}(\mu_{1},\sigma^2_{1})$ and quadrature ${X}_2\!\sim\!\mathcal{M}_{\nu_{2}}(\mu_{2},\sigma^2_{2})$ are \ac{c.i.d.} two distributions correlated by $\rho\!\in\![-1,1]$. It is denoted by $Z\!\sim\!\allowbreak\mathcal{EM}_{\nu}(\mu,\sigma^2,\rho)$, the mean is $\mu\!=\!\mu_{1}+\imaginary\mu_{2}$, the normality is $\nu\!=\!\nu_{1}\!=\!\nu_{2}$, 
the variance is $\sigma^2\!=\!{2}\sigma^2_{1}\!=\!2\sigma^2_{2}$, and the correlation is  
\begin{equation}
	\rho=\frac{\Covariance{{X}_1}{{X}_2}}{\sqrt{\Variance{{X}_1}\Variance{{X}_2}}}
	    =\frac{2}{\sigma^2}(\mathbb{E}[{X}_1{X}_2]-\mu_{{X}_1}\mu_{{X}_2}).
\end{equation}
Accordingly, we present the definition of the \ac{CES} MacLeish distribution in the following theorem.

\begin{theorem}\label{Theorem:CESMcLeishDefinition}
Under the condition of being \ac{CES}, the definition of $Z\!\sim\!\mathcal{EM}_{\nu}(\mu,\sigma^2,\rho)$ can be decomposed as
\begin{subequations}
\setlength\arraycolsep{1.4pt}
\begin{eqnarray}\label{Eq:CESMcLeishDistributionDefinition}
\label{Eq:CESMcLeishDistributionDefinitionA}
Z&=&\sqrt{G}{Z_0}+\mu,\\
\label{Eq:CESMcLeishDistributionDefinitionB}
&=&\sqrt{G}\bigl({X_0}+\imaginary({\rho{X}_{0}+\sqrt{1-\rho^2}Y_0})\bigr)+\mu,
\end{eqnarray}
\end{subequations}
where $Z_{0}\!\sim\!\mathcal{EN}(0,\sigma^2,\rho)$,  
$X_{0}\!\sim\!\mathcal{N}(0,\sigma^2_{1})$ and $Y_{0}\!\sim\!\mathcal{N}(0,\sigma^2_{2})$.
$X_{0}$ and $Y_{0}$ are independent and identically distributed random distributions (i.e., $2\sigma^2_{1}\!=\!2\sigma^2_{2}\!=\!\sigma^2$). Further, ${G}\!\sim\!Gamma(\nu,1)$.
\end{theorem} 

\begin{proof}
Referring to \theoremref{Theorem:CCSMcLeishDefinition}, the correlation between the inphase and quadrature of $Z\!\sim\!\mathcal{EM}_{\nu}(\mu,\sigma^2,\rho)$ is certainly determined by that between the inphase and quadrature of $Z_{0}\!\sim\!\mathcal{EN}(0,\sigma^2,\rho)$. For a certain correlation
$\rho\!\in\![-1,1]$, the inphase and quadrature of $Z_{0}\!\sim\!\mathcal{EN}(0,\sigma^2,\rho)$ are respectively written as
\setlength\arraycolsep{1.4pt}
\begin{eqnarray}
	\RealPart{Z_{0}}&=&{X}_{0} \\
	\ImagPart{Z_{0}}&=&\rho{X}_{0}+\sqrt{1-\rho^2}{Y}_{0},
\end{eqnarray}
such that $\Covariance{\RealPart{Z_{0}}}{\ImagPart{Z_{0}}}\!=\!\rho\,\sigma^2/2$ and $\Variance{\RealPart{Z_{0}}}\!=\!\Variance{\ImagPart{Z_{0}}}\!=\!\sigma^2/2$.
Accordingly, the correlation between the inphase and quadrature of $Z\!\sim\!\mathcal{EM}_{\nu_{Z}}(\mu_{Z},\sigma^2_{Z},\rho)$
is written in terms of that between $\RealPart{Z_{0}}$ and $\ImagPart{Z_{0}}$, that is
\begin{equation}
	 \!\!\!\!\rho=\frac{\Covariance{{X}_1}{{X}_2}}{\sqrt{\Variance{{X}_1}\Variance{{X}_2}}}=
		\frac{\Covariance{\RealPart{Z_{0}}}{\ImagPart{Z_{0}}}}{\sqrt{\Variance{\RealPart{Z_{0}}}\Variance{\ImagPart{Z_{0}}}}}.\!\!
\end{equation}
Accordingly, the proof is obvious.
\end{proof}

With the aid of \theoremref{Theorem:CESMcLeishDefinition}, the \ac{PDF} of $Z$ is given in the following theorem.
 
\begin{theorem}\label{Theorem:CESMcLeishPDF}
Under the condition of being \ac{CES}, the \ac{PDF} of $Z\!\sim\!\mathcal{EM}_{\nu}(\mu,\sigma^2,\rho)$ is given by
\begin{equation}\label{Eq:CESMcLeishPDF}
\!\!f_{Z}(z)=\frac{2}{\pi\Gamma(\nu)}
		\frac{\abs{z-\mu}_{\rho}^{\nu-1}}{\sqrt{1-\rho^2}\,\lambda^{\nu+1}}\BesselK[\nu-1]{\frac{2\abs{z-\mu}_{\rho}}{\lambda}}\!\!
\end{equation}
defined over $z\!\in\!\mathbb{C}$, where the deviation factor $\lambda=\sqrt{2\sigma^2/\nu}$. 
\end{theorem}

\begin{proof}
With the aid of \eqref{Eq:CESMcLeishPDF}, the \ac{PDF} of $Z$ conditioned on $G$ is readily written as \cite[Eq.\!~(2.3-78)]{BibProakisBook}  
\begin{equation}\label{Eq:CESMcLeishConditionedPDF}
	f_{Z|G}(z|g)=\frac{1}{\pi{g}\sqrt{1-\rho^2}\sigma^2}
		\exp\Bigl(-\scalemath{0.95}{0.95}{\frac{\abs{z-\mu}^2_{\rho}}{g\sigma^2}}\Bigr),
\end{equation}
for $g\!\in\!\mathbb{R}_{+}$, where setting the correlation $\rho\!=\!0$ yields 
into \eqref{Eq:ConditionedCCSMcLeishPDF} as expected. Accordingly, 
the \ac{PDF} of $Z$ can be expressed as $f_{Z}(z)\!=\!\int_{0}^{\infty}f_{Z|G}(z|g)f_{G}(g)dg$. 
Then, the proof is obvious following the same steps 
in the proof of \theoremref{Theorem:CCSMcLeishPDF}.
\end{proof} 

The \ac{PDF} contour curves of $Z\!\sim\!\mathcal{CM}_{\nu}(\mu,\sigma^2)$ are clearly illustrated in \figref{Figure:CESMcLeishPDF} for $\rho\!=\!\pm{3}/{4}$. In addition to them, let us consider the consistency of \eqref{Eq:CESMcLeishPDF}. Setting the correlation $\rho\!=\!0$ yields \eqref{Eq:CCSMcLeishPDF} as expected. Furthermore, setting $\nu\!=\!1$ reduces \eqref{Eq:CESMcLeishPDF} to the \ac{PDF} of \ac{CES} Laplacian distribution, and equivalently so does $\nu\!\rightarrow\!\infty$ to the \ac{PDF} of the bivariate correlated Gaussian distribution\cite[Eq.\!~(2.3-78)]{BibProakisBook}, whose inphase and quadrature are mutually correlated with $\rho\!\neq\!0$. In contrast to the evidence that zero correlation implies independence between Gaussian distributions, the two uncorrelated McLeish distributions are not independent of each other unless $\nu\!\rightarrow\!\infty$.
Eventually, having treated the correlation, it is useful to define the McLeish's bivariate \ac{Q-function} with the aid of \eqref{Eq:CESMcLeishPDF}. 
\begin{figure}[tp] 
\centering
\begin{subfigure}{0.7\columnwidth}
    \centering
    \includegraphics[clip=true, trim=0mm 0mm 0mm 0mm, width=1.0\columnwidth,height=0.85\columnwidth]{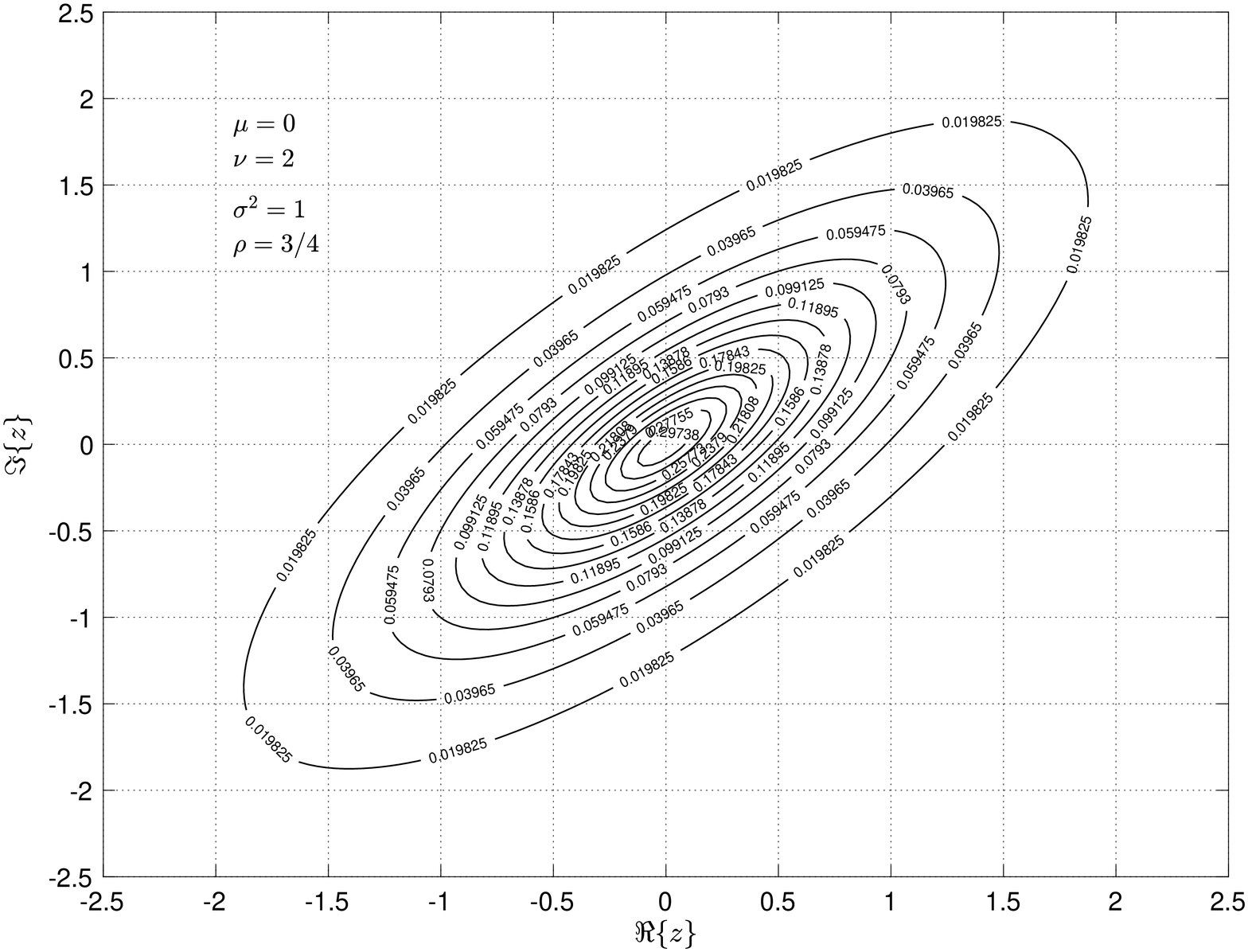} 
    \caption{The contour curves for $\rho\!=\!3/4$.}
    \vspace{5mm}
    \label{Figure:CESMcLeishPDFA}
\end{subfigure}
~~~
\begin{subfigure}{0.7\columnwidth}
    \centering
    \includegraphics[clip=true, trim=0mm 0mm 0mm 0mm, width=1.0\columnwidth,height=0.85\columnwidth]{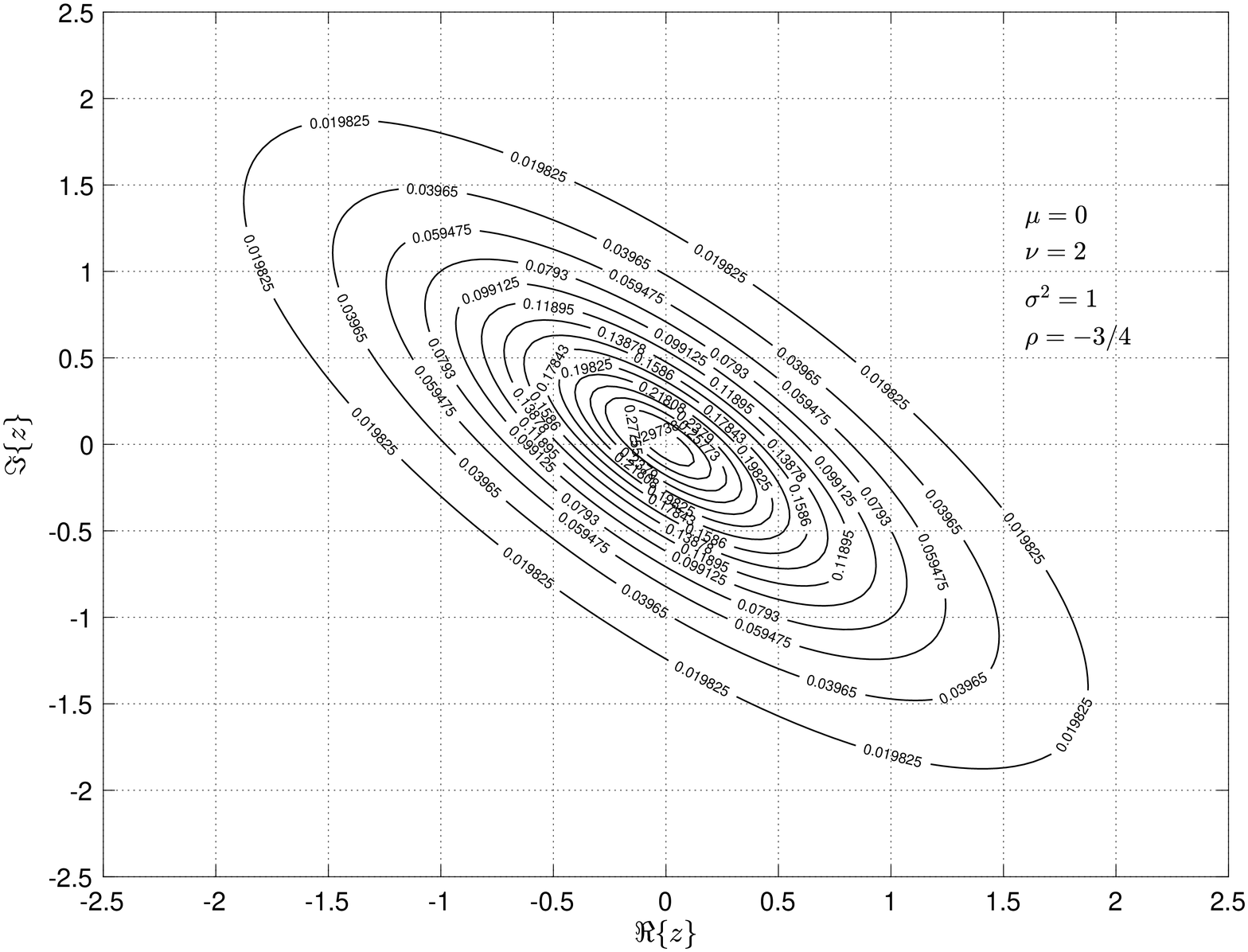}
    \caption{The contour curves for $\rho\!=\!-3/4$.}
    \label{Figure:CESMcLeishPDFB}
\end{subfigure}
\caption{The \ac{PDF} contour curves of $\mathcal{EM}_{\nu}(0,\sigma^2,\rho)$ (i.e., the illustration of \eqref{Eq:CESMcLeishPDF} for $\mu\!=\!0$).}
\label{Figure:CESMcLeishPDF}
\vspace{-2mm} 
\end{figure}

\begin{definition}[McLeish's Bivariate Quantile]\label{Definition:McLeishBivariateQFunction}
The McLeish's bivariate \ac{Q-function} is defined for $x\!\in\!\mathbb{R}$ and $y\!\in\!\mathbb{R}$ by
\begin{equation}\label{Eq:McLeishBivariateQFunction}
Q_{\nu}(x,y,\rho)=\int_{x}^{\infty}\int_{y}^{\infty}
\frac{2}{\pi\Gamma(\nu)}
	\frac{\abs{z_{\ell}}_{\rho}^{\nu-1}}{\sqrt{1-\rho^2}\,\lambda_{0}^{\nu+1}}
	    \BesselK[\nu-1]{\frac{2\abs{z_{\ell}}_{\rho}}{\lambda_{0}}}\,dx_{\ell}\,dy_{\ell},
\end{equation}
where $z_{\ell}\!=\!{x}_{\ell}+\imaginary{y}_{\ell}\in\mathbb{C}$.
\end{definition}
\begin{theorem}\label{Theorem:CESMcLeishCDF}
Under the condition of being \ac{CES}, the \ac{CDF} of $Z\!\sim\!\mathcal{EM}_{\nu_{Z}}(\mu,\sigma^2,\rho)$ is given by
\begin{subequations}\label{Eq:CESMcLeishCDF}
\begin{multline}\label{Eq:CESMcLeishCDFA} 
	\!\!F_{Z}(z)=1
		-Q_{\nu}\Bigl(\sqrt{2}\Bigl\langle{1,\frac{z-\mu}{\sigma}}\Bigr\rangle\Bigr)
		-Q_{\nu}\Bigl(\sqrt{2}\Bigl\langle{\imaginary,\frac{z-\mu}{\sigma}}\Bigr\rangle\Bigr)\\
		+Q_{\nu}\Bigl(
			\sqrt{2}\Bigl\langle{1,\frac{z-\mu}{\sigma}}\Bigr\rangle,
			\sqrt{2}\Bigl\langle{\imaginary,\frac{z-\mu}{\sigma}}\Bigr\rangle,
			\rho\Bigr),
\end{multline}
for the upper right quadrant (i.e., $\RealPart{z}\!\geq\!0$ and $\ImagPart{z}\!\geq\!0$);
\begin{equation}\label{Eq:CESMcLeishCDFB} 
	\!\!F_{Z}(z)=
		Q_{\nu}\Bigl(\sqrt{2}\Bigl\langle{1,\frac{\mu-z}{\sigma}}\Bigr\rangle\Bigr)
		-Q_{\nu}\Bigl(
			\sqrt{2}\Bigl\langle{1,\frac{\mu-z}{\sigma}}\Bigr\rangle,
			\sqrt{2}\Bigl\langle{\imaginary,\frac{z-\mu}{\sigma}}\Bigr\rangle,
			\rho\Bigr),
\end{equation}
for the upper left quadrant (i.e., $\RealPart{z}\!<\!0$ and $\ImagPart{z}\!\geq\!0$);
\begin{equation}\label{Eq:CESMcLeishCDFC} 
	\!\!F_{Z}(z)=
		Q_{\nu}\Bigl(
			\sqrt{2}\Bigl\langle{1,\frac{\mu-z}{\sigma}}\Bigr\rangle,
			\sqrt{2}\Bigl\langle{\imaginary,\frac{\mu-z}{\sigma}}\Bigr\rangle,
			\rho\Bigr),
\end{equation}
for the lower left quadrant (i.e., $\RealPart{z}\!<\!0$ and $\ImagPart{z}\!<\!0$);
\begin{equation}\label{Eq:CESMcLeishCDFD} 
	\!\!F_{Z}(z)=
		Q_{\nu}\Bigl(\sqrt{2}\Bigl\langle{\imaginary,\frac{\mu-z}{\sigma}}\Bigr\rangle\Bigr)
		-Q_{\nu}\Bigl(
			\sqrt{2}\Bigl\langle{1,\frac{z-\mu}{\sigma}}\Bigr\rangle,
			\sqrt{2}\Bigl\langle{\imaginary,\frac{\mu-z}{\sigma}}\Bigr\rangle,
			\rho\Bigr),
\end{equation}
\end{subequations}
for the lower right quadrant (i.e., $\RealPart{z}\!\geq\!0$ and $\ImagPart{z}\!<\!0$).
\end{theorem}

\begin{proof}
Note that the \ac{CDF} of $Z_{0}\!\sim\!\mathcal{EN}(0,\sigma^2_{Z},\rho)$ is defined by $F_{Z_0}(z_{\ell}|\sigma_{Z})\!=\!\Pr\{{X_{0}\leq\langle{1,z_{\ell}}\rangle}\,\cap\,{Y_{0}\leq\langle{\imaginary,z_{\ell}}\rangle}\,|\,\sigma_{Z}\}$ conditioned on $\sigma_{Z}$ and expressed for~a~certain~$z\!=\!{x}+\imaginary{y}\in\mathbb{C}$~as
\begin{equation}\label{Eq:ComplexGaussianCCDF}
	F_{Z_{0}}(z|\sigma_{Z})=
		\int_{-\infty}^{x}
			\int_{-\infty}^{y}
				\frac{\exp\bigl(-{{\langle{z_{\ell},z_{\ell}}\rangle}_{\rho}}/{\sigma^2}\bigr)}{\pi\sigma^2\sqrt{1-\rho^2}}\,dx_{\ell}\,dy_{\ell},
\end{equation}
where $z_{\ell}\!=\!{x}_{\ell}+\imaginary{y}_{\ell}\in\mathbb{C}$. Utilizing \cite[Eqs.\!~(2.3-10)\!~and\!~(2.3-11)]{BibProakisBook}
and \cite[Eqs.\!~(4.3)]{BibAlouiniBook} with $\langle{1,z}\rangle\!=\!\RealPart{z}$ and $\langle{\imaginary,z}\rangle\!=\!\ImagPart{z}$, \eqref{Eq:ComplexGaussianCCDF} simplifies for the quadrants of complex plane, that is 
\begin{subequations}\label{Eq:CESGaussianCDF} 
\begin{equation}\label{Eq:CESGaussianCDFA} 
	\!\!{F}_{Z_0}(z|\sigma)=1-
		Q\bigl(\sqrt{2}\bigl\langle{1,{z}/{\sigma}}\bigr\rangle\bigr)-
		Q\bigl(\sqrt{2}\bigl\langle{\imaginary,{z}/{\sigma}}\bigr\rangle\bigr)+
		Q\bigl(\sqrt{2}\bigl\langle{1,{z}/{\sigma}}\bigr\rangle,\sqrt{2}\bigl\langle{\imaginary,{z}/{\sigma}}\bigr\rangle,\rho\bigr),
\end{equation}
for the upper right quadrant (i.e., $\RealPart{z}\!\geq\!0$ and $\ImagPart{z}\!\geq\!0$);
\begin{equation}\label{Eq:CESGaussianCDFB}
	\!\!{F}_{Z_0}(z|\sigma)=
		Q\bigl(\sqrt{2}\bigl\langle{\imaginary,{z}/{\sigma}}\bigr\rangle\bigr)
		-Q\bigl(-\sqrt{2}\bigl\langle{1,{z}/{\sigma}}\bigr\rangle,\sqrt{2}\bigl\langle{\imaginary,{z}/{\sigma}}\bigr\rangle,\rho\bigr),
\end{equation}
for the upper left quadrant (i.e., $\RealPart{z}\!<\!0$ and $\ImagPart{z}\!\geq\!0$);
\begin{equation}\label{Eq:CESGaussianCDFC}
	\!\!{F}_{Z_0}(z|\sigma)=Q\bigl(-\sqrt{2}\bigl\langle{1,{z}/{\sigma}}\bigr\rangle,-\sqrt{2}\bigl\langle{\imaginary,{z}/{\sigma}}\bigr\rangle,\rho\bigr),
\end{equation}
for the lower left quadrant (i.e., $\RealPart{z}\!<\!0$ and $\ImagPart{z}\!<\!0$);
\begin{equation}\label{Eq:CESGaussianCDFD}
	\!\!{F}_{Z_0}(z|\sigma)=
		Q\bigl(\sqrt{2}\bigl\langle{1,{z}/{\sigma}}\bigr\rangle\bigr)
		-Q\bigl(\sqrt{2}\bigl\langle{1,{z}/{\sigma}}\bigr\rangle,-\sqrt{2}\bigl\langle{\imaginary,{z}/{\sigma}}\bigr\rangle,\rho\bigr),
\end{equation}
\end{subequations}
for the lower right quadrant (i.e., $\RealPart{z}\!\geq\!0$ and $\ImagPart{z}\!<\!0$). Accordingly, referring to \eqref{Eq:CESMcLeishDistributionDefinition}, the \ac{CDF} of $Z\!\sim\!\mathcal{CM}_{\nu}(\mu,\sigma^2,\rho)$ is explicitly written as $F_{Z}(z)\!=\!\int_{0}^{\infty}F_{Z_0}(z-\mu|\sqrt{g}\sigma)f_{G}(g)\,dg$. With the aid of \theoremref{Theorem:CCSMcLeishCDF}, we rewrite \eqref{Eq:McLeishQFunction} and \eqref{Eq:McLeishBivariateQFunction} as
\setlength\arraycolsep{1.4pt}
\begin{eqnarray}
	\label{Eq:McLeishQFunctionIntegralWithCorrelation}
	Q_{\nu}(x)&=&\int_{0}^{\infty}Q\bigl(\sqrt{2g}x\bigr)f_{G}(g)\,dg,\\
	\label{Eq:McLeishBivariateQFunctionIntegralWithCorrelation}
	Q_{\nu}(x,y,\rho)&=&\int_{0}^{\infty}Q\bigl(\sqrt{2g}x,\sqrt{2g}y,\rho\bigr)f_{G}(g)\,dg,
\end{eqnarray}
the \ac{CDF} $F_{Z}(z)$ is readily obtained as \eqref{Eq:CESMcLeishCDF}, which completes the proof of \theoremref{Theorem:CESMcLeishCDF}.
\end{proof}

\begin{theorem}\label{Theorem:CESMcLeishMGF}
Under the condition of being \ac{CES}, the \ac{MGF} of $Z\!\sim\!\mathcal{EM}_{\nu}(\mu,\sigma^2,\rho)$ is given by
\begin{equation}\label{Eq:CESMcLeishMGF}
    M_{Z}(s)={e}^{-\langle{s,\mu}\rangle}\Bigl(1-\frac{\lambda^2}{8}(1-\rho^2)\bigl\langle{s,s}\bigr\rangle_{-\rho}\Bigr)^{-\nu},
\end{equation}
where $s\!=\!{s}_{X}+\imaginary{s}_{Y}\in\mathbb{C}$ within the existence region $s\!\in\!\mathbb{C}_{0}$, and the region $\mathbb{C}_{0}$ is given by 
\begin{equation}\label{Eq:CESMcLeishMGFExistenceRegion}
    \mathbb{C}_0=\Bigl\{
        s\,    
        \Bigl|\,
        {\lambda^2} 
        (1-\rho^2)
        \bigl\langle{s,s}\bigr\rangle_{-\rho}
        \leq
        {8}
        \Bigr.
        \Bigr\}.
\end{equation}
\end{theorem}

\begin{proof}
Note that, referring to \theoremref{Theorem:CESMcLeishDefinition}, the \ac{MGF} of $Z\!\sim\!\mathcal{EM}_{\nu}(\mu,\sigma^2,\rho)$ conditioned on $G$ is written as 
\begin{equation}
    M_{Z|G}(s|g)=\exp\Bigl(
        -\langle{s,\mu}\rangle
        +\frac{g}{4}\sigma^2(1-\rho^2)\langle{s,s}\rangle_{-\rho}\Bigr).
\end{equation}
Then, performing the almost same steps followed in the proof of \theoremref{Theorem:CCSMcLeishMGF}, the \ac{MGF} of $Z\!\sim\!\mathcal{EM}_{\nu}(\mu,\sigma^2,\rho)$ is obtained as \eqref{Eq:CESMcLeishMGF}, which completes the proof of \theoremref{Theorem:CESMcLeishMGF}.
\end{proof}

\subsection{Multivariate McLeish Distribution}
\label{Section:StatisticalBackground:MultivariateMcLeishDistribution}
In this subsection, we deal with random vectors instead~of just individual random distributions. We define~multivariate McLeish distribution and derive its statistical characterization, where we begin with~a~vector~of~independent McLeish distributions and work ourselves up to the general case where they are no longer~mutually~independent. Let us start with a vector that consists of \emph{uncorrelated and identically distributed} random distributions of the same family, that is
\begin{equation}\label{Eq:StandardMultivariateMcLeishRandomVector}
    \defrmat{S}=[S_1,S_2,\ldots,S_L]^T,
\end{equation}
where $S_\ell$ denotes a random distribution with zero mean and unit variance, i.e., $\mathbb{E}[S_\ell]\!=\!{0}$ and $\mathbb{V}[S_\ell]\!=\!{1}$,  ${1}\!\leq\!\ell\!\leq\!{L}$ such that any pair of ${S}_k$ and ${S}_\ell$, $k\!\neq\!\ell$ must be uncorrelated (i.e., $\mathbb{E}[S_{k}S_{\ell}]\!=\!{0}$). Hence, the mean vector $\defvec{\mu}\!=\!\mathbb{E}[\defrmat{S}]$ is given by 
\begin{equation}
	\defvec{\mu}=[0,0,\ldots,0]^{T},
\end{equation}
and the covariance matrix $\defmat{\Sigma}\!=\!\mathbb{E}[\defrmat{S}\defrmat{S}^T]$ is given by
\begin{equation}
	\defmat{\Sigma}=
     \begin{bmatrix}
    	1      & 0      & \dots  & 0      \\
    	0      & 1      & \dots  & 0      \\
    	\vdots & \vdots & \ddots & \vdots \\
    	0      & 0      & \dots  & 1
	\end{bmatrix}.
\end{equation}
By definition of standard multivariate distribution \cite{BibAndersenBook1995,BibTongBook1989,BibWangKotzNgBook1989,BibBilodeauBrennerBook1999,BibKotzBalakrishnanBook2004}, $\defrmat{S}$ follows a standard multivariate distribution with zero mean vector and unit covariance matrix iff $\forall\defvec{a}\in\mathbb{R}^{L}$, $\defvec{a}^{T}\defrmat{S}$ follows a random distribution of the same family with zero mean and $\defvec{a}^T\defvec{a}$ variance. Accordingly, in case of that all marginal distributions $S_\ell\!\sim\!\mathcal{M}_{\nu_\ell}(0,1)$, ${1}\!\leq\!\ell\!\leq\!{L}$, if $\defrmat{S}$ follows a standard multivariate McLeish distribution with zero mean vector and unit covariance matrix, $\defvec{a}^{T}\defrmat{S}$ should have to follow a McLeish distribution with zero mean and $\defvec{a}^T\defvec{a}$ variance, which surely imposes that there must be a condition among $\nu_\ell$, ${1}\!\leq\!\ell\!\leq\!{L}$. By the uniqueness property of \ac{MGF} \cite{BibCurtissAMS1942}, we know that the \ac{PDF} is uniquely determined by the \ac{MGF}, and therefore the \ac{MGF} of $\defvec{a}^{T}\defrmat{S}$ has to be in the same form of the \ac{MGF} of $S_\ell\!\sim\!\mathcal{M}_{\nu_\ell}(0,1)$ for all ${1}\!\leq\!\ell\!\leq\!{L}$. With the aid of \theoremref{Theorem:McLeishMGF}, the \ac{MGF} of  $\defvec{a}^{T}\defrmat{S}$, i.e., $M_{\defvec{a}^{T}\defrmat{S}}(s)\!=\!\mathbb{E}[\exp(-s\,\defvec{a}^{T}\defrmat{S})]$ can be written as the product of the \acp{MGF} of all marginal distributions $S_\ell\!\sim\!\mathcal{M}_{\nu_\ell}(0,1)$ for all ${1}\!\leq\!\ell\!\leq\!{L}$, that is 
$M_{\defvec{a}^{T}\defrmat{S}}(s)\!=\!\prod_{\ell=1}^{L}\bigl(1-{\frac{1}{4}}\lambda^2_{\ell}s^2\bigr)^{-\nu_{\ell}}$
with $\lambda_{\ell}\!=\!\sqrt{2{a}^2_{\ell}/\nu_{\ell}}$. 
When the all component deviation factors are exactly the same (i.e., 
$\lambda_{\ell}\!=\!\lambda_{\Sigma}$, ${1}\!\leq\!\ell\!\leq\!{L}$), we can rewrite it in the form of \eqref{Eq:McLeishMGF}, that is 
$M_{\defvec{a}^{T}\defrmat{S}}(s)\!=\!\bigl(1-{\frac{1}{4}}\lambda^2_{\Sigma}s^2\bigr)^{-\nu_{\Sigma}}$,
where $\nu_{\Sigma}\!=\!\sum_{\ell=1}^{L}\nu_{\ell}$ and $\sigma^2_{\Sigma}\!=\!\defvec{a}^T\defvec{a}$, and therefore  $\lambda_{\Sigma}\!=\!\sqrt{2\sigma^2_{\Sigma}/\nu_{\Sigma}}$. Eventually, we reach $\nu_{\Sigma}\!=\!L\nu_{\ell}$, ${1}\!\leq\!\ell\!\leq\!{L}$, where each equality can be satisfied when and only when $\nu_\ell\!=\!\nu_k\!=\!\nu$ for any $\ell\!\neq\!k$. Consequently, $\defrmat{S}$ follows a standard multivariate McLeish distribution iff ${S}_\ell\!\sim\!\mathcal{M}_{\nu}(0,1)$ for all $1\!\leq\!{\ell}\!\leq\!{L}$.
There hence, each marginal distribution is decomposed as ${S}_\ell\!=\!\sqrt{G_{\ell}}N_{\ell}$ with $G_{\ell}\!\sim\!\mathcal{G}(\nu,1)$ and ${N}_{\ell}\!\sim\!\mathcal{N}(0,1)$ for all $1\!\leq\!{\ell}\!\leq\!{L}$. Owing to preserving the being \ac{CS}, any given pair of ${S}_k\!\sim\!\mathcal{M}_{\nu}(0,1)$ and
${S}_\ell\!\sim\!\mathcal{M}_{\nu}(0,1)$, $k\!\neq\!\ell$, must be uncorrelated, 
and what is more accordingly, $\Phi_{k,\ell}\!=\!\arctan({S}_k,{S}_\ell)$ 
has to be uniformly distributed over $[-\pi,\pi)$ and independent of
both ${S}_k$ and ${S}_\ell$. Referring to the proof of \theoremref{Theorem:CCSMcLeishDefinition}, we notice that $G_{\ell}$, ${1}\!\leq\!\ell\!\leq\!{L}$, are the same distribution (i.e., 
the correlation between any pair of ${G}_k\!\sim\!\mathcal{G}(\nu,1)$ and
${G}_\ell\!\sim\!\mathcal{G}(\nu,1)$, $k\!\neq\!\ell$ is surely $1$ without loss of generality), and thus $\defrmat{S}$ certainly follows a \ac{CS} standard multivariate distribution, denoted by $\defrmat{S}\!\sim\!\mathcal{M}^L_{\nu}(\defvec{0},\defmat{I})$  and decomposed in the following theorem. 

\begin{theorem}\label{Theorem:StandardMultivariateMcLeishDefinition}
A standard multivariate McLeish distribution, denoted by $\defrmat{S}\!\sim\!\mathcal{M}^L_{\nu}(\defvec{0},\defmat{I})$, is decomposed as
\begin{equation}\label{Eq:StandardMultivariateMcLeishDefinition}
    \defrmat{S}=\sqrt{G}\defrmat{N},
\end{equation}
where $\defrmat{N}\!\sim\!\mathcal{N}^L(\defvec{0},\defmat{I})$. 
\end{theorem}

\begin{proof}
The proof is obvious from the pivotal and tractable details mentioned before \theoremref{Theorem:StandardMultivariateMcLeishDefinition}.
\end{proof}

With \theoremref{Theorem:StandardMultivariateMcLeishDefinition}, we conclude that since~any non-empty subset of multivariate Gaussian distribution follows a multivariate Gaussian distribution \cite{BibAndersenBook1995,BibTongBook1989,BibWangKotzNgBook1989,BibBilodeauBrennerBook1999,BibKotzBalakrishnanBook2004}, the random vector $\defrmat{W}\!=\![{S}_{k_1},{S}_{k_2},\ldots,\allowbreak{S}_{k_K}]^T$ constructed from $\defrmat{S}$ for a subset $\{k_1,k_2,\ldots,k_K\}$ of $\{1,2,\ldots,L\}$ with cardinal ${K}\!\leq\!{L}$ follows a standard multivariate CS McLeish distri\-bu\-tion. Eventually, the \ac{PDF} of standard multi\-variate \ac{CS} McLeish distribution denoted by $\defrmat{S}\!\sim\!\mathcal{M}^L_{\nu}(\defvec{0},\defmat{I})$ is given in the following theorem.

\begin{theorem}\label{Theorem:StandardMultivariateMcLeishPDF}
The \ac{PDF} of $\defrmat{S}\!\sim\!\mathcal{M}^{L}_{\nu}(\defvec{0},\defmat{I})$ is given by 
\begin{equation}\label{Eq:StandardMultivariateMcLeishPDF}
f_{\defrmat{S}}(\defvec{x})=
	\frac{2}{\sqrt{\pi^L}}
	\frac{{\lVert{\defvec{x}}\rVert}^{\nu-{L}/{2}}}
		{\Gamma(\nu)\lambda_{0}^{\nu+{L}/{2}}}
					 {K}_{\nu-{L}/{2}}
					 \Bigl(
							\frac{2}{\lambda_{0}}
							{\bigl\lVert{\defvec{x}}\bigr\rVert}
					 \Bigr),
\end{equation}
for a certain $\defvec{x}\!=\![x_1,x_2,\ldots,x_L]^T\in\mathbb{R}^L$. 
\end{theorem}

\begin{proof}
Referring to \eqref{Eq:StandardMultivariateMcLeishDefinition}, the \ac{PDF} of 
$\defrmat{S}$ conditioned on $G$, i.e., $f_{\defrmat{S}|G}(\defvec{x}|g)$ can be readily written as\cite[Eq. (2.3-74)]{BibProakisBook}
\begin{equation}\label{Eq:StandardMultivariateMcLeishConditionalPDF}
\!\!f_{\defrmat{S}|G}(\defvec{x}|g)=
	\frac{1}{(2\pi)^{{L}/{2}}g^{{L}/{2}}}
		\exp\Bigl(-\frac{\lVert\defvec{x}\rVert^2}{2g}\Bigr),
\end{equation} 
for $g\!\in\!\mathbb{R}_{+}$. In accordance, the joint \ac{PDF} $f_{\defrmat{S}}(\defvec{x})$ can be readily expressed as
$f_{\defrmat{S}}(\defvec{x})\!=\!\int_{0}^{\infty}f_{\defrmat{S}|G}(\defvec{x}|g)\,f_{G}(g)dg$, that is 
\begin{equation}\label{Eq:StandardMultivariateMcLeishPDFIntegral}
f_{\defrmat{S}}(\defvec{x})= \frac{1}{(2\pi)^{{L}/{2}}}\int_{0}^{\infty}\!
\frac{1}{g^{{L}/{2}}}\exp\Bigl(-\frac{\lVert\defvec{x}\rVert^2}{2g}\Bigr)
\,f_{G}(g)dg,
\end{equation}
where $f_{G}(g)$ denotes the \ac{PDF} of $G\!\sim\!\mathcal{G}(\nu,1)$ ( i.e., given in \eqref{Eq:ProportionPDF}). Subsequently, using
\cite[Eq.\!~(3.471/9)]{BibGradshteynRyzhikBook}, \eqref{Eq:StandardMultivariateMcLeishPDFIntegral} simplifies to \eqref{Eq:StandardMultivariateMcLeishPDF}, which proves \theoremref{Theorem:StandardMultivariateMcLeishPDF}.
\end{proof}

Note that $\defrmat{S}\!\sim\!\mathcal{M}^{L}_{\nu}(\defvec{0},\defmat{I})$ is termed as standard multivariate McLeish distribution which is a collection of identical standard McLeish distributions. As observed in \theoremref{Theorem:StandardMultivariateMcLeishPDF}, the \ac{PDF} of $\defrmat{S}\!\sim\!\mathcal{M}^{L}_{\nu}(\defvec{0},\defmat{I})$ is given by $f_{\defrmat{S}}(\defvec{x})$, and it does only depend on the squared Euclidean distance $\lVert{\defvec{x}}\rVert^2\!=\!\defvec{x}^T\defvec{x}$ of $\defvec{x}$ from the origin.
That is, there exists a circularly symmetry among all ${S}_\ell\!\sim\!\mathcal{M}_{\nu}(0,1)$, $1\!\leq\!{\ell}\!\leq\!{L}$.
However, we cannot partition \eqref{Eq:StandardMultivariateMcLeishPDF} into the product of the \acp{PDF} of marginal distributions even in spite of that they are uncorrelated. However, it simplifies to \eqref{Eq:McLeishPDF} for ${L}\!=\!{1}$ as expected. Furthermore, since an orthogonal transformation $\defmat{O}$ (i.e., $\defmat{O}^T\defmat{O}\!=\!\defmat{O}\defmat{O}^T\!=\!\defmat{I})$ 
preserves the norm of any vector (i.e., $\lVert{\defmat{O}\defvec{x}}\rVert\!=\!\lVert{\defvec{x}}\rVert$), we can immediately conclude $\defmat{O}\defrmat{S}\!\sim\!\mathcal{M}^{L}_{\nu}(\defvec{0},\defmat{I})$, which remarks that 
$\defrmat{S}\!\sim\!\mathcal{M}^{L}_{\nu}(\defvec{0},\defmat{I})$ has the same distribution in any orthonormal basis. Geometrically, it is invariant to rotations and reflections and hence does not prefer any specific direction.

\begin{definition}[McLeish's Multivariate Quantile and Complementary Quantile]
\label{Definition:McLeishMultivariateQAndComplementaryQFunction}
For a fixed $\defvec{x}\in\mathbb{R}^L$ in higher dimensional space, the McLeish's multivariate \ac{Q-function} is defined by 
\ifCLASSOPTIONtwocolumn
\begin{multline}\label{Eq:McLeishMultivariateQFunction}
Q^{L}_{\nu}(\defvec{x})=\int^{\infty}_{x_1}
    \int^{\infty}_{x_2}
    \cdots
    \int^{\infty}_{x_L}
	\frac{2}{\sqrt{\pi^L}}
	\frac{{\lVert{\defvec{u}}\rVert}^{\nu-{L}/{2}}}
		{\Gamma(\nu)\lambda_{0}^{\nu+{L}/{2}}}\\\times
			{K}_{\nu-{L}/{2}}
			 \Bigl(
					\frac{2}{\lambda_{0}}{\bigl\lVert{\defvec{u}}\bigr\rVert}
				 \Bigr){du_1}{du_2}\ldots{du_L},
\end{multline}
\else
\begin{equation}\label{Eq:McLeishMultivariateQFunction}
Q^{L}_{\nu}(\defvec{x})=\int^{\infty}_{x_1}
    \int^{\infty}_{x_2}
    \cdots
    \int^{\infty}_{x_L}
	\frac{2}{\sqrt{\pi^L}}
	\frac{{\lVert{\defvec{u}}\rVert}^{\nu-{L}/{2}}}
		{\Gamma(\nu)\lambda_{0}^{\nu+{L}/{2}}}
			{K}_{\nu-{L}/{2}}
			 \Bigl(
					\frac{2}{\lambda_{0}}{\bigl\lVert{\defvec{u}}\bigr\rVert}
				 \Bigr){du_1}{du_2}\ldots{du_L},
\end{equation}
\fi
and the corresponding complementary \ac{Q-function} by 
\ifCLASSOPTIONtwocolumn
\begin{multline}\label{Eq:McLeishMultivariateComplementaryQFunction}
\widehat{Q}^{L}_{\nu}(\defvec{x})=\int_{-\infty}^{x_1}
    \int_{-\infty}^{x_2}
    \cdots
    \int_{-\infty}^{x_L}
        \frac{2}{\sqrt{\pi^L}}
	    \frac{{\lVert{\defvec{u}}\rVert}^{\nu-{L}/{2}}}
		    {\Gamma(\nu)\lambda_{0}^{\nu+{L}/{2}}}\\\times
			 {K}_{\nu-{L}/{2}}
				 \Bigl(
					\frac{2}{\lambda_{0}}{\bigl\lVert{\defvec{u}}\bigr\rVert}
				 \Bigr){du_1}{du_2}\ldots{du_L}.
\end{multline}
\else
\begin{equation}\label{Eq:McLeishMultivariateComplementaryQFunction}
\widehat{Q}^{L}_{\nu}(\defvec{x})=\int_{-\infty}^{x_1}
    \int_{-\infty}^{x_2}
    \cdots
    \int_{-\infty}^{x_L}
        \frac{2}{\sqrt{\pi^L}}
	    \frac{{\lVert{\defvec{u}}\rVert}^{\nu-{L}/{2}}}
		    {\Gamma(\nu)\lambda_{0}^{\nu+{L}/{2}}}\\\times
			 {K}_{\nu-{L}/{2}}
				 \Bigl(
					\frac{2}{\lambda_{0}}{\bigl\lVert{\defvec{u}}\bigr\rVert}
				 \Bigr){du_1}{du_2}\ldots{du_L}.
\end{equation}
\fi
\end{definition}

The \ac{CDF} of $\defrmat{S}\!\sim\!\mathcal{M}^{L}_{\nu}(\defvec{0},\defmat{I})$ is completely descriptive of
the probability of that $\defrmat{S}$ are less than or equal to $\defvec{x}$, and defined by
$F_{\defrmat{S}}(\defvec{x})\!=\!\Pr\{\defrvec{S}\!\leq\!\defvec{x}\}\!=\!\allowbreak\Pr\{{S}_{1}\!\leq\!x_1,\allowbreak{S}_{2}\!\leq\!x_2,\ldots,{S}_{L}\!\leq\!x_L\}$ and obtained in the following.
It is worth noting the properties of the \ac{CDF} $F_{\defrmat{S}}(\defvec{x})$;
${0}\!\leq\!{F_{\defrmat{S}}(\defvec{x})}\!\leq\!{1}$, ${F_{\defrmat{S}}(-\defvec{\infty})}\!=\!{0}$, and
${F_{\defrmat{S}}(\defvec{\infty})}\!=\!{1}$.
Furthermore, $F_{\defrmat{S}}(\defvec{x})$ is a monotonically increasing function of $\defvec{x}$, that is
$F_{\defrmat{S}}(\defvec{x})\!\leq\!F_{\defrmat{S}}(\defvec{x}+\Delta)$ for $\Delta\in\mathbb{R}_{+}$.

\begin{theorem}\label{Theorem:StandardMultivariateMcLeishCDF}
The \ac{CDF} of $\!\defrmat{S}\!\sim\!\mathcal{M}_{\nu}^L\!\bigl(\defvec{0},\defmat{I}\bigr)$ is given by
\begin{equation}\label{Eq:StandardMultivariateMcLeishCDF}
F_{\defrmat{S}}(\defvec{x})=\widehat{Q}^{L}_{\nu}(\defvec{x}),
\end{equation}
defined over $\boldsymbol{x}\!\in\!\mathbb{R}^L$.
\end{theorem}

\begin{proof}
The \ac{CDF} of $\!\defrmat{S}\!\sim\!\mathcal{M}_{\nu}^L\!\bigl(\defvec{0},\defmat{I}\bigr)$ is readily given by
$F_{\defrmat{S}}(\defvec{x})\!=\!\allowbreak\int_{-\infty}^{x_1}
    \int_{-\infty}^{x_2} \cdots \int_{-\infty}^{x_L} f_{\defrmat{S}}(\defvec{u})\,\allowbreak{du_1}{du_2}\ldots{du_L}$
defined over $\defvec{x}\!\in\!\mathbb{R}^L$, where $f_{\defrmat{S}}(\defvec{x})$ is given in
\eqref{Eq:StandardMultivariateMcLeishPDF}.
Therewith, exploiting \eqref{Eq:McLeishMultivariateComplementaryQFunction}, the proof is obvious.
\end{proof}

Note that the \ac{CCDF} of $\defrmat{S}\!\sim\!\mathcal{M}^{L}_{\nu}(\defvec{0},\defmat{I})$ is also useful to derive especially when considering tail probabilities, and defined by $\widehat{F}_{\defrmat{S}}(\defvec{x})\!=\!\Pr\{\defrvec{S}>\!\defvec{x}\}\!=\!\allowbreak\Pr\{{S}_{1}\!>\!x_1,\allowbreak{S}_{2}\!>\!x_2,\ldots,{S}_{L}\!>\!x_L\}$ and obtained in the following. As opposite to the \ac{CDF}, $\widehat{F}_{\defrmat{S}}(\defvec{x})$ has the following properties: ${0}\!\leq\!{\widehat{F}_{\defrmat{S}}(\defvec{x})}\!\leq\!{1}$, ${\widehat{F}_{\defrmat{S}}(-\defvec{\infty})}\!=\!{1}$, and ${\widehat{F}_{\defrmat{S}}(\defvec{\infty})}\!=\!{0}$, and it is a monotonically decreasing function of $\defvec{x}$, that is
$\widehat{F}_{\defrmat{S}}(\defvec{x})\!\geq\!\widehat{F}_{\defrmat{S}}(\defvec{x}+\Delta)$ for $\Delta\in\mathbb{R}_{+}$.

\begin{theorem}\label{Theorem:StandardMultivariateMcLeishCCDF}
The \ac{CCDF} of $\!\defrmat{S}\!\sim\!\mathcal{M}_{\nu}^L\!\bigl(\defvec{0},\defmat{I}\bigr)$ is given by
\begin{equation}\label{Eq:StandardMultivariateMcLeishCCDF}
\widehat{F}_{\defrmat{S}}(\defvec{x})={Q}^{L}_{\nu}(\defvec{x}),
\end{equation}
defined over $\boldsymbol{x}\!\in\!\mathbb{R}^L$.
\end{theorem}

\begin{proof}
The proof is obvious following almost the same steps performed in the proof of
\theoremref{Theorem:StandardMultivariateMcLeishCDF}.
\end{proof}

Since any (non-empty) subset of multivariate McLeish distribution is a multivariate McLeish distribution, both the \ac{CDF} and \ac{CCDF} of any subset of multivariate McLeish distribution can be obtained by respectively using \eqref{Eq:StandardMultivariateMcLeishCDF} and \eqref{Eq:StandardMultivariateMcLeishCDF}, where setting 
$x_\ell\!=\!0$ for $X_{\ell}$ which is not in the subset of interest, i.e.,  
the \ac{CDF} of $S_1\!\sim\!\mathcal{M}_{\nu}(0,1)$ is $F_{S_1}(x)\!=\!F_{\defrvec{S}}([{x},0,\ldots,0]^T)$ and the corresponding \ac{CCDF} is $\widehat{F}_{S_1}(x)\!=\!\widehat{F}_{\defrvec{S}}([{x}_1,0,\ldots,0]^T)$, which are respectively as expected the special case of \eqref{Eq:McLeishCDF} and \eqref{Eq:McLeishCCDF} with zero mean and unit variance. Besides, in the case of the bivariate distribution of any pair of $S_k$ and $S_\ell$, $k\neq\ell$, we readily obtain the bivariate \ac{CDF} as follows
$F_{S_k,S_\ell}(x_k,x_\ell)\!=\!F_{\defrvec{S}}([0,\ldots,0,\allowbreak{x}_k,0,\ldots,0,{x}_\ell,0,\ldots,0]^T)$ as expected. In the similar manner, the bivariate \ac{CCDF} $\widehat{F}_{S_k,S_\ell}(x_k,x_\ell)$ can also be readily obtained using \theoremref{Theorem:StandardMultivariateMcLeishCCDF}. 

\begin{theorem}\label{Theorem:StandardMultivariateMcLeishMGF}
The \ac{MGF} of $\defrmat{S}\!\sim\!\mathcal{M}^{L}_{\nu}(\defvec{0},\defmat{I})$ is given by 
\begin{equation}\label{Eq:StandardMultivariateMcLeishMGF}
M_{\defrmat{S}}(\defvec{s})=\Bigl(1-\frac{\lambda^2_{0}}{4}\defvec{s}^T\defvec{s}\Bigr)^{-\nu},
\end{equation}
for a certain $\defvec{s}\!\in\!\mathbb{R}^L$ within the existence region $\defvec{s}\!\in\!\mathbb{C}_{0}$, where the region $\mathbb{C}_{0}$ is given by 
\setlength\jot{3pt}
\begin{equation}\label{Eq:StandardMultivariateMcLeishMGFExistenceRegion}
    \mathbb{C}_0=\Bigl\{
        \defvec{s}\,    
        \Bigl|\,
        \lambda^2_{0}\defvec{s}^T\defvec{s}
        \leq
        {4}
        \Bigr.
        \Bigr\}.
\end{equation}
\end{theorem}

\begin{proof}
The \ac{MGF} of $\defrmat{S}\!\sim\!\mathcal{M}^{L}_{\nu}(\defvec{0},\defmat{I})$ is described by 
$M_{\defrmat{S}}(\defvec{s})\!=\!\mathbb{E}\bigl[\exp(-\defvec{s}^T\defrmat{S})\bigr]\!=\!\int_{-\infty}^{\infty}\!\cdots\int_{-\infty}^{\infty}\!\allowbreak\exp(-\defvec{s}^T\!\defvec{x})f_{\defrmat{S}}(\defvec{x})\,{dx_1}\ldots{dx_L}$, where substituting \eqref{Eq:StandardMultivariateMcLeishPDFIntegral} yields 
\begin{equation}\label{Eq:StandardMultivariateMcLeishMGFIntegral}
M_{\defrmat{S}}(\defvec{s})=\int_{0}^{\infty}\!\frac{1}{g^{{L}/{2}}}I(g){f}_{G}(g)\,dg,
\end{equation}
where $f_{G}(g)$ denotes the \ac{PDF} of $G\!\sim\!\mathcal{G}(\nu,1)$ ( i.e., given in \eqref{Eq:ProportionPDF}) and $I(g)$ is given by 
\begin{equation}
I(g)=\int_{-\infty}^{\infty}\!\cdots\int_{-\infty}^{\infty}\!
        \frac{{e}^{-\frac{1}{2g}\left(\lVert\defvec{x}\rVert^2+g\defvec{s}^T\!\defvec{x}\right)}}{(2\pi)^{{L}/{2}}}\,{dx_1}\ldots{dx_L},
\end{equation}
where achieving the equivalent of completing the square, i.e., substituting  $\lVert\defvec{x}\rVert^2+2g\defvec{s}^T\!\defvec{x}\!=\!{\lVert\defvec{x}+g\defvec{s}\rVert^2}-g^2\defvec{s}^T\!\defvec{s}$ readily results in $I(g)\!=\!\exp(\frac{g}{2}\defvec{s}^T\defvec{s})$. Accordingly, \eqref{Eq:StandardMultivariateMcLeishMGFIntegral} simplifies with the aid of \cite[Eq.\!~(3.381/4)]{BibGradshteynRyzhikBook} to \eqref{Eq:StandardMultivariateMcLeishMGF} with the convergence \eqref{Eq:StandardMultivariateMcLeishMGFExistenceRegion}, which proves \theoremref{Theorem:StandardMultivariateMcLeishMGF}.
\end{proof}

As similar to the \ac{CDF} and C\textsuperscript{2}DF of the subset of multivariate McLeish distribution, the corresponding \ac{MGF} is obtained utilizing \eqref{Eq:StandardMultivariateMcLeishMGF}. For instance, we can easily obtain the \ac{MGF} of $S_1\!\sim\!\mathcal{M}_{\nu}(0,1)$ by means of $M_{S_1}(s)\!=\!M_{\defrvec{S}}([{s}_1,0,\ldots,0]^T)=(1-{\lambda^2_{0}}{s}^2/{4})^{-\nu}$, which is consistent with \eqref{Eq:McLeishMGF} for zero mean and unit variance. Besides, in the case of the bivariate distribution of any given pair of $S_k$ and $S_\ell$, $k\neq\ell$, we readily obtain  $M_{S_k,S_\ell}(s_k,s_\ell)\!=\!M_{\defrvec{S}}([0,\ldots,0,\allowbreak{s}_k,0,\ldots,0,{s}_\ell,0,\ldots,0]^T)=(1-{\lambda^2_{0}}(s_1^2+s_2^2)/{4})^{-\nu}$ as expected. It is lastly worth noting that these results and the ones given above are restricted to the case where all ${S}_\ell\!\sim\!\mathcal{M}_{\nu}(0,1)$, $1\!\leq\!{\ell}\!\leq\!{L}$, are identically distributed. A more general case is investigated in the following.

Let us have a vector of \acf{u.n.i.d.} McLeish distributions, that is
\begin{equation}\label{Eq:INIDMultivariateMcLeishRandomVector}
    \defrmat{X}=[X_1,X_2,\ldots,X_L]^T,
\end{equation}
where $X_{\ell}\!\sim\!\mathcal{M}_{\nu}(0,\sigma^2_{k})$ for all $1\!\leq\!\ell\!\leq\!L$, and~any~given pair of $X_{k}\!\sim\!\mathcal{M}_{\nu}(0,\sigma^2_{\ell})$ and $X_{\ell}\!\sim\!\mathcal{M}_{\nu}(0,\sigma^2_{\ell})$, $k\!\neq\!\ell$ are assumed uncorrelated (i.e., $\Covariance{X_{k}}{X_{\ell}}\!=\!{0}$). 
It is worth noticing that $\defrmat{X}$ follows a multivariate McLeish distribution iff $\defvec{a}^T\defrmat{X}$ for all $\defvec{a}\!\in\!\mathbb{R}^L$ follows a McLeish distribution by the definition of multivariate distribution. Define  $\defvec{\sigma}^2\!=\![\sigma^2_1,\sigma^2_2,\ldots,\sigma^2_L]^T\!$ consisting of variances of marginal distributions, and accordingly $\defvec{\sigma}\!=\![\sigma_1,\sigma_2,\ldots,\sigma_L]^T$. Due to possessing $\Covariance{X_{k}}{X_{\ell}}\!=\!{0}$ for any $k\!\neq\!\ell$, the random vector $\defrmat{X}$ certainly follows a multivariate \acf{ES} McLeish distribution denoted by $\defrmat{X}\!\sim\!\mathcal{M}^L_{\nu}(\defvec{0},\diag(\defvec{\sigma}^2))$ and decomposed as in the following. 

\begin{theorem}\label{Theorem:INIDMultivariateMcLeishDefinition}
A multivariate McLeish distribution of uncorrelated and not identically distributed McLeish distributions, denoted by 
$\defrmat{X}\!\sim\!\mathcal{M}^L_{\nu}(\defvec{0},\diag(\defvec{\sigma}^2))$, is decomposed as
\begin{equation}\label{Eq:INIDMultivariateMcLeishDefinition}
    \defrmat{X}=\diag(\defvec{\sigma})\defrmat{S}.
\end{equation}
where $\defrmat{S}\!\sim\!\mathcal{M}_{\nu}(0,\defvec{I})$. 
\end{theorem}

\begin{proof}
The proof is obvious since $\defvec{\sigma}^T\defrmat{S}\!\sim\!\mathcal{M}_{\nu}(0,\defvec{\sigma}^T\defvec{\sigma})$. 
\end{proof}

Accordingly, the \ac{PDF} of a multivariate \acf{ES} McLeish distribution, denoted by $\defrmat{X}\!\sim\!\mathcal{M}^L_{\nu}(\defvec{0},\diag(\defvec{\sigma}^2))$, is given in the following.

\begin{theorem}\label{Theorem:INIDMultivariateMcLeishPDF}
\!\!The \ac{PDF} of $\defrmat{X}\!\sim\!\mathcal{M}^L_{\nu}(\defvec{0},\!\diag(\defvec{\sigma}^2))$~is~given~by
\begin{equation}\label{Eq:INIDMultivariateMcLeishPDF}
\!\!\!\!f_{\defrmat{X}}(\defvec{x})=
		\frac{2}{\pi^{L/2}}
		\frac{{\bigl\lVert{\defmat{\Lambda}^{-1}\defvec{x}}\bigr\rVert}^{\nu-{L}/{2}}}
			{\Gamma(\nu)\det(\defmat{\Lambda})}
				{K}_{\nu-{L}/{2}}\Bigl({2}{\bigl\lVert{\defmat{\Lambda}^{-1}\defvec{x}}\bigr\rVert}\Bigr)
\end{equation}
for a certain $\defvec{x}\!=\![x_1,x_2,\ldots,x_L]^T\!\in\!\mathbb{R}^L$, where $\defrmat{\Lambda}\!=\!\diag(\defvec{\lambda})$,
and $\defvec{\lambda}\!=\!\lambda_{0}\,\defvec{\sigma}$ denotes the component deviation vector.
\end{theorem}

\begin{proof}
Note that, referring to \eqref{Eq:INIDMultivariateMcLeishDefinition}, we express  $\defrmat{S}\!\sim\!\mathcal{M}^{L}_{\nu}(\defvec{0},\defmat{I})$ with the aid of a linear transform, that is $\defrmat{S}\!=\!\diag(\defvec{\sigma})^{-1}\defrmat{X}$, and therefrom we notice the Jacobian $J_{\defrmat{X}|\defrmat{S}}\!=\!\det(\defvec{\sigma})^{-1}$. Hence, we can write the \ac{PDF} of $\defrmat{X}$ as
\begin{equation}\label{Eq:INIDMultivariateMcLeishPDFTransform}
    f_{\defrmat{X}}(\defvec{x})=f_{\defrmat{S}}(\diag(\defvec{\sigma})^{-1}\defvec{x})
        J_{\defrmat{X}|\defrmat{S}}.
\end{equation}
Further, defining the component deviation factor matrix as
\begin{equation}\label{Eq:INIDMultivariateMcLeishDeviationFactorMatrix}
    \defmat{\Lambda}=\lambda_0\diag(\defvec{\sigma})=
    \begin{bmatrix}
    	\lambda_1      & 0      & \dots  & 0      \\
    	0      & \lambda_2      & \dots  & 0      \\
    	\vdots & \vdots & \ddots & \vdots \\
    	0      & 0      & \dots  & \lambda_L
	\end{bmatrix}.
\end{equation}
where $\lambda_{\ell}=\sqrt{2\sigma_{\ell}^2/{\nu}}$, ${1}\!\leq\!{\ell}\!\leq\!{L}$, we directly acknowledge that $\diag(\defvec{\sigma})^{-1}\!=\!\lambda_{0}\,\defrmat{\Lambda}^{-1}$ and $\det(\diag(\defvec{\sigma}))^{-1}\!=\!\lambda_{0}^{L}\det(\defmat{\Lambda})^{-1}$. Finally, with these results, substituting \eqref{Eq:StandardMultivariateMcLeishPDF} into \eqref{Eq:INIDMultivariateMcLeishPDFTransform} results in \eqref{Eq:INIDMultivariateMcLeishPDF}, which proves \theoremref{Theorem:INIDMultivariateMcLeishPDF}. 
\end{proof}

For consistency, accuracy, and clarity, setting $\diag(\defvec{\sigma}^2)\!=\!\sigma^2\defmat{I}$ (i.e., making each component have equal power), we can readily reduce \eqref{Eq:INIDMultivariateMcLeishPDF} to the \ac{PDF} of $\defrmat{X}\!\sim\!\mathcal{M}^L_{\nu}(\defvec{0},\!\sigma^2\defmat{I})$ given by
\begin{equation}\label{Eq:IIDMultivariateMcLeishPDF}
\!\!\!\!f_{\defrmat{X}}(\defvec{x})=
		\frac{2}{\pi^{L/2}}
		\frac{{\bigl\lVert{\defvec{x}}\bigr\rVert}^{\nu-{L}/{2}}}
			{\Gamma(\nu)\lambda^{\nu+{L}/{2}}}
				{K}_{\nu-{L}/{2}}\Bigl(\frac{2}{\lambda}{\bigl\lVert{\defvec{x}}\bigr\rVert}\Bigr)
\end{equation}
where $\lambda\!=\!\sqrt{2\sigma^2/\nu}$ is the component deviation defined before. 

\begin{theorem}\label{Theorem:INIDMultivariateMcLeishCDF}
\!\!The \ac{CDF} of $\defrmat{X}\!\sim\!\mathcal{M}^L_{\nu}(\defvec{0},\!\diag(\defvec{\sigma}^2))$~is~given~by
\begin{equation}\label{Eq:INIDMultivariateMcLeishCDF}
    F_{\defrmat{X}}(\defvec{x})=\widehat{Q}^{L}_{\nu}\bigl(\lambda_{0}\defmat{\Lambda}^{-1}\defvec{x}\bigr),
\end{equation}
defined over $\defvec{x}\!\in\!\mathbb{R}^L$.
\end{theorem}

\begin{proof}
Using \eqref{Eq:INIDMultivariateMcLeishDefinition} and $\diag(\defvec{\sigma})^{-1}\!=\!\lambda_{0}\,\defrmat{\Lambda}^{-1}$, we have $\defrvec{S}\!=\!\lambda_{0}\,\defrmat{\Lambda}^{-1}\!\defrvec{S}$. The proof is then obvious using \theoremref{Theorem:StandardMultivariateMcLeishCDF}.
\end{proof}

\begin{theorem}\label{Theorem:INIDMultivariateMcLeishCCDF}
\!\!The \ac{CCDF} of $\defrmat{X}\!\sim\!\mathcal{M}^L_{\nu}(\defvec{0},\!\diag(\defvec{\sigma}^2))$~is~given~by
\begin{equation}\label{Eq:INIDMultivariateMcLeishCCDF}
    F_{\defrmat{X}}(\defvec{x})={Q}^{L}_{\nu}\bigl(\lambda_{0}\defmat{\Lambda}^{-1}\defvec{x}\bigr),
\end{equation}
defined over $\defvec{x}\!\in\!\mathbb{R}^L$.
\end{theorem}

\begin{proof}
The proof is obvious following almost the same steps performed in the proof of \theoremref{Theorem:INIDMultivariateMcLeishCDF}.
\end{proof}

\begin{theorem}\label{Theorem:INIDMultivariateMcLeishMGF}
\!\!The \ac{MGF} of $\defrmat{X}\!\sim\!\mathcal{M}^L_{\nu}(\defvec{0},\!\diag(\defvec{\sigma}^2))$~is~given~by
\begin{equation}\label{Eq:INIDMultivariateMcLeishMGF}
M_{\defrmat{S}}(\defvec{s})=\Bigl(1-\frac{1}{4}\defvec{s}^T\defmat{\Lambda}^{2}\defvec{s}\Bigr)^{-\nu},
\end{equation}
for a certain $\defvec{s}\!\in\!\mathbb{R}^L$ within the existence region $\defvec{s}\!\in\!\mathbb{C}_{0}$, where the region $\mathbb{C}_{0}$ is given by 
\begin{equation}\label{Eq:INIDMultivariateMcLeishMGFExistenceRegion}
    \mathbb{C}_0=\Bigl\{
        \defvec{s}\,    
        \Bigl|\,
        \defvec{s}^T\defmat{\Lambda}^{2}\defvec{s}
        \leq
        {4}
        \Bigr.
        \Bigr\}.
\end{equation}
\end{theorem}

\begin{proof}
Note that, with the aid of \eqref{Eq:INIDMultivariateMcLeishDefinition}, we can readily rewrite
$M_{\defrmat{X}}(\defvec{s})\!=\!\mathbb{E}[\exp(-\defvec{s}^T\!\defrvec{X})]$ as
$M_{\defrmat{X}}(\defvec{s})\!=\!M_{\defrmat{S}}(\diag(\defvec{\sigma})\defvec{s})$. Then, using
\theoremref{Theorem:StandardMultivariateMcLeishMGF}, $M_{\defrmat{X}}(\defvec{s})$ is  expressed as
\begin{equation}
M_{\defrmat{X}}(\defvec{s})=\Bigl(1-\frac{\lambda^2_{0}}{4}\defvec{s}^T\!\diag(\defvec{\sigma})^2\defvec{s}\Bigr)^{-\nu},
\end{equation}
within the region $\mathbb{C}_0\!=\!\bigl\{\defvec{s}\,\bigl|\,\lambda^2_{0}\,\defvec{s}^T\!\diag(\defvec{\sigma})^2\defvec{s}\leq{4}\bigr.\bigr\}$, where substituting \eqref{Eq:INIDMultivariateMcLeishDeviationFactorMatrix} yields \eqref{Eq:INIDMultivariateMcLeishMGF}
within the region \eqref{Eq:INIDMultivariateMcLeishMGFExistenceRegion}, which completes the proof of \theoremref{Theorem:INIDMultivariateMcLeishMGF}.
\end{proof}

Due to the main importance of special cases for clarity and consistency, let us consider a special case in which $\sigma_{\ell}\!=\!\sigma$ for all $1\!\leq\!\ell\!\leq\!{L}$. Appropriately, we can readily simplify \eqref{Eq:INIDMultivariateMcLeishPDF} to \eqref{Eq:StandardMultivariateMcLeishPDF}, and accordingly, \eqref{Eq:INIDMultivariateMcLeishCDF} to \eqref{Eq:StandardMultivariateMcLeishCDF}, \eqref{Eq:INIDMultivariateMcLeishCCDF} to \eqref{Eq:StandardMultivariateMcLeishCCDF}, \eqref{Eq:INIDMultivariateMcLeishMGF} to \eqref{Eq:StandardMultivariateMcLeishMGF}, as respectively expected. In addition, both the results and conclusions presented above are restricted only to the case, where McLeish distributions are assumed to be uncorrelated. Deducing statistical structures benefiting from these results, we investigate in the following the most general case in which McLeish distributions are assumed to be correlated and non-identically distributed.

Let us consider a vector of \ac{c.n.i.d.} McLeish distributions with $\defvec{\mu}$
mean vector and $\defmat{\Sigma}$ covariance matrix, that is
\begin{equation}
    \defrmat{X}=[X_1,X_2,\ldots,X_L],
\end{equation}
where $X_{\ell}\!\sim\!\mathcal{M}_{\nu_\ell}(\mu_{\ell},\sigma^2_{\ell})$, $1\!\leq\!\ell\!\leq\!L$.
Accordingly, $\defvec{\mu}$ is defined by $\defvec{\mu}\!=\!\mathbb{E}[\defrmat{X}]$, that is
\begin{equation}\label{Eq:MultivariableMcLeishMeanVector}
	\defvec{\mu}=[\mu_1,\mu_2,\ldots,\mu_L]^{T},
\end{equation}
where $\mu_\ell\!=\!\mathbb{E}[X_{\ell}]$, $1\!\leq\!\ell\!\leq\!L$. $\defmat{\Sigma}$ is defined by 
$\defmat{\Sigma}\!=\!\mathbb{E}[\defrmat{X}\defrmat{X}^T]-\defvec{\mu}\defvec{\mu}^T$, that is
\begin{equation}\label{Eq:MultivariableMcLeishCovarainceMatrix}
	\defmat{\Sigma}=
     \begin{bmatrix}
    	\sigma_{11} & \sigma_{12} & \dots  & \sigma_{1L} \\
    	\sigma_{21} & \sigma_{22} & \dots  & \sigma_{2L} \\
    	\vdots      & \vdots      & \ddots & \vdots \\
    	\sigma_{L1} & \sigma_{L2} & \dots  & \sigma_{LL}
	\end{bmatrix},
\end{equation}
where $\sigma_{k\ell}\!=\!\Covariance{X_{k}}{X_{\ell}}\!=\!\mathbb{E}[X_{k}X_{\ell}]-\mu_{k}\mu_{\ell}$ for
${1}\!\leq\!{k,\ell}\!\leq\!{L}$. Note that the covariance matrix $\defmat{\Sigma}$ is by
construction~a~symmetric~matrix, i.e.,  $\defmat{\Sigma}\!=\!\defmat{\Sigma}^T$. It is also a positive definite matrix,
i.e., $\defvec{x}^T\defmat{\Sigma}\,\defvec{x}\!\geq\!{0}$ for all $\defvec{x}\!\in\!\mathbb{R}^L$, which immediately
implies that $\rank(\defmat{\Sigma})\!=\!{L}$ and $\det(\defmat{\Sigma})\!\geq\!{0}$, and therefrom
$\min_{\defvec{x}}\defvec{x}^T\defmat{\Sigma}\defvec{x}\!=\!\trace(\defmat{\Sigma})$.
In terms of the entries $\sigma_{k\ell}$ of $\defmat{\Sigma}\!=\![\sigma_{k\ell}]_{{L}\times{L}}$, the preceding imposes the following necessary conditions:
\begin{itemize}
  \item $\sigma_{k\ell}\!=\!\sigma_{k\ell}$, ${1}\!\leq\!{k,\ell}\!\leq\!{L}$ (symmetry),
  \item $\sigma_{\ell\ell}\!>\!0$ for all ${1}\!\leq\!{\ell}\!\leq\!{L}$ since $\sigma_{\ell\ell}=\sigma^2_{\ell}$ which is the variance of $X_{\ell}$ (i.e., $\Variance{X_{\ell}}=\sigma^2_{\ell}$),   
  \item $\sigma_{k\ell}\!\leq\!\sigma_{kk}\sigma_{\ell\ell}$ for all ${1}\!\leq\!{k,\ell}\!\leq\!{L}$ due to Cauchy-Schwarz' inequality\cite[Sec. 2.3]{BibGarlingBook2007}.
\end{itemize}
Since $\defmat{\Sigma}$ is a positive definite matrix, there is a certain triangular decomposition, which is known as Cholesky decomposition\cite[Chap.~\!10]{BibHighamBook2002}, \cite[Sec.~\!2.2]{BibHammerlinBook2012}, in reduced form of $\defmat{\Sigma}\!=\!\defmat{L}^T\defmat{L}$ with a uniquely defined non-singular lower triangular matrix $\defmat{L}\!=\!\begin{bmatrix}L_{k\ell}\end{bmatrix}_{{L}\times{L}}$ such that $L_{\ell\ell}\!>\!{0}$ for ${1}\!\leq\!{\ell}\!\leq\!{L}$. Consequently, we are certain that $\defmat{L}^{-1}\!$ exists, and accordingly we indicate in the following the existence of multivariate McLeish distribution. By definition of multivariate distribution  \cite{BibAndersenBook1995,BibTongBook1989,BibWangKotzNgBook1989,BibBilodeauBrennerBook1999,BibKotzBalakrishnanBook2004}, $\defrmat{X}$ follows a multivariate McLeish distribution iff
\begin{equation}\label{Eq:MultivariateNormalization}
	\defrmat{Y}=\defmat{L}^{-1}(\defrmat{X}-\defvec{\mu})=[{Y}_1,{Y}_2,\ldots,{Y}_L]^T,
\end{equation}
jointly follows a multivariate McLeish distribution with zero mean vector and unit covariance matrix. As explained before \theoremref{Theorem:StandardMultivariateMcLeishPDF}, if $\defvec{a}^T\defrmat{Y}$ for all vectors $\defvec{a}\!\in\!\mathbb{R}_{+}$ follows a McLeish distribution, then we can declare that  $\defrmat{Y}$ follows a multivariate McLeish distribution. Therefore, $\nu_{\ell}\!=\!\nu$ for all ${1}\!\leq\!\ell\!\leq{L}$ since circularity imposes that
$\arctan({Y}_k,{Y}_\ell)$, $k\!\neq\!\ell$ has to follow a  uniform distribution over $[-\pi,\pi)$. By the virtue of both \eqref{Eq:StandardMultivariateMcLeishDefinition} and \eqref{Eq:MultivariateNormalization}, we find out $\defrmat{Y}\!\sim\!\mathcal{M}^{L}_{\nu}(\defvec{0},\defmat{I})$, and therefore, we can decompose $\defrmat{X}$ as
\begin{equation}\label{Eq:MultivariateMcLeishDistributionDecomposition}
	\defrmat{X}=\sqrt{G}\defrmat{N}+\defvec{\mu},
\end{equation}
where $G\!\sim\!\mathcal{G}(\nu,1)$, and $\defrmat{N}\!\sim\!\mathcal{N}^L(\defvec{0},\defmat{\Sigma})$.
In consequence, $\defrmat{X}$ follows a multivariate \ac{ES} McLeish distribution due to the both facts: (i) the types of all
marginal distributions are the same, (ii) for any pair of ${X}_k\!\sim\!\mathcal{M}_{\nu}(\mu_{k},\sigma^2_{k})$ and
${X}_\ell\!\sim\!\mathcal{M}_{\nu}(\mu_{\ell},\sigma^2_{\ell})$, $k\!\neq\!\ell$,
$\arctan({({X}_k-\mu_k)}/{\sigma_k},{({X}_\ell-\mu_\ell)}/{\sigma_\ell})$ follows uniform distribution over $[-\pi,\pi)$. Since it is uniquely determined by its mean vector, covariance matrix and normality, it is denoted by $\defrmat{X}\!\sim\!\mathcal{M}_{\nu}^L\bigl(\defvec{\mu},\defmat{\Sigma}\bigr)$, whose decomposition and \ac{PDF} are obtained in the following.

\begin{theorem}\label{Theorem:MultivariateMcLeishDecomposition}
If $\defrmat{X}\!\sim\!\mathcal{M}_{\nu}^L\bigl(\defvec{\mu},\defmat{\Sigma}\bigr)$, then
it is decomposed as 
\begin{equation}\label{Eq:MultivariateMcLeishDecomposition}
	\defrmat{X}=\defmat{\Sigma}^{1/2}\defrmat{S}+\defvec{\mu},
\end{equation}
where $\defrmat{S}\!\sim\!\mathcal{M}_{\nu}^L(\defvec{0},\defmat{I})$. 
\end{theorem}

\begin{proof}
Note that, using \cite[Eq. (??)]{BibProakisBook}, we can decompose $\defrvec{N}\!\sim\!\mathcal{N}^L(\defvec{0},\defrmat{\Sigma})$ as $\defrvec{N}\!=\!\defmat{\Sigma}^{1/2}\defrvec{U}$, where
$\defrvec{U}\!\sim\!\mathcal{N}^L(\defvec{0},\defmat{I})$. Furthermore, with the aid of
\eqref{Eq:StandardMultivariateMcLeishDefinition}, we can also decompose 
$\defrvec{S}\!\sim\!\mathcal{M}^L(\defvec{0},\defmat{I})$ as   $\defrvec{S}\!=\!{G}\,\defrvec{U}$, where
$G\!\sim\!\mathcal{G}(\nu,1)$. Then, substituting these results into
\eqref{Eq:MultivariateMcLeishDistributionDecomposition} yields \eqref{Eq:MultivariateMcLeishDecomposition}, which proves
\theoremref{Theorem:MultivariateMcLeishDecomposition}.
\end{proof}

\begin{theorem}\label{Theorem:MultivariateMcLeishPDF}
The \ac{PDF} of $\defrmat{X}\!\sim\!\mathcal{M}_{\nu}^L\!\bigl(\defvec{\mu},\defmat{\Sigma}\bigr)$ is given by
\begin{equation}\label{Eq:MultivariateMcLeishPDF}
f_{\defrmat{X}}(\defvec{x})=\frac{2}{\sqrt{\pi^L}\Gamma(\nu)}
					\frac{{\lVert{\defvec{x}-\defvec{\mu}}\rVert}_{\defmat{\Sigma}}^{\nu-{L}/{2}}}
						{\sqrt{\det(\defmat{\Sigma})}\lambda_{0}^{\nu+{L}/{2}}}
					 {K}_{\nu-{L}/{2}}
					 \Bigl(
							\frac{2}{\lambda_{0}}
							{\bigl\lVert{\defvec{x}-\defvec{\mu}}\bigr\rVert}_{\defmat{\Sigma}}
					 \Bigr),
\end{equation}
defined over $\boldsymbol{x}\!\in\!\mathbb{R}^L$, where ${\lVert{\defvec{x}-\defvec{\mu}}\rVert}_{\defmat{\Sigma}}\!=\!(\defvec{x}-\defvec{\mu})^{T}\defmat{\Sigma}^{-1}(\defvec{x}-\defvec{\mu})$.
\end{theorem}

\begin{proof}
With the aid of \theoremref{Theorem:MultivariateMcLeishDecomposition}, we readily recognize that $\defrmat{X}\!\sim\!\mathcal{M}_{\nu}^L\!\bigl(\defvec{\mu},\defmat{\Sigma}\bigr)$ is a linear transform of $\defrmat{S}\!\sim\!\mathcal{M}^{L}_{\nu}(\defvec{0},\defmat{I})$. Hence, we can write $\defrmat{S}\!=\!\defmat{\Sigma}^{-1/2}(\defrvec{X}-\defvec{\mu})$ and therefrom immediately obtain its Jacobian $J_{\defrmat{X}|\defrmat{S}}\!=\!\det(\defmat{\Sigma})^{-1/2}$in order to express the \ac{PDF} of $\defrmat{X}$ in terms of the \ac{PDF} of $\defrmat{S}$, that is
\begin{equation}\label{Eq:MultivariateMcLeishPDFTransform}
    f_{\defrmat{X}}(\defvec{x})=f_{\defrmat{S}}(\defmat{\Sigma}^{-1/2}(\defrvec{X}-\defvec{\mu}))
        J_{\defrmat{X}|\defrmat{S}}.
\end{equation}
where $f_{\defrmat{S}}(\defvec{x})$ has been already given in \eqref{Eq:StandardMultivariateMcLeishPDF}.
Finally, substituting \eqref{Eq:StandardMultivariateMcLeishPDF} into \eqref{Eq:MultivariateMcLeishPDFTransform} and 
utilizing the symmetry of $\defmat{\Sigma}$ (i.e., $\defmat{\Sigma}\!=\!\defmat{\Sigma}^T$) with the results given above, we obtain \eqref{Eq:MultivariateMcLeishPDF}, which completes the proof of \theoremref{Theorem:MultivariateMcLeishPDF}\footnote{An alternative proof of \theoremref{Theorem:MultivariateMcLeishPDF} can be found as follows. According to \eqref{Eq:MultivariateMcLeishDecomposition}, the \ac{PDF} of $\defrmat{X}$ conditioned on $G$, i.e., the conditional \ac{PDF} $f_{\defrmat{X}|G}(\defvec{x}|g)$ can be readily written as\cite[Eq. (2.3-74)]{BibProakisBook}
\begin{equation}
\!\!f_{\defrmat{X}|G}(\defvec{x}|g)=
	\frac{1}
	{\sqrt{(2\pi)^{L}g^{L}\det(\defmat{\Sigma})}}
	    \exp\biggl(-\frac{{\lVert\defvec{x}-\defvec{\mu}\rVert}^{2}_\defmat{\Sigma}}{2g}\biggr),
	    \tag{F-\thefootnote.1}
\end{equation} 
for $g\!\in\!\mathbb{R}_{+}$. Then, performing the almost same steps followed in the proof of
\theoremref{Theorem:StandardMultivariateMcLeishPDF}, the \ac{PDF} $f_{\defrmat{X}}(\defvec{x})$ is expressed as
\eqref{Eq:MultivariateMcLeishPDF}, which proves \theoremref{Theorem:MultivariateMcLeishPDF}.}.
\end{proof}

Note that we can compute the \ac{CDF} of $\defrmat{X}\!\sim\!\mathcal{M}^{L}_{\nu}(\defvec{\mu},\defmat{\Sigma})$ as 
$F_{\defrmat{X}}(\defvec{x})\!=\!\allowbreak\Pr\{{X}_{1}\!\leq\!x_1,{X}_{2}\!\leq\!x_2\ldots{X}_{L}\!\leq\!x_L\}$, and similarly, its \ac{CCDF} as $\widehat{F}_{\defrmat{X}}(\defvec{x})\!=\!\allowbreak\Pr\{{X}_{1}\!>\!x_1,{X}_{2}\!>\!x_2\ldots{X}_{L}\!>\!x_L\}$, and obtain them in the following. 

\begin{theorem}\label{Theorem:MultivariateMcLeishCDF}
\!\!The \ac{CDF} of $\defrmat{X}\!\sim\!\mathcal{M}^L_{\nu}(\defvec{\mu},\!\defvec{\Sigma})$~is~given~by
\begin{equation}\label{Eq:MultivariateMcLeishCDF}
    F_{\defrmat{X}}(\defvec{x})=\widehat{Q}^{L}_{\nu}\bigl(\defmat{\Sigma}^{-1/2}(\defrvec{X}-\defvec{\mu})\bigr),
\end{equation}
defined over $\defvec{x}\!\in\!\mathbb{R}^L$.
\end{theorem}

\begin{proof}
With the aid of \theoremref{Theorem:MultivariateMcLeishDecomposition}, we have $\defrmat{S}\!=\!\defmat{\Sigma}^{-1/2}(\defrvec{X}-\defvec{\mu})$. Then, using \eqref{Eq:McLeishMultivariateComplementaryQFunction}, the proof is obvious. 
\end{proof}

\begin{theorem}\label{Theorem:MultivariateMcLeishCCDF}
\!\!The \ac{CCDF} of $\defrmat{X}\!\sim\!\mathcal{M}^L_{\nu}(\defvec{\mu},\!\defvec{\Sigma})$~is~given~by
\begin{equation}\label{Eq:MultivariateMcLeishCCDF}
    \widehat{F}_{\defrmat{X}}(\defvec{x})={Q}^{L}_{\nu}\bigl(\defmat{\Sigma}^{-1/2}(\defrvec{X}-\defvec{\mu})\bigr),
\end{equation}
defined over $\defvec{x}\!\in\!\mathbb{R}^L$.
\end{theorem}

\begin{proof}
The proof is obvious using \theoremref{Theorem:MultivariateMcLeishCDF}. 
\end{proof}

As expected based on the mentioned above, the marginal \ac{CDF} of $X_{\ell}\!\sim\!(\mu_{\ell},\sigma^2_{\ell})$ is given by  $F_{X_{\ell}}\bigl(x_{\ell}\bigr)\!=\!{F}_{\defrmat{X}}\bigl(\infty,\ldots,\infty,\allowbreak{x_{\ell}},\infty,\ldots,\infty\bigr)$. In the same manner, the bivariate \ac{CDF} of $X_{k}$ and $X_{\ell}$, $k\!<\!\ell$, is derived as  $F_{X_{k},X_{\ell}}\bigl(x_{k},x_{\ell}\bigr)\!=\!F_{\defrmat{X}}\bigl(\infty,\ldots,\infty,\allowbreak{x_{k}},\allowbreak\infty,\ldots,\infty,{x_{\ell}},\infty,\ldots,\allowbreak\infty\bigr)$, which can be readily generalized for the case more than two marginal distributions. The same manner is also valid for the \ac{CCDF}. 

We further note that the \ac{MGF} of $\defrmat{X}\!\sim\!\mathcal{M}^{L}_{\nu}(\defvec{\mu},\defmat{\Sigma})$, defined by $M_{\defrmat{X}}(\defvec{s})\!=\!\allowbreak\mathbb{E}\bigl[\exp(-\defvec{s}^T\defrmat{X})\bigr]$, is obtained in the following.

\begin{theorem}\label{Theorem:MultivariateMcLeishMGF}
The \ac{MGF} of $\defrmat{X}\!\sim\!\mathcal{M}_{\nu}^L\!\bigl(\defvec{\mu},\defmat{\Sigma}\bigr)$ is given by
\begin{equation}\label{Eq:MultivariateMcLeishMGF}
M_{\defrmat{X}}(\defvec{s})=\exp(-\defvec{s}^T\!\defvec{\mu})\Bigl(1-\frac{\lambda_0^2}{4}\defvec{s}^T\defmat{\Sigma}\,\defvec{s}\Bigr)^{-\nu},
\end{equation}
for a certain $\defvec{s}\!\in\!\mathbb{R}^L$ within 
the existence region $\defvec{s}\!\in\!\mathbb{C}_{0}$, where the region $\mathbb{C}_{0}$ 
is given by 
\begin{equation}\label{Eq:MultivariateMcLeishMGFExistenceRegion}
    \mathbb{C}_0=\Bigl\{
        \defvec{s}\,    
        \Bigl|\,
        \lambda_{0}^2\,
        \defvec{s}^T\defmat{\Sigma}\,\defvec{s}
        \leq
        {4}
        \Bigr.
        \Bigr\}.
\end{equation}
\end{theorem}

\begin{proof}
Using \eqref{Eq:MultivariateMcLeishDecomposition} with $M_{\defrmat{X}}(\defvec{s})\!=\!\allowbreak\mathbb{E}\bigl[\exp(-\defvec{s}^T\!\defrmat{X})\bigr]$, we have
\begin{subequations}\label{Eq:MultivariateMcLeishMGFTransformation}
\setlength\arraycolsep{1.4pt}
\begin{eqnarray}
\label{Eq:MultivariateMcLeishMGFTransformationA}
M_{\defrvec{X}}(\defvec{s})
&=&\mathbb{E}\bigl[\exp\bigl(-\defvec{s}^T(\defmat{\Sigma}^{1/2}\defrmat{S}+\defvec{\mu})\bigr)\bigr],\\
\label{Eq:MultivariateMcLeishMGFTransformationB}
&=&\exp\bigl(-\defvec{s}^T\!\defvec{\mu}\bigr)\mathbb{E}\bigl[\exp\bigl(-\defvec{s}^T\!\defmat{\Sigma}^{1/2}\defrmat{S}\bigr)\bigr],\\
\label{Eq:MultivariateMcLeishMGFTransformationC}
&=&\exp\bigl(-\defvec{s}^T\!\defvec{\mu}\bigr)M_{\defrvec{S}}\bigl(\defmat{\Sigma}^{1/2}\defvec{s}\bigr),
\end{eqnarray}
\end{subequations}   
Eventually, substituting \eqref{Eq:StandardMultivariateMcLeishMGF} into \eqref{Eq:MultivariateMcLeishMGFTransformationC}
yields \eqref{Eq:MultivariateMcLeishMGF} with the existence region  \eqref{Eq:MultivariateMcLeishMGFExistenceRegion}, which proves \theoremref{Theorem:MultivariateMcLeishMGF}\footnote{An alternative proof of \theoremref{Theorem:MultivariateMcLeishMGF} can be done using $\defrmat{X}\!=\!{G}\,\defmat{\Sigma}^{\frac{1}{2}}\defrmat{N}+\defvec{\mu}$ 
derived from \eqref{Eq:StandardMultivariateMcLeishDefinition} and \eqref{Eq:MultivariateMcLeishDecomposition}. Thus, the \ac{MGF} of $\defrmat{X}$ conditioned on $G$ is 
\begin{equation}
M_{\defrmat{X}|G}(\defvec{s}|g)=\exp\Bigl(-\defvec{s}^T\defvec{\mu}+\frac{g}{2}\defvec{s}^T\defmat{\Sigma}\defvec{s}\Bigr),  \tag{F-\thefootnote.1}
\end{equation}
for $g\!\in\!\mathbb{R}_{+}$. In accordance, $M_{\defrmat{S}}(\defvec{s})\!=\!\int_{0}^{\infty}M_{\defrmat{S}|G}(\defvec{s}|g)\,f_{G}(g)dg$ is written as
\begin{equation}\label{Eq:MultivariateMcLeishMGFIntegral}
M_{\defrmat{X}}(\defvec{s})=\exp\bigl(-\defvec{s}^T\defvec{\mu}\bigr)\int_{0}^{\infty}\!\exp\Bigl(\frac{g}{2}\defvec{s}^T\defmat{\Sigma}\defvec{s}\Bigr)
\,f_{G}(g)dg,\tag{F-\thefootnote.2}
\end{equation}
where $f_{G}(g)$ denotes the \ac{PDF} of $G\!\sim\!\mathcal{G}(\nu,1)$ ( i.e., given in \eqref{Eq:ProportionPDF}). So, using \cite[Eq.\!~(3.381/4)]{BibGradshteynRyzhikBook}, \eqref{Eq:MultivariateMcLeishMGFIntegral} simplifies to \eqref{Eq:MultivariateMcLeishMGF}, which proves \theoremref{Theorem:MultivariateMcLeishMGF}.}.
\end{proof} 

Given a non-singular covariance matrix $\defmat{\Sigma}$, the correlation matrix $\defmat{P}$ can be expressed as
\begin{subequations}\label{Eq:MultivariableMcLeishCorrelationMatrix}
\setlength\arraycolsep{1.4pt}
\begin{eqnarray}\label{Eq:MultivariableMcLeishCorrelationMatrixA}
\defmat{P}&=&\begin{bmatrix}
    	1         & \rho_{12} & \dots  & \rho_{1L} \\
    	\rho_{21} & 1         & \dots  & \rho_{2L} \\
    	\vdots    & \vdots    & \ddots & \vdots \\
    	\rho_{L1} & \rho_{L2} & \dots  & 1
	\end{bmatrix},\\
	\label{Eq:MultivariableMcLeishCorrelationMatrixB}
	&=&\diag(\defvec{\sigma})^{-1}\defmat{\Sigma}\diag(\defvec{\sigma})^{-1},
\end{eqnarray}
\end{subequations}
where for ${1}\!\leq\!{k,\ell}\!\leq\!{L}$, $\rho_{k\ell}\in[-1,1]$ denotes the correlation between 
$X_k$ and $X_\ell$, and it is defined by
\begin{equation}
	\rho_{k\ell}=\frac{\Covariance{X_k}{X_\ell}}{\sqrt{\Variance{X_k}\Variance{X_\ell}}}
	    =\frac{\mathbb{E}[{X_k}{X_\ell}]-\mu_k\mu_\ell}{\sigma_k\sigma_\ell}.
\end{equation}  
After using \eqref{Eq:MultivariableMcLeishCorrelationMatrixB}, the inverse of $\defmat{\Sigma}$ is readily rewritten as 
\begin{subequations}
\label{Eq:CovarianceMatrixDecomposition}
\setlength\arraycolsep{1.4pt}
\begin{eqnarray}
    \label{Eq:CovarianceMatrixDecompositionA}
	\defmat{\Sigma}^{-1}&=&\diag(\defvec{\sigma})^{-1}\defmat{P}^{-1}\diag(\defvec{\sigma})^{-1},\\
	\label{Eq:CovarianceMatrixDecompositionB}
		&=&\lambda^2_{0}\,\defmat{\Lambda}^{-1}\defmat{P}^{-1}\defmat{\Lambda}^{-1}
\end{eqnarray}
\end{subequations}
where $\defmat{\Lambda}\!=\!\lambda_{0}\diag(\defvec{\sigma})$. 
In case of $\defmat{\Lambda}\!=\!\lambda\defmat{I}$ with $\lambda\!=\!\sigma\lambda_{0}$, we have  
$\defmat{\Sigma}\!=\!\sigma^2\defmat{P}$, and thus \eqref{Eq:MultivariateMcLeishPDF} simplifies to
\ifCLASSOPTIONtwocolumn
\begin{multline}\label{Eq:NIIDMultivariateMcLeishPDF}
f_{\defrmat{X}}(\defvec{x})=
	\frac{2}{\sqrt{\pi^L}\Gamma(\nu)}
	\frac{{\lVert{\defvec{x}-\defvec{\mu}}\rVert}_{\defmat{P}}^{\nu-{L}/{2}}}
		{\sqrt{\det(\defmat{P})}\lambda^{\nu+{L}/{2}}}
					\\[-1mm]\times		
					 {K}_{\nu-{L}/{2}}
					 \Bigl(
							\frac{2}{\lambda}
							{\bigl\lVert{\defvec{x}-\defvec{\mu}}\bigr\rVert}_{\defmat{P}}
					 \Bigr).
\end{multline}
\else
\begin{equation}\label{Eq:NIIDMultivariateMcLeishPDF}
f_{\defrmat{X}}(\defvec{x})=
	\frac{2}{\sqrt{\pi^L}\Gamma(\nu)}
	\frac{{\lVert{\defvec{x}-\defvec{\mu}}\rVert}_{\defmat{P}}^{\nu-{L}/{2}}}
		{\sqrt{\det(\defmat{P})}\lambda^{\nu+{L}/{2}}}
					 {K}_{\nu-{L}/{2}}
					 \Bigl(
							\frac{2}{\lambda}
							{\bigl\lVert{\defvec{x}-\defvec{\mu}}\bigr\rVert}_{\defmat{P}}
					 \Bigr).
\end{equation}
\fi
Accordingly, we can readily simplify \eqref{Eq:MultivariateMcLeishCDF} to
\begin{equation}\label{Eq:NIIDMultivariateMcLeishCDF}
    F_{\defrmat{X}}(\defvec{x})=\widehat{Q}^{L}_{\nu}\bigl(\lambda_{0}\,\defmat{P}^{-1/2}\defmat{\Lambda}^{-1/2}(\defrvec{X}-\defvec{\mu})\bigr),
\end{equation}
and \eqref{Eq:MultivariateMcLeishCCDF} to
\begin{equation}\label{Eq:NIIDMultivariateMcLeishCCDF}
    \widehat{F}_{\defrmat{X}}(\defvec{x})={Q}^{L}_{\nu}\bigl(\lambda_{0}\,\defmat{P}^{-1/2}\defmat{\Lambda}^{-1/2}(\defrvec{X}-\defvec{\mu})\bigr),
\end{equation}
and \eqref{Eq:MultivariateMcLeishMGF} to
\begin{equation}\label{Eq:NIIDMultivariateMcLeishMGF}
M_{\defrmat{X}}(\defvec{s})=\exp(-\defvec{s}^T\!\defvec{\mu})\Bigl(1-\frac{1}{4}\defvec{s}^T\!\defmat{\Lambda}\,\defmat{P}\,\defmat{\Lambda}\,\defvec{s}\Bigr)^{-\nu},
\end{equation}
In addition, in case of no~correlation~among~marginal~McLeish distributions (i.e., when $\defmat{P}\!=\!\defmat{I}$), we have the covariance matrix $\defmat{\Sigma}\!=\!\defmat{\Lambda}^{2}/\lambda^2_{0}$. Accordingly, for zero mean 
$\defvec{\mu}\!=\!\defvec{0}$, we simplify \eqref{Eq:NIIDMultivariateMcLeishPDF} to \eqref{Eq:INIDMultivariateMcLeishPDF}, \eqref{Eq:NIIDMultivariateMcLeishCDF} to \eqref{Eq:INIDMultivariateMcLeishCDF}, \eqref{Eq:NIIDMultivariateMcLeishCCDF} to \eqref{Eq:INIDMultivariateMcLeishCCDF}, and \eqref{Eq:NIIDMultivariateMcLeishMGF} to \eqref{Eq:INIDMultivariateMcLeishMGF}, as respectively expected.  

There are also two notable properties of multivariate ES McLeish distributions to be explicitly considered: (i)~any non-degenerate affine transformation of $\!\defrmat{X}\!\sim\!\mathcal{M}^{L}_{\nu}(\defvec{\mu},\defmat{\Sigma})$ is also a multivariate \ac{ES} McLeish distri\-bution, (ii) its conditional and marginal distributions are jointly multivariate \ac{ES} McLeish distribution. The first property is given in the following.

\begin{theorem}\label{Theorem:MultivariateMcLeishLinearityProperty}
If $\defrmat{X}\!\sim\!\mathcal{M}^{L}_{\nu}(\defvec{\mu},\defmat{\Sigma})$ and if $\defrmat{Y}\!=\!\defmat{B}\defrmat{X}+\defvec{b}$, where $\rank(\defmat{B})\!\leq\!{L}$, then     
$\defrmat{Y}\!\sim\!\mathcal{M}^{L}_{\nu}(\defmat{B}\defvec{\mu}+\defvec{b},\defmat{B}\defmat{\Sigma}\defmat{B}^T)$.
\end{theorem}

\begin{proof}
Using \theoremref{Theorem:MultivariateMcLeishDecomposition}, we have 
$\defrmat{Y}\!=\!\defmat{B}(\defmat{\Sigma}^{1/2}\defrmat{S}+\defvec{\mu})+\defvec{b}$, which can be rearranged as $\defrmat{Y}\!=\!\defmat{B}\defmat{\Sigma}^{1/2}\defrmat{S}+(\defmat{B}\defvec{\mu}+\defvec{b})$ with $\defmat{B}\defvec{\mu}+\defvec{b}$ mean vector and $\defmat{B}\defmat{\Sigma}\defmat{B}^T$ covariance matrix. 
\end{proof}

As for the second property, the conditional distribution of
$\defrmat{X}\!\sim\!\mathcal{M}^{L}_{\nu}(\defvec{\mu},\defmat{\Sigma})$ is given in the following.

\begin{theorem}\label{Theorem:MultivariateMcLeishPartioningProperty}
Let $\defrmat{X}\!\sim\!\mathcal{M}^{L}_{\nu}(\defvec{\mu},\defmat{\Sigma})$ be $\defrmat{X}\!=\![\defrmat{X}_1^T,\defrmat{X}_2^T]^T$ 
with $\defrmat{X}_1\!\sim\!\mathcal{M}^{L_1}_{\nu}(\defvec{\mu}_1,\defmat{\Sigma}_{11})$ and $\defrmat{X}_2\!\sim\!\mathcal{M}^{L_2}_{\nu}(\defvec{\mu}_2,\defmat{\Sigma}_{22})$, where $L\!=\!{L_1+L_2}$, and $\defmat{\mu}$ and $\defmat{\Sigma}$ are respectively by
\begin{equation}
    \defmat{\mu}=
        \begin{bmatrix}
            \defmat{\mu}_{1}\\
            \defmat{\mu}_{2}
        \end{bmatrix},
    \text{~and~}    
    \defmat{\Sigma}=
        \begin{bmatrix}
            \defmat{\Sigma}_{11} & \defmat{\Sigma}_{12}\\
            \defmat{\Sigma}_{21} & \defmat{\Sigma}_{22}
        \end{bmatrix}.
\end{equation}
The conditional distribution of $\defrmat{X}_1$ given $\defrmat{X}_2\!=\!\defvec{x}_2$ is given by
$\defrmat{X}_1|\defrmat{X}_2\!\sim\!
        \mathcal{M}^{L_1}_{\nu}(
        \defvec{\mu}_1\!+\!\defmat{\Sigma}_{12}\defmat{\Sigma}^{-1}_{22}(\defvec{x}_2\!-\!\defvec{\mu}_2),
        \defmat{\Sigma}_{11}\!-\!\defmat{\Sigma}_{12}\defmat{\Sigma}^{-1}_{22}\defmat{\Sigma}_{21})$. 
\end{theorem}

\begin{proof}
As substituting \eqref{Eq:StandardMultivariateMcLeishDefinition} in  \theoremref{Theorem:MultivariateMcLeishDecomposition}, we can decompose $\defrmat{X}\!\sim\!\mathcal{M}^{L}_{\nu}(\defvec{\mu},\defmat{\Sigma})$ as follows
\begin{equation}
    \defrmat{X}=\sqrt{G}\defrmat{N}=\sqrt{G}
        \begin{bmatrix}
            \defrmat{N}_{1}\\
            \defrmat{N}_{2}
        \end{bmatrix}+
        \begin{bmatrix}
            \defvec{\mu}_{1}\\
            \defvec{\mu}_{2}
        \end{bmatrix},
\end{equation}
with definitions of $\defrmat{X}_1\!=\!\sqrt{G}\defrmat{N}_1+\defvec{\mu}_1$ and $\defrmat{X}_2\!=\!\sqrt{G}\defrmat{N}_2+\defvec{\mu}_2$, 
where ${G}\!\sim\!\mathcal{G}(\nu,1)$, $\defrmat{N}_1\!\sim\!\mathcal{N}^{L}(\defvec{0},\defmat{\Sigma}_{11})$ and $\defrmat{N}_2\!\sim\!\mathcal{N}^{L}(\defvec{0},\defmat{\Sigma}_{22})$. 
The conditional distri\-bution of $\defrmat{X}_1$ given both $G\!=\!{g}$ and $\defrmat{X}_2\!=\!\defvec{x}_2$ is therefore defined by the ratio between two multivariate Gaussian densities, that is $f_{\defrmat{X}_1|\defrmat{X}_2,G}(\defvec{x}_1|\defvec{x}_2,g)\!=\!    {f_{\defrmat{X}|G}(\defvec{x}|g)}/{f_{\defrmat{X}_{2}|G}(\defvec{x}_{2}|g)}$ given by 
\begin{equation}\label{Eq:MultivariateMcLeishConditionalPDFWithoutNormality}
    f_{\defrmat{X}_1|\defrmat{X}_2,G}(\defvec{x}_1|\defvec{x}_2,g)=
    \frac{\sqrt{\det(\defmat{\Sigma}_{22})}}
        {\sqrt{(2\pi{g})^{L_1}\det(\defmat{\Sigma})}}
            \exp\Bigl(
                -\frac{1}{2g}
                \bigl(
                    \bigr\lVert\defvec{x}-\defvec{\mu}\bigl\rVert^2_{\defmat{\Sigma}}
                    -
                    \bigr\lVert\defvec{x}_2-\defvec{\mu}_2\bigl\rVert^2_{\defmat{\Sigma}_{22}}
                \bigl)
            \Bigr)
\end{equation}
for $\defvec{x}\!=\![\defvec{x}^T_1,\defvec{x}^T_2]^T\in\mathbb{R}^{L}$, where $\lVert\defvec{x}-\defvec{\mu}\rVert^2_{\defmat{\Sigma}}$ can be given by
\begin{equation}\label{Eq:CovarianceMatrixExpansion}
    \bigr\lVert\defvec{x}-\defvec{\mu}\bigl\rVert^2_{\defmat{\Sigma}}=
        \bigl\lVert{\defvec{x}_2-\defvec{\mu}_2}\bigr\rVert^2_{\defmat{\Sigma}_{22}}
       +\bigl\lVert
            {\defvec{x}_1-\defvec{\mu}_1-\defmat{\Sigma}_{12}\defmat{\Sigma}_{22}^{-1}\defvec{x}_2}
        \bigr\rVert^2_{\defmat{\Sigma}_{11}-\defmat{\Sigma}_{12}\defmat{\Sigma}_{22}^{-1}\defmat{\Sigma}_{21}},
\end{equation}
After substituting \eqref{Eq:CovarianceMatrixExpansion} in \eqref{Eq:MultivariateMcLeishConditionalPDFWithoutNormality}, the \ac{PDF} of $\defrmat{X}_1$ given $\defrmat{X}_2$ is written as $f_{\defrmat{X}_1|\defrmat{X}_2}(\defvec{x}_1|\defvec{x}_2)\!=\!\int_{0}^{\infty}f_{\defrmat{X}_1|\defrmat{X}_2,G}(\defvec{x}_1|\defvec{x}_2,g)\allowbreak{f}_{G}(g)dg$. Accordingly, and pursuant to utilizing \theoremref{Theorem:MultivariateMcLeishPDF} with \cite[Eq.\!~(3.381/4)]{BibGradshteynRyzhikBook}, the \ac{PDF} of $\defrmat{X}_1|\defrmat{X}_2$ is obtained in the form of \eqref{Eq:MultivariateMcLeishPDF} with mean vector 
$\defvec{\mu}_1\!+\!\defmat{\Sigma}_{12}\defmat{\Sigma}^{-1}_{22}(\defvec{x}_2\!-\!\defvec{\mu}_2)$ and covariance matrix $\defmat{\Sigma}_{11}\!-\!\defmat{\Sigma}_{12}\defmat{\Sigma}^{-1}_{22}\defmat{\Sigma}_{21}$. Then, the proof is obvious.
\end{proof}

Note that, when $\defmat{\Sigma}_{12}\!=\!\defmat{\Sigma}_{12}\!=\!\defmat{0}$, \eqref{Eq:CovarianceMatrixExpansion} reduces to 
\begin{equation}
    \bigr\lVert\defvec{x}-\defvec{\mu}\bigl\rVert^2_{\defmat{\Sigma}}=
       \bigl\lVert{\defvec{x}_1-\defvec{\mu}_1}\bigr\rVert^2_{\defmat{\Sigma}_{11}}
       +\bigl\lVert{\defvec{x}_2-\defvec{\mu}_2}\bigr\rVert^2_{\defmat{\Sigma}_{22}},
\end{equation}
which implies that $\defrmat{X}_1$~and~$\defrmat{X}_2$~are~mutually~uncorrelated, and thus we have
$\defrmat{X}_1|\defrmat{X}_2\!\sim\!\mathcal{M}^{L_1}_{\nu}(\defvec{\mu}_1,\defmat{\Sigma}_{11})$ and $\defrmat{X}_2|\defrmat{X}_1\!\sim\!\mathcal{M}^{L_2}_{\nu}(\defvec{\mu}_2,\defmat{\Sigma}_{22})$. 

\subsection{Multivariate Complex McLeish Distribution}
\label{Section:StatisticalBackground:MultivariateComplexMcLeishDistribution}
Let us have $\defrvec{S}\!\sim\!\mathcal{M}_{\nu}^{2L}(\defvec{0},\defmat{I})$ be represented by
\begin{equation}\label{Eq:StandardMultivariateMcLeishPartitioning}
    \defrvec{S}\!=\!
        \begin{bmatrix}
        \defrvec{S}_1\\
        \defrvec{S}_2
        \end{bmatrix},
\end{equation}
where both $\defrvec{S}_1\!\sim\!\mathcal{M}_{\nu}^L(\defvec{0},\defmat{I})$ and $\defrvec{S}_2\!\sim\!\mathcal{M}_{\nu}^L(\defvec{0},\defmat{I})$ are two such uncorrelated standard multivariate McLeish distributions that 
$\mathbb{E}[\defrvec{S}_{1}\defrvec{S}_{2}^T]\!=\!\defvec{0}$ and $\mathbb{E}[\defrvec{S}_{2}\defrvec{S}_{1}^T]\!=\!\defvec{0}$. Form this point of view, we can define a  multivariate complex McLeish distribution as   
\begin{equation}\label{Eq:StandardMultivariateComplexMcLeishDefinition}
	\defrvec{W}=\defrvec{S}_{1}+\imaginary\defrvec{S}_{2},
\end{equation}
which can be considered as a vector of uncorrelated and iden\-tically distributed standard \ac{CCS} McLeish distributions, i.e., $\defrvec{W}\!=\![W_1,W_2,\ldots,W_L]^T$, where ${W}_{\ell}\!\sim\!\mathcal{CM}(0,1)$, ${1}\!\leq\!{\ell}\!\leq\!{L}$ such that the inphase and quadrature parts of any given pair of ${W}_{k}\!\sim\!\mathcal{CM}(0,1)$ and ${W}_{\ell}\!\sim\!\mathcal{CM}(0,1)$, $k\!\neq\!\ell$ are \ac{CS} by default. By the definition of multi\-variate distribution \cite{BibAndersenBook1995,BibTongBook1989,BibWangKotzNgBook1989,BibBilodeauBrennerBook1999,BibKotzBalakrishnanBook2004}, $\defrmat{W}$ has a multivariate complex distribution iff $\forall\defvec{a}\!\in\!\mathbb{C}^{L}\!$, $\defvec{a}^{T}\defrmat{W}$ follows a complex random distribution of the same family. Accordingly, our intention is to come up to the \ac{PDF} of $\defrvec{W}$, denoted by $f_{\defrvec{W}}(\defvec{z})$, to check its distribution family. Taking into account the definition of multivariate distribution, and pursuant to what presented in \secref{Section:StatisticalBackground:MultivariateMcLeishDistribution} above, we conclude that the \ac{PDF} of $\defrvec{W}$ is exactly the same as the \ac{PDF} of $\defrvec{S}\!\sim\!\mathcal{M}_{\nu}^{2L}(\defvec{0},\defmat{I})$,~i.e.,~$f_{\defrvec{W}}(\defvec{z})\!=\!f_{\defrvec{S}}(\defvec{z})$. The multivariate distribution $\defrmat{W}$ is therefore explicitly termed as standard multivariate \ac{CCS} McLeish distribution and properly denoted by $\defrmat{W}\!\sim\!\mathcal{CM}^{L}_{\nu}(\defvec{0},\defmat{I})$, whose decomposition is given in the following.

\begin{theorem}\label{Theorem:StandardMultivariateCCSMcLeishDefinition}
A standard multivariate \ac{CCS} McLeish distribution, denoted by $\defrmat{W}\!\sim\!\mathcal{CM}^L_{\nu}(\defvec{0},\defmat{I})$, is decomposed as
\begin{equation}\label{Eq:StandardMultivariateCCSMcLeishDefinition}
    \defrmat{W}=\sqrt{G}(\defrmat{N}_1+\imaginary\defrmat{N}_2),
\end{equation}
where $\defrvec{N}_1\!\sim\!\mathcal{N}^{L}(\defvec{0},\defmat{I})$ and $\defrvec{N}_2\!\sim\!\mathcal{N}^{L}(\defvec{0},\defmat{I})$~are~such~two~standard multivariate Gaussian distributions that $\mathbb{E}[\defrvec{N}_{1}\defrvec{N}_{2}^T]\!=\!\defvec{0}$ and $\mathbb{E}[\defrvec{N}_{2}\defrvec{N}_{1}^T]\!=\!\defvec{0}$. Furthermore, $G\!\sim\!\mathcal{G}(\nu,1)$. 
\end{theorem}

\begin{proof}
Using \eqref{Eq:StandardMultivariateMcLeishDefinition} in \eqref{Eq:StandardMultivariateComplexMcLeishDefinition}, we rewrite $\defrvec{W}=[W_1,W_2,\ldots,\allowbreak{W}_L]^T$ as follows 
\begin{equation}
    \defrvec{W}=\sqrt{G_1}\defrvec{N}_{1}+\imaginary\sqrt{G_2}\defrvec{N}_{2},
\end{equation}
where the inphase and quadrature parts of any given pair of ${W}_{k}\!\sim\!\mathcal{CM}(0,1)$ and ${W}_{\ell}\!\sim\!\mathcal{CM}(0,1)$, $k\!\neq\!\ell$ have to be \ac{CS} according to the pivotal and tractable details mentioned before \theoremref{Theorem:StandardMultivariateCCSMcLeishDefinition}. Therefore, $G_1$ and $G_2$ has to be the same distribution, which completes the proof of \theoremref{Theorem:StandardMultivariateCCSMcLeishDefinition}.  
\end{proof}

With the aid of \theoremref{Theorem:StandardMultivariateCCSMcLeishDefinition}, we give the PDF of standard multivariate \ac{CCS} McLeish distribution, denoted by $\defrmat{W}\!\sim\!\mathcal{CM}_{\nu}^L\!\bigl(\defvec{0},\defmat{I}\bigr)$, in the following theorem. 

\begin{theorem}\label{Theorem:StandardMultivariateCCSMcLeishPDF}
The \ac{PDF} of $\defrmat{W}\!\sim\!\mathcal{CM}_{\nu}^L\!\bigl(\defvec{0},\defmat{I}\bigr)$ is given by
\begin{equation}\label{Eq:StandardMultivariateCCSMcLeishPDF}
\!\!\!\!f_{\defrvec{W}}(\defvec{z})=\frac{2}{\pi^L}
	\frac{{\lVert{\defvec{z}}\rVert}^{\nu-L}}
		{\Gamma(\nu)\lambda_{0}^{\nu+L}}
			 {K}_{\nu-L}\Bigl(\frac{2}{\lambda_{0}}
					{\bigl\lVert{\defvec{z}}\bigr\rVert}
					 \Bigr),\!\!
\end{equation}
for a certain $\defvec{z}\!=\![z_1,z_2,\ldots,z_L]^T\in\mathbb{C}^L$, where $\lVert{\defvec{z}}\rVert\!=\!\defvec{z}^H\!\defvec{z}$.
\end{theorem}

\begin{proof}
Referring to the distributional equality between \eqref{Eq:StandardMultivariateMcLeishPartitioning} and \eqref{Eq:StandardMultivariateComplexMcLeishDefinition}, well explained above, we acknowledge that both $\defrvec{S}\!\sim\!\mathcal{M}_{\nu}^{2L}(\defvec{0},\defmat{I})$ and $\defrmat{W}\!\sim\!\mathcal{CM}_{\nu}^L\!\bigl(\defvec{0},\defmat{I}\bigr)$ have the same PDF, i.e. 
\begin{equation}
 f_{\defrmat{S}}(\defvec{x})=
    f_{\defrmat{W}}(\defvec{z}_{I}+\imaginary\defvec{z}_{Q}),
\end{equation}
where $\defvec{z}_{I}\!\in\!\mathbb{R}^{L}$ and $\defvec{z}_{Q}\!\in\!\mathbb{R}^{L}$ such that $\defvec{z}\!=\!\defvec{z}_{I}+\imaginary\defvec{z}_{Q}$ and  
\begin{equation}
    \defvec{x}=
        \begin{bmatrix}
            \defvec{z}_{I}\\
            \defvec{z}_{Q}
        \end{bmatrix}.
\end{equation}
Then, using \theoremref{Theorem:StandardMultivariateMcLeishPDF}, we easily deduce the \ac{PDF} of $\defrmat{W}$ as in \eqref{Eq:StandardMultivariateCCSMcLeishPDF}, which completes the proof of \theoremref{Theorem:StandardMultivariateCCSMcLeishPDF}.
\end{proof}


As observed in \theoremref{Theorem:StandardMultivariateCCSMcLeishPDF}, the \ac{PDF} $f_{\defrmat{W}}(\defvec{z})$ is a function of squared Euclidean norm $\lVert{\defvec{z}}\rVert^2\!=\!\defvec{z}^H\!\defvec{z}$ in complex space. Since a unitary transformation $\defmat{U}$ (i.e., $\defmat{U}\defmat{U}^H\!\!=\!\defmat{U}^H\defmat{U}\!=\!\defmat{I})$ 
preserves the Euclidean norm of all complex vectors (i.e., $\lVert{\defmat{U}\defvec{z}}\rVert\!=\!\lVert{\defvec{z}}\rVert$), we immediately obtain the covariance matrix of $\defmat{U}\defrmat{W}$ as 
\begin{equation}\label{Eq:CovarianceMatrixForUnitaryTransform}
    \mathbb{E}[\defmat{U}\defrmat{W}(\defmat{U}\defrmat{W})^H]=\defmat{U}\mathbb{E}[\defrmat{W}\defrmat{W}^H]\defmat{U}^H\!=2\defmat{I},
\end{equation}
and its pseudo-covariance matrix as
\begin{equation}\label{Eq:PseudoCovarianceMatrixForUnitaryTransform}
    \mathbb{E}[\defmat{U}\defrmat{W}(\defmat{U}\defrmat{W})^T]=\defmat{U}\mathbb{E}[\defrmat{W}\defrmat{W}^T]\defmat{U}^T\!=\defmat{0}.
\end{equation}
These same conclusions are also being drawn for an orthogonal transformations. Further, we notice that 
\begin{subequations}
\setlength\arraycolsep{1.4pt}   
\begin{eqnarray}
\!\!\!\!\trace(\mathbb{E}[\defrvec{W}\defrvec{W}^H])
&=&\trace(\mathbb{E}[\defrvec{S}\defrvec{S}^T]),\\
&=&2\trace(\mathbb{E}[\defrvec{S}_j\defrvec{S}_j^T]),~j\in\{1,2\},\\
&=&2L,
\end{eqnarray}
\end{subequations}
Both \eqref{Eq:CovarianceMatrixForUnitaryTransform} and \eqref{Eq:PseudoCovarianceMatrixForUnitaryTransform} together  impose that $f_{\defrmat{U}\defrmat{W}}(\defvec{z})\!=\!f_{\defrmat{W}}(\defvec{z})$, and therefore $\defmat{U}\defrmat{W}\!\sim\!\mathcal{CM}^{L}_{\nu}(\defvec{0},\defmat{I})$. In addition, for clarity and consistency, we readily rewrite $f_{\defrmat{W}}(\defvec{z})$ in terms of Meijer's G function using \cite[Eq. (8.4.23/1)]{BibPrudnikovBookVol3}, that is
\begin{equation}
    f_{\defrmat{W}}(\defvec{z})=\frac{1}{\pi^{L}\lambda_0^{2L}\Gamma(\nu)}\MeijerG[right]{2,0}{0,2}{\frac{\lVert\defvec{z}\rVert^2}{\lambda_0^2}}{\emptycoefficient}{0,\nu-L}.
\end{equation}
With the aid of whose Mellin-Barnes countour integration\cite[Eq. (8.2.1/1)]{BibPrudnikovBookVol3}, we rewrite 
\begin{equation}\label{Eq:StandardMultivariateCCSMcLeishPDFContourIntegral}
\!\!\!\!f_{\defrmat{W}}(\defvec{z})=\frac{1}{2\pi\imaginary}
        \int_{c-\imaginary\infty}^{c+\imaginary\infty}\frac{\Gamma(s)\Gamma(\nu-L+s)}{\pi^{L}\lambda_0^{2L}\Gamma(\nu)}
            \lVert\defvec{z}\rVert^{-2s}ds
\end{equation}
within the existence region $s\!\in\!\Omega_0$, where $\Omega_0\!=\!\bigl\{s\bigl|\RealPart{s}\!>\!\max(0,L-\nu)\bigr.\bigr\}$. As observing $\defvec{z}\!=\!\defvec{x}+\imaginary\defvec{y}$ and employing both \eqref{Eq:StandardMultivariateCCSMcLeishPDFContourIntegral} and \cite[Eq. (3.241/4)]{BibGradshteynRyzhikBook} together, we have both $\int_{\mathbb{R}^L}f_{\defrmat{W}}(\defvec{x}+\imaginary\defvec{y})\,d\defvec{x}$ and $\int_{\mathbb{R}^L}f_{\defrmat{W}}(\defvec{x}+\imaginary\defvec{y})\,d\defvec{y}$ reduced to \eqref{Eq:StandardMultivariateMcLeishPDF} as intuitively expected. In addition, when $\nu\!=\!1$, \eqref{Eq:StandardMultivariateCCSMcLeishPDF} is then reduced to the \ac{PDF} of standard multivariate \ac{CCS} Laplacian  distribution, that is
\begin{equation}\label{Eq:StandardMultivariateCCSLaplacianPDF}
f_{\defrvec{W}}(\defvec{z})=\frac{1}{2^{(L-1)/2}\pi^L}
	{\lVert{\defvec{z}}\rVert}^{1-L}
			 {K}_{1-L}\Bigl(\sqrt{2}{\lVert{\defvec{z}}\rVert}
					 \Bigr),\!\!
\end{equation}
which simplifies more to \cite[Eq. (5.1.2)]{BibKotzKozubowskiPodgorskiBook2012} for $L\!=\!1$. The other special case, which is obtained when $\nu\!\rightarrow\!\infty$, is 
\begin{equation}\label{Eq:StandardMultivariateCCSGaussianPDF}
    f_{\defrvec{W}}(\defvec{z})=\frac{1}{(2\pi)^L}
        \exp\Bigl(-\frac{1}{2}{\lVert{\defvec{z}}\rVert}^2\Bigr),
\end{equation}
which is the \ac{PDF} of standard multivariate Gaussian distribution \cite[Eq. (2.6-29)]{BibProakisBook} as expected. 

\begin{definition}[McLeish's Multivariate Complex Quantile and Complementary Complex Quantile]
\label{Definition:McLeishMultivariateComplexQAndComplementaryComplexQFunction}
For a fixed $\defvec{z}\!\in\!\mathbb{C}^L$ in higher dimensional complex space, the McLeish's multivariate complex \ac{Q-function} is defined by 
\begin{equation}\label{Eq:McLeishMultivariateComplexQFunction}
    Q^{L}_{\nu}(\defvec{z})=
        Q^{2L}_{\nu}([\RealPart{\defvec{z}}^T,\ImagPart{\defvec{z}}^T]^T),
\end{equation}
and whose complementary complex \ac{Q-function} is defined by 
\begin{equation}\label{Eq:McLeishMultivariateComplementaryComplexQFunction}
    \widehat{Q}^{L}_{\nu}(\defvec{z})=
        Q^{2L}_{\nu}([\RealPart{\defvec{z}}^T,\ImagPart{\defvec{z}}^T]^T),
\end{equation}
where $Q^{2L}_{\nu}(\defvec{x})$ and $\widehat{Q}^{2L}_{\nu}(\defvec{x})$, defined for real vectors $\defvec{x}\!\in\!\mathbb{R}^L$, are given in \eqref{Eq:McLeishMultivariateQFunction} and \eqref{Eq:McLeishMultivariateComplementaryQFunction}, respectively.  
\end{definition}

As we mentioned above, referring to both \eqref{Eq:StandardMultivariateMcLeishPartitioning} and \eqref{Eq:StandardMultivariateComplexMcLeishDefinition} together, we have $f_{\defrvec{W}}(\defvec{z})\!=\!f_{\defrvec{S}}(\defvec{z})$. Therefore, we can readily obtain the \ac{CDF} and \ac{CCDF} of $\defrmat{W}\!\sim\!\mathcal{CM}^{L}_{\nu}(\defvec{0},\defmat{I})$, especially by using \theoremref{Theorem:StandardMultivariateMcLeishCDF} and \theoremref{Theorem:StandardMultivariateMcLeishCCDF}, respectively. Accordingly, the \ac{CDF} of $\defrmat{W}\!\sim\!\mathcal{CM}^{L}_{\nu}(\defvec{0},\defmat{I})$ is properly defined in complex space by $F_{\defrmat{W}}(\defvec{z})\!=\!\Pr\{\defrvec{W}\!\leq\!\defvec{z}\}\!=\!\allowbreak\Pr\{{W}_{1}\!\leq\!z_1,\allowbreak{W}_{2}\!\leq\!z_2,\ldots,{W}_{L}\!\leq\!z_L\}$ and obtained in the following.  

\begin{theorem}\label{Theorem:StandardMultivariateCCSMcLeishCDF}
The \ac{CDF} of $\!\defrmat{W}\!\sim\!\mathcal{CM}_{\nu}^L\!\bigl(\defvec{0},\defmat{I}\bigr)$ is given by
\begin{equation}\label{Eq:StandardMultivariateCCSMcLeishCDF}
    F_{\defrmat{W}}(\defvec{z})=
        \widehat{Q}^{L}_{\nu}\bigl(\defvec{z}\bigr),
\end{equation}
defined over $\boldsymbol{z}\!\in\!\mathbb{C}^L$, where $\widehat{Q}^{L}_{\nu}\bigl(\defvec{z}\bigr)$ is given in \eqref{Eq:McLeishMultivariateComplementaryComplexQFunction}.
\end{theorem}

\begin{proof}
From the distributional equality between between \eqref{Eq:StandardMultivariateMcLeishPartitioning} and \eqref{Eq:StandardMultivariateComplexMcLeishDefinition}, the proof is obvious using \eqref{Eq:McLeishMultivariateComplementaryComplexQFunction}.
\end{proof}

The \ac{CCDF} of $\defrmat{W}\!\sim\!\mathcal{CM}^{L}_{\nu}(\defvec{0},\defmat{I})$ is defined by $\widehat{F}_{\defrmat{W}}(\defvec{z})\!=\!\Pr\{\defrvec{W}>\!\defvec{z}\}\!=\!\allowbreak\Pr\{{W}_{1}\!>\!z_1,\allowbreak{W}_{2}\!>\!z_2,\ldots,{W}_{L}\!>\!z_L\}$ and obtained in the following.

\begin{theorem}\label{Theorem:StandardMultivariateCCSMcLeishCCDF}
The \ac{CCDF} of $\!\defrmat{W}\!\sim\!\mathcal{CM}_{\nu}^L\!\bigl(\defvec{0},\defmat{I}\bigr)$ is given by
\begin{equation}\label{Eq:StandardMultivariateCCSMcLeishCCDF}
\widehat{F}_{\defrmat{W}}(\defvec{x})=
    {Q}^{L}_{\nu}\bigl(\defvec{z}\bigr),
\end{equation}
defined over $\boldsymbol{z}\!\in\!\mathbb{C}^L$, where ${Q}^{L}_{\nu}\bigl(\defvec{z}\bigr)$ is given in \eqref{Eq:McLeishMultivariateComplexQFunction}.
\end{theorem}

\begin{proof}
The proof is obvious using \eqref{Eq:McLeishMultivariateComplexQFunction}.
\end{proof}

In $L$-dimensional complex space $\defvec{s}\!\in\!\mathbb{C}^L$, we can define~the \ac{MGF} by $M_{\defrmat{W}}(\defvec{s})\!=\!\allowbreak\mathbb{E}\bigl[\exp(-\langle\defvec{s},\defrmat{W}\rangle)\bigr]$ that uniquely determines the distribution of $\defrmat{W}\!\sim\!\mathcal{CM}^{L}_{\nu}(\defvec{0},\defmat{I})$ and is obtained in the following.

\begin{theorem}\label{Theorem:StandardMultivariateCCSMcLeishMGF}
The \ac{MGF} of $\defrmat{W}\!\sim\!\mathcal{CM}^{L}_{\nu}(\defvec{0},\defmat{I})$ is given by 
\begin{equation}\label{Eq:StandardMultivariateCCSMcLeishMGF}
M_{\defrmat{W}}(\defvec{s})=\Bigl(1-\frac{\lambda^2_{0}}{4}\defvec{s}^H\defvec{s}\Bigr)^{-\nu},
\end{equation}
for a certain $\defvec{s}\!\in\!\mathbb{C}^L$ within the existence region $\defvec{s}\!\in\!\mathbb{C}_{0}$, where the region $\mathbb{C}_{0}$ is given by 
\begin{equation}\label{Eq:StandardMultivariateCCSMcLeishMGFExistenceRegion}
    \mathbb{C}_0=\Bigl\{
        \defvec{s}\,    
        \Bigl|\,
        \lambda^2_{0}\defvec{s}^H\defvec{s}
        \leq
        {4}
        \Bigr.
        \Bigr\}.
\end{equation}
\end{theorem}

\begin{proof}
Following the same logic presented in the proof of \theoremref{Theorem:StandardMultivariateCCSMcLeishPDF}, and noticing that MGF uniquely determines the distributions, we can conclude that the distributional equality between \eqref{Eq:StandardMultivariateMcLeishPartitioning} and \eqref{Eq:StandardMultivariateComplexMcLeishDefinition} also makes both $\defrvec{S}\!\sim\!\mathcal{M}_{\nu}^{2L}(\defvec{0},\defmat{I})$ and $\defrmat{W}\!\sim\!\mathcal{CM}_{\nu}^L\!\bigl(\defvec{0},\defmat{I}\bigr)$ have the same MGF, i.e. 
\begin{equation}
 \mathcal{M}_{\defrmat{S}}(\hat{\defvec{s}})=
    \mathcal{M}_{\defrmat{W}}(\defvec{s}_{I}+\imaginary\defvec{s}_{Q}),
\end{equation}
where $\defvec{x}\!\in\!\mathbb{R}^{L}$ and $\defvec{y}_{Q}\!\in\!\mathbb{R}^{L}$ such that $\defvec{s}\!\in\!\mathbb{R}^{2L}$, that is
\begin{equation}
    \hat{\defvec{s}}=
        \begin{bmatrix}
            \defvec{s}_{I}\\
            \defvec{s}_{Q}
        \end{bmatrix}.
\end{equation}
Then, using \theoremref{Theorem:StandardMultivariateMcLeishMGF}, we easily deduce the \ac{MGF} of $\defrmat{W}$ as in \eqref{Eq:StandardMultivariateCCSMcLeishMGF}, which completes the proof of \theoremref{Theorem:StandardMultivariateCCSMcLeishMGF}.
\end{proof}

Let us have a vector of \acf{u.n.i.d.} \ac{CCS} McLeish distributions, that is 
\begin{equation}\label{Eq:INIDMultivariateCESMcLeishRandomVector}
    \defrmat{Z}=[Z_1,Z_2,\ldots,Z_L]^T,
\end{equation}
where $Z_{\ell}\!=\!{X}_{\ell}+\imaginary{Y}_{\ell}$ such that $X_{\ell}\!\sim\!\mathcal{M}_{\nu}(0,\sigma^2_{\ell})$ and $Y_{\ell}\!\sim\!\allowbreak\mathcal{M}_{\nu}(0,\sigma^2_{\ell})$ (i.e., $Z_{\ell}\!\sim\!\mathcal{CM}_{\nu}(0,\sigma^2_{\ell})$), 
$1\!\leq\!\ell\!\leq\!L$. Furthermore, we assume $\Covariance{X_{k}}{X_{\ell}}\!=\!{0}$ and $\Covariance{Y_{k}}{Y_{\ell}}\!=\!{0}$ for all $k\!\neq\!\ell$, and more $\Covariance{X_{k}}{Y_{\ell}}\!=\!{0}$ for all $1\!\leq\!k,\ell\!\leq\!L$. In accordance with the definition of multivariate distribution, $\defrmat{Z}$ follows a multivariate \ac{CES} McLeish distribution because $\defvec{a}^T\defrmat{Z}$ for all $\defvec{a}\!\in\!\mathbb{C}^L$ follows a McLeish distribution. It is then worth noticing that $\Variance{X_{\ell}}\!=\!\Variance{Y_{\ell}}\!=\!\sigma^2_{\ell}$ and 
and $\Variance{Z_{\ell}}\!=\!\Variance{X_{\ell}}+\Variance{Y_{\ell}}\!=\!2\sigma^2_{\ell}$. Herewith, as similar to what defined before, let us define $\defvec{\sigma}^2\!=\![\sigma^2_1,\sigma^2_2,\ldots,\sigma^2_L]^T\!$, and therefrom $\defvec{\sigma}\!=\![\sigma_1,\sigma_2,\ldots,\sigma_L]^T$. 
Owing to processing $\defvec{\sigma}^T\defrmat{W}\!\sim\!\mathcal{M}_{\nu}(0,\defvec{\sigma}^T\defvec{\sigma})$, we conclude that $\defrmat{Z}$ certainly follows a multivariate \ac{CES} McLeish distribution~with~a~diagonal~covariance matrix, denoted by $\defrmat{Z}\!\sim\!\mathcal{CM}^L_{\nu}(\defvec{0},\diag(\defvec{\sigma}^2))$. Thus, we can decompose $\defrmat{Z}\!\sim\!\mathcal{CM}^L_{\nu}(\defvec{0},\diag(\defvec{\sigma}^2))$~as~an~affine~transformation of standard multivariate \ac{CCS} McLeish distribution as shown in the following theorem.

\begin{theorem}\label{Theorem:INIDMultivariateCESMcLeishDefinition}
A multivariate \ac{CCS} McLeish distribution of uncorrelated and not identically distributed \ac{CCS} McLeish distributions, denoted by 
$\defrmat{Z}\!\sim\!\mathcal{CM}^L_{\nu}(\defvec{0},\diag(\defvec{\sigma}^2))$, is decomposed as
\begin{equation}\label{Eq:INIDMultivariateCESMcLeishDefinition}
    \defrmat{Z}=\diag(\defvec{\sigma})\defrmat{W}.
\end{equation}
where $\defrmat{W}\!\sim\!\mathcal{CM}_{\nu}(0,\defvec{I})$. 
\end{theorem}

\begin{proof}
The proof is obvious using the fact that, for all $\defvec{a}\in\mathbb{C}^{L}$, we have $\defvec{a}^T\defrmat{Z}\!\sim\!\mathcal{M}_{\nu}(0,\defvec{a}^T\!\diag(\defvec{\sigma})\,\defvec{a})$ with the pivotal details mentioned before \theoremref{Theorem:INIDMultivariateMcLeishDefinition}.
\end{proof}

Accordingly, the \ac{PDF}, \ac{CDF}, \ac{CCDF} and \ac{MGF} of a multivariate \acf{CES} McLeish distribution, denoted by $\defrmat{Z}\!\sim\!\mathcal{CM}^L_{\nu}(\defvec{0},\diag(\defvec{\sigma}^2))$, is given in the following.

\begin{theorem}\label{Theorem:INIDMultivariateCESMcLeishPDF}
\!\!The \ac{PDF} of $\defrmat{Z}\!\sim\!\mathcal{CM}^L_{\nu}(\defvec{0},\!\diag(\defvec{\sigma}^2))$~is~given~by
\begin{equation}\label{Eq:INIDMultivariateCESMcLeishPDF}
\!\!\!\!f_{\defrmat{Z}}(\defvec{z})=
		\frac{2}{\pi^{L}}
		\frac{{\bigl\lVert{\defmat{\Lambda}^{-1}\defvec{z}}\bigr\rVert}^{\nu-{L}}}
			{\Gamma(\nu)\det(\defmat{\Lambda})}
				{K}_{\nu-{L}}\Bigl({2}{\bigl\lVert{\defmat{\Lambda}^{-1}\defvec{z}}\bigr\rVert}\Bigr),
\end{equation}
for a certain $\defvec{z}\!=\![z_1,z_2,\ldots,z_L]^T\!\in\!\mathbb{C}^L$, where $\defrmat{\Lambda}\!=\!\diag(\defvec{\lambda})$
and $\defvec{\lambda}\!=\!\lambda_{0}\,\defvec{\sigma}$ denotes the component deviation vector.
\end{theorem}

\begin{proof}
Note that, using \eqref{Eq:INIDMultivariateCESMcLeishDefinition}, we can write $\defrmat{W}\!=\!\diag(\defvec{\sigma})^{-1}\defrvec{Z}$ and therefrom obtain its Jacobian $J_{\defrmat{W}|\defrmat{Z}}\!=\!\det(\diag(\defvec{\sigma})^{-1})$. We can write the \ac{PDF} of $\defrmat{X}$ as 
\begin{subequations}
\setlength\arraycolsep{1.4pt}
\begin{eqnarray}
    f_{\defrmat{Z}}(\defvec{z})&=&
        f_{\defrmat{W}}(\diag(\defvec{\sigma})^{-1}\defrvec{z})\,
            J_{\defrmat{W}|\defrmat{Z}},\\
    &=&f_{\defrmat{W}}(\diag(\defvec{\sigma})^{-1}\defrvec{z})\,
            \det(\diag(\defvec{\sigma})^{-1}),
\end{eqnarray}
\end{subequations}
where substituting \eqref{Eq:StandardMultivariateCCSMcLeishPDF} and utilizing both $\det(\diag(\defvec{\sigma})^{-1})\!=\!\det(\diag(\defvec{\sigma}))^{-1}$ and $\det(\diag(\defvec{\sigma})^2)\!=\!\det(\diag(\defvec{\sigma}))^{2}$
yields \eqref{Eq:INIDMultivariateCESMcLeishPDF}, which completes the proof of 
\theoremref{Theorem:StandardMultivariateCCSMcLeishPDF}.
\end{proof}

Note that, for consistency and clarity, setting $\diag(\defvec{\sigma}^2)\!=\!\sigma^2\defmat{I}$ (i.e., making each component have equal power) reduces \eqref{Eq:INIDMultivariateCESMcLeishPDF} to the \ac{PDF} of $\defrmat{Z}\!\sim\!\mathcal{CM}^L_{\nu}(\defvec{0},\!\sigma^2\defmat{I})$~given~by
\begin{equation}\label{Eq:INIDMultivariateCESMcLeishPDFWithUniformVariances}
\!\!\!\!f_{\defrmat{Z}}(\defvec{z})=
		\frac{2}{\pi^{L}}
		\frac{{\bigl\lVert{\defvec{z}}\bigr\rVert}^{\nu-{L}}}
			{\Gamma(\nu)\lambda^{\nu+{L}}}
				{K}_{\nu-{L}}\Bigl(\frac{2}{\lambda}{\bigl\lVert{\defvec{z}}\bigr\rVert}\Bigr),
\end{equation}
where $\lambda\!=\!\sqrt{2\sigma^2/\nu}$ as defined before.

\begin{theorem}\label{Theorem:INIDMultivariateCESMcLeishCDF}
\!\!The \ac{CDF} of $\defrmat{Z}\!\sim\!\mathcal{CM}^L_{\nu}(\defvec{0},\!\diag(\defvec{\sigma}^2))$~is~given~by
\begin{equation}\label{Eq:INIDMultivariateCESMcLeishCDF}
    F_{\defrmat{Z}}(\defvec{z})=
        \widehat{Q}^{L}_{\nu}\bigl(\lambda_{0}\defmat{\Lambda}^{-1}\defvec{z}\bigr),
\end{equation}
defined over $\boldsymbol{z}\!\in\!\mathbb{C}^L$.
\end{theorem}

\begin{proof}
With the aid of the distributional relation between $\defrmat{Z}\!\sim\!\mathcal{CM}^L_{\nu}(\defvec{0},\!\diag(\defvec{\sigma}^2))$ and $\defrvec{W}\!\sim\!\mathcal{CM}_{\nu}^{L}(\defvec{0},\defmat{I})$, presented in \eqref{Eq:INIDMultivariateCESMcLeishDefinition}, we have $\defrvec{W}=\diag(\defvec{\sigma})^{-1}\defrvec{Z}$ and therefrom write
\begin{subequations}\label{Eq:INIDMultivariateCESMcLeishCDFTransform}
\setlength\arraycolsep{1.4pt}
\begin{eqnarray}
    \label{Eq:INIDMultivariateCESMcLeishCDFTransformA}
    F_{\defrmat{Z}}(\defvec{z})
        &=&F_{\defrmat{W}}(\defvec{w}),\\
    \label{Eq:INIDMultivariateCESMcLeishCDFTransformB}
        &=&F_{\defrmat{W}}(\diag(\defvec{\sigma})^{-1}\defvec{z}), 
\end{eqnarray}
\end{subequations}
Finally, substituting the \ac{CDF} $F_{\defrmat{W}}(\defvec{z})$, which is given in \eqref{Eq:StandardMultivariateCCSMcLeishCDF}, into \eqref{Eq:INIDMultivariateCESMcLeishCDFTransformB} and therein using $\diag(\defvec{\sigma})^{-1}\!=\!\lambda_{0}\defmat{\Lambda}^{-1}$, we readily obtain \eqref{Eq:INIDMultivariateCESMcLeishCDF}, which~proves~\theoremref{Theorem:INIDMultivariateCESMcLeishCDF}. 
\end{proof}

\begin{theorem}\label{Theorem:INIDMultivariateCESMcLeishCCDF}
\!\!The\!~\ac{CCDF}\!~of\!~$\defrmat{Z}\!\sim\!\mathcal{CM}^L_{\nu}(\defvec{0},\!\diag(\defvec{\sigma}^2))$~is~given~by\!
\begin{equation}\label{Eq:INIDMultivariateCESMcLeishCCDF}
    \widehat{F}_{\defrmat{X}}(\defvec{z})=
        {Q}^{L}_{\nu}\bigl(\lambda_{0}\defmat{\Lambda}^{-1}\defvec{z}\bigr),
\end{equation}
defined over $\boldsymbol{z}\!\in\!\mathbb{C}^L$.
\end{theorem}

\begin{proof}
The proof is obvious using \eqref{Eq:StandardMultivariateCCSMcLeishCCDF} and  \theoremref{Theorem:StandardMultivariateCCSMcLeishCCDF} and then performing almost same steps followed in the proof of \theoremref{Theorem:INIDMultivariateCESMcLeishCDF}.
\end{proof}

\begin{theorem}\label{Theorem:INIDMultivariateCESMcLeishMGF}
\!\!The \ac{MGF} of $\defrmat{Z}\!\sim\!\mathcal{CM}^L_{\nu}(\defvec{0},\!\diag(\defvec{\sigma}^2))$~is~given~by
\begin{equation}\label{Eq:INIDMultivariateCESMcLeishMGF}
    M_{\defrmat{Z}}(\defvec{s})=\Bigl(1-\frac{1}{4}\defvec{s}^H\defmat{\Lambda}^{2}\defvec{s}\Bigr)^{-\nu},
\end{equation}
for a certain $\defvec{s}\!\in\!\mathbb{C}^L$ within the existence region $\defvec{s}\!\in\!\mathbb{C}_{0}$, where the region $\mathbb{C}_{0}$ is given by 
\begin{equation}\label{Eq:INIDMultivariateCESMcLeishMGFExistenceRegion}
    \mathbb{C}_0=\Bigl\{
        \defvec{s}\,    
        \Bigl|\,
        \defvec{s}^H\defmat{\Lambda}^{2}\defvec{s}
        \leq
        {4}
        \Bigr.
        \Bigr\}.
\end{equation}
\end{theorem}

\begin{proof}
We can write the \ac{MGF} of $\defrmat{Z}\!\sim\!\mathcal{CM}^L_{\nu}(\defvec{0},\!\diag(\defvec{\sigma}^2))$ as $M_{\defrmat{Z}}(\defvec{s})\!=\!\mathbb{E}[\exp\bigl(-\langle\defvec{s},\defrmat{Z}\rangle\bigr)]$, where putting \eqref{Eq:INIDMultivariateCESMcLeishDefinition} gives 
\setlength\arraycolsep{1.4pt}
\begin{eqnarray}
    M_{\defrmat{Z}}(\defvec{s})
        &=&\mathbb{E}[\exp\bigl(-\langle\defvec{s},\diag(\defvec{\sigma})\defrmat{W}\rangle\bigr)],\\
        &=&\mathbb{E}[\exp\bigl(-\langle\diag(\defvec{\sigma})\defvec{s},\defrmat{W}\rangle\bigr)],
\end{eqnarray}
and therefrom we conclude that $M_{\defrmat{Z}}(\defvec{s})\!=\!M_{\defrmat{W}}(\diag(\defvec{\sigma})\defvec{s})$, where $M_{\defrmat{W}}(\defvec{s})$ denotes the \ac{MGF} of $\defrmat{W}$ and is given in \eqref{Eq:StandardMultivariateCCSMcLeishMGF}. Finally, substituting $\diag(\defvec{\sigma})\defvec{s}\!=\!\defmat{\Lambda}\defvec{s}/\lambda_{0}$ into \eqref{Eq:StandardMultivariateCCSMcLeishMGF} results in \eqref{Eq:INIDMultivariateCESMcLeishMGF}, which completes the proof of \theoremref{Theorem:INIDMultivariateCESMcLeishMGF}.
\end{proof}

In what follows, the most general case in which we assume that complex McLeish distributions are mutually correlated and non-identically distributed is investigated using the results obtained previously. Referring to \eqref{Eq:StandardMultivariateComplexMcLeishDefinition}, let us have a random vector of complex McLeish distributions given as 
\begin{equation}\label{Eq:MultivariateComplexMcLeishDefinition}
	\defrvec{Z}\!=\!\defrvec{X}_1+\imaginary\defrvec{X}_2,
\end{equation}
where $\defrvec{X}_1\!\sim\!\mathcal{M}_{\nu_1}^{L}(\defvec{\mu}_{1},\defmat{\Sigma}_{11})$ and $\defrvec{X}_2\!\sim\!\mathcal{M}_{\nu_2}^{L}(\defvec{\mu}_{2},\defmat{\Sigma}_{22})$.~Moreover, we assume that both $\defrvec{X}_1$ and $\defrvec{X}_2$ are without loss of generality correlated with each other, i.e., 
\setlength\arraycolsep{1.4pt} 
\begin{eqnarray}
    \defmat{\Sigma}_{12}&=&
        \mathbb{E}[(\defrvec{X}_{1}-\defvec{\mu}_{1})\allowbreak(\defrvec{X}_{2}-\defvec{\mu}_{2})^T]\neq\defmat{0},\\
    \defmat{\Sigma}_{21}&=&
        \mathbb{E}[(\defrvec{X}_{2}-\defvec{\mu}_{2})(\defrvec{X}_{1}-\defvec{\mu}_{1})^T]\neq\defmat{0}.
\end{eqnarray}
As noticing the mean vector of $\defrvec{Z}$ is readily obtained as $\defvec{\mu}\!=\!\mathbb{E}[\defrvec{Z}]\!=\!\defvec{\mu}_1+\defvec{\mu}_2$, then we properly write its pseudo-covariance matrix as follows
\begin{equation}\label{Eq:MultivariateComplexMcLeishPseudoCovariance}
\mathbb{E}[(\defrvec{Z}-\defvec{\mu})(\defrvec{Z}-\defvec{\mu})^T]
    =\defmat{\Sigma}_{11}-\defmat{\Sigma}_{22}+\imaginary(\defmat{\Sigma}_{12}+\defmat{\Sigma}_{21}),
\end{equation}
and its covariance matrix as follows
\begin{equation}\label{Eq:MultivariateComplexMcLeishCovariance}
\mathbb{E}[(\defrvec{Z}-\defvec{\mu})(\defrvec{Z}-\defvec{\mu})^H]
    =\defmat{\Sigma}_{11}+\defmat{\Sigma}_{22}+\imaginary(\defmat{\Sigma}_{12}-\defmat{\Sigma}_{21}),
\end{equation}
We acknowledge that circular symmetry for McLeish random vectors is more detailed than circular symmetry for individual McLeish distributions. For preserving the circularly symmetry around the mean \cite{BibGallagerPRPRT2008}, i.e., in order to have the components of $\defrvec{X}_1$ become circular to those of $\defrvec{X}_2$, we should provide that, as well explained in  \cite{BibGallagerPRPRT2008}, $\mathbb{E}[(\defrvec{Z}-\defvec{\mu})(\defrvec{Z}-\defvec{\mu})^T]$ has to be a null matrix \cite{BibGallagerPRPRT2008}. For that purpose, we strictly impose from \eqref{Eq:MultivariateComplexMcLeishPseudoCovariance} that $\defmat{\Sigma}_{11}\!=\!\defmat{\Sigma}_{22}\!=\!\defmat{R}$ and $\defmat{\Sigma}_{12}\!=\!-\defmat{\Sigma}_{21}\!=\!\defmat{J}$. Accordingly, we have 
\setlength\arraycolsep{1.4pt} 
\begin{eqnarray}
    \mathbb{E}[(\defrvec{Z}-\defvec{\mu})(\defrvec{Z}-\defvec{\mu})^T]&=&\defmat{0},\\
    \mathbb{E}[(\defrvec{Z}-\defvec{\mu})(\defrvec{Z}-\defvec{\mu})^H]&=&2(\defmat{R}+\imaginary\defmat{J})=2\defmat{\Sigma},
\end{eqnarray}
where $\defmat{\Sigma}\!=\!\defmat{R}+\imaginary\defmat{J}$ such that $\defmat{\Sigma}$ is a complex symmetric matrix (i.e., $\defmat{\Sigma}^H\!\!=\!\defmat{\Sigma}$). Furthermore, we acknowledge that $\ImagPart{\defmat{\Sigma}}\!=\!\defmat{0}$ when $\defmat{\Sigma}_{12}\!=\!\defmat{\Sigma}_{21}\!=\!\defmat{0}$. 
By the definition of multivariate distribution \cite{BibAndersenBook1995,BibTongBook1989,BibWangKotzNgBook1989,BibBilodeauBrennerBook1999,BibKotzBalakrishnanBook2004}, $\defrmat{Z}$ is a multivariate complex distribution iff $\defvec{a}^{T}\defrmat{Z}$ for all $\defvec{a}\!\in\!\mathbb{C}^{L}$ follows a complex random distribution of the same family. Taking into account this definition, and pursuant to what presented in \secref{Section:StatisticalBackground:MultivariateMcLeishDistribution} above, we note that $\defrmat{Z}$ follows a multivariate complex distribution only when $\nu_1\!=\!\nu_2\!=\!\nu$ with $\defmat{\Sigma}_{11}\!=\!\defmat{\Sigma}_{22}$ and $\defmat{\Sigma}_{12}\!=\!-\defmat{\Sigma}_{21}$.  Since being an Hermitian positive definite matrix, $\defmat{\Sigma}$ is decomposed using Cholesky decomposition as 
\begin{equation}\label{Eq:MultivariateComplexMcLeishCovarianceMatrixCholeskyDecomposition}
    \defmat{\Sigma}=\defmat{D}\defmat{D}^{H}.
\end{equation}
When there is no correlation between quadrature and inphase components of $\defrvec{Z}$ (i.e., when $\defmat{\Sigma}_{12}\!=\!\defmat{\Sigma}_{21}\!=\!\defmat{0}$), we have $\defmat{J}\!=\!\defmat{0}$, and therefrom $\defmat{D}=\defmat{\Sigma}^{-{1}/{2}}$. We conclude that
\begin{equation}\label{Eq:MultivariateComplexMcLeishComponentDecomposition}
    \defrvec{X}_1=\sqrt{G}\defmat{D}\defrvec{N}_1
    \text{~and~}
    \defrvec{X}_2=\sqrt{G}\defmat{D}\defrvec{N}_2,    
\end{equation}
where $\defrvec{N}_1\!\sim\!\mathcal{N}^L(0,\defmat{I})$, $\defrvec{N}_2\!\sim\!\mathcal{N}^L(0,\defmat{I})$
and $G\!\sim\!\mathcal{G}(\nu,1)$. As a consequence, $\defrmat{Z}$ follows a multivariate \ac{CES} McLeish distribution, denoted by $\defrmat{Z}\!\sim\!\mathcal{CM}_{\nu}^{L}(\defvec{\mu},\defmat{\Sigma})$, whose decomposition is given in the following.

\begin{theorem}\label{Theorem:MultivariateCESMcLeishDecomposition}
If $\defrmat{Z}\!\sim\!\mathcal{CM}_{\nu}^L\bigl(\defvec{\mu},\defmat{\Sigma}\bigr)$, then it is decomposed as 
\begin{equation}\label{Eq:MultivariateCESMcLeishDecomposition}
	\defrmat{Z}=\defmat{D}\defrmat{W}+\defvec{\mu},
\end{equation}
where $\defrmat{W}\!\sim\!\mathcal{CM}_{\nu}^L(\defvec{0},\defmat{I})$. Further, $\defmat{D}$, given in \eqref{Eq:MultivariateComplexMcLeishCovarianceMatrixCholeskyDecomposition}, is the Cholesky decomposition of $\defmat{\Sigma}$.
\end{theorem}

\begin{proof}
The proof is obvious using the pivotal details mentioned before \theoremref{Theorem:MultivariateCESMcLeishDecomposition}.
\end{proof}

Accordingly, the \ac{PDF} of multivariate \ac{CES} McLeish distribution with a covariance matrix is given in the following. 

\begin{theorem}\label{Theorem:MultivariateCESMcLeishPDF}
The \ac{PDF} of $\defrmat{Z}\!\sim\!\mathcal{CM}_{\nu}^L(\defvec{\mu},\defmat{\Sigma})$ is given by
\begin{equation}\label{Eq:MultivariateCESMcLeishPDF}
\!\!\!\!f_{\defrmat{Z}}(\defvec{z})=\frac{2}{\pi^L\Gamma(\nu)}
	\frac{{\lVert{\defvec{z}-\defvec{\mu}}\rVert}_{\defmat{\Sigma}}^{\nu-L}}
		{\det(\defmat{\Sigma})\,\lambda_{0}^{\nu+L}}
					 {K}_{\nu-L}
					 \Bigl(
							\frac{2}{\lambda_{0}}
					{\bigl\lVert{\defvec{z}-\defvec{\mu}}\bigr\rVert}_{\defmat{\Sigma}}
					 \Bigr),\!\!
\end{equation}
defined in $\boldsymbol{z}\!\in\!\mathbb{C}^L$,~where~${\lVert{\defvec{z}-\defvec{\mu}}\rVert}_{\defmat{\Sigma}}\!=\!(\defvec{z}-\defvec{\mu})^{H}\defmat{\Sigma}^{-1}(\defvec{z}-\defvec{\mu})$.
\end{theorem}

\begin{proof}
Note that $\defrmat{Z}\!\sim\!\mathcal{CM}_{\nu}^L\!\bigl(\defvec{\mu},\defmat{\Sigma}\bigr)$ is, as observed~in~\eqref{Eq:MultivariateCESMcLeishDecomposition}, described by an affine transformation of $\defrmat{W}\!\sim\!\mathcal{CM}_{\nu}^L\!\bigl(\defvec{0},\defmat{I}\bigr)$. Appropriately, using  $\defmat{\Sigma}\!=\!\defmat{D}\defmat{D}^{H}$, we have 
\begin{equation}
    \defrmat{W}=\defmat{D}^{-1}(\defrvec{Z}-\defvec{\mu})    
\end{equation}
and therefrom find the Jacobian $J_{\defrmat{Z}|\defrmat{W}}\!=\!\det(\defmat{D})$ and $J_{\defrmat{W}|\defrmat{Z}}\!=\allowbreak\!\det(\defmat{D})^{-1}$ Then, using $\det(\defmat{\Sigma})=\det(\defmat{D})^2$, we have the \ac{PDF} of $\defrmat{Z}$ using \eqref{Eq:StandardMultivariateCCSMcLeishPDF}, i.e.,
\begin{equation}\label{Eq:MultivariateCESMcLeishPDFTransform}
    f_{\defrmat{Z}}(\defvec{z})=f_{\defrmat{W}}(\defmat{D}^{-1}(\defrvec{z}-\defvec{\mu}))
        J_{\defrmat{W}|\defrmat{Z}}.
\end{equation}
Finally, using $\defmat{\Sigma}\!=\!\defmat{\Sigma}^H$ with these results, substituting \eqref{Eq:MultivariateCESMcLeishPDF} into \eqref{Eq:MultivariateCESMcLeishPDFTransform} results in \eqref{Eq:MultivariateMcLeishPDF}, which proves \theoremref{Theorem:MultivariateCESMcLeishPDF}.
\end{proof}

For consistency and clarity, note that the complex covariance matrix $\defmat{\Sigma}$ can also be rewritten as $\defmat{\Sigma}\!=\!\lambda^{-2}_{0}\,\defmat{\Lambda}\defmat{P}\defmat{\Lambda}$, where $\defmat{\Lambda}\!=\!\lambda_{0}\diag(\defvec{\sigma})\!=\!\diag(\lambda_1,\lambda_2,\ldots,\lambda_L)$ is previously defined. Moreover, $\defmat{P}\!\in\!\mathbb{C}^{{L}\times{L}}$ denotes the complex correlation matrix. When the variance of all the components are the same (i.e., when $\sigma^2_{\ell}\!=\!\sigma^2$, and thus $\lambda_{\ell}\!=\!\lambda\!=\!\sqrt{2\sigma^2/\nu}$, ${1}\!\leq\!\ell\!\leq\!{L}$), we have $\defmat{\Sigma}\!=\!\lambda^2\defmat{P}$ and $\det(\defmat{\Sigma})\!=\!\lambda^{2L}\det(\defmat{P})$, and correspondingly 
simplify \eqref{Eq:MultivariateCESMcLeishPDF} to
\begin{equation}\label{Eq:MultivariateCESIdenticallyDistributedMcLeishPDF}
\!\!\!\!f_{\defrmat{Z}}(\defvec{z})=\frac{2}{\pi^L\Gamma(\nu)}
	\frac{{\lVert{\defvec{z}-\defvec{\mu}}\rVert}_{\defmat{P}}^{\nu-L}}
		{\det(\defmat{P})\lambda^{\nu+L}}
					 {K}_{\nu-L}
					 \Bigl(
							\frac{2}{\lambda}
					    {\bigl\lVert{\defvec{z}-\defvec{\mu}}\bigr\rVert}_{\defmat{P}}
					 \Bigr).\!\!
\end{equation}
In addition, in case of no correlation and zero mean (i.e., when $\defmat{P}\!=\!\defmat{I}$ and $\defvec{\mu}\!=\!\defvec{0}$), we also simplify \eqref{Eq:MultivariateCESMcLeishPDF} to \eqref{Eq:INIDMultivariateCESMcLeishPDF} as expected. 

\begin{theorem}\label{Theorem:MultivariateCESMcLeishCDF}
The \ac{CDF} of $\defrmat{Z}\!\sim\!\mathcal{CM}_{\nu}^L(\defvec{\mu},\defmat{\Sigma})$ is given by
\begin{equation}\label{Eq:MultivariateCESMcLeishCDF}
    F_{\defrmat{Z}}(\defvec{z})=
        \widehat{Q}^{L}_{\nu}\bigl(\defmat{D}(\defrvec{z}-\defvec{\mu})\bigr),
\end{equation}
defined over $\boldsymbol{z}\!\in\!\mathbb{C}^L$, where $\defmat{D}$ is given in \eqref{Eq:MultivariateComplexMcLeishCovarianceMatrixCholeskyDecomposition}.
\end{theorem}

\begin{proof}
Following almost the same steps presented~in~the~proof of \theoremref{Theorem:INIDMultivariateCESMcLeishCDF}, the proof is quite obvious. Specifically, from \eqref{Eq:MultivariateCESMcLeishDecomposition}, we have $F_{\defrmat{Z}}(\defvec{z})\!=\!F_{\defrmat{W}}(\defvec{w})$ with $\defrvec{w}\!=\!\defmat{D}(\defrvec{z}-\defvec{\mu})$, where substituting the \ac{CDF} $F_{\defrmat{W}}(\defvec{z})$, given in \eqref{Eq:StandardMultivariateCCSMcLeishCDF}, we readily obtain \eqref{Eq:MultivariateCESMcLeishCDF}, which proves \theoremref{Theorem:MultivariateCESMcLeishCDF}. 
\end{proof}

\begin{theorem}\label{Theorem:MultivariateCESMcLeishCCDF}
The \ac{CCDF} of $\defrmat{Z}\!\sim\!\mathcal{CM}_{\nu}^L(\defvec{\mu},\defmat{\Sigma})$ is given by
\begin{equation}\label{Eq:MultivariateCESMcLeishCCDF}
    \widehat{F}_{\defrmat{X}}(\defvec{z})=
        {Q}^{L}_{\nu}\bigl(\defmat{D}(\defrvec{z}-\defvec{\mu})\bigr),
\end{equation}
defined over $\boldsymbol{z}\!\in\!\mathbb{C}^L$, where $\defmat{D}$ is given in \eqref{Eq:MultivariateComplexMcLeishCovarianceMatrixCholeskyDecomposition}.
\end{theorem}

\begin{proof}
The proof is obvious using \eqref{Eq:MultivariateCESMcLeishDecomposition} and \theoremref{Theorem:StandardMultivariateCCSMcLeishCCDF} and then performing nearly same steps taken after within the proof of \theoremref{Theorem:MultivariateCESMcLeishCDF}.
\end{proof}

\begin{theorem}\label{Theorem:MultivariateCESMcLeishMGF}
The \ac{MGF} of $\defrmat{Z}\!\sim\!\mathcal{CM}_{\nu}^L(\defvec{\mu},\defmat{\Sigma})$ is given by
\begin{equation}\label{Eq:MultivariateCESMcLeishMGF}
    M_{\defrmat{Z}}(\defvec{s})=\exp\bigl(-\defvec{s}^H\defvec{\mu}\bigr)\Bigl(1-\frac{1}{4}\defvec{s}^H\defmat{\Sigma}\defvec{s}\Bigr)^{-\nu},
\end{equation}
for a certain $\defvec{s}\!\in\!\mathbb{C}^L$ within the existence region $\defvec{s}\!\in\!\mathbb{C}_{0}$, where the region $\mathbb{C}_{0}$ is given by 
\begin{equation}\label{Eq:MultivariateCESMcLeishMGFExistenceRegion}
    \mathbb{C}_0=\Bigl\{
        \defvec{s}\,    
        \Bigl|\,
        \defvec{s}^H\defmat{\Sigma}\,\defvec{s}
        \leq
        {4}
        \Bigr.
        \Bigr\}.
\end{equation}
\end{theorem}

\begin{proof}
With the aid of \eqref{Eq:MultivariateCESMcLeishDecomposition}, we can write the \ac{MGF} of $\defrmat{Z}\!\sim\!\mathcal{CM}_{\nu}^L(\defvec{\mu},\defmat{\Sigma})$ in terms of the \ac{MGF} of $\defrmat{W}\!\sim\!\mathcal{CM}_{\nu}^L(\defvec{0},\defmat{I})$, i.e.
\begin{subequations}
\setlength\arraycolsep{1.4pt}
\begin{eqnarray}
    M_{\defrmat{Z}}(\defvec{s})
        &=&\mathbb{E}[\exp\bigl(-\langle\defvec{s},\defrmat{Z}\rangle\bigr)],\\
        &=&\mathbb{E}[\exp\bigl(-\langle\defvec{s},\defmat{D}\defrvec{W}+\defvec{\mu}\rangle\bigr)],\\
        &=&\exp\bigl(-\langle\defvec{s},\defvec{\mu}\rangle\bigr)
                \mathbb{E}[\exp\bigl(-\langle\defvec{s},\defmat{D}\defrvec{W}\rangle\bigr)],\\
        &=&\exp\bigl(-\langle\defvec{s},\defvec{\mu}\rangle\bigr)
                \mathbb{E}[\exp\bigl(-\langle\defmat{D}\defvec{s},\defrvec{W}\rangle\bigr)],\\
        &=&\exp\bigl(-\langle\defvec{s},\defvec{\mu}\rangle\bigr)
                M_{\defrmat{W}}(\defmat{D}\defvec{s}),
\end{eqnarray}
\end{subequations}
where $M_{\defrmat{W}}(\defvec{s})$ denotes the \ac{MGF} of $\defrmat{W}$ and is given in \eqref{Eq:StandardMultivariateCCSMcLeishMGF}, and where both substituting \eqref{Eq:StandardMultivariateCCSMcLeishMGF} and using $\langle\defvec{s},\defvec{x}\rangle\!=\!\defvec{s}^{H}\defvec{x}$ yields \eqref{Eq:INIDMultivariateCESMcLeishMGF}, which completes the proof of \theoremref{Theorem:INIDMultivariateCESMcLeishMGF}.
\end{proof}

Eventually, we will exploit the closed-form results obtained in the preceding as a statistical and mathematical framework to introduce in the following sections some preliminary and fundamental results not only about how to properly exercise McLeish distribution to model the additive non-Gaussian white noise in wireless communications, but also about how to use the statistical characterization of McLeish distribution to obtain closed-form \ac{BER}\,/\,\ac{SER} expressions of modulation schemes and develop an analytical approach for the averaged \ac{BER}\,/\,\ac{SER} performance of diversity reception in slowly time-varying flat fading environments.

\section{Additive White McLeish Noise Channels}
\label{Section:AWMNChannels}
In wireless communications, modulation schemes are used to map the digital information sequence into a set of signal waveforms to transmit them over a communication channel. Within each symbol transmission time $t\!\in\!(0,T_{S}]$, the communication channel is without loss of generality described by the mathematical relation given by 
\begin{equation}\label{Eq:ReceivedSignal}
    R(t)=h(t)\,S(t)+Z(t), \quad{t}\in(0,T_{S}]
\end{equation}
where $T_{S}$ denotes the symbol transmission time, $s(t)$ denotes the transmitted symbol, and with respect to the information, it is chosen from the set of all possible modulation symbols $\{s_1(t),s_2(t),\ldots,s_M(t)\}$ such that $\sum_{m}\!\Pr\bigl\{s_{m}(t)\bigr\}\!=\!1$, where $M\!\in\!\mathbb{N}$ is the modulation~level. $h(t)$ denotes the fading process originating from the random nature of diffraction, refraction, and reflection within the channel,~and~due~to~coherence in time, it is assumed to be approximately constant for a number of symbol intervals. $Z(t)$ denotes a sample waveform of a zero-mean additive McLeish noise process, and $R(t)$ denotes the received waveform. The receiver makes observations on the received signal $R(t)$ and then makes an optimal decision based on the detection of which symbol $m$, ${1}\!\leq \!{m}\!\leq\!{M}$, was transmitted. 
As well explained in  \cite{BibProakisBook,BibAlouiniBook,BibGoldsmithBook}, not only can an $L$-orthonormal basis be used to represent each modulation symbol with a $L$-dimensional vector but it can also used to represent a zero-mean additive noise process as a vector of additive \ac{CES} noise distributions. With the aid of this observation, for the $n$th symbol received over additive noise channels, we can readily give a well-known mathematical base-band model in vector form \cite{BibProakisBook,BibRappaportBook,BibGoldsmithBook,BibAlouiniBook}, while we assume that symbols are sequentially transmitted, that is
\begin{equation}\label{Eq:ReceivedSignalInVectorForm}
	\defrmat{R}[n]={H}[n]\,\exp(\imaginary\,\Theta[n])\,\defvec{S}[n]+\defrmat{Z}[n]  ,
\end{equation}
where all vectors are, without loss of generality, assumed $L$-dimensional complex vectors. Specifically, $\defvec{S}[n]$ denotes the vector form of the $n$th transmitted symbol, and thus during each symbol transmission, it is randomly chosen from the set of all possible vectors $\{\defvec{S}_{1},\defvec{S}_{2},\ldots,\defvec{S}_{M}\}$. ${H}[n]$ denotes the fading envelope following a non-negative random distribution whereas $\Theta[n]$ denotes the fading phase following a random distribution over $[-\pi,\pi]$. Further, both ${H[n]}$ and $\Theta[n]$ are assumed constant during symbol duration due to the existence of channel coherence in time \cite{BibProakisBook,BibGoldsmithBook,BibAlouiniBook}. $\defrmat{Z}[n]$ denotes the additive noise, and it is always present in all communication channels and it is the major cause of impairment in many communication systems. Further, modeling $\defrmat{Z}[n]$ by a Gaussian distribution is well supported and widely evidenced from both theoretical and practical viewpoints. However, we show in what follows that the random power nature of the additive noise indicates that $\defrmat{Z}[n]$ follows non-Gaussian distribution. It is thus prudent to pick a non-Gaussian noise model, which will let us to find out the performance and bottlenecks of non-Gaussian communication channels. Accordingly, for the first time in the literature, we introduce McLeish distribution as an additive noise model that approaches to Gaussian distribution in the worst case scenarios. We call the additive McLeish noise channel to the communication channel that is subjected to the additive noise modeled by McLeish distribution.

\subsection{Random Fluctuations of Noise Variance}
\label{Section:AWMNChannels:RandomNoiseVarianceFluctuations}
In wireless digital communications, we assume that the total variance of the additive noise vector $\defrmat{Z}[n]\!\sim\!\mathcal{CM}_{\nu}^L(\defvec{\mu},\defmat{\Sigma})$ is constant~for~short-term conditions, and actually observe that it is a stationary random process in long-term conditions. We further estimate both the mean and the total variance of $\defrmat{Z}[n]$, respectively, as 
\begin{eqnarray}
    \label{Eq:InstantaneousMean}
	\defvec{\mu}_\tau[n]&=&
	    \frac{1}{\lfloor\frac{\tau}{\tau_0}\rfloor}
	        \sum_{k=n-\lfloor\frac{\tau}{\tau_0}\rfloor}^{n}
	            \!\!\!\!
	            \defrmat{Z}[k], \\
	\label{Eq:InstantaneousVariance}
	\sigma^2_\tau[n]&=&
	    \frac{1}{\lfloor\frac{\tau}{\tau_0}\rfloor}
	        \sum_{k=n-\lfloor\frac{\tau}{\tau_0}\rfloor}^{n}
	            \!\!\!\!
	            (\defrmat{Z}[k]-\defvec{\mu}_\tau[n])^H
	                (\defrmat{Z}[k]-\defvec{\mu}_\tau[n]),{~~~~~~}
\end{eqnarray}
where $\tau\in\mathbb{R}_{+}$ denotes the coherence window that characterizes the dispersive nature of the total variance, $\tau_0$ denotes the sample duration, and $\lfloor{x}\rfloor$ yields the maximum integer less that or equal to $x$. It is important for theoreticians and practitioners to be aware that the total variance contains fluctuations over time (i.e., the total variance is not constant over time), and be able to precisely quantify the amount of fluctuations associated with the total variance. Accordingly, we can write the exact total variance of $\defrmat{Z}[n]$ as
\begin{equation}\label{Eq:TotalVariance}
    \sigma^2=\lim_{\tau\rightarrow\infty}\sigma^2_\tau[n].
\end{equation}
As matter of fact that the stability of the total variance depends on the chosen window $\tau$, we can perform the Allan's variance \cite{BibAllanIEEE1966,BibMillerVandomeJohnBook2010,BibGalleaniTUFFC2009}, which is a time domain measure representing \ac{RMS} random drift within the total variance as a function of averaged time, on $\sigma^2_\tau[n]$
to express the stability the total variance with respect to $\tau\!\in\!\mathbb{R}_{+}$ and write 
\begin{equation}\label{Eq:AllanVarianceOfVariance}
	\mathbb{A}\!\left[\defrmat{Z}[n];\tau\right]
	    =\frac{1}{2}
	        \Expected{\bigl(
	            \sigma^2_\tau[n]-
	                \sigma^2_\tau[n-\tau]\bigr)^{2}},
\end{equation}
where $\Expected{\cdot}$ denotes the expectation operator, and $\mathbb{A}[y[n];\tau]$ is termed as Allan's operator applied on
the sequence of $y[n]$. By means of \eqref{Eq:InstantaneousVariance} and \eqref{Eq:TotalVariance}, we introduce 
\begin{equation}\label{Eq:TotalVarianceDrift}
    \Delta\sigma^2_\tau[n]\!=\!\sigma^2_\tau[n]-\sigma^2,
\end{equation}
which is the variance fluctuation (i.e., the random drift within the total variance over samples) such that $\mathbb{E}\bigl[\Delta\sigma^2_\tau[n]\bigr]\!=\!0$ for $\tau\!\in\!\mathbb{R}_{+}$. From \eqref{Eq:TotalVariance} and \eqref{Eq:TotalVarianceDrift}, we observe  $\lim_{\tau\rightarrow\infty}\Delta\sigma^2_\tau[n]\!=\!0$. Substituting $\sigma^2_\tau[n]\!=\!\Delta\sigma^2_\tau[n]+\sigma^2$ into \eqref{Eq:AllanVarianceOfVariance}, we can rewrite $\mathbb{A}\!\bigl[\defrmat{Z}[n];\tau\bigr]$ in terms of the statistics of $\Delta\sigma^2_\tau[n]$ as follows 
\begin{equation}\label{Eq:AllanVarianceOfVarianceDecomposition}
	\mathbb{A}\!\left[\defrmat{Z}[n];\tau\right]
	=\frac{1}{2}\Bigl(
	        \Expected{(\Delta\sigma^2_\tau[n])^2}+
	        \Expected{(\Delta\sigma^2_\tau[n-\tau])^2}
	        -2\Expected{\Delta\sigma^2_\tau[n]\Delta\sigma^2_\tau[n-\tau]}
	    \Bigr).
\end{equation}
After recognizing the variance and covariance terms associated with the variance fluctuation, i.e., using 
\begin{eqnarray}
    \label{Eq:VarianceOfVarianceFluctuations}
    \Variance{\sigma^2_\tau[n]}&=&\Expected{(\Delta\sigma^2_\tau[n])^2},\\
    \label{Eq:CovarianceOfVarianceFluctuations}
    \Covariance{\sigma^2_\tau[n]}{\sigma^2_\tau[n-\tau]}&=&\Expected{\Delta\sigma^2_\tau[n]\Delta\sigma^2_\tau[n-\tau]},
\end{eqnarray}
we eventually rewrite   \eqref{Eq:AllanVarianceOfVarianceDecomposition} as 
\begin{equation}\label{Eq:AllanVarianceOfVarianceDecompositionII}
	\mathbb{A}\!\left[\defrmat{Z}[n];\tau\right]
	=\frac{1}{2}\Bigl(
	        \Variance{\sigma^2_\tau[n]}+
	        \Variance{\sigma^2_\tau[n-\tau]}
	        -2\,\Covariance{\sigma^2_\tau[n]}{\sigma^2_\tau[n-\tau]}
	    \Bigr).
\end{equation}
Note that, without loss of generality, we can consider $\sigma^2_\tau[n]$ as a \ac{WSS} random process with~respect~to $n\!\in\!\mathbb{N}$, especially since $\defrmat{Z}[n]$ is a sample vector of \ac{WSS}~random processes. Consequently, from the \ac{WSS} feature of $\sigma^2_\tau[n]$ with respect~to~$n$, we write 
\begin{equation}\label{Eq:AllanVarianceBounds}
    0<\mathbb{A}\!\left[\defrmat{Z}[n];\tau\right]<2\,\Variance{\sigma^2_\tau[n]}
\end{equation}
for all $\tau\!\in\!\mathbb{N}$, and further we write 
\begin{equation}\label{Eq:AllanVarianceInfiniteBound}
    \lim\nolimits_{\tau\rightarrow\infty}\mathbb{A}\!\left[\defrmat{Z}[n];\tau\right]
        \leq\lim\nolimits_{\tau\rightarrow\infty}\Variance{\sigma^2_\tau[n]}.
\end{equation}
\begin{figure}[tp] 
    \centering
    \includegraphics[clip=true, trim=0mm 0mm 0mm 0mm,width=0.8\columnwidth,keepaspectratio=true]{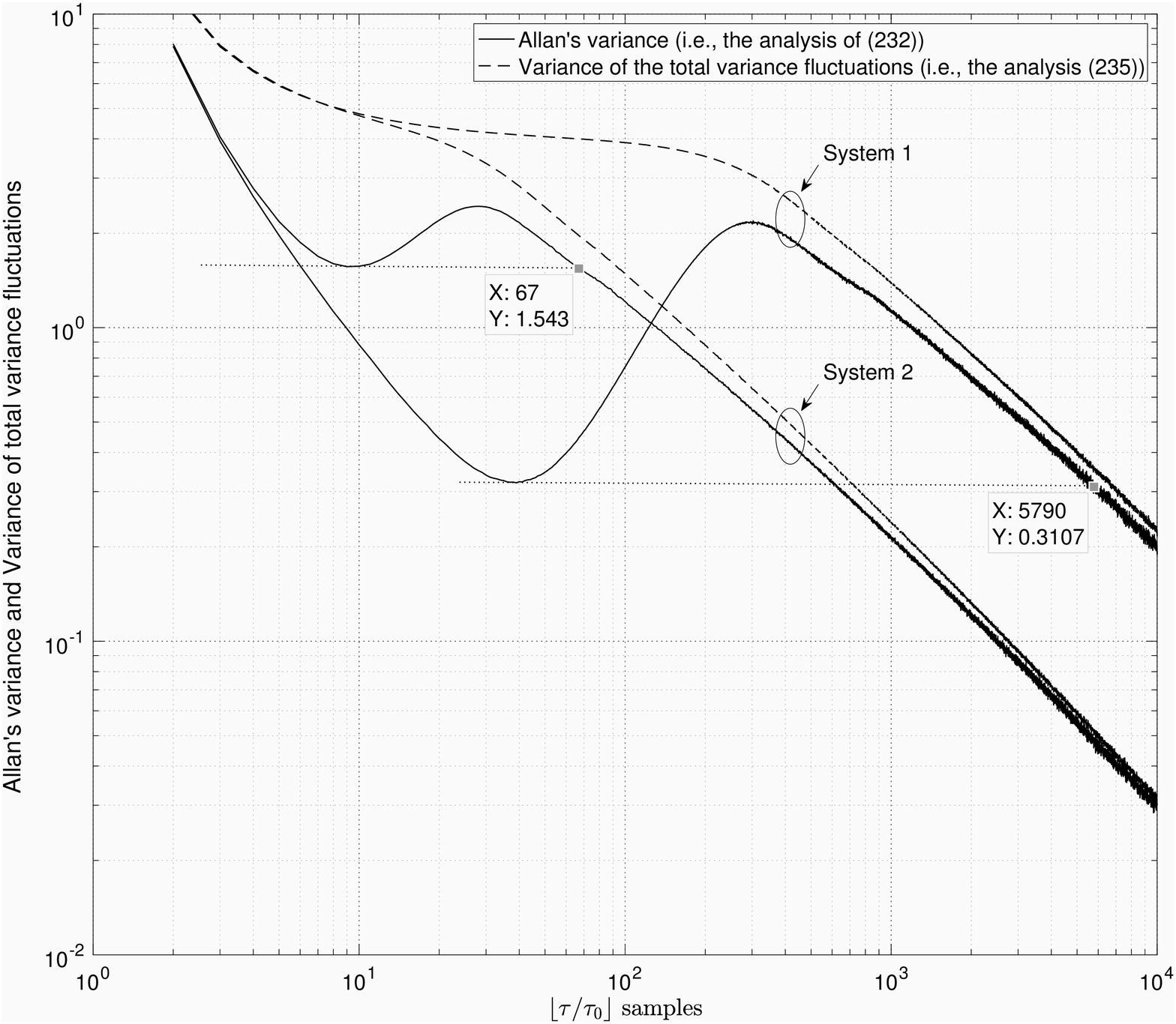} 
    \caption{The Allan's variance $\mathbb{A}\![\defrmat{Z}[n];\tau]$ and the variance of the total variance fluctuations ${\mathrm{Var}}[\sigma^2_\tau[n]]$ with respect to $\tau$, where $\sigma^2_\tau[n]$ follows a \ac{WSS} random process.}
    \label{Figure:AllanVarianceAndVarianceError}
    \vspace{-3mm}  
\end{figure}
With together aid of \eqref{Eq:AllanVarianceBounds} and \eqref{Eq:AllanVarianceInfiniteBound}, we notice that the Allan's variance $\mathbb{A}\!\left[\defrmat{Z}[n];\tau\right]$ is not a monotonically decreasing function with respect to $\tau$, which suggest some $\tau$-values for which the variance of the variance fluctuations, which is denoted by ${\mathrm{Var}}[\sigma^2_\tau[n]]$, is at desired level. Accordingly, ${\mathrm{Var}}[\sigma^2_\tau[n]]$ with respect to $\lfloor\tau/\tau_0\rfloor$ is depicted in \figref{Figure:AllanVarianceAndVarianceError} for the additive noise data that belongs to different two systems, where the variances of these additive noise are not constant and follow a \ac{WSS} non-negative random process. As such, the variance for system 1 is much more auto-correlated than that for system 2. Herein, we readily observe that, as $\tau$ increases, the variance of the total variance fluctuation decreases as expected. This fact does not reveal a minimum $\tau$ value that will keep the total variance fluctuation as small as possible. On the other hand, as demonstrated in \figref{Figure:AllanVarianceAndVarianceError}, the fact that the Allan's variance is not a monotonic function of $\tau$ can help to determine this minimum $\tau$ value, namely $\tau\!\approx\!5790\tau_0$ for system 1 and $\tau\!\approx\!67\tau_0$ for system 2. 

\begin{theorem}[Autocorrelation of Noise Variance]\label{Theorem:VarianceCorrelation}
The correlation between $\sigma^2_\tau[n]$ and $\sigma^2_\tau[n-\tau]$ is given by
\begin{equation}\label{Eq:VarianceAutoCorrelation}
    \Covariance{\sigma^2_\tau[n]}{\sigma^2_\tau[n-\tau]}
        =\Variance{\sigma^2_\tau[n]}-
            \mathbb{A}\!\left[\defrvec{Z}[n];\tau\right]
\end{equation}
for any window $\tau\in\mathbb{N}$.
\end{theorem}

\begin{proof}
From the \ac{WSS} view of $\sigma^2_\tau[n]$, we have $\Variance{\sigma^2_\tau[n]}\!=\!\Variance{\sigma^2_\tau[n-t]}$ for all $t\!\in\!\mathbb{N}$, and then simplify \eqref{Eq:AllanVarianceOfVarianceDecompositionII} to 
\begin{equation}\label{Eq:AllanVarianceOfVarianceDecompositionIII}
	\mathbb{A}\!\left[\defrmat{Z}[n];\tau\right]
	=\Variance{\sigma^2_\tau[n]}
	    -\Covariance{\sigma^2_\tau[n]}{\sigma^2_\tau[n-\tau]}.
\end{equation}
which completes the proof of \theoremref{Theorem:VarianceCorrelation}.
\end{proof}

The correlation between two consecutive estimated variances for a certain $\tau$ is given by \theoremref{Theorem:VarianceCorrelation},~from~which we observe that, when $\tau$ becomes~as~large~as~possible, this correlation ${\mathrm{Cov}}[\sigma^2_\tau[n],\sigma^2_\tau[n-\tau]]$ closes to zero,~and~therefrom with \eqref{Eq:AllanVarianceInfiniteBound}, the total variance fluctuation becomes~minimized. In the context of correlation, the auto-correlation coefficient between two consecutive estimated variances is obtained in the following.  

\begin{theorem}[Auto-correlation Coefficient of Noise Variance]
\label{Theorem:VarianceCorrelationCoefficient}
The correlation coefficient between $\sigma^2_\tau[n]$ and
$\sigma^2_\tau[n-\tau]$ is given by
\begin{equation}\label{Eq:VarianceAutoCorrelationCoefficient}
	\Correlation{\sigma^2_\tau[n]}{\tau}=
	    {1}-\frac{\mathbb{A}\!\left[\defrmat{Z}[n];\tau\right]}
	        {\Variance{\sigma^2_\tau[n]}},
\end{equation}
for any window $\tau\!\in\!\mathbb{R}_{+}$ such that
\begin{equation}\label{Eq:VarianceAutoCorrelationCoefficientRange}
    {-1}<\Correlation{\sigma^2_\tau[n]}{\tau}<{1}.
\end{equation}
\end{theorem}

\begin{proof}
The correlation coefficient between $\sigma^2[n]$ and $\sigma^2[n-\tau]$ is readily written as 
\begin{equation}
    \Correlation{\sigma^2_\tau[n]}{\tau}=
        \frac{\Covariance{\sigma^2_\tau[n]}{\sigma^2_\tau[n-\tau]}}
            {\sqrt{\Variance{\sigma^2_\tau[n]}
                \Variance{\sigma^2_\tau[n-\tau]}}}.
\end{equation}
Noticing $\Variance{\sigma^2_\tau[n]}\!=\!\Variance{\sigma^2_\tau[n-\tau]}$ from~the~\ac{WSS}~feature and subsequently substituting \eqref{Eq:VarianceAutoCorrelation} into \eqref{Eq:VarianceAutoCorrelation}, we obtain \eqref{Eq:VarianceAutoCorrelationCoefficient}. Further, from \eqref{Eq:AllanVarianceBounds} and \eqref{Eq:VarianceAutoCorrelationCoefficient}, we readily observe the existence of \eqref{Eq:VarianceAutoCorrelationCoefficientRange}, which proves \theoremref{Theorem:VarianceCorrelationCoefficient}.
\end{proof}

Note that, according to \theoremref{Theorem:VarianceCorrelationCoefficient}, $\mathbb{R}[{\sigma^2_\tau[n]};{\tau}]\!\in\![-1,1]$ is such a measurement that it describes the degree to which $\sigma^2_\tau[n]$ and $\sigma^2_\tau[n-\tau]$ are
correlated with each other. For a specific~coherence window $0\!\leq\!\tau\!\leq\!\tau_{\ell}$, if the consecutively-estimated two variances $\sigma^2_\tau[n]$ and
$\sigma^2_\tau[n-\tau]$ are highly~correlated, then we have  $\mathbb{R}[{\sigma^2_\tau[n]};{\tau}]\!\approx\!{1}$ and thus $\mathbb{A}[\defrmat{Z}[n];\tau]\!\ll\!\Variance{\sigma^2_\tau[n]}$, which means that the estimation $\sigma^2_\tau[n]$ has the minimum error, i.e., $\sigma^2_\tau[n]$ is approximately constant. Accordingly, we can exploit \theoremref{Theorem:VarianceCorrelationCoefficient} to estimate the coherence window $\tau$ of the random fluctuations in the nature of variance.

\begin{theorem}[Coherence of Noise Variance]\label{Theorem:VarianceCoherenceWindow}
The length of the coherence window $[0,\tau_{C}]$ of the additive noise variance can be estimated as
\begin{equation}\label{Eq:VarianceCoherenceWindowEstimation}
\tau_{C}=\arg\min_{\tau\in\mathbb{R}_{+}}
		\biggl(
			\frac{\mathbb{A}[\defrmat{Z}[n];\tau]}
			    {\Variance{\sigma^2_\tau[n]}}+R-1
		\biggr)^{2}\!,
\end{equation}
where $R\!\in\![0,1]$ denotes a certain correlation level, typically chosen as $0.95$, $0.68$, or $0.5$.\QEDsymbol
\end{theorem}

\begin{proof}
Note that $\abs{\mathbb{R}[{\sigma^2_\tau[n]};\tau]}$ decreases monotonically with respect to $\tau\!\in\!\mathbb{R}_{+}$, i.e., $\abs{\mathbb{R}[\sigma^2_\tau[n];\tau]}\!\leq\!\mathbb{R}[\sigma^2_\tau[n];{0}]$. Hence, we can determine the width $\tau_{C}$ of the coherence window as that of $\abs{\mathbb{R}[\sigma^2_\tau[n];\tau]}$ where it drops to a certain level $R$. Having an objective to minimize the Euclidean distance between $R$ and $\abs{\mathbb{R}[\sigma^2_\tau[n];\tau]}$, we can formulate this problem as
\begin{equation}\label{Eq:CoherenceWindowEstimation}
\tau_{C}=\arg\min_{\tau\in\mathbb{N}}
		\Bigl(
			R-\abs{\mathbb{R}[\sigma^2_\tau[n];\tau]}
		\Bigr)^{2}\!.~~\\
\end{equation}
where substituting \eqref{Eq:VarianceAutoCorrelationCoefficient} and using ${1-\abs{x}}\!\leq\!\abs{1-x}$ results in
\eqref{Eq:VarianceCoherenceWindowEstimation}, which proves \theoremref{Theorem:VarianceCoherenceWindow}.
\end{proof}

Based on the concepts and procedures mentioned above for the random fluctuations of noise-variance,~let~us~now~briefly consider which values of $\tau_{C}$ cause some uncertainty (random fluctuations) in noise variance. Let $T_{C}\!\in\!\mathbb{R}_{+}$ be the coherence time of the fading conditions in the wireless channel, and $T_{S}$ be the symbol duration. In literature, it is widely assumed that $T_{C}\!\gg\!T_{S}$ for a reliable transmission~in~flat~fading~environments. In order to get the idea how to elucidate which values of $\tau_{C}$ cause the random fluctuations in noise variance, we need to compare both $\tau_{C}$ and $T_{C}$ with each other with regard to $T_{S}$. Regarding the random fluctuations of noise variance, there are three distinct variance-uncertainties observed in wireless communications and listed as follows. 
\begin{itemize}[label=$\square$]
\setlength\itemsep{1mm}
\item \textit{(Constant variance).} In the literature of wireless communications \cite[and references therein]{BibAlouiniBook,BibGoldsmithBook,BibProakisBook,BibRappaportBook,BibLapidothBook2017}, $\tau_{C}$ is often assumed to be pretty much large enough as compared both to $T_{C}$ and $T_{S}$ such that $\tau_{C}/T_{C}\!\gg\!{T_{S}}$. In such a case, we observe that $\sigma^2[n]$ does actually have no fluctuations, namely, that it is constant (i.e., $\sigma^2_{\tau}[n]={2}N_{0}$ for all $n\in\mathbb{N}$ and $\tau\in\mathbb{R}_{+}$, where $2N_{0}$ denotes the power spectral density of $\defrmat{Z}[n]$) since
\begin{equation}
	\lim\limits_{\tau_{C}\rightarrow\infty}\mathbb{A}[\defrmat{Z}[n];\tau]={0}^{+},
\end{equation}
In other words, since $\lim_{\tau_{C}\rightarrow\infty}\sigma^2_\tau[n]\!=\!\sigma^2[n]\!=\!2N_{0}$, we notice that the random fluctuations of the variance vanish when $\tau_{C}\!\rightarrow\!\infty$ as expected.~Accordingly~and~conveniently, we can use multivariate \ac{CES} Gaussian distribution instead of multivariate \ac{CES} non-Gaussian distribution to model the additive noise $\defrmat{Z}[n]$.
\item\textit{(Slow variance-uncertainty).} If $\tau_{C}$ is either comparable to or greater than $T_{C}$ with respect to $T_{S}$, i.e. when $\tau_{C}/T_{C}\geq{T_{S}}$, then it is observed that the instantaneous variance $\sigma^2[n]$ is approximately constant during the symbols transmitted in the coherence time $T_{C}$ of fading conditions but fluctuates arbitrarily over all transmitted symbols. For example, either in high-speed transmission in ultra-high frequencies, or in wireless powered diversity receivers, $\defrmat{Z}[n]$ follows a multivariate \ac{CES} Gaussian distribution whose total variance $\sigma^{2}[n]$ fluctuates randomly in long-term conditions. This phenomenon is called \emph{noise uncertainty}  \cite{BibShellhammerTandraIEEEStd2006}.~That~is~to~say,~ $\sigma^{2}[n]$~follows a non-negative distribution, which modulates complex Gaussian distribution, and thus causes impulsive effects on the performance of the transmission system. Accordingly, we show that $\defrmat{Z}[n]$ is accurately modeled in terms of Hall's noise model  \cite{BibHall1966TechReport,BibHedin1968TechReport} as follows
\begin{equation}\label{Eq:HallNoiseModel}
	\defrmat{Z}[n]=\sigma[n]\,\defrmat{N}[n],
\end{equation}
where $\defrmat{N}[n]$ is a multivariate \ac{CES} Gaussian distribution with zero mean vector and $\defmat{\Sigma}$ covariance matrix, and independent of $\sigma^2[n]$. Thus, according to \eqref{Eq:HallNoiseModel}, $\defrmat{Z}[n]$ follows a multivariate \ac{CES} Gaussian distribution given $\sigma^2[n]$. Therefore, \eqref{Eq:HallNoiseModel} is found to be a \ac{SIRP} \cite{BibYaoKung1973}, which has been widely adopted in wireless communications \cite[and~references~therein]{BibAlouiniBook}. It is worth mentioning that, as well explained in the following sections, $\sigma^{2}[n]$ can be perfectly estimated in the coherence time $T_{C}$ as a \ac{CSI} to maximize the \ac{SNR} in the case of signal reception over generalized fading environments. 

\item \textit{(Fast variance-uncertainty).} When $\tau_{C}$ is much smaller than $T_{C}$ such that $\tau_{C}/T_{C}\ll{T_{S}}$, the estimation of $\sigma^{2}[n]$ within the coherence time $T_{C}$ is a more difficult task, and mostly not possible. In such a case, the noise model presented in \eqref{Eq:HallNoiseModel} still applies, but optimum detection
and optimum combining schemes have to be reconsidered to minimize the performance degradation originated from the variance uncertainity.
\end{itemize}

\vspace{1mm}
Eventually, from the statements given above, we conclude that in both slow and fast uncertainty (random fluctuations) of the noise variance, the multiplication of $\sigma[n]$ and $G[n]$ leads to some impulsive random fluctuations. As such, the random distribution of $\sigma[n]$ modulates the inphase and quadrature parts of $\defrmat{N}[n]$ since the inphase and quadrature parts of $\defrmat{N}[n]$ belong to the same channel. In the following, we show that the variance fluctuations exists in real life scenarios, and the additive noise, whose model is introduced in \eqref{Eq:HallNoiseModel}, follows McLeish distribution.

For the sake of brevity, clarity and readability, the symbol indexing $[n]$ is
deliberately omitted in the following. 

\subsection{Existence of McLeish noise distribution}
\label{Section:AWMNChannels:McLeishNoiseExistence}
The existence of McLeish noise in communication systems is observed in many forms and in various ways.

\subsubsection{Thermal noise}
\label{Section:AWMNChannels:McLeishNoiseExistence:JohnsonNoise}
If the additive noise is primarily originated from electronic materials at the receiver, it is then called thermal noise. The electrical conduction is governed by how freely mobile electrons can move throughout the electronic material while their movements are hindered and impeded by scattering with other electrons, as well as with impurities or thermal excitations (phonons) \cite{BibSchefflerDresselJourdanAdrianNAT2005}. At this point, the thermal noise is explained as a phenomenon associated with the discreteness and random motion of the electrons, and always exists in varying degrees in all electrical parts of systems. Regarding the model of thermal agitation \cite{BibHenryAJP1973},\cite[Sec.~8.10]{BibEngelbergBook}, which goes back to the classical theory introduced by Drude in 1900 \cite{BibDrudePZ1900}, let us consider a steady electrical current composed of many electrons, each passing through a resistor which is illustrated in \figref{Figure:ResistorModel} as a cylinder of finite conductive material of length $L$ and cross-sectional area $A$ (i.e., its volume is $V\!=\!AL$). The velocity of an electron in the $x$-direction (i.e. the velocity along the direction of the steady electric field impressed upon the resistor by the battery) is given by ${v}_{x}\!=\!{v}_{d}+{v}_{t}$, where ${v}_{d}$ is the drift velocity due to electric field and ${v}_{t}$ is the x-velocity due to the \emph{thermal agitation} of the electrons. Further, since the electric field inside the resistor is, without loss of generality, assumed to be constant, 
the field-based velocity ${v}_{d}$ has no random nature. However, as a result of ${v}_{d}\gg{v}_{t}$, the thermal-based velocity ${v}_{t}$ has random nature in the $x$-direction, following Gaussian \ac{PDF} given by
\begin{equation}\label{Eq:ElectronRandomMotionPDF}
    f_{{v}_{t}}(v)=\sqrt{\frac{m_{0}}{2{\pi}KT}}\exp\Bigl(-m_{0}\frac{v^2}{2KT}\Bigr),\quad{v}\in\mathbb{R},
\end{equation}
with mean $\mathbb{E}[{v}_{t}]\!=\!0$ and variance $\mathbb{E}[{v}^2_{t}]\!=\!\frac{KT}{2m_{0}}$, where the constant  $m_{0}\!\approx\!9.10938356\times{10}^{-31}$ kg is the mass of an electron, and $K$ and $T$ are respectively the Boltzmann constant and the absolute temperature. If temperature is measured in Kelvins, and energy is measured in Joules, then the Boltzmann constant is \emph{approximately} given by $K\!\approx\!1.38064852\!\times\!{10}^{-23}$ J/K. Accordingly, average thermal kinetic energy of an electron can be written as 
\begin{equation}
    E_{t}=\mathbb{E}\Bigl[\frac{1}{2}m_{0}{v}^2_{t}\Bigr]=\frac{1}{2}KT
\end{equation}
in accordance within the literature \cite{BibDrudePZ1900,BibJohnsonPR1928,BibKhinchinBook,BibJohnsonIEEESPEC1971,BibHenryAJP1973,BibRomeroP1998,BibEngelbergBook,BibAntsiferovBook,BibSchefflerDresselJourdanAdrianNAT2005,BibThierryBook,BibVasileskaGoodnickKlimeckBook2016,BibDresselhausCroninSouzaFilhoBook}. Further, free electrons will move randomly due to thermal energy, so they experience many collisions. Let $C$ be the number of collisions in 1 second and $\tau$ be the interval time between any two sequential collisions of an electron. In accordance with the statistical theory of collisions, we notice that $C$ has a random relaxation nature that follows Poisson process
 \cite{BibVasileskaGoodnickKlimeckBook2016,BibDresselhausCroninSouzaFilhoBook} with the \ac{PMF} given by
\begin{figure}[tp] 
    \centering
    \includegraphics[clip=true, trim=0mm 0mm 0mm 0mm,width=0.7\columnwidth,keepaspectratio]{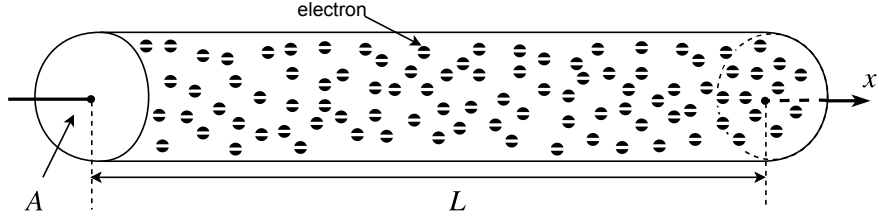}
    \caption{Finite conductive material.}
    \label{Figure:ResistorModel}
    \vspace{-2mm} 
\end{figure} 
\begin{subequations}
\setlength\arraycolsep{1.4pt}
\begin{eqnarray}
    f_{C}(n)&=&
        \Pr\{\text{$n$~collisions occurs in 1 second}\},\\
    &=&\frac{1}{n!}\Bigl(\frac{1}{\Delta\tau}\Bigr)^{n}
        \exp\Bigl(-\frac{1}{\Delta\tau}\bigr),
\end{eqnarray}
\end{subequations}
where $\Delta\tau$ is the mean relaxation time between collisions (i.e. $\Delta\tau\!=\!\mathbb{E}[\tau]$) and decreases as with temperature $T$, i.e., $\Delta\tau\!\propto\!{1}/{\sqrt{T}}$. In average sense, each electron should experience ${1}/{\Delta\tau}$ collisions per $1$ second. In~the~best electron excitation, $\tau$ follows an exponential distribution, that is   
\begin{equation}\label{Eq:ElectronCollisionMeanTimePDF}
    f_{\tau}(t)=\frac{1}{\Delta\tau}\exp\Bigl(-\frac{t}{\Delta\tau}\Bigr),
\end{equation}
for ${t}\!\in\!\mathbb{R}_{+}$. It is worth emphasizing either $\Pr\{\tau<\Delta\tau\}\!>\!1-\Pr\{\tau<\Delta\tau\}$ under the best electron excitation conditions or $\Pr\{\tau<\Delta\tau\}\!\leq\!1-\Pr\{\tau<\Delta\tau\}$ otherwise.~In~other~words, $\tau$ is the most probably less than $\Delta\tau$ under the best electron excitation conditions. Let us denote the electron excitation condition by $\nu\!\in\!\mathbb{R}_{+}$. We notice that $\nu$ increases while the electron excitation conditions get worse, which results the fact that each electron displacement occurs after more than one collisions under the worst electron excitation conditions. Therefore, under the best electron excitation conditions, we have $\Pr\{\tau<\Delta\tau\}\!\leq\!1-\Pr\{\tau<\Delta\tau\}$ and therefrom notice that $\tau$ is the most probably larger than or equal to $\Delta\tau$. In pursuance of the electron excitation conditions, in which the variation in time between any two sequential collisions of an electron arises from fluctuations in the momentum of electrons created by collisions, we conveniently deduce that $\tau$ follows a Gamma distribution, that is
\begin{equation}\label{Eq:ElectronCollisionMeanTimeGAMMAPDF}
f_{\tau}(t)=\frac{1}{\Gamma(\nu)}
    \Bigl(\frac{\nu}{\Delta\tau}\Bigr)^{\nu}t^{\nu-1}
        \exp\Bigl(-\frac{\nu}{\Delta\tau}t\Bigr),
\end{equation}
which readily simplifies to \eqref{Eq:ElectronCollisionMeanTimePDF} for the best electron excitation conditions $\nu\!=\!1$ as expected. But, for the worst electron excitation conditions, we have $\nu\!\rightarrow\!\infty$, and correspondingly we notice that  \eqref{Eq:ElectronCollisionMeanTimeGAMMAPDF} approximates to the Dirac's distribution, that is $f_{\tau}(t)\!=\!\DiracDelta{t-\Delta\tau}$. This fact means that the randomness of $\tau$ disappears (i.e., constantly $\tau=\Delta\tau$), and implies in other words that the thermal displacement of each electron along the $x$-direction for a period of 1 second will precisely occur as a result of its certain $1/\Delta\tau$ number of collisions. 

Note that the number of free electrons causing thermal noise depends on the finite conductivity of the resistor. Accordingly, let $\rho$ denote the density of free electrons, then the total number of free electrons in the finite conductive material, depicted in \figref{Figure:ResistorModel}, is given by
$\eta_{f}\!=\!\rho\,{A}{L}$, and then the total number of possible displacement steps taken by all the free electrons in $1$ seconds should be $\eta\!\approx\!{\eta_{f}}/{\tau}\!=\!{\rho\,{A}{L}}/{\tau}$. In accordance with the velocity of an electron explained above, let $v_{t}[n]$ be the $n$th thermal displacement of an electron along the $x$-direction for the period of 1 second. The distribution of $v_{t}[n]$ is given in \eqref{Eq:ElectronRandomMotionPDF}. Accordingly, in terms of fractional sum, we can write the total charge movement due to thermal energy, i.e., the additive noise current passing though the resistor of length $L$, that is
\begin{equation}
    I=\sum_{n=0}^{\eta}{e}_{0}\,\tau\,\frac{v_{t}[n]}{L}=\sum_{n=0}^{\eta}Q[n],
\end{equation}
where ${e}_{0}\!\approx\!{1.60217662}\times{10}^{-19}$ C denotes the charge on each electron. Under the assumption that $\tau$ is instantaneously known, $Q[n]$, ${0}\!\leq\!{n}\!\leq\!{\eta}$, has Gaussian distribution. Therefore, the additive noise current $I$ conditioned on $\tau$, which is denoted by $I|\tau$, will follow Gaussian distribution with mean and variance, respectively obtained with the aid of the Euler-like identities of fractional sums \cite{BibMullerJCAM2005,BibMullerRJ2010,BibAbdeljawadBaleanuJaradAgarwalDDNS2013} as follows
\setlength\arraycolsep{1.4pt}
\begin{eqnarray}
    \label{Eq:AdditiveNoiseMeanConditionedOnCollisionMeanTime}
    \mu_{I|\tau}&=&
        \mathbb{E}\bigl[\bigl.I\bigr|\tau\bigr]=0,\\
    \label{Eq:AdditiveNoiseVarianceConditionedOnCollisionMeanTime}
    \sigma^{2}_{I|\tau}&=&
        \mathbb{E}\bigl[\bigl.I^2\bigr|\tau\bigr]=\tau\rho\,{e}^{2}_{0}\frac{AKT}{2{m}_{0}L}.
\end{eqnarray}
Accordingly, the \ac{PDF} of $I$ given $\tau$, i.e., $f_{I|\tau}(x)$ is written as 
\begin{equation}\label{Eq:AdditiveNoisePDFConditionedOnCollisionMeanTime}
    f_{I|\tau}(x)=\sqrt{\frac{{m}_{0}L}{\pi\tau\rho\,{e}^{2}_{0}{AKT}}}
        \exp\Bigl(-\frac{{m}_{0}L}{\tau\rho\,{e}^{2}_{0}{AKT}}x^{2}\Bigr).
\end{equation}
In pursuance, the \ac{PDF} of $I$ is readily expressed as
$f_{I}(x)\!=\!\int_{0}^{\infty}f_{I|\tau}(x|t)f_{\tau}(t)dt$,  
where substituting \eqref{Eq:AdditiveNoisePDFConditionedOnCollisionMeanTime} and \eqref{Eq:ElectronCollisionMeanTimeGAMMAPDF}, and subsequently employing
\cite[Eq.\!~(3.471/9)]{BibGradshteynRyzhikBook} results in 
\begin{equation}\label{Eq:AdditiveNoisePDF}
	f_{I}(x)=\frac{2}{\sqrt{\pi}}
		\frac{\abs{x}^{\nu-\frac{1}{2}}}{\Gamma(\nu)\,\lambda^{\nu+\frac{1}{2}}}
			\BesselK[\nu-\frac{1}{2}]{\frac{2\abs{x}}{\lambda}},
\end{equation}
which is surprisingly the \ac{PDF} of McLeish distribution with zero mean and $\sigma^2$ variance. Hence, we have $I\!\sim\!\mathcal{M}_{\nu}(0,\sigma^2)$, where the admittance per collision is given by $\lambda\!=\!\sqrt{2\sigma^2/\nu}$. We obtain the variance $\sigma^2\!=\!\mathbb{E}[I^{2}]$ by $\sigma^{2}\!=\!\int_{0}^{\infty}\!\!\sigma^{2}_{I|t}\,f_{\tau}(t)\,dt$, where 
substituting \eqref{Eq:AdditiveNoiseVarianceConditionedOnCollisionMeanTime} and \eqref{Eq:ElectronCollisionMeanTimeGAMMAPDF} results in 
\begin{equation}\label{Eq:AdditiveNoiseVariance}
    \sigma^{2}=\Delta\tau\,\rho\,{e}^{2}_{0}\frac{AKT}{2{m}_{0}L}.
\end{equation}
According to the Nyquist's theorem \cite{BibNyquistPR1928,BibRomeroP1998,BibEngelbergBook}, the power spectral density of the additive noise current is given by $S_{I}(f)\!=\!{2KT}/{R}$ for all $f\!\in\!\mathbb{R}$, where $R$ denotes the thermal resistance of the finite conductive material, to which the additive thermal noise is associated. With the aid of $S_{I}(0)=\sigma^2$, we obtain the resistance as 
\begin{equation}\label{Eq:AdditiveNoiseResistance}
    R=\frac{2KT}{S_{I}(0)}=\frac{4{m}_{0}L}{{\Delta\tau}\,\rho\,{e}^{2}_{0}A}.
\end{equation} 
Let us consider some crucial special cases. For the best electron excitation conditions (i.e., $\nu\!=\!1$), we can simplify \eqref{Eq:AdditiveNoisePDF} to the \ac{PDF} of Laplacian distribution with zero mean and $\sigma^2$ variance, that is $f_{I}(x)\!=\!\frac{1}{\sqrt{2\sigma^2}}\exp\bigl(-{\sqrt{{2}/{\sigma^2}}}\abs{x}\bigr)$. On the other hand, for the worst electron excitation conditions (i.e., $\nu\!\rightarrow\!\infty$), we can simplify \eqref{Eq:AdditiveNoisePDF} to $f_{I}(x)\!=\!\frac{1}{\sqrt{2\pi\sigma^2}}\exp\bigl(-{x^2}/{2\sigma^2}\bigr)$, which the \ac{PDF} of Gaussian distribution with zero mean and $\sigma^2$ variance as expected. We notice that these facts are compliant for the fact that the additive noise following Gaussian distribution the worst-case noise distribution for communication channels \cite{BibLapidothTIT1996,BibShomoronyAvestimehrISIT2012,BibShomoronyAvestimehrTIT2013}. Furthermore, we observe both from \eqref{Eq:AdditiveNoiseVariance} and \eqref{Eq:AdditiveNoiseResistance} that the variance of the additive noise proportional to both the temperature $T$ and the length $L$ but inversely to the cross-sectional area $A$ as expected. 

In addition to all stated above, we acknowledge one extra point in which McLeish distribution also occurs in resistance circuits. Let us assume that there exist $N$ resistors connected in parallel, then we will observe the total additive noise current as the sum of these numerous low-power impulsive noise sources $I_{\Sigma}\!=\!\sum_{n=1}^{N}I_{n}$, where $I_{n}$, ${1}\!\leq\!{n}\!\leq\!{N}$, denotes the additive noise originated from the $n$th resistor, and therein we have  $I_{n}\!\sim\!\mathcal{L}_{\nu}(0,\sigma^2)$ under the best electron excitation conditions. Consequently, the total additive noise $I$ follows a McLeish distribution, i.e, $I_{\Sigma}\!\sim\!\mathcal{M}_{\nu}(0,N\sigma^2)$. As the number of resistors increases, the number of additive Laplacian components increases, which yields the convergence of the additive noise to a Gaussian distribution according to the \ac{CLT}. Consequently, we remark that McLeish distribution is found to be a noise model capturing different impulsive noise environment. 

\subsubsection{Multiple access\texorpdfstring{\,/\,}{}User interference}
\label{Section:AWMNChannels:McLeishNoiseExistence:MAIAndMUIInterference}
In wireless communications, both \ac{MAI} and \ac{MUI} resembles impulse noise more than Gaussian noise was rigorously investigated and soundly concluded in  \cite{BibGowdaAnnampeduViswanathanICC1998,BibDhibiTSPKaiser2006,BibDhibiKaiserMNA2006,BibFiorinaGLOBECOM2006,BibHuBeaulieuIEEERWS2008,BibThompsonChang1994,BibBeaulieuNiranjayan2010}, and the impulsive effects of the interference caused by each one of the other multiple users is often reasonably be modeled by Laplacian process. It is reported in  \cite{BibGowdaAnnampeduViswanathanICC1998} that \ac{MAI} follows Laplacian  distribution in \ac{DS} \ac{CDMA} systems. Not only the theoretical background necessary to understand why \ac{MAI} and \ac{MUI} have Laplace distribution but also the further details are presented in the following. The total interference a user experiences in a \ac{MAI}\,/\,\ac{MUI} communication system can be written as
\begin{equation}\label{Eq:InterferenceModel}
    I=\sum_{n=1}^{N}I_{n},
\end{equation}
due to a small number of interfering users at close range, where the configuration of the interference originating from the $n$th interferer can be written as
\begin{equation}\label{Eq:InterferenceComponentModel}
    I_{n}=\sum_{k=1}^{\infty}\alpha_{k}e^{\imaginary\theta_{k}}I_{nk},
\end{equation}
where $\{I_{nk}\}_{k=1}^{\infty}$ denotes the set of interference components originating from the signaling of the $n$th interfering user, where $I_{nk}$ is the interference originating from the $k$th signaling configuration the $n$th interfering user employs, and modeled as $I_{nk}\!\sim\!\mathcal{CN}(0,\sigma^2_{nk})$. In accordance, let us assume that the interference components are without loss of generality  ordered with respect to their variances, i.e.,
\begin{equation}
    \sigma_{n1}\geq\sigma_{n2}\geq\sigma_{n3}\geq\ldots\geq\sigma_{nk}\geq\ldots\geq{0}.
\end{equation}
As a result of $\lim_{k\rightarrow\infty}\sigma^2_{nk}\!=\!{0}$ using the strong law of large numbers, we have $\sum_{k=1}^{\infty}\sigma^{k}_{n}\!<\!\infty$. Moreover, in \eqref{Eq:InterferenceComponentModel}, $\alpha_{k}$, ${1}\!\leq\!{k}\!\leq\!{\infty}$ is the indicator for the $k$th possible signaling configuration, and modeled as Bernoulli distribution taking values $1$ and $0$ with probabilities $p$ and $1-p$, respectively, such that ${0}\!<\!{p}\!<\!{1}$.
The phase $\theta_{k}$ is the component phase with respect to user, and it is uniformly distributed over $[-\pi,\pi)$. We can easily show that each interference component $I_{n}$, which is given by \eqref{Eq:InterferenceComponentModel}, is decomposed as 
\begin{equation}\label{Eq:InterferenceComponentDecomposition}
    I_{n}={\sigma}\sqrt{E}({X}_{0}+\imaginary{Y}_{0}),
\end{equation}
where $X_{0}\!\sim\!\mathcal{N}(0,1)$, $Y_{0}\!\sim\!\mathcal{N}(0,1)$ and $E\!\sim\!\mathcal{G}(1,1)$. Upon using \theoremref{Theorem:CCSMcLeishDefinition} with \ac{CS} property and making use of \eqref{Eq:CCSLaplacePDF}, we note that $I_{n}$ follows a Laplace distribution that has zero mean, i.e., $\mathbb{E}[I_{n}]\!=\!0$, and has a variance given by
\begin{equation}
    \sigma^2={p}\sum_{k=1}^{\infty}\sigma^2_{nk}<\infty.
\end{equation}
Since $I_{n}\!\sim\!\mathcal{CL}(0,\sigma^2)$, ${1}\!\leq\!{n}\!\leq\!{N}$, the total interference, given in \eqref{Eq:InterferenceModel}, follows a \ac{CCS} McLeish distribution with zero mean and $\nu\sigma^2$ variance (i.e., $I\!\sim\!\mathcal{CM}_{\nu}(0,\nu\sigma^2)$). Consequently, we have remarked that \ac{CCS} McLeish distribution is found to be a better model for the total \ac{MAI}\,/\,\ac{MUI} interference.


\subsubsection{Versatility}
\label{Section:AWMNChannels:McLeishNoiseExistence:Versatility}
The additive noise in most communication systems is supposed to be modeled as Gaussian distribution \cite[and\!~references\!~therein]{BibAlouiniBook, BibGoldsmithBook,BibProakisBook,BibRappaportBook}. These systems are also subjected to impulsive noise effects. Many statistical distributions have been proposed in the literature to model impulsive noise effects. As such, the so-called non-Gaussian distributions such as Laplacian, symmetric $\alpha$-stable (${S}\alpha{S}$), and generalized Gaussian distributions have attracted the interest of the research community due to their ability to capture different impulsive noise effects  \cite{BibBernsteinIEEE1974,BibDurisiUWST2002,BibHu2004Accurate,BibFiorinaGLOBECOM2006,BibDhibiMNA2006,BibDhibiTSP2006Interference, BibBeaulieuTVT2008,BibHu2008Characterizing,BibBeaulieuIEEE2009,BibChianiIEEE2009Coexistence,BibIlowHatzinakosTSP1998,BibJiang2010BER,BibYangPetropuluTSP2003,BibHaenggiAndrewsBaccelliDousseFranceschettiJSAC2009,BibBrownZoubirTSP2000,BibTsihrintzisMILCOM1996,BibKuruogluRaynerFitzgeraldSSAP1998,BibRajanTepedelenliogluTWC2010,BibViswanathanAnsariASSP1989,BibZahabiTadaionICT2010,BibBiglieriYaoYangWCOM2015}. The lack of characterizing the impulsive noise effects from non-Gaussianity to Gaussianity is one of the essential weaknesses of these distributions mentioned above. On the other hand, note that 
the statistical description of McLeish distribution is typically defined according to the two observations, one of which is that the additive noise is caused by the summation of numerous impulsive noise sources of low power, each of which is found to be properly characterized by Laplacian distribution. The other observation is that, according to
the \ac{CLT}, the additive noise certainly converges to follow Gaussian distribution as the limit case of that the number of impulsive noise sources. As a result, the McLeish distribution demonstrates a superior fit to the different impulsive noise characteristics from non-Gaussian to Gaussian distributions with respect to its normality parameter $\nu\!\in\!\mathbb{R}_{+}$. As such, let $W$ be a additive noise distribution we would like to fit the \ac{PDF} of McLeish distribution by using \ac{MOM} estimation technique. Then, we can estimate the mean by $\widehat{\mu}=\Expected{W}$, and further the variance and the normality respectively by 
\begin{equation}
    \widehat{\sigma}^2=\Variance{W}, 
	        {~}\text{and}{~}
		    \widehat\nu=\frac{3}{\Kurtosis{W}-3}, 
\end{equation}
where $\Variance{\cdot}$ and $\Kurtosis{\cdot}$ denote the well-known variance and Kurtosis operators, respectively. Consequently, we have remarked that the McLeish distribution is a very useful additive noise model that can be used in wireless communication performance analysis and research due to its versatility, experimental validity and analytical tractability.

\section{Signalling over \ac{AWMN} Channels}
\label{Section:SignallingOverAWMNChannels}
In what follows, for signaling over impulsive additive noise channels, we will introduce complex correlated \ac{AWMN} vector channels and therein benefit from the vectorization that removes the redundancy in signal waveforms and that provides a compact presentation for them. Let us proceed to establish a mathematical model, which is in vector form using \eqref{Eq:ReceivedSignal}, for the baseband signaling over complex correlated \ac{AWMN} vector channels, that is  \cite{BibProakisBook,BibRappaportBook,BibGoldsmithBook,BibAlouiniBook}
\begin{equation}\label{Eq:ComplexAWMNVectorChannel}
	\defrmat{R}={H}{e}^{\imaginary\Theta}\defmat{F}\defrvec{S}+\defrmat{Z},
\end{equation}
where all vectors are without loss~of~generality~$L$-dimen\-sional complex vectors. Specifically, $\defrmat{R}\!=\![R_1,R_2,\ldots,R_L]^T$ denotes the received signal vector. When we start explaining from the right of \eqref{Eq:ComplexAWMNVectorChannel}, the random vector $\defrmat{Z}$ is the additive noise modeled as multivariate \ac{CES} McLeish distribution with $\nu$ normality, zero mean vector and $\defmat{\Sigma}$ covariance matrix, and it is denoted by $\defrmat{Z}\!\sim\!\mathcal{CM}^{L}_{\nu}(\defvec{0},\defmat{\Sigma})$. With the aid of \theoremref{Theorem:MultivariateCESMcLeishPDF}, we readily write the \ac{PDF} of $\defrmat{Z}$ as  
\begin{equation}\label{Eq:AWMNAdditiveNoiseVectorPDF}
    f_{\defrmat{Z}}(\defvec{z})=\frac{2}{\pi^L\Gamma(\nu)}
	    \frac{{\lVert{\defvec{z}}\rVert}_{\defmat{\Sigma}}^{\nu-L}}
		    {\det(\defmat{\Sigma})\,\lambda_{0}^{\nu+L}}
			    {K}_{\nu-L}
					\Bigl(
						\frac{2}{\lambda_{0}}
					        {\bigl\lVert{\defvec{z}}\bigr\rVert}_{\defmat{\Sigma}}
					\Bigr),
\end{equation} 
where $\lambda_{0}\!=\!\sqrt{2/\nu}$ denotes the standard component deviation. \emph{It is worth noticing that $\defrmat{Z}$ has a \ac{CES} distribution (i.e., it is a colored (non-white) additive complex noise)}, which is the most essential issue at the receiver to be solved in making a decision of which symbol vector was transmitted based on the observation of $\defrmat{R}$. Moreover, for a fixed modulation level $M\!\in\!\mathbb{N}$, the random vector $\defvec{S}$ denotes the modulation symbol vector randomly chosen from the set of possible fixed modulation symbols $\{\defvec{s}_{1},\defvec{s}_{2},\ldots,\defvec{s}_{M}\}$ according to \emph{a priori} probabilities $\{{p}_{1},\allowbreak{p}_{2},\ldots,{p}_{M}\}$, where $p_{m}=\Pr\{\defrvec{S}=\defvec{s}_{m}\}$, ${1}\!\leq\!{m}\!\leq\!{M}$ with the fact that $\sum_{m}p_{m}\!=\!1$. As such, upon while considering the overall transmission, we write the \ac{PMF} of $\defvec{S}$ in continuous form \cite[Eq. (4-15)]{BibPapoulisBook}, that is
\begin{equation}\label{Eq:AWMNModulationSymbolVectorPMF}
    f_{\defrvec{S}}(\defvec{s})=\sum_{m=1}^{M}p_{m}\DiracDelta{\lVert\defvec{s}-\defvec{s}_{m}\rVert}.
\end{equation} 
Further, in \eqref{Eq:ComplexAWMNVectorChannel}, $\defmat{F}\in\mathbb{C}^{{L}\times{L}}$ is a precoding matrix~filter~that precodes each modulation symbol before transmission in order to compensate the performance degradation originating from the correlation between channels. In addition, in \eqref{Eq:ComplexAWMNVectorChannel}, ${H}$ denotes the fading envelope following a non-negative random distribution whereas  $\Theta$ denotes the fading phase uniformly distributed over $[-\pi,\pi]$. As well explained in \secref{Section:AWMNChannels} above, both ${H}$ and $\Theta$ are assumed constant during~the~period of each~modulation symbol because of the transmission coherence time arising out of fading conditions \cite{BibProakisBook,BibGoldsmithBook,BibAlouiniBook}, but each has a random nature while considering the overall transmission. Therefore, in coherent receiver, both ${H}$ and $\Theta$ is required to be without loss of generality perfectly estimated at the receiver during the period of each modulation symbol vector. However, there is no need to estimate ${H}$ and $\Theta$ in non-coherent receiver. Additionally, the covariance matrix $\defmat{\Sigma}$ of $\defrmat{Z}\!\sim\!\mathcal{CM}^{L}_{\nu}(\defvec{0},\defmat{\Sigma})$ is assumed perfectly estimated during that period. Eventually, thanks to the \ac{ES} property of $\defrmat{Z}\!\sim\!\mathcal{CM}^{L}_{\nu}(\defvec{0},\defmat{\Sigma})$ (i.e., with the aid of $f_{\defrmat{Z}}(\defvec{z})\!=\!f_{\defrmat{Z}}({e}^{\imaginary\Theta}\defvec{z})$ when $\mathbb{E}[\defrvec{Z}]\!=\!\defvec{0}$), the received vector $\defrvec{R}$ depends statistically on $\defrvec{S}$ with the conditional \ac{PDF} $f_{\defrvec{R}|\defrvec{S}}(\defvec{r}|\defvec{s})$, which we derive from \eqref{Eq:ComplexAWMNVectorChannel} with the aid of \theoremref{Theorem:MultivariateCESMcLeishCDF} as
\begin{equation}\label{Eq:AWMNReceivedVectorConditionalPDF}
f_{\defrvec{R}|\defrvec{S}}(\defvec{r}|\defvec{s})=
    \frac{2}{\pi^L\Gamma(\nu)}
	\frac{{\lVert{\defvec{r}-{H}{e}^{\imaginary\Theta}\defmat{F}\defvec{s}}\rVert}_{\defmat{\Sigma}}^{\nu-L}}
		{\det(\defmat{\Sigma})\,\lambda_{0}^{\nu+L}}
		    {K}_{\nu-L}
			    \Bigl(
					\frac{2}{\lambda_{0}}
				{\bigl\lVert{\defvec{r}-{H}{e}^{\imaginary\Theta}\defmat{F}\defvec{s}}\bigr\rVert}_{\defmat{\Sigma}}
	        	\Bigr).\!\!
\end{equation}
Having the joint \ac{PDF} of $\defrvec{R}$ and $\defrvec{S}$, i.e., $f_{\defrvec{R},\defrvec{S}}(\defvec{r},\defvec{s})\!=\!{f}_{\defrvec{R}|\defrvec{S}}(\defvec{r}|\defvec{s})\allowbreak{f}_{\defrvec{S}}(\defvec{s})$ by means of \eqref{Eq:AWMNModulationSymbolVectorPMF} and \eqref{Eq:AWMNReceivedVectorConditionalPDF}, we obtain the \ac{PDF} of the received vector $\defrvec{R}$ as 
\begin{subequations}
\setlength\arraycolsep{1.4pt}
\begin{eqnarray}
f_{\defrvec{R}}(\defvec{r})
    &=&\int\!{f}_{\defrvec{R},\defrvec{S}}(\defvec{r},\defvec{s})\,d\defvec{s},\\
    &=&\sum_{m=1}^{M}f_{\defrvec{R}|\defrvec{S}}(\defvec{r}|\defvec{s}_m)\Pr\{\defrvec{S}=\defvec{s}_{m}\},\\
    &=&\sum_{m=1}^{M}p_{m}\frac{2}{\pi^L\Gamma(\nu)}
	    \frac{{\lVert{\defvec{r}-{H}{e}^{\imaginary\Theta}\defmat{F}\defvec{s}_m}\rVert}_{\defmat{\Sigma}}^{\nu-L}}
		    {\det(\defmat{\Sigma})\,\lambda_{0}^{\nu+L}}
		        {K}_{\nu-L}\Bigl(
					\frac{2}{\lambda_{0}}
				        {\bigl\lVert{\defvec{r}-{H}{e}^{\imaginary\Theta}\defmat{F}\defvec{s}_m}\bigr\rVert}_{\defmat{\Sigma}}
	        	\Bigr).{~~~}
\end{eqnarray}
\end{subequations}
After transmission of each modulation symbol, if the transmitted symbol $m$ and the optimally detected symbol $\widehat{m}$ are not the same, then we say that a transmission error has occurred with the probability given by 
\begin{equation}
    \Pr\{e\,|\,m\}=\Pr\{\widehat{m}\neq{m}\},
\end{equation}
whose averaging with respect to all possible modulation symbols results in the \ac{SER} of the transmission, that is 
\begin{equation}
    \Pr\{e\}=
        \sum_{m=1}^{M}\Pr\{e\,|\,m\}
            \Pr\{\defrvec{S}=\defvec{s}_{m}\},
\end{equation}
which will be derived for coherent\,/\,non-coherent signaling using digital modulation schemes over \ac{CES} \ac{AWMN} channels. 

\subsection{Coherent Signalling}
\label{Section:SignallingOverAWMNChannels:CoherentSignalling}
As referring to the mathematical model given by \eqref{Eq:ComplexAWMNVectorChannel}, we assume that the receiver has a perfect knowledge of the phase, or in some cases, that of both the amplitude and the phase in coherent signaling. As such, during the transmission of each modulation symbol while being conditioned on $H$ and $\Theta$, if the transmitted symbol $m$ and the optimally detected symbol $\widehat{m}$ are not the same, then we say that an instantaneous symbol error has occurred with the probability given by 
\begin{equation}
    \Pr\{e\,|\,H,\Theta\}=\Pr\{\widehat{m}\neq{m}\,|\,H,\Theta\}.
\end{equation}
whose averaging with respect to $H$ and $\Theta$ while considering all symbols results in the averaged \ac{SER} of the transmission. The receiver observes $\defrvec{R}$, and based on this observation, decides which modulation symbol was transmitted, essentially by an optimal detection rule that minimizes the error probability or equivalently maximizes correct decision. The optimal detection rule, which is also occasionally called \ac{MAP} rule  \cite{BibProakisBook,BibAlouiniBook,BibGoldsmithBook}, produces the index of the most probable transmitted symbol that maximizes $f_{\defrvec{R},\defrvec{S}}(\defvec{r},\defvec{s})$. In more details, in order to acquire the index of the most probable transmitted symbol, we write the \ac{MAP} decision rule accordingly as follows
\begin{subequations}\label{Eq:AWMNMAPRule}
\setlength\arraycolsep{1.4pt}
\begin{eqnarray}
    \label{Eq:AWMNMAPRuleA}
    \widehat{m}&=&
        \argmax_{{1}\leq{m}\leq{M}}
            f_{\defrvec{R},\defrvec{S}}(\defrvec{R},\defvec{s}_{m}),\\
    \label{Eq:AWMNMAPRuleB}
        &=&
        \argmax_{{1}\leq{m}\leq{M}}
        f_{\defrvec{S}|\defrvec{R}}(\defvec{s}_{m}|\defrvec{R})
            f_{\defrvec{r}}(\defvec{r}),\\
    \label{Eq:AWMNMAPRuleC}
        &=&
        \argmax_{{1}\leq{m}\leq{M}}
            f_{\defrvec{S}|\defrvec{R}}(\defvec{s}_{m}|\defrvec{R}),
\end{eqnarray}
\end{subequations}
which decides in favor of the modulation symbol that maximizes the conditional \ac{PDF} $f_{\defrvec{S}|\defrvec{R}}(\defvec{s}_{m}|\defvec{r})$. Further, we simplify the \ac{MAP} rule more to 
\begin{equation}\label{Eq:AWMNMAPRuleII}
    \widehat{m}=\argmax_{{1}\leq{m}\leq{M}}
        f_{\defrvec{R}|\defrvec{S}}
            (\defrvec{R}|\defvec{s}_{m})\Pr\{\defrvec{S}=\defvec{s}_{m}\},
\end{equation}
where we often call $f_{\defrvec{R}|\defrvec{S}}(\defrvec{R}|\defvec{s}_{m})$ the likelihood of the symbol $\defvec{s}_{m}$ given the received vector $\defrvec{R}$. Hence also, we often remark that the \ac{MAP} rule, given above, clearly illustrates how each decision given the received vector $\defrvec{R}$ maps into one of the $M$ possible transmitted modulation symbols. Corresponding to the $M$ possible decisions, we partition the sample space of $\defrvec{R}$ into $M$ regions, and therefrom define the decision region for the symbol $\widehat{m}$ as
\begin{equation}\label{Eq:AWMNMAPRuleDecisionRegion}
\mathbb{D}^{\text{MAP}}_{\widehat{m}}=\Bigl\{
        \Bigl.
        \defvec{r}\in\mathbb{C}^{L}
        \,\Bigr|\,
        f_{\defrvec{R}|\defrvec{S}}(\defvec{r}|\defvec{s}_{\widehat{m}})\Pr\{\defrvec{S}=\defvec{s}_{\widehat{m}}\}
        \geq
        f_{\defrvec{R}|\defrvec{S}}(\defvec{r}|\defvec{s}_{m})\Pr\{\defrvec{S}=\defvec{s}_{m}\},
        \forall{m}\neq{\widehat{m}}
    \Bigr\},
\end{equation}
which imposes that the decision regions are non-overlapping (i.e., $\mathbb{D}_{m}\cap\mathbb{D}_{n}\!=\!\emptyset$ for all ${m}\!\neq\!{n}$). In addition, \eqref{Eq:AWMNMAPRuleDecisionRegion} stipulates that each decision region can be described in terms of at most $M-1$ inequalities. In general, these $M$ decision regions need not be connected with each other. When the receiver observes that the received vector $\defrvec{R}$ has fallen into the region $\mathbb{D}_{m}$ (i.e., when $\defrvec{R}\!\in\!\mathbb{D}_{m}$), it decides that the transmitted symbol is the modulation symbol $m$.
Eventually, substituting \eqref{Eq:AWMNReceivedVectorConditionalPDF} into \eqref{Eq:AWMNMAPRuleII} yields the \ac{MAP} decision rule as follows 
\begin{equation}\label{Eq:AWMNMAPRuleIII}
\widehat{m}=\argmax_{{1}\leq{m}\leq{M}}
    \frac{2{p}_{m}}{\pi^L\Gamma(\nu)}
	\frac{{\lVert{\defrvec{R}-{H}{e}^{\imaginary\Theta}\defmat{F}\defvec{s}_{m}}\rVert}_{\defmat{\Sigma}}^{\nu-L}}{\det(\defmat{\Sigma})\,\lambda_{0}^{\nu+L}}
		    {K}_{\nu-L}
			    \Bigl(
					\frac{2}{\lambda_{0}}
				{\bigl\lVert{\defrvec{R}-{H}{e}^{\imaginary\Theta}\defmat{F}\defvec{s}_{m}}\bigr\rVert}_{\defmat{\Sigma}}
	        	\Bigr),
\end{equation}
which can be even simplified more using the \ac{CES} property around mean, as shown in the following. 

\begin{theorem}\label{Theorem:MAPDecisionRuleForComplexAWMNVectorChannel}
For the complex vector channel introduced in \eqref{Eq:ComplexAWMNVectorChannel}, the coherent \ac{MAP} detection  rule is given by 
\begin{equation}\label{Eq:MAPDecisionRuleForComplexAWMNVectorChannel}
\!\!\widehat{m}=\argmax_{{1}\leq{m}\leq{M}}
        \,\Bigl(
        2\log(p_m)-
            {\bigl\lVert{\defvec{R}-{H}{e}^{\imaginary\Theta}\defmat{F}\defvec{s}_m}\bigr\rVert}^{2}_{\defmat{\Sigma}}
        \Bigr),\!\!
\end{equation}
under the condition that ${H}\!\in\!\mathbb{R}_{+}$ and $\Theta\!\in\![-\pi,\pi)$ are assumed perfectly estimated during each modulation symbol.
\end{theorem}

\begin{proof}
Note that the received vector $\defrmat{R}$ given the transmitted symbol $\defrmat{S}\!=\!\defvec{s}_{m}$ follows a multivariate \ac{CES} McLeish distribution, i.e., 
$\defrmat{R}\!\sim\!\mathcal{CM}_{\nu}^L\!\bigl({H}{e}^{\imaginary\Theta}\defvec{s}_{m},\defmat{\Sigma}\bigr)$. According to both \eqref{Eq:MultivariateComplexMcLeishComponentDecomposition} and \eqref{Eq:MultivariateCESMcLeishDecomposition}, the received vector $\defrmat{R}$ given the transmitted symbol $\defrmat{S}$ can be decomposed as 
\begin{equation}
(\defrvec{R}|\defrvec{S})={H}{e}^{\imaginary\Theta}\defmat{F}\defrvec{S}+\sqrt{G}\,\defmat{D}\,(\defrvec{N}_1+\imaginary\defrvec{N}_2),
\end{equation}
where $\defmat{D}$ is the Cholesky decomposition of $\defmat{\Sigma}$ such that $\defmat{\Sigma}\!=\!\defmat{D}\defmat{D}^{H}$, and where $\defrvec{N}_1\!\sim\!\mathcal{N}^L(0,\defmat{I})$, $\defrvec{N}_2\!\sim\!\mathcal{N}^L(0,\defmat{I})$
and $G\!\sim\!\mathcal{G}(\nu,1)$.
Accordingly, the \ac{PDF} of $\defrmat{R}$ conditioned on both $\defrvec{S}$ and $G$, i.e., $f_{\defrmat{R}|\defrmat{S},G}(\defvec{z}|\defvec{s},g)$ can be written as
\begin{equation}\label{Eq:MAPDecisionConditionalPDF}
\!\!f_{\defrmat{R}|\defrmat{S},G}(\defvec{r}|\defvec{s},g)=
	\frac{\exp\bigl(-\frac{1}{2g}{\lVert\defvec{r}-{H}{e}^{\imaginary\Theta}\defmat{F}\defvec{s}\rVert}^{2}_\defmat{\Sigma}\bigr)}
	{(2\pi)^{L}g^{L}\det(\defmat{\Sigma})},
\end{equation} 
for $g\!\in\!\mathbb{R}_{+}$. Then, the conditional \ac{PDF} $f_{\defrvec{R}|\defrvec{S}}(\defrvec{R}|\defvec{s})$ is obtained by $f_{\defrvec{R}|\defrvec{S}}(\defrvec{R}|\defvec{s})\!=\!\int_{0}^{\infty}f_{\defrvec{R}|\defrvec{S},G}(\defrvec{R}|\defvec{s},g)\,f_{G}(g)\,dg$, where $f_{G}(g)$ is the \ac{PDF} of $G\!\sim\!\mathcal{G}(\nu,1)$, and given in \eqref{Eq:ProportionPDF}. Upon substituting $f_{\defrvec{R}|\defrvec{S}}(\defrvec{R}|\defvec{s}_{m})$ into \eqref{Eq:AWMNMAPRuleII}, we rewrite the \ac{MAP} rule as
\begin{subequations}\label{Eq:MAPDecisionDerivation}
\setlength\arraycolsep{1.4pt}
\begin{eqnarray}
    \label{Eq:MAPDecisionDerivationA}
    \!\!\!\!\widehat{m}
        &\overset{(a)}{=}&\argmax_{{1}\leq{m}\leq{M}}
            {p}_{m}\int_{0}^{\infty}               
                f_{\defrvec{R}|\defrvec{S},G}
                    (\defrvec{R}|\defvec{s}_{m},g)f_{G}(g)\,dg,\quad\quad\\
    \label{Eq:MAPDecisionDerivationB}                
        &\overset{(b)}{=}&\argmax_{{1}\leq{m}\leq{M}}
            {p}_{m}f_{\defrvec{R}|\defrvec{S},G}
                (\defrvec{R}|\defvec{s}_{m},\mathbb{E}[G]),   
\end{eqnarray}
\end{subequations}
where we have used the following steps in simplifying the expression. In step $(a)$, we observe that \eqref{Eq:MAPDecisionConditionalPDF} is being averaged by the \ac{PDF} $f_{G}(g)$, and notice that $f_{G}(g)\!\geq\!{0}$ for all $g\in\mathbb{R}_{+}$, which simplifies \eqref{Eq:MAPDecisionDerivationA} to \eqref{Eq:MAPDecisionDerivationB} with $\mathbb{E}[G]\!=\!{1}$. Then, in step $(b)$, we substitute \eqref{Eq:MAPDecisionConditionalPDF} into \eqref{Eq:MAPDecisionDerivationB} and drop all the positive constant terms. Accordingly, we obtain
\begin{equation}\label{Eq:MAPDecisionDerivationII}     
    \widehat{m}=\argmax_{{1}\leq{m}\leq{M}}
        {p}_{m}\exp\Bigl(
                -\frac{1}{2}{\bigl\lVert
                    \defrvec{R}-{H}{e}^{\imaginary\Theta}\defmat{F}\defvec{s}_{m}
                        \bigr\rVert}^{2}_\defmat{\Sigma}\Bigr).
\end{equation}
We acknowledge that, since the $\log(\cdot)$ function is a monotonically increasing function, we simplify this maximization by applying the $\log(\cdot)$ function to \eqref{Eq:MAPDecisionDerivationII}. Eventually, multiplying the resultant by $2$, we obtain \eqref{Eq:MAPDecisionRuleForComplexAWMNVectorChannel}, which proves \theoremref{Theorem:MAPDecisionRuleForComplexAWMNVectorChannel}.
\end{proof}

It is worth mentioning that, in some signaling conditions, some parameters within \eqref{Eq:MAPDecisionRuleForComplexAWMNVectorChannel} may be discarded without loss of performance. Appropriately, the \ac{MAP} rule can be even reduced more to a simple form. Namely, in case of that the modulation symbol vectors are equiprobable (i.e., when $\Pr\{\defrvec{S}\!=\!\defvec{s}_{m}\}\!=\!\Pr\{\defrvec{S}\!=\!\defvec{s}_{n}\}$, ${1}\!\leq\!{m,n}\!\leq\!{M}$), we ignore the term $\Pr\{\defrvec{S}\!=\!\defvec{s}_{m}\}$ in \eqref{Eq:AWMNMAPRuleII}, and thereby further simplify the \ac{MAP} decision rule to \begin{equation}\label{Eq:AWMNMaximumLikelihoodRule}
    \widehat{m}=\argmax_{{1}\leq{m}\leq{M}}
        f_{\defrvec{R}|\defrvec{S}}(\defrvec{R}|\defvec{s}_{m}),
\end{equation}
which we call the \ac{ML} decision rule. Appropriately, we simply define the decision region for the symbol $\widehat{m}$ as follows 
\begin{equation}\label{Eq:AWMNMLRuleDecisionRegion}
\!\!\!\!\!\!\mathbb{D}^{\text{ML}}_{\widehat{m}}=\Bigl\{
        \Bigl.
        \defvec{r}\in\mathbb{C}^{L}
        \,\!\Bigr|\,\!
        f_{\defrvec{R}|\defrvec{S}}(\defvec{r}|\defvec{s}_{\widehat{m}})
        \geq
        f_{\defrvec{R}|\defrvec{S}}(\defvec{r}|\defvec{s}_{m}),
        {\widehat{m}}\neq{m}
    \Bigr\}.\!\!\!\!
\end{equation}
Further, in \eqref{Eq:AWMNMaximumLikelihoodRule}, We calculate the likelihood of the modulation symbol $m$, i.e.,  $f_{\defrvec{R}|\defrvec{S}}(\defrvec{R}|\defvec{s}_{m})$ by using the conditional \ac{PDF} given in \eqref{Eq:AWMNReceivedVectorConditionalPDF}, and we simplify it more in the following. 

\begin{theorem}\label{Theorem:MLDecisionRuleForComplexAWMNVectorChannel}
For the complex vector channel introduced in \eqref{Eq:ComplexAWMNVectorChannel}, the coherent \ac{ML} detection rule is given by 
\begin{equation}\label{Eq:MLDecisionRuleForComplexAWMNVectorChannel}
    \widehat{m}=\argmin_{{1}\leq{m}\leq{M}}
        \,{\bigl\lVert{\defvec{R}-{H}{e}^{\imaginary\Theta}\defmat{F}\defvec{s}_m}\bigr\rVert}^{2}_{\defmat{\Sigma}},
\end{equation}
under the condition that ${H}\!\in\!\mathbb{R}_{+}$ and $\Theta\!\in\![-\pi,\pi)$ are assumed perfectly estimated during each modulation symbol.
\end{theorem}

\begin{proof}
The \ac{ML} decision rule states that each modulation symbol has the same probability of transmission. In accordance, in \eqref{Eq:MAPDecisionRuleForComplexAWMNVectorChannel}, we make $p_{m}\!=\!1/M$ for all ${1}\!\leq\!{m}\!\leq\!{M}$ and therein ignore the term $2\log(p_m)$ same for all modulation symbols. Finally, changing the maximization to the minimization, we readily deduce \eqref{Eq:MLDecisionRuleForComplexAWMNVectorChannel}, which completes the proof of \theoremref{Theorem:MAPDecisionRuleForComplexAWMNVectorChannel}. 
\end{proof}

As an interpretation of \eqref{Eq:MLDecisionRuleForComplexAWMNVectorChannel}, we explicate that the receiver observes the received vector $\defrvec{R}$. Then, using a decision rule, it searches a symbol among all modulation symbols $\{\defvec{s}_m\}_{m=1}^{M}$, that is closest to the received vector $\defrvec{R}$ by using Mahalanobis distance. When the modulation symbols are equiprobable, the optimal detector uses the \ac{ML} decision rule, and therefore we occasionally call it the minimum-distance (or nearest-neighbor) detector. In this case, we corroborate the finding that the boundaries between the decision region of $\defvec{s}_m$ and that of $\defvec{s}_n$ are the set of hyper-plane points that are equidistant from these two modulation symbols.

In case of that the modulation symbols are equiprobable and have equal power (i.e., when $\Pr\{\defrvec{S}\!=\!\defvec{s}_{m}\}\!=\!\Pr\{\defrvec{S}\!=\!\defvec{s}_{n}\}$ and  ${\lVert\defvec{s}_{m}\rVert}^{2}\!=\!{\lVert\defvec{s}_{n}\rVert}^{2}$ for all ${1}\!\leq\!{m,n}\!\leq\!{M}$), we revise the optimal detection rule either from the \ac{MAP} rule or the \ac{ML} rule and accordingly we put it in much simpler form, that is   
\begin{equation}
\label{Eq:OptimalDecisionRuleForPrecodedComplexAWMNVectorChannel}
    \widehat{m}=\argmax_{{1}\leq{m}\leq{M}}\,\Re\bigl\{{e}^{-\imaginary\Theta}\defvec{s}^{H}_{m}\defrvec{R}\bigr\},
\end{equation}
whose decision region $\mathbb{D}_{\widehat{m}}$ is given by 
\begin{equation}
\mathbb{D}^{\text{ML}}_{\widehat{m}}=\Bigl\{\Bigl.\defvec{r}\in\mathbb{C}^{L}\,\Bigr|\,
        \Re\bigl\{{e}^{-\imaginary\Theta}\defvec{s}^{H}_{\widehat{m}}\defvec{r}\bigr\}\geq
        \Re\bigl\{{e}^{-\imaginary\Theta}\defvec{s}^{H}_{m}\defvec{r}\bigr\},
        \forall{m}\neq\widehat{m}
    \Bigr\},
\end{equation}
where $\Re\bigl\{{e}^{-\imaginary\Theta}\defvec{s}^{H}_{m}\defvec{r}\bigr\}$, ${1}\!\leq\!{m}\!\leq\!{M}$ can be readily rewritten as 
\begin{equation}
    \Re\bigl\{{e}^{-\imaginary\Theta}\defvec{s}^{H}_{m}\defvec{r}\bigr\}=
    \frac{1}{2}\bigl(
        {e}^{-\imaginary\Theta}\defvec{s}^{H}_{m}\defvec{r}+
        {e}^{\imaginary\Theta}\defvec{r}^{H}\defvec{s}_{m}
    \bigr).
\end{equation}
It is worth mentioning that when we compare both \ac{ML} and \ac{MAP} decision rules given above, we differ only the inclusion of a priori probabilities $\Pr\{\defrvec{S}\!=\!\defvec{s}_{m}\}$, ${1}\!\leq\!{m}\!\leq\!{M}$ in the \ac{MAP} rule, otherwise we observe that they are conceptually identical. This means that we perceive the \ac{MAP} rule when we weight the \ac{ML} rule with a priori probabilities. In addition, in both \theoremref{Theorem:MAPDecisionRuleForComplexAWMNVectorChannel} and \theoremref{Theorem:MLDecisionRuleForComplexAWMNVectorChannel}, the term  ${\lVert{\defrvec{R}-{H}{e}^{\imaginary\Theta}\defmat{F}\defvec{s}_{m}}\rVert}^{2}_{\defmat{\Sigma}}$ is the square of the Mahalanobis distance between the received vector $\defrvec{R}$ and its mean  ${H}{e}^{\imaginary\Theta}\defmat{F}\defvec{s}_{m}$. We decompose it as 
\begin{subequations}\label{Eq:MAPRuleMahalanobisDistance} 
\setlength\arraycolsep{1.4pt}
\begin{eqnarray}
    \label{Eq:MAPRuleMahalanobisDistanceA} 
    {\bigl\lVert{
        \defrvec{R}-{H}{e}^{\imaginary\Theta}\defmat{F}\defvec{s}_{m}
    }\bigr\rVert}^{2}_{\defmat{\Sigma}}
    &=&
    {\bigl\lVert{
        {e}^{\imaginary\Theta}({e}^{-\imaginary\Theta}\defrvec{R}-{H}\defmat{F}\defvec{s}_{m})
    }\bigr\rVert}^{2}_{\defmat{\Sigma}},\quad\quad\quad\\
    \label{Eq:MAPRuleMahalanobisDistanceB} 
    &\overset{(a)}{=}&
    {\bigl\lVert{
        {e}^{-\imaginary\Theta}\defrvec{R}-{H}\defmat{F}\defvec{s}_{m}
    }\bigr\rVert}^{2}_{\defmat{\Sigma}},\\
    \label{Eq:MAPRuleMahalanobisDistanceC} 
    &\overset{(b)}{\equiv}&
    {\bigl\lVert{
        \defrvec{R}-{H}\defmat{F}\defvec{s}_{m}
    }\bigr\rVert}^{2}_{\defmat{\Sigma}},
\end{eqnarray}
\end{subequations}
Thanks to the \ac{ES} property of $\defrmat{Z}\!\sim\!\mathcal{CM}^{L}_{\nu}(\defvec{0},\defmat{\Sigma})$, i.e., with~the aid of the fact that $f_{\defrmat{Z}}(\defvec{z})\!=\!f_{\defrmat{Z}}({e}^{\imaginary\Theta}\defvec{z})$, we progress \eqref{Eq:MAPRuleMahalanobisDistance} from step $(a)$ to step $(b)$. Being aware of that $\defrmat{Z}$ and ${e}^{\imaginary\Theta}\defrmat{Z}$ follow the same distribution, we have 
\begin{subequations}\label{Eq:ComplexAWMNChannelsESProperty} 
\setlength\arraycolsep{1.4pt}
\begin{eqnarray}
    \label{Eq:ComplexAWMNChannelsESPropertyA} 
        {e}^{-\imaginary\Theta}\defrvec{R}&=&{e}^{-\imaginary\Theta}({H}{e}^{\imaginary\Theta}\defmat{F}\defvec{S}+\defrmat{Z}),\\
    \label{Eq:ComplexAWMNChannelsESPropertyB}
        &=&{H}\defmat{F}\defvec{S}+{e}^{\imaginary\Theta}\defrmat{Z},\\
    \label{Eq:ComplexAWMNChannelsESPropertyC} 
        &\equiv&{H}\defmat{F}\defvec{S}+\defrmat{Z},
\end{eqnarray}
\end{subequations}
from which we notice that, without any performance degradation, the receiver completely compensate the fading phase $\Theta$ by co-phasing the received vector $\defrvec{R}$ with $\exp(-\imaginary\Theta)$ before the optimal detection (i.e., \ac{MAP}\,/\,\ac{ML} decision rules). The other crucial point we notice is the decorrelation of the channels to further simplify the receiver. For this purpose, we decompose 
\begin{equation}\label{Eq:MahalanobisDistanceExpansion}
{\bigl\lVert{\defrvec{R}-{H}{e}^{\imaginary\Theta}\defmat{F}\defvec{s}_{m}}\bigr\rVert}^{2}_{\defmat{\Sigma}}
    ={H}^{2}\defvec{s}_{m}^H\defmat{F}^{H}\defmat{\Sigma}^{-1}\defmat{F}\defvec{s}_{m}
        -2H\Re\bigl\{{e}^{-\imaginary\Theta}\defvec{s}_{m}^H\defmat{F}^{H}\defmat{\Sigma}^{-1}\defmat{F}\defrvec{R}\bigr\}
        +\defrvec{R}^H\defmat{\Sigma}^{-1}\defrvec{R},\!\!
\end{equation}
where, in order to avoid the performance degradation resulting from non-zero cross correlation between channels, we need to carefully choose the precoding matrix filter $\defmat{F}$ in such a way that eliminates the term $\defmat{F}^{H}\defmat{\Sigma}^{-1}\defmat{F}$ while maximizing the power of the received signal. The covariance matrix and total power of $\defrvec{Z}\!\sim\!\mathcal{M}^L_{\nu}(\defvec{0},\defmat{\Sigma})$ are given by
\begin{eqnarray}
    \mathbb{E}\bigl[
        \defrvec{Z}\defrvec{Z}^H
    \bigr]&=&2\defmat{\Sigma},\\
    \mathbb{E}\bigl[
        \defrvec{Z}^H\defrvec{Z}
    \bigr]&=&2\trace(\defmat{\Sigma}),   
\end{eqnarray}
respectively, where we remark that $\defmat{\Sigma}$ is a square and conjugate symmetric matrix and hence lets us use Cholesky's decomposition\cite[Chap.~\!10]{BibHighamBook2002}, \cite[Sec.~\!2.2]{BibHammerlinBook2012} to map $\defmat{\Sigma}$ into the product of $\defmat{\Sigma}\!=\!\defmat{D}\defmat{D}^{H}\!$, where $\defmat{D}$ is the lower triangular matrix and $\defmat{D}^{H}$ is the transposed, complex conjugate, and therefore of upper triangular form. We find that
\begin{equation}\label{Eq:ComplexAWMNChannelPrecodingMatrixFilter}
    \defmat{F}=\sqrt{\frac{2L}{\trace(\defmat{\Sigma})}}\defmat{D}=
        \sqrt{\frac{2}{N_{0}}}\defmat{D},
\end{equation}
where $N_{0}$ is the averaged total variance per noise component in the complex vector channel. Accordingly, we express $\defmat{\Sigma}$ as
\begin{equation}\label{Eq:ComplexAWMNChannelCovarianceMatrixUsingPrecodingMatrix}
    \defmat{\Sigma}=\frac{N_{0}}{2}\defmat{F}\defmat{F}^{H}.
\end{equation}
Substituting $\eqref{Eq:ComplexAWMNChannelCovarianceMatrixUsingPrecodingMatrix}$ into \eqref{Eq:MahalanobisDistanceExpansion}, we obtain 
\begin{equation}\label{Eq:MahalanobisDistancePrecodedExpansion}
{\bigl\lVert{
        \defrvec{R}-{H}{e}^{\imaginary\Theta}\defmat{F}\defvec{s}_{m}
    }\bigr\rVert}^{2}_{\defmat{\Sigma}}=
    2\frac{{H}^{2}}{N_{0}}{\bigl\lVert\defvec{s}_{m}\bigr\rVert}^{2}
            -4\frac{{H}}{N_{0}}\Re\bigl\{{e}^{-\imaginary\Theta}\defvec{s}^{H}_{m}\defrvec{R}\bigr\}
            +{\bigl\lVert\defrvec{R}\bigr\rVert}^{2}_{\defmat{\Sigma}}.
\end{equation}
Accordingly, choosing $\defmat{F}$ as in  \eqref{Eq:ComplexAWMNChannelPrecodingMatrixFilter} equalizes the received vector $\defrvec{R}$ from the channel, introduced in \eqref{Eq:ComplexAWMNVectorChannel}, to yield the equalized version before it is fed to the optimal detector, that is given by
\begin{subequations}\label{Eq:ComplexAWMNVectorChannelEqualization}
\begin{eqnarray}
    \label{Eq:ComplexAWMNVectorChannelEqualizationA}
    \defmat{F}^{-1}\defrvec{R}
    &=&\defmat{F}^{-1}
            \bigl({H}{e}^{\imaginary\Theta}\defmat{F}\defvec{S}+\defrmat{Z}\bigr),\\
    \label{Eq:ComplexAWMNVectorChannelEqualizationB}
    &=&H{e}^{\imaginary\Theta}\defvec{S}+\defmat{F}^{-1}\defrmat{Z},\\
    \label{Eq:ComplexAWMNVectorChannelEqualizationC}
    &=&H{e}^{\imaginary\Theta}\defvec{S}+\defrvec{Z}_{c},
\end{eqnarray}
\end{subequations}
where $\defrvec{Z}\!\sim\!\mathcal{CM}_{\nu}^{L}(\defmat{0},\defmat{\Sigma})$ whose \ac{PDF} is already given by \eqref{Eq:AWMNAdditiveNoiseVectorPDF}, and  $\defrvec{Z}_{c}\!\sim\!\mathcal{CM}_{\nu}^{L}(\defmat{0},\frac{N_{0}}{2}\defmat{I})$ follows the \ac{PDF} obtained with the aid of both \theoremref{Theorem:INIDMultivariateCESMcLeishPDF} and the special case \eqref{Eq:INIDMultivariateCESMcLeishPDFWithUniformVariances}, that is
\begin{equation}\label{Eq:AWMNINIDVectorChannelNoisePDF}
\!\!\!\!f_{\defrvec{Z}_{c}}(\defvec{z})=
		\frac{2}{\pi^{L}}
		\frac{{\bigl\lVert{\defvec{z}}\bigr\rVert}^{\nu-{L}}}
			{\Gamma(\nu)\Lambda_{0}^{\nu+{L}}}
				{K}_{\nu-{L}}\Bigl(\frac{2}{\Lambda_{0}}{\bigl\lVert{\defvec{z}}\bigr\rVert}\Bigr)
\end{equation}
where $\Lambda_{0}$ is the component deviation (i.e., the variance per each Laplacian component) and  obtained by 
\begin{equation}
\Lambda_{0}=\sqrt{\frac{2}{\nu}\frac{\trace(\defmat{F}^{H}\defmat{\Sigma}^{-1}\defmat{F})}
    {\trace(\defmat{D}^{H}\defmat{\Sigma}^{-1}\defmat{D})}}
    =\sqrt{\frac{N_{0}}{\nu}}.
\end{equation}
Properly, both from the phase compensation~presented~in \eqref{Eq:ComplexAWMNChannelsESProperty} and the  equalization steps presented in \eqref{Eq:ComplexAWMNVectorChannelEqualization}, we~conclude~that, thanks to the coherence time of the vector channel, the received vector can be equalized by the precoding matrix filter $\defmat{F}$ and also can be maximized by phase compensation before the optimal detection as follows
\begin{subequations}\label{Eq:ComplexAWMNVectorChannelCoherentEqualization}
\begin{eqnarray}
    \label{Eq:ComplexAWMNVectorChannelCoherentEqualizationA}
    \defrvec{R}_{c}&=&{e}^{-\imaginary\Theta}\defmat{F}^{-1}\defrvec{R},\\
    \label{Eq:ComplexAWMNVectorChannelCoherentEqualizationB}
    &=&{e}^{-\imaginary\Theta}\defmat{F}^{-1}
            \bigl({H}{e}^{\imaginary\Theta}\defmat{F}\defvec{S}+\defrmat{Z}\bigr),\\
    \label{ComplexAWMNVectorChannelCoherentEqualizationC}
    &\equiv&H\defvec{S}+\defmat{F}^{-1}\defrmat{Z},\\
    \label{ComplexAWMNVectorChannelCoherentEqualizationD}
    &=&H\defvec{S}+\defrvec{Z}_{c},
\end{eqnarray}
\end{subequations}
which simplifies the complex correlated \ac{AWMN} vector channel, introduced above in \eqref{Eq:ComplexAWMNVectorChannel}, to the simple one, which we call the uncorrelated complex \ac{AWMN} vector channels, whose mathematical model is typically given by
\begin{equation}\label{Eq:PrecodedComplexAWMNVectorChannel}
	\defrvec{R}_{c}={H}\defvec{S}+\defrvec{Z}_{c}.
\end{equation}
where during each modulation symbol, $\defrvec{R}_{c}$ depends statistically on $\defrvec{S}$. With the aid of \eqref{Eq:AWMNINIDVectorChannelNoisePDF}, we obtain the conditional \ac{PDF} $f_{\defrvec{R}_{c}|\defrvec{S}}(\defvec{r}|\defvec{s})$ as
\begin{equation}\label{Eq:PrecodedComplexAWMNVectorConditionalPDF}
    \!\!\!\!f_{\defrvec{R}_{c}|\defrvec{S}}(\defvec{r}|\defvec{s})=
		\frac{2}{\pi^{L}}
		\frac{{\bigl\lVert{\defvec{r}-H\defvec{s}}\bigr\rVert}^{\nu-{L}\!\!\!}}
			{\Gamma(\nu)\Lambda_{0}^{\nu+{L}}}
				{K}_{\nu-{L}}\Bigl(\frac{2}{\Lambda_{0}}{\bigl\lVert{\defvec{r}-H\defvec{s}}\bigr\rVert}\Bigr).\!\!
\end{equation}
Accordingly, thanks to the \ac{CS} property of multivariate \ac{CCS} McLeish distribution (for more details, see \secref{Section:StatisticalBackground:MultivariateComplexMcLeishDistribution}), we just state that the \ac{BER}\,/\,\ac{SER} performance of the vector channel in \eqref{Eq:PrecodedComplexAWMNVectorChannel} is completely the same as that of one in \eqref{Eq:ComplexAWMNVectorChannel} when we choose the precoding matrix $\defmat{F}$ as  $\defmat{\Sigma}\!=\!{N_{0}}/{2}\defmat{F}\defmat{F}^{H}\!$. 

\begin{theorem}\label{Theorem:MAPDecisionRuleForPrecodedComplexAWMNVectorChannel}
The \ac{MAP} rule for complex uncorrelated \ac{AWMN} vector channels, defined in \eqref{Eq:PrecodedComplexAWMNVectorChannel}, is given by
\begin{subequations}\label{Eq:MAPDecisionRuleForPrecodedComplexAWMNVectorChannel}
\begin{eqnarray}
    \label{Eq:MAPDecisionRuleForPrecodedComplexAWMNVectorChannelA}
    \widehat{m}
        &=&\argmax_{{1}\leq{m}\leq{M}}
        \,\Bigl(
        N_{0}\log(p_m)-
            {\bigl\lVert{\defvec{R}_{c}-{H}\defvec{s}_m}\bigr\rVert}^{2}
        \Bigr),{~~~~}\\
        \label{Eq:MAPDecisionRuleForPrecodedComplexAWMNVectorChannelB}
        &=&\argmax_{{1}\leq{m}\leq{M}}
        \,\Bigl(
        N_{0}\log(p_m)\Bigr.+
        \Bigl.2H\RealPart{\defvec{s}^{H}_{m}\defvec{R}_{c}}-{H}^2{\bigl\lVert{\defvec{s}_{m}}\bigr\rVert}^{2}
        \Bigr).
\end{eqnarray}
\end{subequations}
with the decision region $\mathbb{D}^{\text{MAP}}_{\widehat{m}}$ given by 
\begin{multline}\label{Eq:MAPDecisionRegionForPrecodedComplexAWMNVectorChannel}
\!\!\!\!\!\!\mathbb{D}^{\text{MAP}}_{\widehat{m}}=\Bigl\{
        \Bigl.
        \defvec{r}\in\mathbb{C}^{L}
        \,\Bigr|\,
        N_{0}\log(p_{\widehat{m}})+2H\RealPart{\defvec{s}^{H}_{\widehat{m}}\defvec{r}}
        -{H}^2{\lVert{\defvec{s}_{\widehat{m}}}\rVert}^{2}
        \!\geq\!\\
        N_{0}\log(p_{m})+2H\RealPart{\defvec{s}^{H}_{m}\defvec{r}}
        -{H}^2{\lVert{\defvec{s}_{m}}\rVert}^{2},
        \forall{m}\neq{\widehat{m}}
    \Bigr\},\!\!    
\end{multline}
\end{theorem}

\begin{proof}
The proof is obvious putting $\defmat{\Sigma}\!=\!\frac{{N}_{0}}{2}\defmat{I}$ in \theoremref{Theorem:MAPDecisionRuleForComplexAWMNVectorChannel} and selecting $\defmat{F}\!=\!{e}^{-\imaginary\Theta}\defmat{I}$ as per the phase compensation.~With the aid of the \ac{MAP} decision rule \eqref{Eq:AWMNMAPRuleII}, we accordingly write the decision region of the modulation symbol $\widehat{m}$ as follows
\begin{equation}
\!\!\!\!\mathbb{D}^{\text{MAP}}_{\widehat{m}}=\Bigl\{
        \bigl.
        \defvec{r}\in\mathbb{C}^{L}
        \,\bigr|\,
        N_{0}\log(p_{\widehat{m}})-{\bigl\lVert{\defvec{r}-{H}\defvec{s}_{\widehat{m}}}\bigr\rVert}^{2}
        \geq
        N_{0}\log(p_{m})-{\bigl\lVert{\defvec{r}-{H}\defvec{s}_{m}}\bigr\rVert}^{2},
        \forall{m}\neq{\widehat{m}}
    \Bigr\},\!\!    
\end{equation}
where using ${\lVert{\defvec{r}-{H}\defvec{s}_{m}}\rVert}^{2}\!=\!{\lVert{\defvec{r}}\rVert}^{2}\!-\!2H\RealPart{\defvec{s}^{H}_{m}\defvec{r}}\!+\!{H}^2{\lVert{\defvec{s}_{m}}\rVert}^{2}$ and therein ignoring the term ${\bigl\lVert{\defvec{r}}\bigr\rVert}^{2}$, we immediately derive \eqref{Eq:MAPDecisionRegionForPrecodedComplexAWMNVectorChannel}, which completes the proof of  \theoremref{Theorem:MAPDecisionRuleForPrecodedComplexAWMNVectorChannel}.
\end{proof}

In case of that the modulation symbols are transmitted with equal a priori probabilities (i.e., $p_{m}=1/M$ for all ${1}\!\leq\!{m}\!\leq\!{M}$), the \ac{MAP} rule decision given in \theoremref{Theorem:MAPDecisionRuleForPrecodedComplexAWMNVectorChannel} is readily reduced to the \ac{ML} decision rule given in the following. 

\begin{theorem}\label{Theorem:MLDecisionRuleForPrecodedComplexAWMNVectorChannel}
The \ac{ML} rule for complex uncorrelated \ac{AWMN} vector channels, defined in \eqref{Eq:PrecodedComplexAWMNVectorChannel}, is given by
\begin{subequations}\label{Eq:MLDecisionRuleForPrecodedComplexAWMNVectorChannel}
\begin{eqnarray}
    \label{Eq:MLDecisionRuleForPrecodedComplexAWMNVectorChannelA}
    \widehat{m}
        &=&\argmin_{{1}\leq{m}\leq{M}}
        \,
        {\bigl\lVert{\defvec{R}_{c}-{H}\defvec{s}_m}\bigr\rVert}^{2}
        ,{~~~~}\\
        \label{Eq:MLDecisionRuleForPrecodedComplexAWMNVectorChannelB}
        &=&\argmin_{{1}\leq{m}\leq{M}}
        \,\Bigl(
        {H}^2{\bigl\lVert{\defvec{s}_{m}}\bigr\rVert}^{2}-2H\RealPart{\defvec{s}^{H}_{m}\defvec{R}_{c}}
        \Bigr),
\end{eqnarray}
\end{subequations}
with the decision region $\mathbb{D}^{\text{ML}}_{\widehat{m}}$ given by 
\begin{equation}\label{Eq:MLDecisionRegionForPrecodedComplexAWMNVectorChannel}
\!\!\mathbb{D}^{\text{ML}}_{\widehat{m}}=\Bigl\{
        \Bigl.
        \defvec{r}\in\mathbb{C}^{L}
        \,\Bigr|\,
        {H}^2{\bigl\lVert{\defvec{s}_{\widehat{m}}}\bigr\rVert}^{2}
        -2H\RealPart{\defvec{s}^{H}_{\widehat{m}}\defvec{r}}
        \leq
        {H}^2{\bigl\lVert{\defvec{s}_{m}}\bigr\rVert}^{2}
        -2H\RealPart{\defvec{s}^{H}_{m}\defvec{r}},
        \forall{m}\neq{\widehat{m}}
    \Bigr\}.\!\!
\end{equation}
\end{theorem}

\begin{proof}
The proof is obvious setting $p_{m}\!=\!1/M$, ${1}\!\leq\!{m}\!\leq\!{M}$ in \theoremref{Theorem:MAPDecisionRuleForPrecodedComplexAWMNVectorChannel} and then ignoring the term $N_{0}\log(p_m)\!=\!-N_{0}\log(M)$ since being the same for all possible modulation symbols. 
\end{proof}

Note that, when the modulation symbols have equal power, we identify that the term ${\lVert{\defvec{s}_{m}}\rVert}^{2}$ in \eqref{Eq:MLDecisionRuleForPrecodedComplexAWMNVectorChannelB} is constant for all $1\!\leq\!{m}\!\leq\!{M}$ and therefore can be ignored. In accordance, the optimal detection rule either from the \ac{MAP} rule or the \ac{ML} rule for complex uncorrelated \ac{AWMN} vector channels, defined in \eqref{Eq:PrecodedComplexAWMNVectorChannel}, reduces to 
\begin{equation}\label{Eq:OptimalMLDecisionRuleForPrecodedComplexAWMNVectorChannel}
    \widehat{m}=\argmax_{{1}\leq{m}\leq{M}}\,\Re\bigl\{\defvec{s}^{H}_{m}\defrvec{R}_{c}\bigr\},
\end{equation}
whose decision region $\mathbb{D}^{\text{ML}}_{\widehat{m}}$ is given by 
\begin{equation}
\!\!\!\!\mathbb{D}^{\text{ML}}_{\widehat{m}}=\bigl\{\bigl.\defvec{r}\in\mathbb{C}^{L}\,\bigr|\,
        \Re\bigl\{\defvec{s}^{H}_{\widehat{m}}\defvec{r}\bigr\}\geq
        \Re\bigl\{\defvec{s}^{H}_{m}\defvec{r}\bigr\},
        \forall{m}\neq\widehat{m}
    \bigr\}.\!\!
\end{equation}
Additionally, we notice that the other important point in~the~nature of complex vector channels, which is  well-known~in~the literature \cite{BibProakisBook,BibGoldsmithBook,BibAlouiniBook,BibRappaportBook}, is the rotational invariance property. As~being typically observed either in \theoremref{Theorem:MAPDecisionRuleForPrecodedComplexAWMNVectorChannel} or \theoremref{Theorem:MLDecisionRuleForPrecodedComplexAWMNVectorChannel} in accordance with the channel model given by \eqref{Eq:PrecodedComplexAWMNVectorChannel}, the \ac{ML} decision rule partitions
the sample space of the received vector $\defrvec{R}$ depending on the modulation constellation. However, the rotation of the modulation constellation does not change the probability of making a decision error, primarily because~of~two facts, one of which corresponds to that the \ac{ML} decision error depends only on distances between modulation symbols. The other fact is that the additive complex noise $\defrmat{Z}_{c}\!\sim\!\mathcal{CM}_{\nu}^{L}(\defmat{0},\frac{N_{0}}{2}\defmat{I})$ is \ac{CS} in all directions in signaling space. 

\subsubsection{Symbol Error Probability}
\label{Section:SignallingOverAWMNChannels:CoherentSignalling:SymbolErrorProbability}
In order to determine~and~assess the \ac{SER} of a detection scheme, let us assume that the modulation symbol $m$ (i.e. $\defvec{s}_{m}$) is randomly selected from~a modulation constellation and then transmitted through the complex vector channel, introduced above in \eqref{Eq:PrecodedComplexAWMNVectorChannel}. Appropriately, we write the received vector $\defrvec{R}$ as
\begin{equation}\label{Eq:PrecodedComplexAWMNVectorChannelII}
    \defrmat{R}_{c}={H}\defvec{s}_{m}+\defrmat{Z}_{c}   
\end{equation}
where $\defrmat{Z}_{c}\!\sim\!\mathcal{M}_{\nu}^{L}(\defmat{0},\frac{N_{0}}{2}\defmat{I})$. A decision error occurs only when the received vector $\defrvec{R}_{c}$ does not fall into the decision region $\mathbb{D}^{\text{MAP}}_{m}$ of the modulation symbol $m$ (i.e., $\defrvec{R}_{c}\!\not\in\!\mathbb{D}^{\text{MAP}}_{m}$ causes an error). Making allowance for all decision regions $\bigl\{\mathbb{D}^{\text{MAP}}_{m},\allowbreak{1}\!\leq\!{m}\!\leq\!{M}\bigr\}$ of the modulation constellation $\bigl\{\defvec{s}_{m},\allowbreak{1}\!\leq\!{m}\!\leq\!{M}\bigr\}$, the probability of that a receiver makes an error in detection of the modulation symbol $m$ is readily written as
\begin{subequations}\label{Eq:ModulationSymbolDecisionError}
\begin{eqnarray}
    \label{Eq:ModulationSymbolDecisionErrorA}
    \Pr\bigl\{\bigl.e\,\bigr|\,H,\defvec{s}_{m}\bigr\}
    &=&\Pr\bigl\{\bigl.\defrvec{R}_{c}\not\in\mathbb{D}^{\text{MAP}}_{m}\,\bigr|\,\defvec{s}_{m}\bigr\},\\
    \label{Eq:ModulationSymbolDecisionErrorB}
    &=&\sum^{M}_{\substack{{n}=1\\{n}\neq{m}}}\Pr\bigl\{\bigl.\defrvec{R}_{c}\in\mathbb{D}^{\text{MAP}}_{n}\,\bigr|\,\defvec{s}_{m}\bigr\},\quad\quad\\
    \label{Eq:ModulationSymbolDecisionErrorC}
    &=&\sum^{M}_{\substack{{n}=1\\{n}\neq{m}}}\int_{\mathbb{D}^{\text{MAP}}_{n}}\!\!f_{\defrvec{R}_{c}|\defrvec{S}}(\defvec{r}|\defvec{s}_{m})d\defvec{r},
\end{eqnarray}
\end{subequations}
where the conditional \ac{PDF} $f_{\defrvec{R}|\defrvec{S}}(\defvec{r}|\defvec{s})$ is given in \eqref{Eq:PrecodedComplexAWMNVectorConditionalPDF}. 
The conditional \ac{SER} of the receiver is therefore given by 
\begin{subequations}\label{Eq:ModulationSymbolErrorProbability}
\begin{eqnarray}
    \label{Eq:ModulationSymbolErrorProbabilityA}
    \Pr\bigl\{\bigl.e\,\bigr|\,H\bigr\}
        &=&\sum_{m=1}^{M}\Pr\bigl\{\defvec{s}_{m}\bigr\}\Pr\bigl\{\bigl.e\,\bigr|H,\defvec{s}_{m}\bigr\},\\
    \label{Eq:ModulationSymbolErrorProbabilityB}        
        &=&\sum_{m=1}^{M}{p}_{m}\Pr\bigl\{\bigl.e\,\bigr|H,\defvec{s}_{m}\bigr\},
\end{eqnarray}
\end{subequations}
where the probability of the modulation symbol $m$ we select to transmit is typically denoted by ${p}_{m}\!=\!\Pr\{\defvec{s}_{m}\}$, 
and where inserting \eqref{Eq:ModulationSymbolDecisionErrorC} yields 
\begin{equation}\label{Eq:ModulationSymbolDecisionErrorAveragedWithRespectToSymbols}
    \Pr\bigl\{\bigl.e\,\bigr|\,H\bigr\}=\sum_{m=1}^{M}p_{m}
        \sum^{M}_{\substack{{\widehat{m}}=1\\{\widehat{m}}\neq{m}}}\int_{\mathbb{D}^{\text{MAP}}_{\widehat{m}}}\!\!f_{\defrvec{R}_{c}|\defrvec{S}}(\defvec{r}|\defvec{s}_{m})d\defvec{r}.
\end{equation}
Accordingly, considering the whole transmission, we express the averaged \ac{SER} of the signaling as 
\begin{equation}
    \Pr\bigl\{e\bigr\}=\int_{0}^{\infty}\!\!\Pr\bigl\{\bigl.e\,\bigr|\,h\bigr\}f_{H}(h)dh,
\end{equation}
where $f_{H}(h)$ is the \ac{PDF} of the channel fading the signaling is subjected to. In this context, we mention that, in many cases, having exact information about a priori probabilities of the modulation symbols is difficult and actually impossible. We thus assume $p_{m}\!=\!1/M$ for all ${1}\!\leq\!{m}\!\leq\!{M}$ and then use the \ac{ML} decision rule at the receiver. Accordingly, we simplify \eqref{Eq:ModulationSymbolDecisionErrorAveragedWithRespectToSymbols} more to 
\begin{equation}\label{Eq:ModulationSymbolDecisionErrorAveragedWithRespectToSymbolsII}
    \Pr\bigl\{\bigl.e\,\bigr|\,H\bigr\}=\frac{1}{M}\sum_{m=1}^{M}
        \sum^{M}_{\substack{{\widehat{m}}=1\\{\widehat{m}}\neq{m}}}\int_{\mathbb{D}^{\text{ML}}_{\widehat{m}}}\!\!f_{\defrvec{R}_{c}|\defrvec{S}}
            (\defvec{r}|\defvec{s}_{m})d\defvec{r}.
\end{equation}
Note that for very few modulation constellations, all decision regions $\bigl\{\mathbb{D}^{\text{ML}}_{m},\allowbreak{1}\!\leq\!{m}\!\leq\!{M}\bigr\}$ are regular enough~to~be~defined mathematically such that we can compute the integrals in  \eqref{Eq:ModulationSymbolDecisionErrorAveragedWithRespectToSymbolsII} in closed forms. But, in cases where these integrals cannot be expressed in a closed form, it is useful to have a union upper bound for the \ac{SER} and hence for averaged \ac{SER} since being quite tight particularly at high \ac{SNR}. From \eqref{Eq:ModulationSymbolDecisionErrorAveragedWithRespectToSymbolsII}, we obtain 
the union upper bound for the averaged \ac{SER} over additive complex \ac{AWMN} channels as 
\begin{equation}\label{Eq:UnionUpperBoundSEPOverPrecodedComplexAWMNVectorChannels}
\!\!\!\!\!\!\Pr\bigl\{\bigl.e\,\bigr|\,H\bigr\}\leq\frac{1}{M}\sum_{m=1}^{M}
        \!\sum^{M}_{\substack{{\widehat{m}}=1\\{\widehat{m}}\neq{m}}}\!\Pr\bigl\{\bigl.\defvec{s}_{\widehat{m}}\,\text{detected}\,\bigr|\,\defvec{s}_{m}\,\text{sent}\bigr\},\!\!
\end{equation}
where $\Pr\bigl\{\bigl.\defvec{s}_{\widehat{m}}\,\text{detected}\,\bigr|\,\defvec{s}_{m}\,\text{sent}\bigr\}$, ${m}\neq{\widehat{m}}$ is the probability of the error as a result of detection of $\defvec{s}_{\widehat{m}}$ given the modulation symbol $\defvec{s}_{m}$ transmitted. Note that the boundary between $\mathbb{D}_{m}$ and $\mathbb{D}_{\widehat{m}}$ is perpendicular bisector of the line connecting $\defvec{s}_{m}$ and $\defvec{s}_{\widehat{m}}$, ${m}\neq{\widehat{m}}$. Accordingly, since $\defvec{s}_{m}$ is transmitted, a decision error occurs considering only $\defvec{s}_{m}$ and $\defvec{s}_{\widehat{m}}$, ${m}\neq{\widehat{m}}$ when the projection of $\defrvec{R}_{c}-H\defvec{s}_{m}$ on $H\defvec{s}_{\widehat{m}}-H\defvec{s}_{m}$ becomes larger than $Hd_{m\widehat{m}}/2$, where $d_{m\widehat{m}}$ is the Euclidean distance between $\defvec{s}_{m}$ and $\defvec{s}_{\widehat{m}}$, and defined by 
\begin{equation}\label{Eq:DistanceBetweenTwoModulationSymbols}
    d^2_{m\widehat{m}}={\bigl\lVert{\defvec{s}_{m}-\defvec{s}_{\widehat{m}}}\bigr\rVert}^2.
\end{equation}
As addressing $\defrvec{Z}_{c}\!=\!\defrvec{R}_{c}-H\defvec{s}_{m}$ and $\defrvec{Z}_{c}\!\sim\!\mathcal{CM}^{L}_{\nu}(\defvec{0},\frac{N_{0}}{2}\defmat{I})$, the probability of making an error when considering only $\defvec{s}_{m}$ and $\defvec{s}_{\widehat{m}}$, ${m}\neq{\widehat{m}}$ is given by
\begin{subequations}\label{Eq:UnionUpperBoundAWMNProjectionProbability}
\begin{eqnarray}
\nonumber
\!\!\Pr\bigl\{\bigl.\defvec{s}_{\widehat{m}}\!\!\!\!&\,&\!\!\!\!\text{detected}\bigr|\defvec{s}_{m}\,\text{sent}\bigr\}\\
    \label{Eq:UnionUpperBoundAWMNProjectionProbabilityA}
        &=&\Pr\biggl\{
            \frac{\RealPart{\defrvec{Z}^H_{c}(H\defvec{s}_{\widehat{m}}-H\defvec{s}_{m})}}{Hd_{m\widehat{m}}}
                >
                \frac{Hd_{m\widehat{m}}}{2}\biggr\},{~~~~~~~~~}\\
    \label{Eq:UnionUpperBoundAWMNProjectionProbabilityB}
        &=&\Pr\biggl\{
                \RealPart{\defrvec{Z}^H_{c}(\defvec{s}_{\widehat{m}}-\defvec{s}_{m})}
                >
                \frac{Hd^2_{m\widehat{m}}}{2}\biggr\},\\
    \label{Eq:UnionUpperBoundAWMNProjectionProbabilityC}
        &=&\Pr\biggl\{
                {N}
                >
                \frac{Hd^2_{m\widehat{m}}}{2}\biggr\},
\end{eqnarray}
\end{subequations}
where ${N}\!\sim\!\mathcal{M}_{\nu}(0,\frac{N_{0}}{2}d^2_{m\widehat{m}})$ as a result from the \ac{CS} property of $\defrvec{Z}_{c}\!\sim\!\mathcal{CM}^{L}_{\nu}(\defvec{0},\frac{N_{0}}{2}\defmat{I})$. 

\begin{theorem}\label{Theorem:UnionUpperBoundSEPForPrecodedComplexAWMNVectorChannel}
The union upper bound of the conditional \ac{SER} of a modulation constellation $\bigl\{\defvec{s}_{m},\allowbreak{1}\!\leq\!{m}\!\leq\!{M}\bigr\}$ is given by
\begin{equation}\label{Eq:UnionUpperBoundSEPForPrecodedComplexAWMNVectorChannel}
\!\!\Pr\bigl\{\bigl.e\,\bigr|\,H\bigr\}\leq\frac{1}{M}\sum_{m=1}^{M}
        \sum^{M}_{\substack{{\widehat{m}}=1\\{\widehat{m}}\neq{m}}}
            Q_{\nu}\biggl(
                \frac{H{\lVert{\defvec{s}_{m}-\defvec{s}_{\widehat{m}}}\rVert}}{\sqrt{2N_{0}}}
            \biggr),
\end{equation}
where $Q_{\nu}(\cdot)$ is the McLeish's \ac{Q-function} defined in \eqref{Eq:McLeishQFunction}.
\end{theorem}

\begin{proof}
From \eqref{Eq:UnionUpperBoundAWMNProjectionProbabilityC}, with the aid of \theoremref{Theorem:McLeishCCDF}, we have 
\begin{equation}\label{Eq:UnionUpperBoundAWMNProjectionProbabilityII}
    \Pr\biggl\{{N}>\frac{Hd^2_{m\widehat{m}}}{2}\biggr\}=
        Q_{\nu}\biggl(\frac{Hd_{m\widehat{m}}}{\sqrt{2N_{0}}}\biggr).
\end{equation}
Eventually, substituting both \eqref{Eq:UnionUpperBoundAWMNProjectionProbability} and \eqref{Eq:UnionUpperBoundAWMNProjectionProbabilityII} into \eqref{Eq:UnionUpperBoundSEPOverPrecodedComplexAWMNVectorChannels} yields
\begin{equation}\label{Eq:PrecodedComplexAWMNUnionUpperBoundSEP}
\!\!\Pr\bigl\{\bigl.e\,\bigr|\,H\bigr\}\leq\frac{1}{M}\sum_{m=1}^{M}
        \sum^{M}_{\substack{{\widehat{m}}=1\\{\widehat{m}}\neq{m}}}
        Q_{\nu}\biggl(\frac{Hd_{m\widehat{m}}}{\sqrt{2N_{0}}}\biggr),
\end{equation}
where inserting \eqref{Eq:DistanceBetweenTwoModulationSymbols} results in \eqref{Eq:UnionUpperBoundSEPForPrecodedComplexAWMNVectorChannel}, which completes the proof of \theoremref{Theorem:UnionUpperBoundSEPForPrecodedComplexAWMNVectorChannel}.
\end{proof}

It is worth noting that \theoremref{Theorem:UnionUpperBoundSEPForPrecodedComplexAWMNVectorChannel} proposes the general union bound expression for the conditional \ac{SER} of modulation constellation over uncorrelated complex \ac{AWMN} vector channels. Let us consider the accuracy and completeness of \theoremref{Theorem:UnionUpperBoundSEPForPrecodedComplexAWMNVectorChannel}, setting $\nu\!\rightarrow\!\infty$ in 
\eqref{Eq:UnionUpperBoundSEPForPrecodedComplexAWMNVectorChannel} yields \cite[Eq. (4.2-72)]{BibProakisBook}
\begin{equation}\label{Eq:UnionUpperBoundSEPForPrecodedComplexAWGNVectorChannel}
\!\!\Pr\bigl\{\bigl.e\,\bigr|\,H\bigr\}\leq\frac{1}{M}\sum_{m=1}^{M}
        \sum^{M}_{\substack{{\widehat{m}}=1\\{\widehat{m}}\neq{m}}}
            Q\biggl(
                \frac{H{\lVert{\defvec{s}_{m}-\defvec{s}_{\widehat{m}}}\rVert}}{\sqrt{2N_{0}}}
            \biggr),
\end{equation}
which is as expected the union upper bound of the conditional \ac{SER} for signaling over complex \ac{AWGN} channels. Further, for $\nu\!=\!1$, \eqref{Eq:UnionUpperBoundSEPForPrecodedComplexAWMNVectorChannel} simplifies to the union upper bound for complex \ac{AWLN} channels, that is 
\begin{equation}\label{Eq:UnionUpperBoundSEPForPrecodedComplexAWLNVectorChannel}
\!\!\Pr\bigl\{\bigl.e\,\bigr|\,H\bigr\}\leq\frac{1}{M}\sum_{m=1}^{M}
        \sum^{M}_{\substack{{\widehat{m}}=1\\{\widehat{m}}\neq{m}}}
            LQ\biggl(
                \frac{H{\lVert{\defvec{s}_{m}-\defvec{s}_{\widehat{m}}}\rVert}}{\sqrt{2N_{0}}}
            \biggr),
\end{equation}
where $LQ(\cdot)$ is the Laplacian \ac{Q-function} defined by \eqref{Eq:LaplacianQFunction}. In addition, if we know the distance structure of the modulation constellation, we can further simplify \eqref{Eq:UnionUpperBoundSEPForPrecodedComplexAWMNVectorChannel} by exploiting the fact that the decision error is mostly contributed by the closest modulation symbols. The distance between the two closest modulation symbols is given by
\begin{equation}
    d_{min}=\min_{m\neq\widehat{m}}
        {\lVert{\defvec{s}_{m}-\defvec{s}_{\widehat{m}}}\rVert}    
\end{equation}
Accordingly, we have 
\begin{equation}
    {Q}_{\nu}\biggl(\frac{Hd_{m\widehat{m}}}{\sqrt{2N_{0}}}\biggr)
        \leq
        {Q}_{\nu}\biggl(\frac{Hd_{min}}{\sqrt{2N_{0}}}\biggr),
\end{equation}
for all $\widehat{m}\!\neq\!{m}$. Therefore, putting this result in \eqref{Eq:PrecodedComplexAWMNUnionUpperBoundSEP} yields 
\begin{equation}\label{Eq:UnionUpperBoundSEPForPrecodedComplexAWMNVectorChannelII}
\!\!\Pr\bigl\{\bigl.e\,\bigr|\,H\bigr\}\leq
    (M-1)\,Q_{\nu}\biggl(\frac{Hd_{min}}{\sqrt{2N_{0}}}\biggr).
\end{equation}

In the following, we consider the well-known modulation constellations such as \ac{BPSK}, \ac{BFSK}, \ac{M-ASK}, \ac{M-PSK}, and \ac{M-QAM}, each of which is mainly characterized by their low bandwidth requirements. Appropriately, we will obtain the conditional \ac{SER} of the coherent optimal detector for these modulation constellations. 

\paragraph{Conditional \ac{BER} of Binary Keying Modulation}
\label{Section:SignallingOverAWMNChannels:CoherentSignalling:SymbolErrorProbability:BinaryKeyingModulation}
When binary signaling is used, let us denote the modulation constellation by $\{\defvec{s}_{+},\defvec{s}_{-}\}$
such that the transmitter transmits $\defvec{s}_{+}$ and $\defvec{s}_{-}$ with priori probabilities $p$ and $1-p$, respectively, and with powers ${E}_{+}\!=\!\lVert\defvec{s}_{+}\rVert^2$ and ${E}_{-}\!=\!\lVert\defvec{s}_{-}\rVert^2$, respectively. Referring to the mathematical model given by \eqref{Eq:PrecodedComplexAWMNVectorChannel}, the received vector $\defrvec{R}_{c}$ is readily written as
\begin{equation}\label{Eq:AWMNBinarySignallingReceivedVector}
    \defrvec{R}_{c}=H\defvec{s}_{\pm}+\defrvec{Z}_{c}
\end{equation}
where $\defrvec{Z}_{c}\!\sim\!\mathcal{M}^{L}_{\nu}(\defvec{0},\frac{N_{0}}{2}\defmat{I})$ and $\defrvec{R}_{c}\!\sim\!\mathcal{M}^{L}_{\nu}(H\defvec{s}_{\pm},\frac{N_{0}}{2}\defmat{I})$ since both the fading envelope $H$ and the modulation symbols $\defvec{s}_{\pm}$ are invariably known during one symbol duration. It is worth re-emphasizing that the received vector $\defrvec{R}_{c}$ depends on the transmitted binary symbol $\defrvec{S}$ through the conditional \ac{PDF} $f_{\defrvec{R}_{c}|\defrvec{S}}(\defvec{r}|\defvec{s})$, which is obtained in \eqref{Eq:PrecodedComplexAWMNVectorConditionalPDF}. Accordingly, utilizing \theoremref{Theorem:MAPDecisionRuleForPrecodedComplexAWMNVectorChannel}, we establish the \ac{MAP} decision rule in the following theorem. 

\begin{theorem}\label{Theorem:MAPDecisionRuleForBinaryCoherentSignallingOverAWMNChannels}
In the case of that coherent binary signaling is used, the \ac{MAP} decision rule given in~\emph{\theoremref{Theorem:MLDecisionRuleForPrecodedComplexAWMNVectorChannel}}~reduces~to
\begin{equation}\label{Eq:MAPDecisionRuleForBinaryCoherentSignallingOverAWMNChannels}
\!\!\text{Decide~$\defvec{s}_{\pm}$~iff~~} 
	{\lVert\defvec{R}_{c}-H\defvec{s}_{\pm}\rVert}^{2}+{\eta}_{\pm}
	\leq
	{\lVert\defvec{R}_{c}-H\defvec{s}_{\mp}\rVert}^{2},\!\!
\end{equation}
with the decision regions $\mathbb{D}^{\emph{\text{MAP}}}_{+}$ and $\mathbb{D}^{\emph{\text{MAP}}}_{-}$, given by 
\begin{equation}\label{Eq:MAPDecisionRegionForBinaryCoherentSignallingOverAWMNChannels}
\mathbb{D}^{\emph{\text{MAP}}}_{\pm}\!=\!
    \Bigl\{\defvec{r}\!\in\!\mathbb{C}^{L}
		\Bigl|
			{\lVert\defvec{r}-H\defvec{s}_{\pm}\rVert}^{2}+R_{\pm}
			\leq
			{\lVert\defvec{r}-H\defvec{s}_{\mp}\rVert}^{2}
		\Bigr.
	\Bigr\},
\end{equation}
where the threshold value, originated from the priori probabilities of modulation symbols, is given by
\begin{equation}
    {\eta}_{\pm}=N_{0}\log\Bigl(\frac{1\mp1\pm{2p}}{1\pm1\mp{2p}}\Bigr).
\end{equation}
\end{theorem}

\begin{proof}
The proof is obvious utilizing \theoremref{Theorem:MAPDecisionRuleForPrecodedComplexAWMNVectorChannel} with the \ac{CS} property of multivariate \ac{CCS} McLeish distribution (for more details, see \secref{Section:StatisticalBackground:MultivariateComplexMcLeishDistribution}). 
\end{proof}

In accordance with \theoremref{Theorem:MAPDecisionRuleForBinaryCoherentSignallingOverAWMNChannels}, the decision regions $\mathbb{D}^{\text{MAP}}_{+}$ and $\mathbb{D}^{\text{MAP}}_{-}$ are separated by a boundary hyperline perpendicular to the hyperline connecting $H\defvec{s}_{+}$ and $H\defvec{s}_{-}$. The decision regions and this boundary line are together illustrated in
\figref{Figure:DecisionRegionsForBinaryEquiprobableSignals}. Let us assume that $\defvec{s}_{+}$ is transmitted, then an error occurs when the received vector $\defrvec{R}_{c}$ falls into $\mathbb{D}_{-}$ instead of $\mathbb{D}_{+}$, which means that the projection of $(\defrvec{R}_{c}-H\defvec{s}_{+})$ on $(H\defvec{s}_{+}-H\defvec{s}_{-})$ is larger than the distance of $H\defvec{s}_{+}$ from the boundary hyperline. 

\begin{figure}[tp]
	\centering
	\includegraphics[clip=true, trim=0mm 0mm 0mm 0mm,width=0.7\columnwidth,keepaspectratio=true]{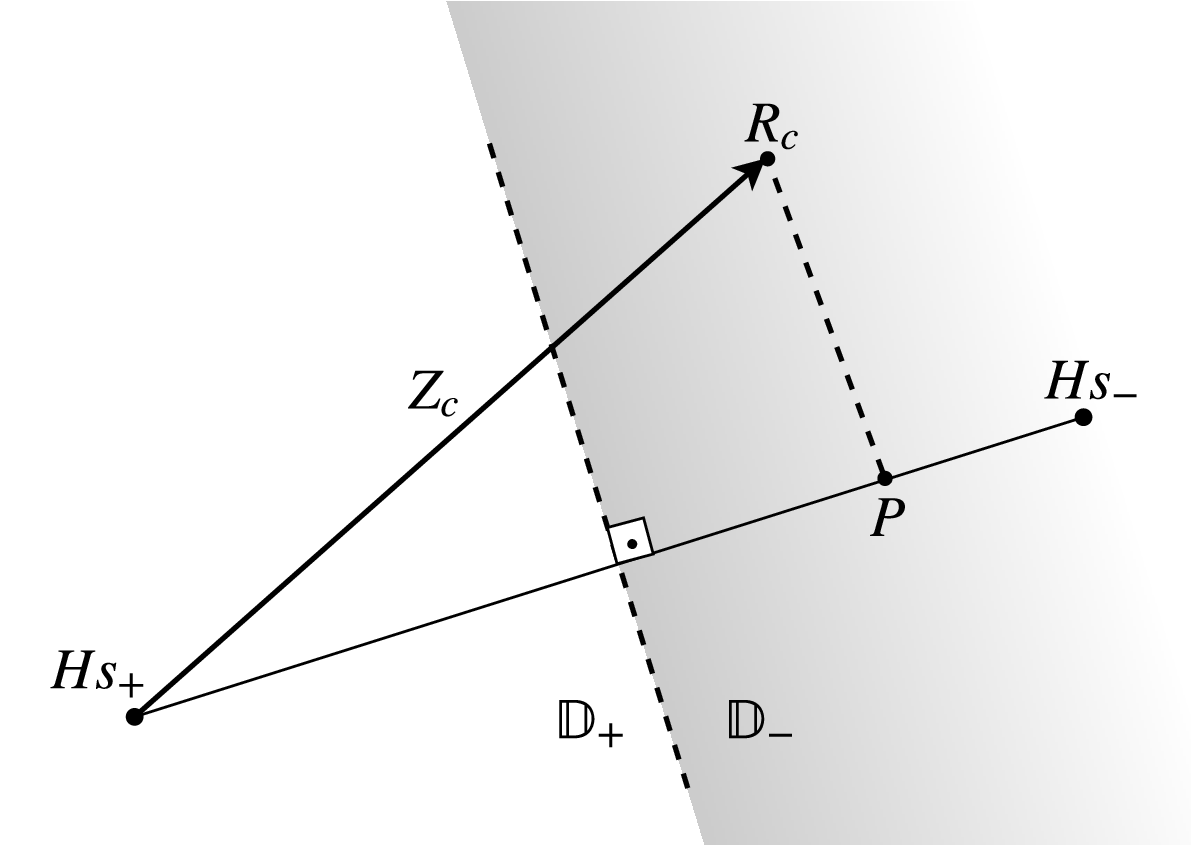}
	\caption{Received vector representation using binary keying symbols $\defvec{s}_{\pm}$ with the decision regions $\mathbb{D}_{\pm}$.}
	\label{Figure:DecisionRegionsForBinaryEquiprobableSignals}
\vspace{-2mm}
\end{figure}

\begin{theorem}\label{Theorem:MAPDecisionErrorProbabilityForBinaryCoherentSignallingOverAWMNChannels}
For the \ac{MAP} decision rule given by \theoremref{Theorem:MAPDecisionRuleForBinaryCoherentSignallingOverAWMNChannels}, the conditional \ac{BER} of binary signaling is given by 
\ifCLASSOPTIONtwocolumn
\begin{multline}\label{Eq:MAPDecisionErrorProbabilityForBinaryCoherentSignallingOverAWMNChannels}
 \Pr\{\bigl.e\,|\,H\}=
        p\,{Q}_{\nu}\biggl(
            \frac{H^{2}{\lVert\defvec{s}_{+}-\defvec{s}_{-}\rVert}^{2}-{\eta}_{+}}
                {H{\lVert\defvec{s}_{+}-\defvec{s}_{-}\rVert}\sqrt{2{N}_{0}}}
            \biggr)
        +\\    
        (1-p){Q}_{\nu}\biggl(
            \frac{H^{2}{\lVert\defvec{s}_{+}-\defvec{s}_{-}\rVert}^{2}-{\eta}_{-}}
                {H{\lVert\defvec{s}_{+}-\defvec{s}_{-}\rVert}\sqrt{2{N}_{0}}}
            \biggr).
\end{multline} 
\else
\begin{equation}\label{Eq:MAPDecisionErrorProbabilityForBinaryCoherentSignallingOverAWMNChannels}
 \Pr\{\bigl.e\,|\,H\}=
        p\,{Q}_{\nu}\biggl(
            \frac{H^{2}{\lVert\defvec{s}_{+}-\defvec{s}_{-}\rVert}^{2}-{\eta}_{+}}
                {H{\lVert\defvec{s}_{+}-\defvec{s}_{-}\rVert}\sqrt{2{N}_{0}}}
            \biggr)
        +    
        (1-p){Q}_{\nu}\biggl(
            \frac{H^{2}{\lVert\defvec{s}_{+}-\defvec{s}_{-}\rVert}^{2}-{\eta}_{-}}
                {H{\lVert\defvec{s}_{+}-\defvec{s}_{-}\rVert}\sqrt{2{N}_{0}}}
            \biggr).
\end{equation}
\fi
\end{theorem}

\begin{proof}
From \eqref{Eq:MAPDecisionRuleForBinaryCoherentSignallingOverAWMNChannels}, we can write the decision correct decision when assuming that $\defvec{s}_{\pm}$ is transmitted as follows 
\ifCLASSOPTIONtwocolumn
\begin{multline}
    {\lVert\defvec{R}_{c}\rVert}^2+
    H^{2}{\lVert\defvec{s}_{\pm}\rVert}^{2}
    -2H\RealPart{\defvec{s}^{H}_{\pm}\defvec{R}_{c}}
    +
    {\eta}_{\pm}\leq\\
    {\lVert\defvec{R}_{c}\rVert}^2+
    H^{2}{\lVert\defvec{s}_{\mp}\rVert}^{2}
    -2H\RealPart{\defvec{s}^{H}_{\mp}\defvec{R}_{c}},
\end{multline}
\else
\begin{equation}
    {\lVert\defvec{R}_{c}\rVert}^2+
    H^{2}{\lVert\defvec{s}_{\pm}\rVert}^{2}
    -2H\RealPart{\defvec{s}^{H}_{\pm}\defvec{R}_{c}}
    +
    {\eta}_{\pm}\leq
    {\lVert\defvec{R}_{c}\rVert}^2+
    H^{2}{\lVert\defvec{s}_{\mp}\rVert}^{2}
    -2H\RealPart{\defvec{s}^{H}_{\mp}\defvec{R}_{c}},
\end{equation}
\fi
where inserting \eqref{Eq:AWMNBinarySignallingReceivedVector} yields 
\begin{equation}
    {D}\leq{H}^{2}{\lVert\defvec{s}_{\pm}-\defvec{s}_{\mp}\rVert}^{2}-{\eta}_{\pm},
\end{equation}
where the decision variable $D$ is given by 
\begin{equation}
    {D}=-2H\RealPart{(\defvec{s}_{\pm}-\defvec{s}_{\mp})^{H}\defrvec{Z}_{c}},
\end{equation}
where $(\defvec{s}_{\mp}-\defvec{s}_{\pm})^{H}\defrvec{Z}_{c}$ follows a \ac{CCS} McLeish distribution with zero mean and ${N_{0}}{\lVert\defvec{s}_{\pm}-\defvec{s}_{\mp}\rVert}^{2}/{2}$ variance per dimension. Therefore, $D\!\sim\!\mathcal{M}_{\nu}(0,2H^2N_{0}{\lVert\defvec{s}_{\pm}-\defvec{s}_{\mp}\rVert}^{2})$, and accordingly, a decision error occurs when ${D}\!>\!H^{2}{\lVert\defvec{s}_{\pm}-\defvec{s}_{\mp}\rVert}^{2}-{\eta}_{\pm}$. With the aid of \theoremref{Theorem:McLeishCCDF}, when $\defvec{s}_{\pm}$ is transmitted, we write the probability of decision error as
\begin{equation}\label{Eq:ErrorneousDecisionProbabilityForAWMNVectorChannels}
    \Pr\bigl\{\bigl.e\,\bigr|\,H,\defvec{s}_{\pm}\bigr\}=
        Q_{\nu}\biggl(
            \frac{H^{2}{\lVert\defvec{s}_{\pm}-\defvec{s}_{\mp}\rVert}^{2}-{\eta}_{\pm}}
                {H{\lVert\defvec{s}_{\pm}-\defvec{s}_{\mp}\rVert}\sqrt{2{N}_{0}}}
            \biggr),
\end{equation}
From \eqref{Eq:ModulationSymbolErrorProbabilityB}, we write $\Pr\{\bigl.e\,|\,H\}\!=\!
\Pr\{\bigl.e\,|\,H,\defvec{s}_{+}\}\Pr\{\defvec{s}_{+}\}+
\Pr\{\bigl.e\,|\,H,\defvec{s}_{-}\}\Pr\{\defvec{s}_{-}\}$, where replacing 
\eqref{Eq:ErrorneousDecisionProbabilityForAWMNVectorChannels} yields \eqref{Eq:MAPDecisionErrorProbabilityForBinaryCoherentSignallingOverAWMNChannels}, which completes the proof of \theoremref{Theorem:MAPDecisionErrorProbabilityForBinaryCoherentSignallingOverAWMNChannels}.
\end{proof}

In the special case where the binary modulation symbols are equiprobable (i.e., when $\Pr\{\defvec{s}_{\pm}\}\!=\!{1}/{2}$), we have the threshold value $\eta_{\pm}\!=\!0$ and then reduce the \ac{MAP} rule to the \ac{ML} rule given below. 

\begin{theorem}\label{Theorem:MLDecisionRuleForBinaryCoherentSignallingOverAWMNChannels}
In the case where coherent binary signaling is used, the \ac{ML} decision rule, given in \emph{\theoremref{Theorem:MLDecisionRuleForPrecodedComplexAWMNVectorChannel}}, reduces to
\begin{equation}\label{Eq:MLDecisionRuleForBinaryCoherentSignallingOverAWMNChannels}
	\text{Decide~$\defvec{s}_{\pm}$~iff~~} 
	\lVert\defvec{R}_{c}-H\defvec{s}_{\pm}\rVert\leq\lVert\defvec{R}_{c}-H\defvec{s}_{\mp}\rVert.
\end{equation}
with the decision regions $\mathbb{D}^{\emph{\text{ML}}}_{+}$ and $\mathbb{D}^{\emph{\text{ML}}}_{-}$, given by 
\begin{equation}\label{Eq:MLDecisionRegionForBinaryCoherentSignallingOverAWMNChannels}
\!\!\!\!\mathbb{D}^{\emph{\text{ML}}}_{\pm}=
    \Bigl\{\defvec{r}\in\mathbb{C}^{L}\,
		\Bigl|\,
			\lVert\defvec{R}_{c}-H\defvec{s}_{\pm}\rVert\leq\lVert\defvec{R}_{c}-H\defvec{s}_{\mp}\rVert
		\Bigr.
	\Bigr\}.\!\!
\end{equation}
\end{theorem}

\begin{proof}
The proof is obvious using \theoremref{Theorem:MAPDecisionRuleForBinaryCoherentSignallingOverAWMNChannels} by assuming that the symbols are equiprobable, i.e., $\Pr\{\defvec{s}_{\pm}\}\!=\!{1}/{2}$. 
\end{proof}

As it can be easily observed from \theoremref{Theorem:MAPDecisionRuleForBinaryCoherentSignallingOverAWMNChannels}, the decision regions $\mathbb{D}^{\text{ML}}_{+}$ and $\mathbb{D}^{\text{ML}}_{-}$ are separated by a perpendicular bisector to the hyperline connecting $H\defvec{s}_{+}$ and $H\defvec{s}_{-}$. As~a~result~of~the fact that the decision error probabilities when the modulation symbol $\defvec{s}_{+}$ or $\defvec{s}_{-}$ is transmitted are equal, we have a symmetry with respect to the perpendicular bisector (i.e., the minimum distance of $\defvec{s}_{+}$ and that of $\defvec{s}_{-}$ from the perpendicular bisector are certainly equal).  

\begin{theorem}\label{Theorem:MLDecisionErrorProbabilityForBinaryCoherentSignallingOverAWMNChannels}
For the \ac{ML} decision rule, given by \theoremref{Theorem:MLDecisionRuleForBinaryCoherentSignallingOverAWMNChannels}, the conditional \ac{BER} of binary signaling is given by 
\begin{equation}\label{Eq:MLDecisionErrorProbabilityForBinaryCoherentSignallingOverAWMNChannels}
 \Pr\bigl\{e\,\bigl|\,H\bigr.\bigr\}=
        {Q}_{\nu}\biggl(
            \frac{H{\lVert\defvec{s}_{+}-\defvec{s}_{-}\rVert}}
                {\sqrt{2{N}_{0}}}
            \biggr).
\end{equation} 
\end{theorem}

\begin{proof}
The proof is obvious setting $p\!=\!\frac{1}{2}$ in \theoremref{Theorem:MAPDecisionErrorProbabilityForBinaryCoherentSignallingOverAWMNChannels}.
\end{proof}

Let us consider the special cases of \theoremref{Theorem:MLDecisionErrorProbabilityForBinaryCoherentSignallingOverAWMNChannels} for certain binary modulation constellations. When the binary modulation symbols 
$\defvec{s}_{+}$ and $\defvec{s}_{-}$ are equiprobable (i.e., $\Pr\{\defvec{s}_{\pm}\}\!=\!{1}/{2}$) and have equal power (i.e., ${\lVert\defvec{s}_{+}\rVert}^2={\lVert\defvec{s}_{-}\rVert}^2$), we can rewrite the distance between $\defvec{s}_{+}$ and $\defvec{s}_{-}$ as
\begin{equation}\label{Eq:DistanceForEqualEnergyBinarySignalling}
\lVert\defvec{s}_{+}-\defvec{s}_{-}\rVert=\sqrt{2E_{\defrvec{S}}(1-\rho)},
\end{equation}
where $E_{\defrvec{S}}\!=\!\mathbb{E}[\defrvec{S}^H\defrvec{S}]$ denotes the transmitted average power and can be written in more details as follows
\begin{subequations}\label{Eq:AverageSignalPowerForBinarySignalling}
\begin{eqnarray}
    \label{Eq:AverageSignalPowerForBinarySignallingA}
    E_{\defrvec{S}}&=&\Pr\{\defvec{s}_{+}\}\lVert\defvec{s}_{+}\rVert^2+\Pr\{\defvec{s}_{-}\}\lVert\defvec{s}_{-}\rVert^2,\\
    \label{Eq:AverageSignalPowerForBinarySignallingB}    
    &=&{\frac{1}{2}}\lVert\defvec{s}_{+}\rVert^2+\frac{1}{2}\lVert\defvec{s}_{-}\rVert^2,\\
    \label{Eq:AverageSignalPowerForBinarySignallingC}    
    &=&E_{+}\text{~(or $E_{-}$)},
\end{eqnarray}
\end{subequations}
Further, in \eqref{Eq:DistanceForEqualEnergyBinarySignalling}, $\rho$ denotes the cross-correlation coefficient between the modulation symbols $\defvec{s}_{+}$ and $\defvec{s}_{-}$, defined by
\begin{subequations}\label{Eq:BinarySignallingCorrelation}
\begin{eqnarray}
    \label{Eq:BinarySignallingCorrelationA}
    \rho
    &=&
        \frac{\RealPart{\defvec{s}^{H}_{+}\defvec{s}_{-}}}
            {\lVert\defvec{s}_{+}\rVert\lVert\defvec{s}_{-}\rVert},\\
    \label{Eq:BinarySignallingCorrelationB}
    &=&\frac{1}{E_{\defrvec{S}}}
        \Bigl(
            \RealPart{\defvec{s}^T_{+}}\RealPart{\defvec{s}_{-}}+
            \ImagPart{\defvec{s}^T_{+}}\ImagPart{\defvec{s}_{-}}
        \Bigr).
\end{eqnarray}
\end{subequations}
\begin{figure}[tp] 
\centering
\begin{subfigure}{0.7\columnwidth}
    \centering
    \includegraphics[clip=true, trim=0mm 0mm 0mm 0mm, width=1.0\columnwidth,height=0.85\columnwidth]{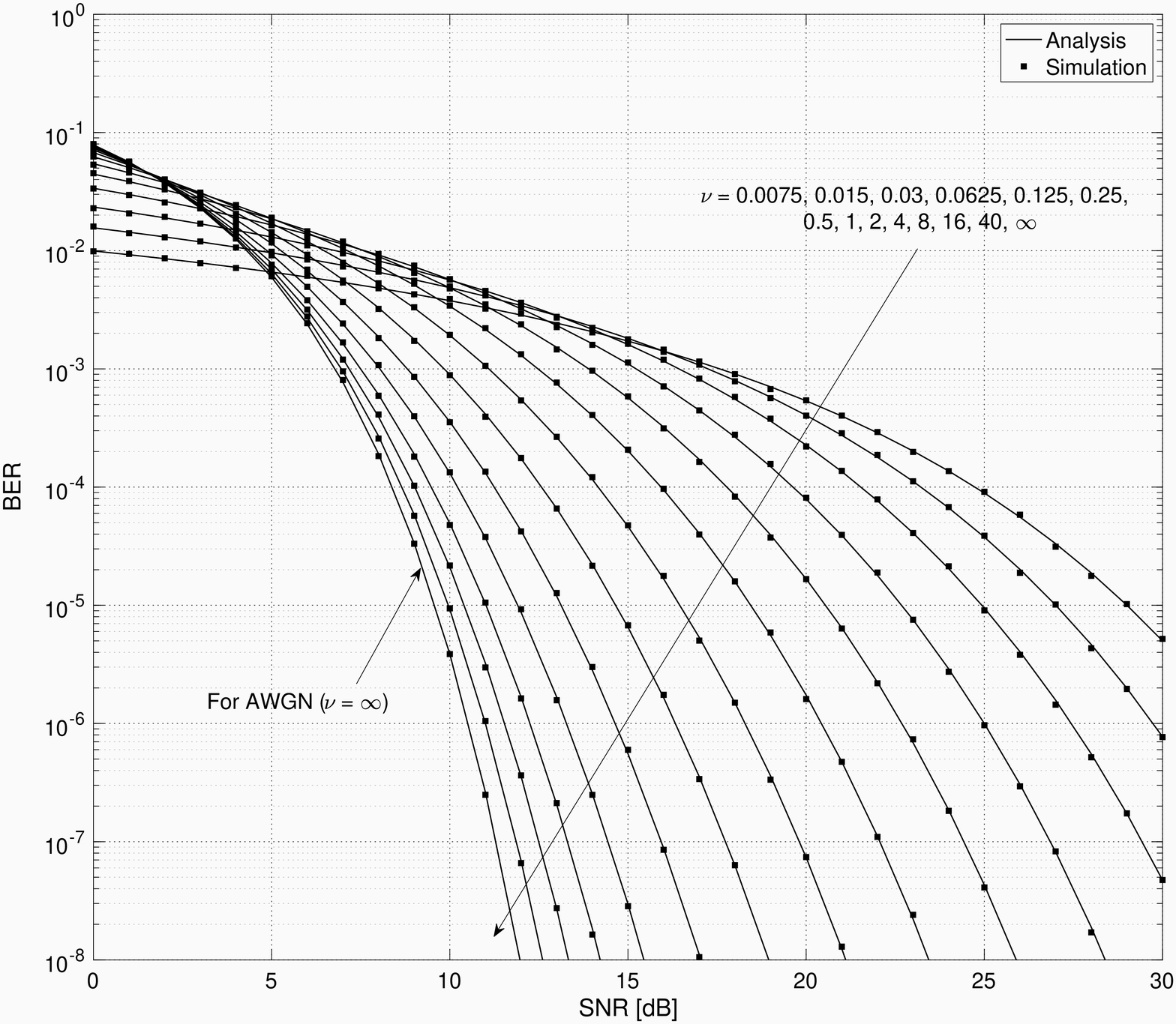}
    \caption{With respect to \ac{SNR}.}
    \vspace{5mm} 
    \label{Figure:ConditionalBEPForBPSKA}
\end{subfigure}
{~~~} 
\begin{subfigure}{0.7\columnwidth}
    \centering
    \includegraphics[clip=true, trim=0mm 0mm 0mm 0mm, width=1.0\columnwidth,height=0.85\columnwidth]{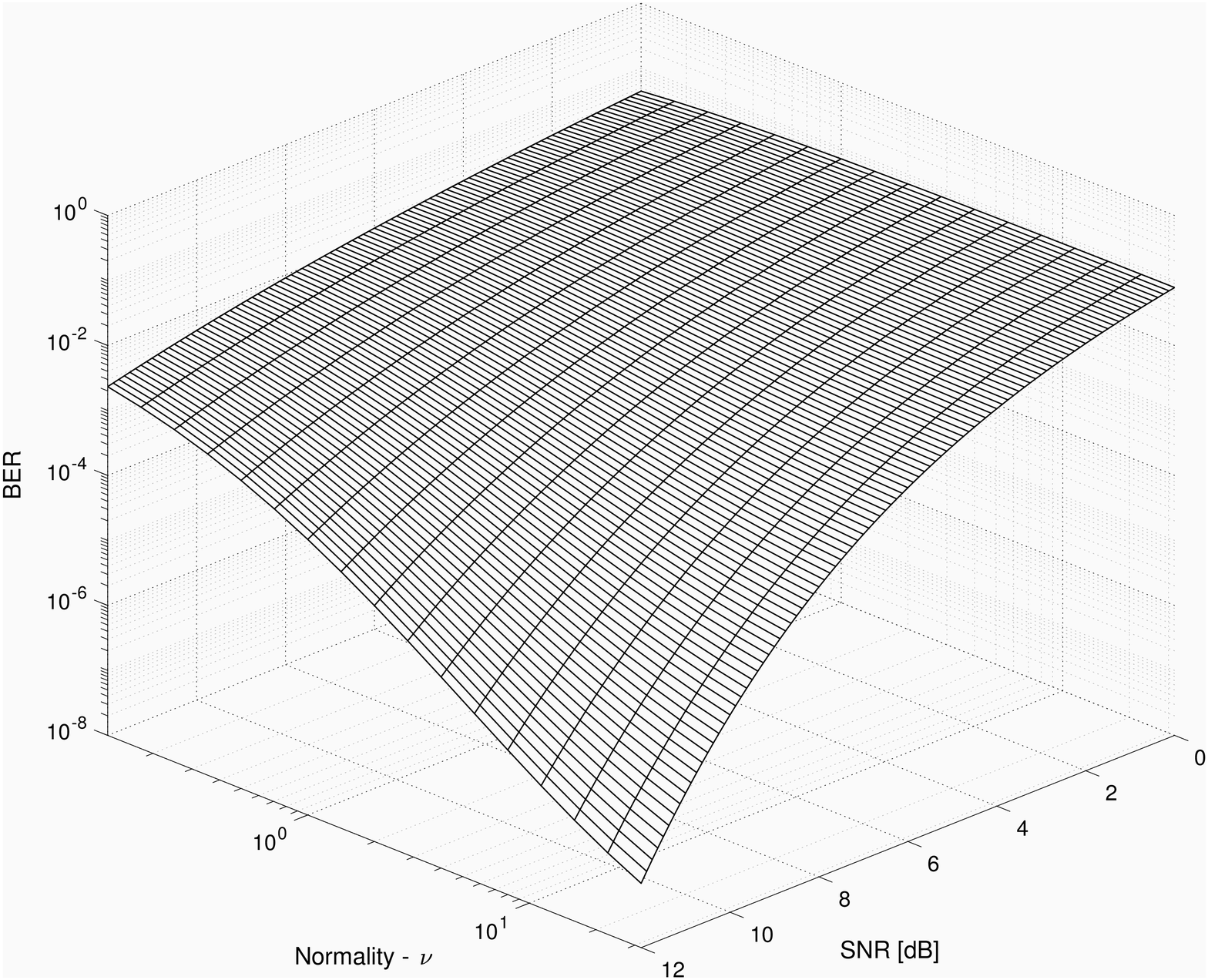}
    \caption{With respect to \ac{SNR} and normality.}
    \label{Figure:ConditionalBEPForBPSKB}
\end{subfigure}
\caption{The \ac{BER} of \ac{BPSK} signaling over \ac{AWMN} channels.}
\label{Figure:ConditionalBEPForBPSK}
\vspace{-2mm}
\end{figure}
\begin{figure}[tp] 
\centering
\begin{subfigure}{0.7\columnwidth}
    \centering
    \includegraphics[clip=true, trim=0mm 0mm 0mm 0mm, width=1.0\columnwidth,height=0.85\columnwidth]{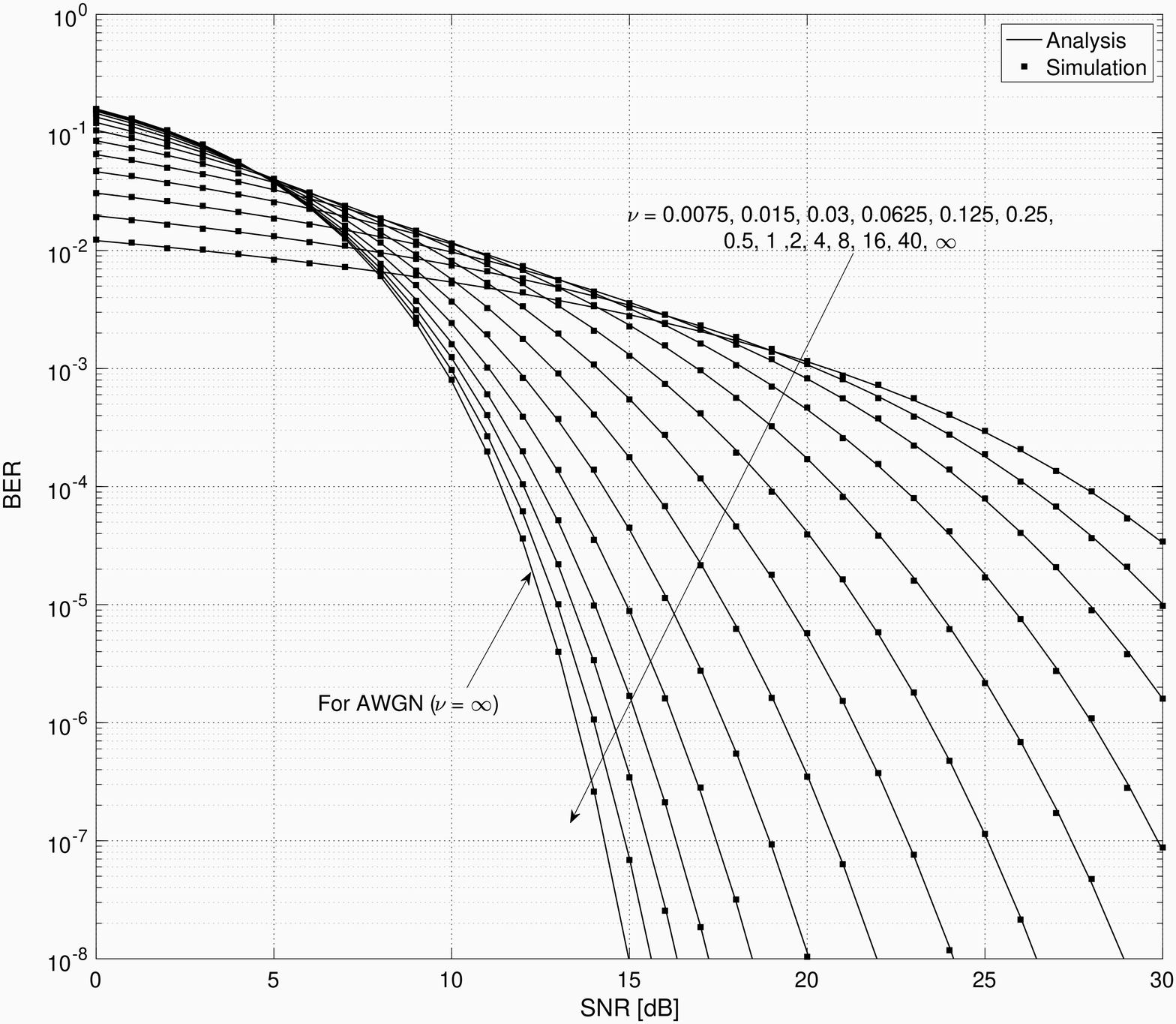}
    \caption{With respect to \ac{SNR}.}
    \vspace{5mm}
    \label{Figure:ConditionalBEPForBFSKA}
\end{subfigure}
{~~~}
\begin{subfigure}{0.7\columnwidth}
    \centering
    \includegraphics[clip=true, trim=0mm 0mm 0mm 0mm, width=1.0\columnwidth,height=0.85\columnwidth]{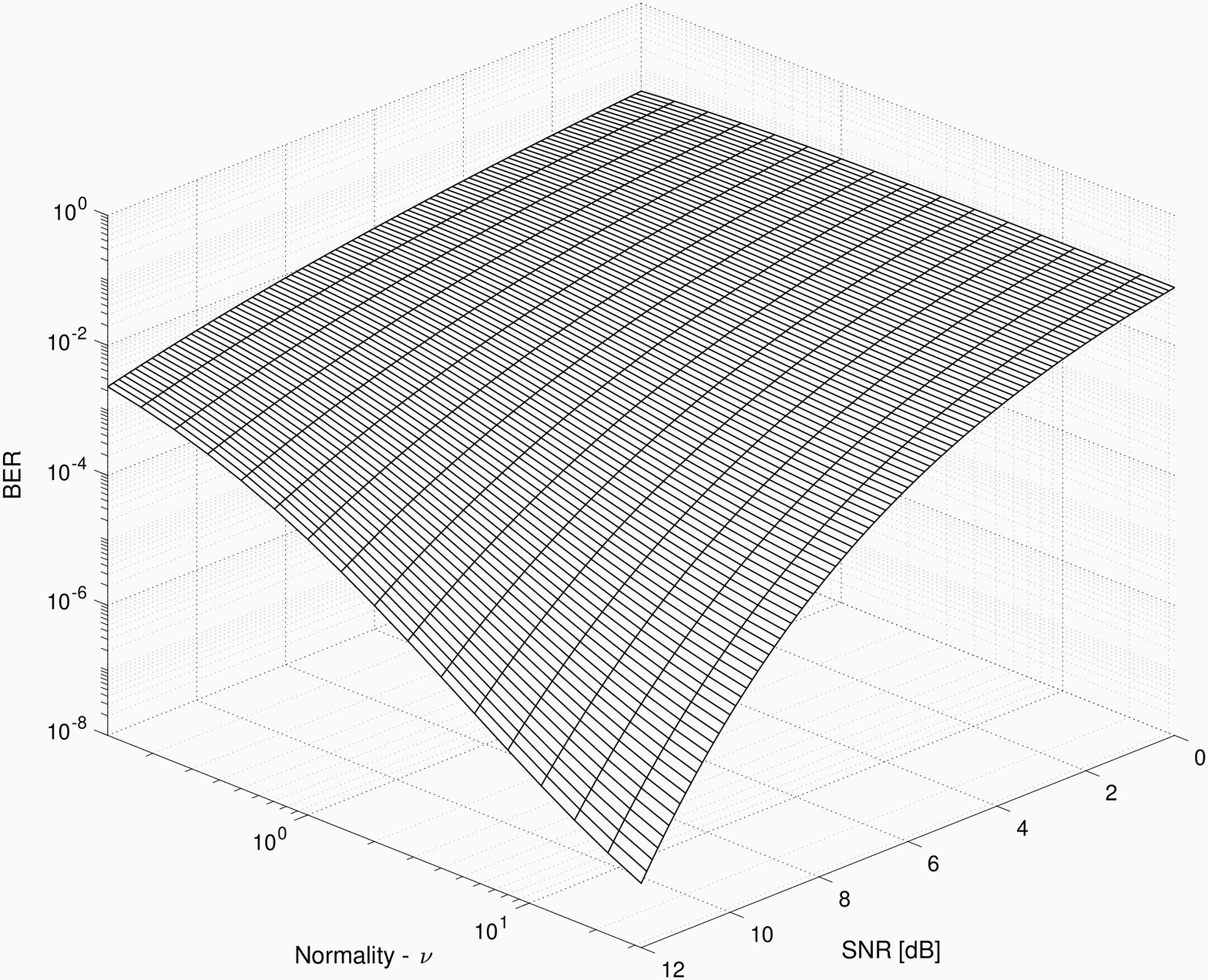}
    \caption{With respect to \ac{SNR} and normality.}
    \label{Figure:ConditionalBEPForBFSKB}
\end{subfigure}
\caption{The \ac{BER} of \ac{BFSK} signaling over \ac{AWMN} channels.}
\label{Figure:ConditionalBEPForBFSK}
\vspace{-2mm} 
\end{figure}
It is consequently valuable to notice that, since ${-1}\leq\!\rho\!\leq{1}$, \eqref{Eq:DistanceForEqualEnergyBinarySignalling} is maximally increased when $\rho\!=\!{-1}$, i.e., when the the binary modulation symbols are antipodal (i.e., when $\defvec{s}_{\pm}\!=\!\mp\defvec{s}_{\mp}$). Consequently, substituting \eqref{Eq:DistanceForEqualEnergyBinarySignalling} into \eqref{Eq:MLDecisionErrorProbabilityForBinaryCoherentSignallingOverAWMNChannels} results in 
\begin{equation}\label{Eq:ConditionalBEPForCoherentBinarySignalingUsingCorrelation}
	\Pr\bigl\{e\,\bigl|\,H\bigr.\bigr\}=Q_{\nu}\Bigl(\sqrt{(1-\rho)\gamma}\Bigr),
\end{equation}
where $\gamma$ is the instantaneous \ac{SNR} during transmission of one modulation symbol and defined by 
\begin{subequations}\label{Eq:SignalToNoiseDefinition}
\begin{eqnarray}
    \label{Eq:SignalToNoiseDefinitionA}
    \gamma
    &=&\frac{\mathbb{E}[\langle{{H}\defrvec{S},\defrvec{R}_{c}}\rangle]^2}
            {\mathrm{Var}[
                \langle{{H}\defrvec{S},\defrvec{R}_{c}}\rangle
            ]},\\
    \label{Eq:SignalToNoiseDefinitionB}        
    &=&\frac{\mathbb{E}[\langle{{H}\defrvec{S},\defrvec{R}_{c}}\rangle]^2}
            {
            \mathbb{E}[
                {\langle{{H}\defrvec{S},\defrvec{R}_{c}}\rangle}^2
            ]
            -
            \mathbb{E}[
                {\langle{{H}\defrvec{S},\defrvec{R}_{c}}\rangle}
            ]^2
            },\\
    \label{Eq:SignalToNoiseDefinitionC}
    &=&{H}^2\frac{E_S}{N_0},
\end{eqnarray}
\end{subequations}
with the aid of the optimal decision rules given above. 

\begin{theorem}\label{Theorem:ConditionalBEPForBPSKSignaling}
The contional \ac{BER} $\Pr\bigl\{e\,\bigl|\,H\bigr.\bigr\}$ of \ac{BPSK} signaling over \ac{CCS} \ac{AWMN} channels is given by
\begin{equation}\label{Eq:ConditionalBEPForBPSK}
	\Pr\bigl\{e\,\bigl|\,H\bigr.\bigr\}=Q_{\nu}\bigl(\sqrt{2\gamma}\bigr),
\end{equation}
where $\gamma$ is the instantaneous \ac{SNR} defined above. 
\end{theorem}

\begin{proof}
Note that the \ac{BPSK} symbols are defined by $\{\defvec{s}_{+},\defvec{s}_{-}\}$ such that $\defvec{s}_{\pm}\!=\!-\defvec{s}_{\mp}$, which means that $\defvec{s}_{+}$ and $\defvec{s}_{+}$ have equal power. In case of that they are equiprobable, we have 
$\lVert\defvec{s}_{\pm}\rVert^2\!=\!E_\mathcal{S}$. Therefore, with the aid of \eqref{Eq:BinarySignallingCorrelationA},  $\rho\!=\!{-1}$, and then
\eqref{Eq:ConditionalBEPForCoherentBinarySignalingUsingCorrelation} simplifies to 
\eqref{Eq:ConditionalBEPForBPSK}, which proves \theoremref{Theorem:ConditionalBEPForBPSKSignaling}.
\end{proof}

\begin{theorem}\label{Theorem:ConditionalBEPForBFSKSignaling}
The contional \ac{BER} $\Pr\bigl\{e\,\bigl|\,H\bigr.\bigr\}$ of \ac{BFSK} signaling over \ac{CCS} \ac{AWMN} channels is given by
\begin{equation}\label{Eq:ConditionalBEPForBFSKSignaling}
	\Pr\bigl\{e\,\bigl|\,H\bigr.\bigr\}=Q_{\nu}\Bigl(\sqrt{\gamma}\Bigr),
\end{equation}
where $\gamma$ is the instantaneous \ac{SNR} defined above.
\end{theorem}

\begin{proof}
Note that the \ac{BFSK} symbols are defined by $\{\defvec{s}_{+},\defvec{s}_{-}\}$ such that $\defvec{s}^H_{\pm}\defvec{s}_{\mp}\!=\!0$. In case where $\defvec{s}_{+}$ and $\defvec{s}_{+}$ are equiprobable and have equal power, we obtain the correlation $\rho\!=\!{0}$ with the aid of \eqref{Eq:BinarySignallingCorrelationA}, and accordingly, we reduce \eqref{Eq:ConditionalBEPForCoherentBinarySignalingUsingCorrelation} into \eqref{Eq:ConditionalBEPForBFSKSignaling}, which proves \theoremref{Theorem:ConditionalBEPForBFSKSignaling}.
\end{proof}

As mentioned before, the impulsive nature of McLeish noise distribution is simply expressed by its normality $\nu\!\in\!\mathbb{R}_{+}$. As such, when $\nu\!\rightarrow\!\infty$, the impulsive nature vanishes and McLeish noise distribution approaches to Gaussian noise distribution. For that purpose, we demonstrated the effect of non-Gaussian noise on communication performance by plotting in \figref{Figure:ConditionalBEPForBPSK} and \figref{Figure:ConditionalBEPForBFSK} the conditional \ac{BER} of \ac{BPSK} and \ac{BFSK}  modulations, respectively, with respect to different normalities $\nu\!\in\!\{0.0075,\allowbreak{0.015},\allowbreak{0.03},\allowbreak{0.0625},\allowbreak{0.125},\allowbreak{0.25},\allowbreak{0.5},\allowbreak{1},2,4,8,16,40,\infty\}$. We evidently observe that the impulsive nature of McLeish noise distribution deteriorates the performance of binary modulations in high-\ac{SNR} regime while negligibly improves it in low-\ac{SNR} regime. 

The other binary keying signaling is the \ac{OOK} modulation, in case of which the binary information is transmitted by the presence or absence of a modulation symbol. Accordingly, the modulation symbols $\defvec{s}_{+}\!\neq\!\defvec{0}$ and $\defvec{s}_{-}\!=\!\defvec{0}$ are employed to transmit $1$ and $0$ binary information, respectively, with equal a priori probabilities $\Pr\{\defvec{s}_{+}\}\!=\!\Pr\{\defvec{s}_{-}\}\!=\!1/2$. At this point, note that the \ac{OOK} constellation can be achieved by shifting the \ac{BPSK}\,/\,\ac{BFSK} constellation up to $\defvec{s}_{-}\!=\!0$. Accordingly, 
$E_{+}\!=\!\lVert\defvec{s}_{+}\rVert^2\!\neq\!0$ and $E_{-}\!=\!\lVert\defvec{s}_{-}\rVert^2\!=\!0$, such that 
the average power of the \ac{OOK} modulation is written as $E_{\defrvec{S}}\!=\!\Pr\{\defvec{s}_{+}\}\lVert\defvec{s}_{+}\rVert^2\!+\!\Pr\{\defvec{s}_{-}\}\lVert\defvec{s}_{-}\rVert^2\!=\!\frac{1}{2}\lVert\defvec{s}_{+}\rVert^2$. With that result, we can rewrite the distance between $\defvec{s}_{+}$ and $\defvec{s}_{-}$ for the \ac{OOK} modulation as 
\begin{equation}\label{Eq:DistanceForOOKSignalling}
    \lVert\defvec{s}_{+}-\defvec{s}_{-}\rVert=\lVert\defvec{s}_{+}\rVert=\sqrt{2E_{\defrvec{S}}},
\end{equation}
 
\begin{theorem}\label{Theorem:ConditionalBEPForOOKSignaling}
The contional \ac{BER} $\Pr\bigl\{e\,\bigl|\,H\bigr.\bigr\}$ of \ac{OOK} signaling over \ac{CCS} \ac{AWMN} channels is given by
\begin{equation}\label{Eq:ConditionalBEPForOOKSignaling}
	\Pr\bigl\{e\,\bigl|\,H\bigr.\bigr\}=Q_{\nu}\Bigl(\sqrt{\gamma}\Bigr),
\end{equation}
where $\gamma$ is the instantaneous \ac{SNR} defined in \eqref{Eq:SignalToNoiseDefinition}. 
\end{theorem}

\begin{proof}
The proof is obvious inserting \eqref{Eq:DistanceForOOKSignalling} into \theoremref{Theorem:MLDecisionErrorProbabilityForBinaryCoherentSignallingOverAWMNChannels} and using \eqref{Eq:SignalToNoiseDefinitionC}.
\end{proof}

At the moment, it has been investigated to obtain closed-form expressions for the conditional \ac{BER} performance of binary signaling over \ac{AWMN} channels. In the following, we consider the conditional \ac{SER} of M-ary signaling over \ac{CCS} \ac{AWMN} vector channels. 

\paragraph{Conditional \ac{SER} of \ac{M-ASK} Modulation}
\label{Section:SignallingOverAWMNChannels:CoherentSignalling:SymbolErrorProbability:MASKModulation}
Let $\defrvec{S}\!=\!\{\defvec{s}_{1},\defvec{s}_{2},\ldots,\defvec{s}_{M}\}$ denote the \ac{M-ASK} constellation such that its constellation center is zero (i.e., $\defvec{s}_{1}+\defvec{s}_{2}+\ldots+\defvec{s}_{M}=\defvec{0}$) and that 
$\defvec{s}^H_{m}\defvec{s}_{\widehat{m}}\!=\!\defvec{s}^H_{\widehat{m}}\defvec{s}_{m}$
for all ${m}\!\neq\!\widehat{m}$. Accordingly, the correlation between $\defvec{s}_{m}$ and $\defvec{s}_{\widehat{m}}$ for all ${m}\neq\widehat{m}$ is given by
\begin{equation}
 \rho_{m\widehat{m}}=\frac{\RealPart{\defvec{s}^H_{m}\defvec{s}_{\widehat{m}}}}
    {\lVert\defvec{s}_{m}\rVert\lVert\defvec{s}_{\widehat{m}}\rVert}=\pm{1},
\end{equation}
which consequence that, without loss of generality, the modulation symbols are ordered by $\lVert\defvec{s}_{\widehat{m}}-\defvec{s}_{1}\rVert\!<\!\lVert\defvec{s}_{m}-\defvec{s}_{1}\rVert$, ${m}\!<\!{\widehat{m}}$ on a hyperline. Therefore, the modulation symbol $m$ can be written as
\begin{equation}
   \defvec{s}_{m}={a}_{m}\defvec{s},\quad{1}\leq{m}\leq{M},
\end{equation}
where $\defvec{s}$ denotes an arbitrary unit vector, i.e., $\lVert\defvec{s}\rVert\!=\!1$, and thus ${a}_{m}$, ${1}\!\leq\!{m}\!\leq\!{M}$ are such real amplitudes that they support $\defvec{s}_{1}+\defvec{s}_{2}+\ldots+\defvec{s}_{M}=\defvec{0}$, which imposes that 
\begin{equation}
    {a}_{1}+{a}_{2}+\ldots+{a}_{M}={0}.
\end{equation}
From the condition that the modulation symbols are ordered, we have ${a}_{1}\!<\!{a}_{2}\!<\!\ldots\!<\!{a}_{M}$. For each $\defvec{s}_{m}$ except for the two outside ones $\defvec{s}_{1}$ and $\defvec{s}_{M}$, the distance of $\defvec{s}_{m}$ from $\defvec{s}_{m\pm{1}}$ is the constant we readily express
\begin{equation}
    \lVert{\defvec{s}_{m}-\defvec{s}_{m\pm{1}}}\rVert\!=\!({a}_{m}-{a}_{m\pm{1}})^2=\Delta~\text{(constant)}.
\end{equation}
Accordingly, we formulate the modulation symbols as
\begin{equation}
    \defvec{s}_{m}=(m-m_{0})\Delta\defvec{s},\quad{1}\leq{m}\leq{M}.
\end{equation}
which imposes that ${a}_{m}=(m-m_{0})\Delta$, ${1}\leq{m}\leq{M}$, where $\Delta$ is the minimum distance between modulation symbols, and the offset $m_{0}$ is found to be $m_{0}\!=\!{(M+1)}/{2}$ due to ${a}_{1}+{a}_{2}\allowbreak+\ldots+{a}_{M}\!=\!{0}$. Then, the power of $\defvec{s}_{m}$, which is written as ${E}_{m}\!=\!\lVert\defvec{s}_{m}\rVert^2$, can be obtained in terms of $\Delta$ as
\begin{equation}\label{Eq:ASKModulationMthSymbolEnergy}
        E_{m}\!=\!(m-m_0)^2\Delta^2
\end{equation}
Correspondingly, since the modulation~symbols~are~equiprobable, we write the average power of the  \ac{M-ASK} modulation~as $E_{\defrvec{S}}\!=\!(\sum_{m=1}^{M}E_{m})/{M}$, and therein substituting \eqref{Eq:ASKModulationMthSymbolEnergy}, we have
\begin{equation}
    E_{\defrvec{S}}=\frac{1}{12}(M^2-1)\Delta^2,   
\end{equation}
from which the value of $\Delta$ can be determined as 
\begin{equation}\label{Eq:ASKMOdulationMinimumDistance}
    \Delta=\sqrt{\frac{12E_{\defrvec{S}}}{M^2-1}}.
\end{equation}
The distance between $\defvec{s}_{m}$ and $\defvec{s}_{n}$, ${m}\!\neq\!{n}$ is written as 
\begin{equation}
    \lVert\defvec{s}_{m}-\defvec{s}_{n}\rVert=\sqrt{\frac{12\abs{m-n}E_{\defrvec{S}}}{M^2-1}}.
\end{equation}

\begin{figure*}[tp] 
\centering
\begin{subfigure}{0.47\columnwidth}
    \centering
    \includegraphics[clip=true, trim=0mm 0mm 0mm 0mm, width=1.0\columnwidth,height=0.85\columnwidth]{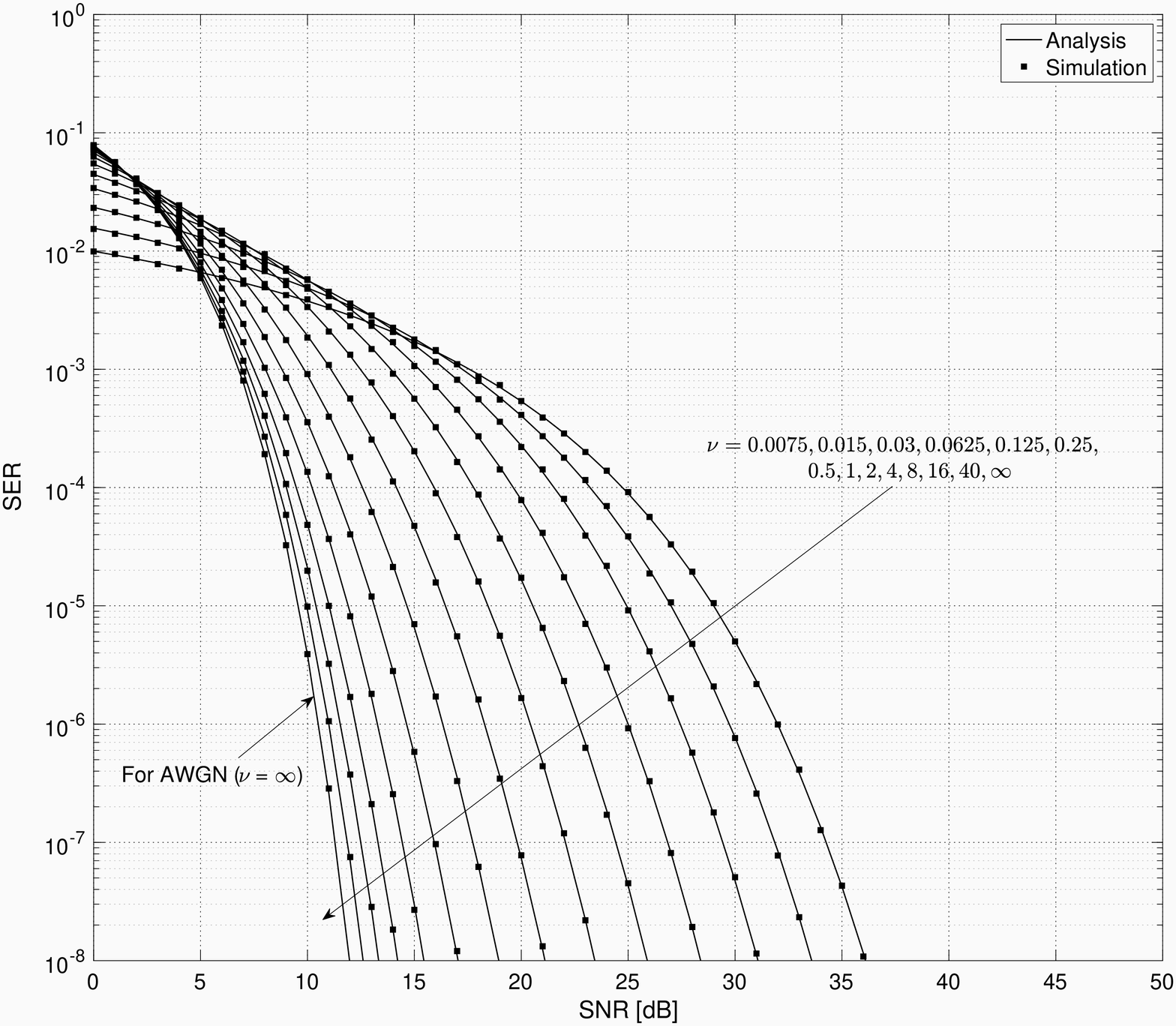}
    \caption{Modulation level $M=2$.}
    \vspace{5mm}
    \label{Figure:ConditionalSEPForMASKA}
\end{subfigure}
{~~~}
\begin{subfigure}{0.47\columnwidth}
    \centering
    \includegraphics[clip=true, trim=0mm 0mm 0mm 0mm, width=1.0\columnwidth,height=0.85\columnwidth]{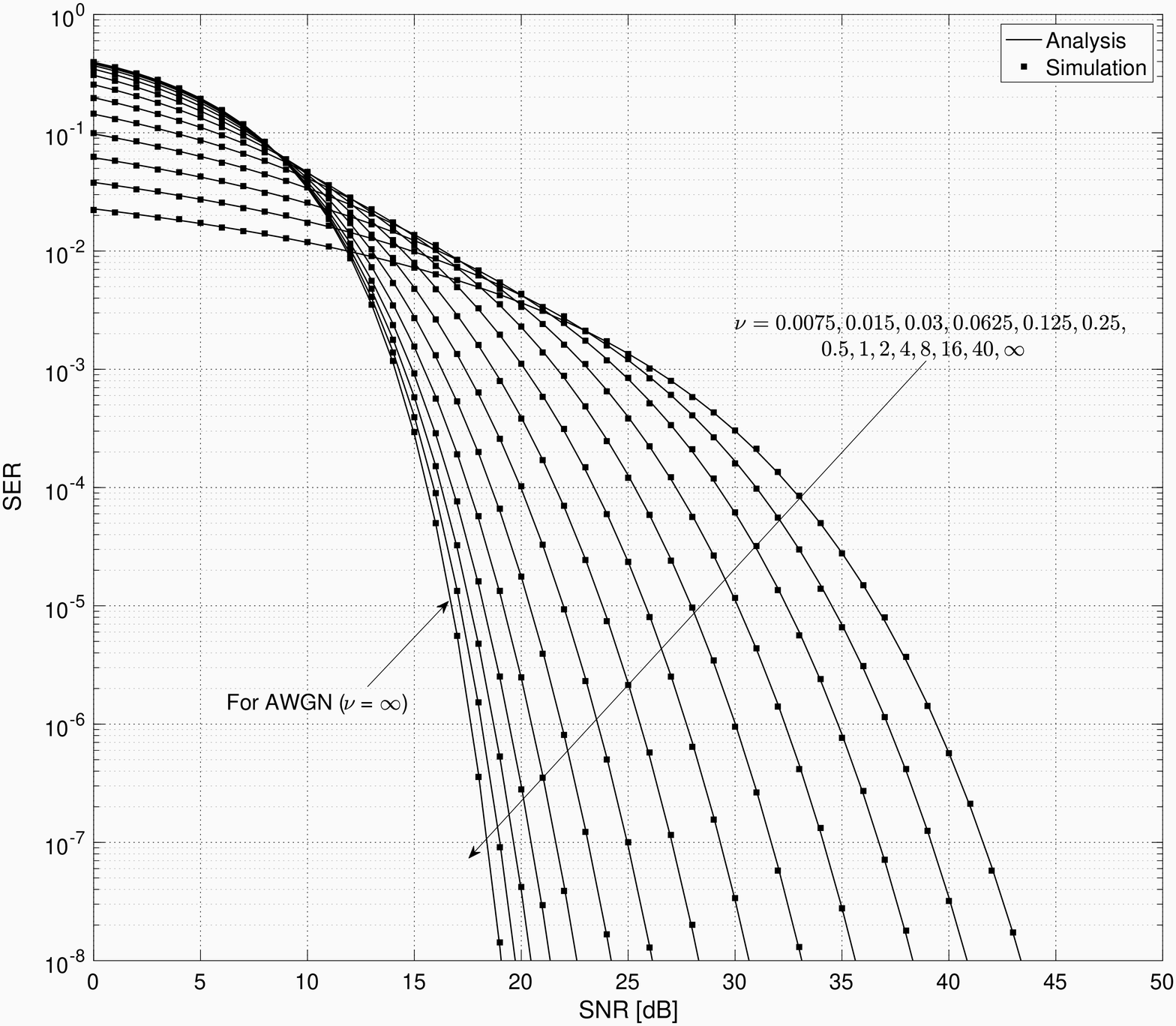}
    \caption{Modulation level $M=4$.}
    \vspace{5mm}
    \label{Figure:ConditionalSEPForMASKB}
\end{subfigure}\\
\begin{subfigure}{0.47\columnwidth}
    \centering
    \includegraphics[clip=true, trim=0mm 0mm 0mm 0mm, width=1.0\columnwidth,height=0.85\columnwidth]{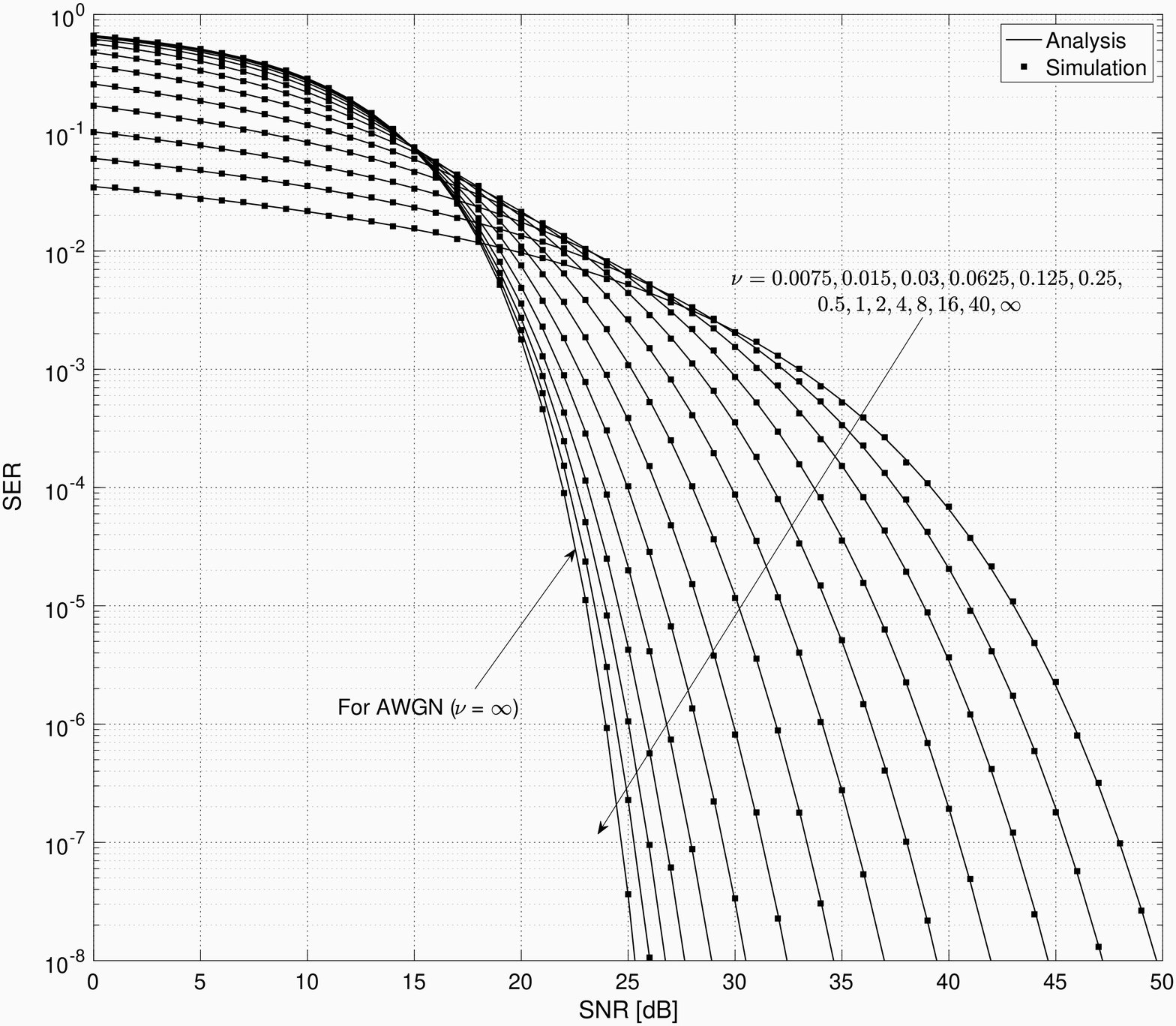}
    \caption{Modulation level $M=8$.}
    \label{Figure:ConditionalSEPForMASKC}
\end{subfigure}
{~~~}
\begin{subfigure}{0.47\columnwidth}
    \centering
    \includegraphics[clip=true, trim=0mm 0mm 0mm 0mm, width=1.0\columnwidth,height=0.85\columnwidth]{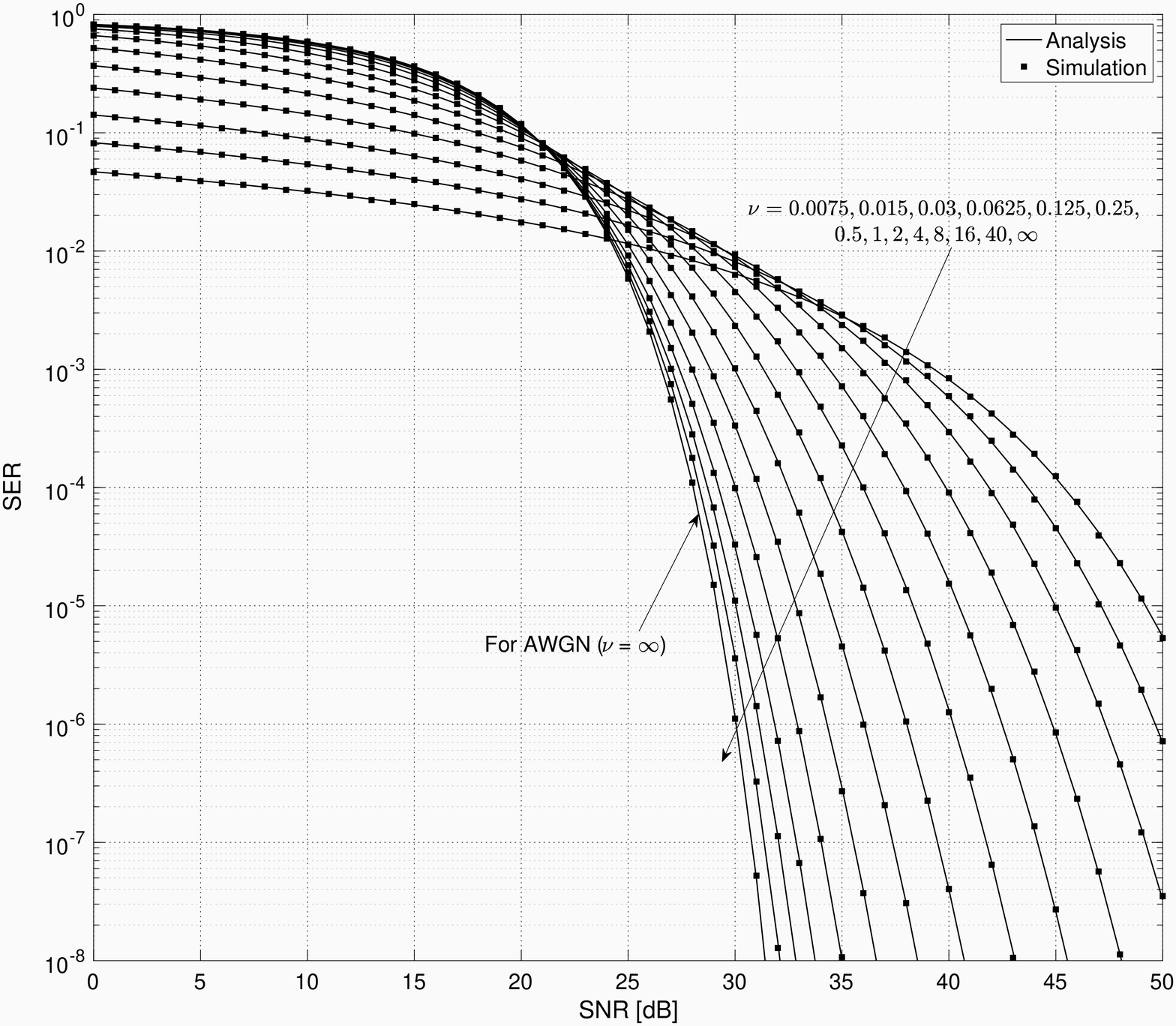}
    \caption{Modulation level $M=16$.}
    \label{Figure:ConditionalSEPForMASKD}
\end{subfigure}
\caption{The \ac{SER} of \ac{M-ASK} signaling over \ac{AWMN} channels.}
\label{Figure:ConditionalSEPForMASK}
\vspace{-2mm} 
\end{figure*} 

Let us find the conditional \ac{SER} for the \ac{M-ASK} modulation. Assuming $\defvec{s}_{m}$ is transmitted, we can write the received vector $\defrvec{R}_{c}$ using the mathematical model given by \eqref{Eq:PrecodedComplexAWMNVectorChannel} as follows 
\begin{equation}
    \defrvec{R}_{c}=H{a}_{m}\defvec{s}+\defrvec{Z}_{c}
\end{equation}
where $\defrvec{Z}_{c}\!\sim\!\mathcal{M}^{L}_{\nu}(\defvec{0},\frac{N_{0}}{2}\defmat{I})$ and hence $\defrvec{R}_{c}\!\sim\!\mathcal{M}^{L}_{\nu}(H{a}_{m}\defvec{s},\frac{N_{0}}{2}\defmat{I})$. Since all modulation symbols are assumed equiprobable, a symbol error occurs for each $\defvec{s}_{m}$ except for the two symbols $\defvec{s}_{1}$ and $\defvec{s}_{M}$ when the the projection of $\defrvec{R}_{c}-H\defvec{s}_{m}$ on $H\defvec{s}_{m\pm{1}}-H\defvec{s}_{m}$, i.e., $\RealPart{(H\defvec{s}_{m\pm{1}}-H\defvec{s}_{m})^H(\defrvec{R}_{c}-H\defvec{s}_{m})}$ is greater~than~the~distance of $H\defvec{s}_{m}$ from the perpendicular bisector of the hyperline that connects $H\defvec{s}_{m}$ and $H\defvec{s}_{m\pm{1}}$, and the probability of this error is written with the aid of
\theoremref{Theorem:MLDecisionErrorProbabilityForBinaryCoherentSignallingOverAWMNChannels} as follows 
\begin{subequations}
\begin{eqnarray}
\Pr\bigl\{\bigl.e\,\bigr|\,H,\defvec{s}_{m}\bigr\}&=&
        Q_{\nu}\biggl(\frac{H{\lVert\defvec{s}_{m}-\defvec{s}_{m+1}\rVert}}{\sqrt{2{N}_{0}}}\biggr)
            +Q_{\nu}\biggl(
            \frac{H{\lVert\defvec{s}_{m}-\defvec{s}_{m-1}\rVert}}{\sqrt{2{N}_{0}}}
            \biggr),{~~~~~~~~}\\
        &=&2Q_{\nu}\biggl(
            \frac{H\Delta}{\sqrt{2{N}_{0}}}
            \biggr),
\end{eqnarray}
\end{subequations}
where substituting \eqref{Eq:ASKMOdulationMinimumDistance} results in
\begin{equation}\label{Eq:ASKMOdulationSmErrorProbability}
    \Pr\bigl\{\bigl.e\,\bigr|\,H,\defvec{s}_{m}\bigr\}=2Q_{\nu}\Biggl(
            \sqrt{\frac{6\gamma}{M^2-1}}
            \Biggr),
\end{equation}
where $\gamma\!=\!{H}^2{E_{\defrvec{S}}}/{N_{0}}$ denotes the \ac{SNR} during transmission of one modulation symbol. Additionally, we also need to obtain $\Pr\bigl\{\bigl.e\,\bigr|\,H,\defvec{s}_{1}\bigr\}$ and $\Pr\bigl\{\bigl.e\,\bigr|\,H,\defvec{s}_{M}\bigr\}$. For the modulation symbol $\defvec{s}_{1}$, we obtain  
\begin{subequations}\label{Eq:ASKMOdulationS1ErrorProbability}
\begin{eqnarray}
    \label{Eq:ASKMOdulationS1ErrorProbabilityA}
    \Pr\bigl\{\bigl.e\,\bigr|\,H,\defvec{s}_{1}\bigr\}&=&
        Q_{\nu}\Biggl(
            \frac{H{\lVert\defvec{s}_{1}-\defvec{s}_{2}\rVert}}{\sqrt{2{N}_{0}}}
            \Biggr),\\
    \label{Eq:ASKMOdulationS1ErrorProbabilityB}
        &=&Q_{\nu}\Biggl(
            \sqrt{\frac{6\gamma}{M^2-1}}
            \Biggr).
\end{eqnarray}
\end{subequations}
Similarly, for the modulation symbol $\defvec{s}_{M}$, we obtain 
\begin{subequations}\label{Eq:ASKMOdulationSMErrorProbability}
\begin{eqnarray}
    \label{Eq:ASKMOdulationSMErrorProbabilityA}
    \Pr\bigl\{\bigl.e\,\bigr|\,H,\defvec{s}_{M}\bigr\}&=&
        Q_{\nu}\Biggl(
            \frac{H{\lVert\defvec{s}_{M}-\defvec{s}_{M-1}\rVert}}{\sqrt{2{N}_{0}}}
            \Biggr),\\
    \label{Eq:ASKMOdulationSMErrorProbabilityB}            
    &=&Q_{\nu}\Biggl(
            \sqrt{\frac{6\gamma}{M^2-1}}
            \Biggr).
\end{eqnarray}
\end{subequations}

\begin{theorem}\label{Theorem:MLDecisionErrorProbabilityForASKCoherentSignallingOverAWMNChannels}
For the \ac{ML} decision rule, the conditional \ac{SER} of the \ac{M-ASK} signaling is given by 
\begin{equation}\label{Eq:MLDecisionErrorProbabilityForASKCoherentSignallingOverAWMNChannels}
 \Pr\bigl\{e\,\bigl|\,H\bigr.\bigr\}=
        2\Bigr(1-\frac{1}{M}\Bigl)Q_{\nu}\Biggl(
            \sqrt{\frac{6\gamma}{M^2-1}}
            \Biggr),
\end{equation} 
where $\gamma$ is the instantaneous \ac{SNR} defined in \eqref{Eq:SignalToNoiseDefinition}.
\end{theorem}

\begin{proof}
When the modulation symbols are equiprobable, we write the conditional \ac{SER} of the \ac{M-ASK} signaling as 
\begin{equation}
    \Pr\bigl\{\bigl.e\,\bigl|\,H\bigr.\bigr\}=
        \frac{1}{M}\sum_{m=1}^{M}
            \Pr\bigl\{\bigl.e\,\bigr|H,\defvec{s}_{m}\bigr\},
\end{equation}
where substituting \eqref{Eq:ASKMOdulationSmErrorProbability}, \eqref{Eq:ASKMOdulationS1ErrorProbabilityB} and \eqref{Eq:ASKMOdulationSMErrorProbabilityB} results in \eqref{Eq:MLDecisionErrorProbabilityForASKCoherentSignallingOverAWMNChannels},
which completes the proof of \theoremref{Theorem:MLDecisionErrorProbabilityForASKCoherentSignallingOverAWMNChannels}.
\end{proof}

Let us check the special cases. First, when the normality factor $\nu\!=\!1$, we reduce \eqref{Eq:MLDecisionErrorProbabilityForASKCoherentSignallingOverAWMNChannels} to the conditional \ac{SER} of the \ac{M-ASK} signaling in \ac{CCS} \ac{AWLN} channels, that is
\begin{equation}\label{Eq:MLDecisionErrorProbabilityForASKCoherentSignallingOverAWLNChannels}
 \Pr\bigl\{e\,\bigl|\,H\bigr.\bigr\}=
        2\Bigr(1-\frac{1}{M}\Bigl)LQ\Biggl(
            \sqrt{\frac{6\gamma}{M^2-1}}
            \Biggr),
\end{equation} 
Secondly, when the normality factor $\nu\!\rightarrow\!\infty$, we also reduce  \eqref{Eq:MLDecisionErrorProbabilityForASKCoherentSignallingOverAWMNChannels} with the aid of \eqref{Eq:McLeishQFunctionandGaussianQFunctionRelation} to \cite[Eq.~(4.3-5)]{BibProakisBook}, \cite[Eq.~(8.3)]{BibAlouiniBook} 
\begin{equation}\label{Eq:MLDecisionErrorProbabilityForASKCoherentSignallingOverAWGNChannels}
 \Pr\bigl\{e\,\bigl|\,H\bigr.\bigr\}=
        2\Bigr(1-\frac{1}{M}\Bigl)Q\Biggl(
            \sqrt{\frac{6\gamma}{M^2-1}}
            \Biggr),
\end{equation} 
which is the conditional \ac{SER} of the \ac{M-ASK} signaling in \ac{CCS} \ac{AWGN} channels as expected. 

Properly with the aid of \theoremref{Theorem:MLDecisionErrorProbabilityForASKCoherentSignallingOverAWMNChannels}, we disclose in \figref{Figure:ConditionalSEPForMASK} the conditional \ac{SER} of \ac{M-ASK} signaling with respect to the different normalities in \ac{AWGN} channels. In addition to our previous observations of that the impulsive nature of the additive noise distribution deteriorates the performance in high-\ac{SNR} regime while negligibly improves in low-\ac{SNR} regime, we observe that the system performance gets more vulnerable to the impulsive nature of the additive noise distribution as the modulation level $M$ increases.     

\paragraph{Conditional \ac{SER} of \ac{M-QAM} Modulation}
\label{Section:SignallingOverAWMNChannels:CoherentSignalling:SymbolErrorProbability:MQAMModulation}
Considering the \ac{M-QAM} constellation as the extension of the two \ac{M-ASK} constellations to the complex amplitude keying, we denote its modulations symbols by $\{\defvec{s}_{1},\defvec{s}_{2},\ldots,\defvec{s}_{M}\}$, where we express each modulation symbol as
\begin{equation}
   \defvec{s}_{m}=({a}_{m}+\imaginary\,{b}_{m})\defvec{s},\quad{1}\leq{m}\leq{M},
\end{equation}
where $\defvec{s}$ denotes an arbitrary unit vector, i.e., $\lVert\defvec{s}\rVert\!=\!1$. Further, the inphase keying ${a}_{m}\!\in\!\mathbb{R}$ and the quadrature keying ${a}_{m}\!\in\!\mathbb{R}$ are chosen such that we can redefine the \ac{M-QAM} modulation by the Cartesian product of two \ac{M-ASK} constellations whose modulation levels are $M_{I}$ and $M_{Q}$, where the modulation level ${M}$ of the \ac{M-QAM} modulation is factorized to $M_I$ and $M_Q$, i.e., $M\!=\!M_{I}M_{Q}$. We write the symbols of the inphase \ac{M-ASK} constellation as
\begin{equation}\label{Eq:MQAMInphase}
    \defvec{s}^I_{m}={\alpha}_{m}\defvec{s},\quad{1}\leq{m}\leq{M}_{I},
\end{equation}
where $\alpha_{m}\!\in\!\mathbb{R}$. Its average power is $E_{I}\!=\!(\sum_{m}{\alpha}^2_{m})/{M}_{I}$ since its modulation symbols are assumed equiprobable. We write the symbols of the quadrature \ac{M-ASK} constellation as
\begin{equation}\label{Eq:MQAMQuadrature}
    \defvec{s}^Q_{n}={\beta}_{n}\defvec{s},\quad{1}\leq{n}\leq{M}_{Q},
\end{equation}
where $\beta_{m}\!\in\!\mathbb{R}$. The average power is $E_{Q}\!=\!(\sum_{n}{\beta}^2_{n})/{{M}_{Q}}$ since the modulation symbols are assumed equiprobable. In terms of ${\alpha}_{m}$ and ${\beta}_{n}$, we can write ${a}_{m}\!\in\!\mathbb{R}$ and ${a}_{m}\!\in\!\mathbb{R}$ as
\begin{equation}\label{Eq:MQAMInphaseAndQuadratureShiftKeying}
    {a}_{m}={\alpha}_{[m/M_{Q}]+1},\text{~and~}{b}_{m}={\beta}_{m-[{m}/{M_{Q}}]M_{Q}},
\end{equation}
for all ${1}\!\leq\!{m}\!\leq\!{M}$. Accordingly and appropriately, we obtain the average power of the \ac{M-QAM} constellation as
\begin{subequations}
\begin{eqnarray}
    E_{\defrvec{S}}
    &=&\frac{1}{M}\sum_{m=1}^{M}\lVert{\defvec{s}_{m}}\rVert^2,\\
    &=&\frac{1}{M}\sum_{m=1}^{M}({a}^2_{m}+{b}^2_{m}),
\end{eqnarray}
\end{subequations}
where substituting \eqref{Eq:MQAMInphaseAndQuadratureShiftKeying} yields 
\begin{subequations}
\begin{eqnarray}
    E_{\defrvec{S}}
    &=&\frac{1}{{M}_{I}}\sum_{m=1}^{{M}_{I}}{\alpha}^2_{m}+\frac{1}{{M}_{Q}}\sum_{n=1}^{{M}_{Q}}{\beta}^2_{n},\\
    &=&E_{I}+E_{Q}.
\end{eqnarray}
\end{subequations}
such that $E_{I}\!=\!(1-\kappa)E_{\defrvec{S}}$ and $E_{Q}\!=\!\kappa{E}_{\defrvec{S}}$, where $\kappa$ denotes the \ac{IQR} given by 
\begin{equation}\label{Eq:MQAMInphaseAndQuadratureRatio}
    \kappa=\frac{(M^2_{Q}-1)\Delta^2_{Q}}{(M^2_{Q}-1)\Delta^2_{Q}+(M^2_{I}-1)\Delta^2_{I}}.
\end{equation}
where $\Delta_{I}$ and $\Delta_{Q}$ are the minimum distance of the inphase and quadrature \ac{M-ASK} constellations, respectively.  In addition, when ${M}_{I}\!=\!{M}_{Q}$ and ${\Delta}_{I}\!=\!{\Delta}_{Q}$, the \ac{M-QAM} signaling is termed as a square \ac{M-QAM} signaling, and otherwise, a rectangular \ac{M-QAM} signaling. Further, with the aid of the definition of the  instantaneous \ac{SNR} given by \eqref{Eq:SignalToNoiseDefinition}, we can rewrite the instantaneous \ac{SNR} as $\gamma\!=\!H^{2}{E_{\defrvec{S}}}/{N_0}\!=\!\gamma_{I}+\gamma_{Q}$, where we have  $\gamma_{I}\!=\!H^{2}{E_{I}}/{N_0}$ and $\gamma_{Q}\!=\!H^{2}{E_{Q}}/{N_0}$ such that $\gamma_{I}\!=\!(1-\kappa)\gamma$ and $\gamma_{I}\!=\!\kappa\gamma$. 

\begin{figure*}[tp]
\begin{tabular}{l}
\begin{minipage}[t]{0.9775\textwidth}
\begin{center}
\normalsize
\setcounter{doublecolumnequation}{\value{equation}}
\setcounter{equation}{377}
\begin{multline}
    \label{Eq:MLDecisionErrorProbabilityForRectangularQAMCoherentSignallingOverAWMNChannels}
    \Pr\bigl\{e\,\bigl|\,H\bigr.\bigr\}=
        2\bigr(1-{1}/{M_{I}}\bigl)
            Q_{\nu}\Bigl(\sqrt{\beta^2_{I}\gamma}\Bigr)
            +2\bigr(1-{1}/{M_{Q}}\bigl)
                Q_{\nu}\Bigl(\sqrt{\beta^2_{Q}\gamma}\Bigr)\\
            -2\bigr(1-{1}/{M_{I}}\bigl)\bigr(1-{1}/{M_{Q}}\bigl)
                Q_{\nu}\Bigl(\sqrt{\beta^2_{I}\gamma},\frac{\pi}{2}-\phi\Bigr)\\
            -2\bigr(1-{1}/{M_{I}}\bigl)\bigr(1-{1}/{M_{Q}}\bigl)
                Q_{\nu}\Bigl(\sqrt{\beta^2_{Q}\gamma},\phi\Bigr),
\end{multline}
\setcounter{equation}{\value{doublecolumnequation}}
\end{center}
\vspace{1pt}
\end{minipage}\\
\hline
\end{tabular}
\vskip -15pt
\end{figure*}

Let us find the conditional \ac{SER} expression for the rectangular \ac{M-QAM} modulation based on the resultants given above. Assuming $\defvec{s}_{m}$ is transmitted, we can readily write the received vector $\defrvec{R}_{c}$ using the mathematical model given by \eqref{Eq:PrecodedComplexAWMNVectorChannel} as follows 
\begin{equation}\label{Eq:AWMNMQAMSignallingReceivedVector}
    \defrvec{R}_{c}=H({a}_{m}+\imaginary{b}_{m})\defvec{s}+\defrvec{Z}_{c}
\end{equation}
where $\defrvec{Z}_{c}\!\sim\!\mathcal{CM}^{L}_{\nu}(\defvec{0},\frac{N_{0}}{2}\defmat{I})$, and then the received vector is $\defrvec{R}_{c}\!\sim\!\mathcal{CM}^{L}_{\nu}(H({a}_{m}+\imaginary{b}_{m})\defvec{s},\frac{N_{0}}{2}\defmat{I})$. Further, we have 
\begin{equation}
    \defrvec{Z}_{c}=\defrvec{I}_{c}+\imaginary\defrvec{Q}_{c}
\end{equation}
where $\defrvec{I}_{c}\!\sim\!\mathcal{M}^{L}_{\nu}(\defvec{0},\frac{N_{0}}{2}\defmat{I})$ and $\defrvec{Q}_{c}\!\sim\!\mathcal{M}^{L}_{\nu}(\defvec{0},\frac{N_{0}}{2}\defmat{I})$. It is further extremely important and necessary to note that $\defrvec{I}_{c}$ and $\defrvec{Q}_{c}$ are mutually uncorrelated but not independent since both are belong to the same \ac{CCS} \ac{AWMN} channel.The projection of the received vector $\defrvec{R}_{c}$ on the space of modulation symbols, i.e., $P_{c}\!=\!\defvec{s}^H\defrvec{R}_{c}$ is given by 
\begin{equation}\label{Eq:AWMNMQAMSignallingProjection}
P_{c}=H({a}_{m}+\imaginary{b}_{m})\defvec{s}+{Z}_{c}
\end{equation}
where we decompose ${Z}_{c}\!\sim\!\mathcal{CM}_{\nu}(\defvec{0},{N_{0}}/{2})$ as 
\begin{equation}
    {Z}_{c}={I}_{c}+\imaginary{Q}_{c}
\end{equation}
where the inphase ${I}_{c}\!\sim\!\mathcal{M}_{\nu}(\defvec{0},{N_{0}}/{2})$ and the quadrature  ${Q}_{c}\!\sim\!\mathcal{M}_{\nu}(\defvec{0},{N_{0}}/{2})$ are mutually uncorrelated but not independent due to  the reason mentioned above. Appropriately, with the aid of \eqref{Eq:MLDecisionErrorProbabilityForASKCoherentSignallingOverAWMNChannels}, the probability of an erroneous detection for this \ac{M-QAM} constellation is given in the following theorem.
\ifCLASSOPTIONtwocolumn
\fi

\begin{theorem}\label{Theorem:MLDecisionErrorProbabilityForRectangularQAMCoherentSignallingOverAWMNChannels}
For the \ac{ML} decision rule, the conditional \ac{SER} of the rectangular \ac{M-QAM} signaling is given by \eqref{Eq:MLDecisionErrorProbabilityForRectangularQAMCoherentSignallingOverAWMNChannels}\setcounter{equation}{378} at the top of this page, in which $\gamma$ is the instantaneous \ac{SNR} defined in \eqref{Eq:SignalToNoiseDefinition}, $\kappa$ is the \ac{IQR} defined in \eqref{Eq:MQAMInphaseAndQuadratureRatio}. Further, $\beta_{I}$ and $\beta_{Q}$ are respectively the minimum inphase and quadrature distances normalized by noise power and are respectively defined by
\begin{equation}\label{Eq:RectangularQAMConditionalSEPInphaseQuadratureRatio}
    \beta_{I}=\sqrt{\frac{6(1-\kappa)}{M^2_{I}-1}},\text{~and~}\beta_{Q}=\sqrt{\frac{6\kappa}{M^2_{Q}-1}}.
\end{equation}
The phase $\phi\!=\!\arctan\bigl(\beta_{I}/\beta_{Q}\bigr)$ is given by
\begin{equation}
    \phi=\arctan\Bigl(\sqrt{\frac{\kappa(M^2_{I}-1)}{(1-\kappa)(M^2_{Q}-1)}}\Bigr).
\end{equation}
\end{theorem}

\begin{proof}
With the aid of \theoremref{Theorem:CCSMcLeishDefinition}, let us further decompose the additive complex noise ${Z}_{c}$ as 
\begin{equation}
    {Z}_{c}=\sqrt{G}({X}_{c}+\imaginary{Y}_{c})
\end{equation}
where ${G}\!\sim\!\mathcal{G}(\nu,1)$, ${X}_{c}\!\sim\!\mathcal{N}(\defvec{0},{N_{0}}/{2})$, and ${Y}_{c}\!\sim\!\mathcal{N}(\defvec{0},{N_{0}}/{2})$ such that we define the inphase ${I}_{c}\!=\!\sqrt{G}{X}_{c}$ and the quadrature ${Q}_{c}\!=\!\sqrt{G}{Y}_{c}$. Hence, we notice that both ${I}_{c}|G$ and ${Q}_{c}|G$ (i.e., both ${I}_{c}$ and ${Q}_{c}$ conditioned on ${G}$) are mutually independent Gaussian distributions with zero mean and $GN_{0}/2$ variance. Appropriately, exploiting \eqref{Eq:MLDecisionErrorProbabilityForASKCoherentSignallingOverAWGNChannels} and using the coefficients \eqref{Eq:RectangularQAMConditionalSEPInphaseQuadratureRatio}, we can write the the conditional \ac{SER} of the inphase  ${\text{M}}_I\text{-ASK}$ as
$\Pr\{e_{I}\,|H,G\}\!=\!2(1-{1}/{M_I})Q(\beta_{I}\sqrt{\gamma/G})$.
Similarly, we can write the conditional \ac{SER} of the quadrature ${\text{M}}_Q\text{-ASK}$ as
$\Pr\{e_{Q}\,|H,G\}\!=\!2(1-{1}/{M_Q})Q(\beta_{Q}\sqrt{\gamma/G})$. 
The mutual independence between ${I}_{c}|G$ and ${Q}_{c}|G$ yields the conclusion that the probability of the correct symbol decision is the product of the conditional probabilities  $\Pr\{c_{I}\,|H,G\}\!=\!1-\Pr\{e_{I}\,|H,G\}$ and $\Pr\{c_{Q}\,|H,G\}\!=\!1-\Pr\{e_{Q}\,|H,G\}$, which are respectively correct decision probabilities for constituent ${\text{M}}_I$-ASK and ${\text{M}}_Q$-ASK constellations when conditioned on ${G}$, we can thus write the probability of an erroneous detection as 
\begin{subequations}
\begin{eqnarray}
    \Pr\{e\,|H,G\}
        &=&1-\Pr\{c\,|H,G\}, \\
        &=&1-\Pr\{c_{I}\,|H,G\}\Pr\{c_{Q}\,|H,G\}, \\
        &=&1-(1-\Pr\{e_{I}\,|H,G\})(1-\Pr\{e_{Q}\,|H,G\}),{~~}
\end{eqnarray}
\end{subequations}
where substituting $\Pr\{e_{I}\,|H,G\}$ and $\Pr\{e_{Q}\,|H,G\}$ yields 
\begin{multline}\label{Eq:RectangularQAMConditionalSEPConditionedOnPowerFluctuations}
 \Pr\{\bigl.e\,|\,H,G\}=
        2\bigr(1-{1}/{M_{I}}\bigl)
            Q\bigl(\beta_{I}\sqrt{\gamma/G}\bigr)\\
            +2\bigr(1-{1}/{M_{Q}}\bigl)
                Q\bigl(\beta_{Q}\sqrt{\gamma/G}\bigr)\\
            -4\bigr(1-{1}/{M_{I}}\bigl)\bigr(1-{1}/{M_{Q}}\bigl)\\
            \times
                Q\bigl(\beta_{I}\sqrt{\gamma/G}\bigr)
                Q\bigl(\beta_{Q}\sqrt{\gamma/G}\bigr).
\end{multline} 
Then, the conditional \ac{SER} of the rectangular \ac{M-QAM} constellation is written as 
$\Pr\{e\,|\,H\}\!=\!\int_{0}^{\infty}\Pr\{e\,|\,H,g\}f_{G}(g)dg$, where substituting \eqref{Eq:ProportionPDF} yields
\ifCLASSOPTIONtwocolumn
\begin{multline}\label{Eq:RectangularQAMConditionalSEP}
\!\!\!\!\Pr\{\bigl.e\,|\,H\}=
        2\bigr(1-{1}/{M_{I}}\bigl)I_1\bigl(\beta_{I}\sqrt{\gamma}\bigr){~~}\\+
        2\bigr(1-{1}/{M_{Q}}\bigl)I_1\bigl(\beta_{Q}\sqrt{\gamma}\bigr){~~}\\-
        4\bigr(1-{1}/{M_{Q}}\bigl)\bigr(1-{1}/{M_{Q}}\bigl)I_2\bigl(\beta_{I}\sqrt{\gamma},\beta_{Q}\sqrt{\gamma}\bigr),
\end{multline}
\else
\begin{multline}\label{Eq:RectangularQAMConditionalSEP}
    \Pr\{\bigl.e\,|\,H\}=
        2\bigr(1-{1}/{M_{I}}\bigl)I_1\bigl(\beta_{I}\sqrt{\gamma}\bigr)+
        2\bigr(1-{1}/{M_{Q}}\bigl)I_1\bigl(\beta_{Q}\sqrt{\gamma}\bigr){~~}\\-
        4\bigr(1-{1}/{M_{Q}}\bigl)\bigr(1-{1}/{M_{Q}}\bigl)I_2\bigl(\beta_{I}\sqrt{\gamma},\beta_{Q}\sqrt{\gamma}\bigr),
\end{multline}
\fi
where $I_1(x)$ and $I_2(x,y)$ are given by 
\begin{eqnarray}
\label{Eq:GaussianQFunctionIntegral}
    I_1(x)&=&\int_{0}^{\infty}Q\bigl(\sqrt{x^2/g}\bigr)f_{G}(g)dg,\\
\label{Eq:DoubleGaussianQFunctionIntegral}
    I_2(x,y)&=&\int_{0}^{\infty}Q\bigl(\sqrt{x^2/g}\bigr)Q\bigl(\sqrt{y^2/g}\bigr)f_{G}(g)dg,{~~~}
\end{eqnarray}
where $x,y\!\in\!\mathbb{R}_{+}$. Inserting 
$Q(x)\!=\!\frac{1}{2}\Erfc{x/\sqrt{2}}$ \cite[Eq.\!~(2.3-18)]{BibProakisBook} and \cite[Eq.\!~(06.27.26.0006.01)]{BibWolfram2010Book} 
into \eqref{Eq:GaussianQFunctionIntegral}, and accordingly using 
\cite[Eqs.\!~(2.8.4)\!~and\!~(2.9.1)]{BibKilbasSaigoBook}, $I_1(x)$ results in \eqref{Eq:McLeishQFunctionUsingFoxH}. Therefore, we have $I_1(x)\!=\!Q_{\nu}(x)$.
In addition, inserting \cite[Eq.\!~(4.6)\!~and\!~(4.8)]{BibSimonDivsalarTCOM1998} into \eqref{Eq:DoubleGaussianQFunctionIntegral} and using \cite[Eq. (3.471/9)]{BibGradshteynRyzhikBook} and then exploiting \defref{Definition:McLeishPartialQFunction}, we obtain $I_2(x,y)$ as $I_2(x,y)\!=\!\frac{1}{2}Q_{\nu}(x,{\pi}/{2}-\phi)+\frac{1}{2}Q_{\nu}(y,\phi)$. Finally, substituting $I_1(x)$ and $I_2(x,y)$ into \eqref{Eq:RectangularQAMConditionalSEP} results in \eqref{Eq:MLDecisionErrorProbabilityForRectangularQAMCoherentSignallingOverAWMNChannels}, which completes the proof of   
\theoremref{Theorem:MLDecisionErrorProbabilityForRectangularQAMCoherentSignallingOverAWMNChannels}.
\end{proof}

\begin{theorem}\label{Theorem:MLDecisionErrorProbabilityForSquareQAMCoherentSignallingOverAWMNChannels}
For the \ac{ML} decision rule, the conditional \ac{SER} of the square \ac{M-QAM} signaling is given by 
\ifCLASSOPTIONtwocolumn
\begin{multline}\label{Eq:MLDecisionErrorProbabilityForSquareQAMCoherentSignallingOverAWMNChannels}
 \Pr\{\bigl.e\,|\,H\}=
        4\Bigr(1-\frac{1}{\sqrt{M}}\Bigl)
            Q_{\nu}\Biggl(\sqrt{\frac{3\gamma}{M-1}}\Biggr)\\
                -4\Bigr(1-\frac{1}{\sqrt{M}}\Bigl)^2Q_{\nu}\Biggl(\sqrt{\frac{3\gamma}{M-1}},\frac{\pi}{4}\Biggr),
\end{multline} 
\else
\begin{multline}\label{Eq:MLDecisionErrorProbabilityForSquareQAMCoherentSignallingOverAWMNChannels}
 \Pr\{\bigl.e\,|\,H\}=
        4\Bigr(1-\frac{1}{\sqrt{M}}\Bigl)
            Q_{\nu}\Biggl(\sqrt{\frac{3\gamma}{M-1}}\Biggr)
                -4\Bigr(1-\frac{1}{\sqrt{M}}\Bigl)^2Q_{\nu}\Biggl(\sqrt{\frac{3\gamma}{M-1}},\frac{\pi}{4}\Biggr),
\end{multline} 
\fi
where $\gamma$ is the \ac{SNR} defined in \eqref{Eq:SignalToNoiseDefinition}.
\end{theorem}

\begin{proof}
When we have $M_{I}\!=\!M_{Q}\!=\!\sqrt{M}$, we perceive that the \ac{M-QAM} constellation becomes a two-dimensional square constellation, where each one of the inphase and quadrature components can be therefore considered as $\sqrt{\text{M}}$-\ac{ASK} constellation. Accordingly, with the aid of \eqref{Eq:ASKMOdulationMinimumDistance}, we find out that the inphase and quadrature minimum distances, i.e., $\Delta_{I}$ and $\Delta_{Q}$ are equal, that is
\begin{equation}
    \Delta_{I}=\Delta_{Q}=\sqrt{\frac{6E_{\defrvec{S}}}{M-1}},
\end{equation}
which yields $\kappa\!=\!1/2$ as observed from \eqref{Eq:MQAMInphaseAndQuadratureRatio}, and further $\beta_{I}\!=\!\beta_{Q}\!=\!\sqrt{{3}/{(M-1)}}$ from \eqref{Eq:RectangularQAMConditionalSEPInphaseQuadratureRatio}. Eventually, substituting these results into \eqref{Eq:MLDecisionErrorProbabilityForRectangularQAMCoherentSignallingOverAWMNChannels} yields \eqref{Eq:MLDecisionErrorProbabilityForSquareQAMCoherentSignallingOverAWMNChannels}, which completes the proof of  \theoremref{Theorem:MLDecisionErrorProbabilityForSquareQAMCoherentSignallingOverAWMNChannels}.
\end{proof}
\begin{figure*}[tp] 
\centering
\begin{subfigure}{0.47\columnwidth}
    \centering
    \includegraphics[clip=true, trim=0mm 0mm 0mm 0mm, width=1.0\columnwidth,height=0.85\columnwidth]{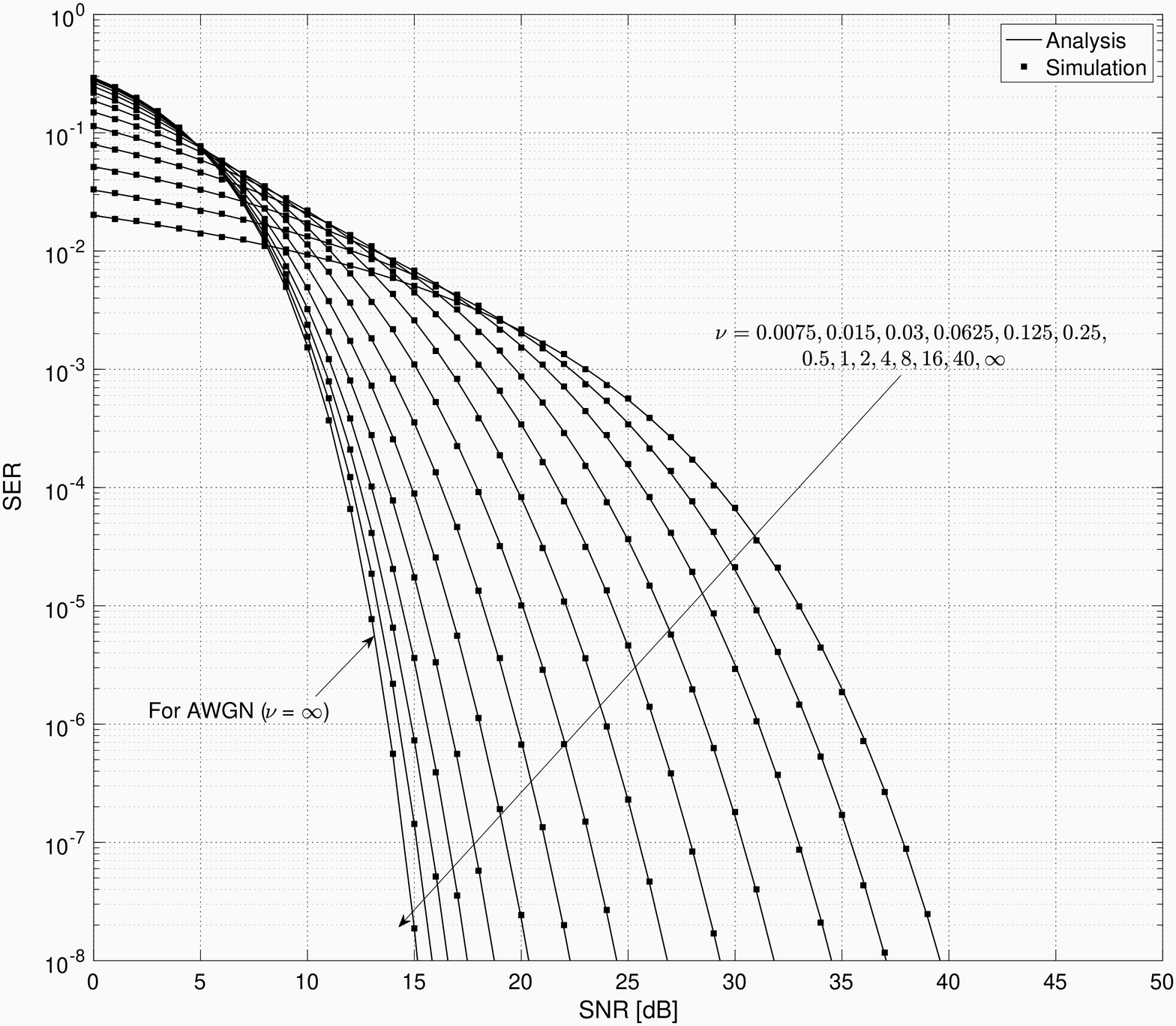}
    \caption{Modulation level $M=4$.}
    \vspace{5mm}
    \label{Figure:ConditionalSEPForMQAMA}
\end{subfigure}
{~~~}
\begin{subfigure}{0.47\columnwidth}
    \centering
    \includegraphics[clip=true, trim=0mm 0mm 0mm 0mm, width=1.0\columnwidth,height=0.85\columnwidth]{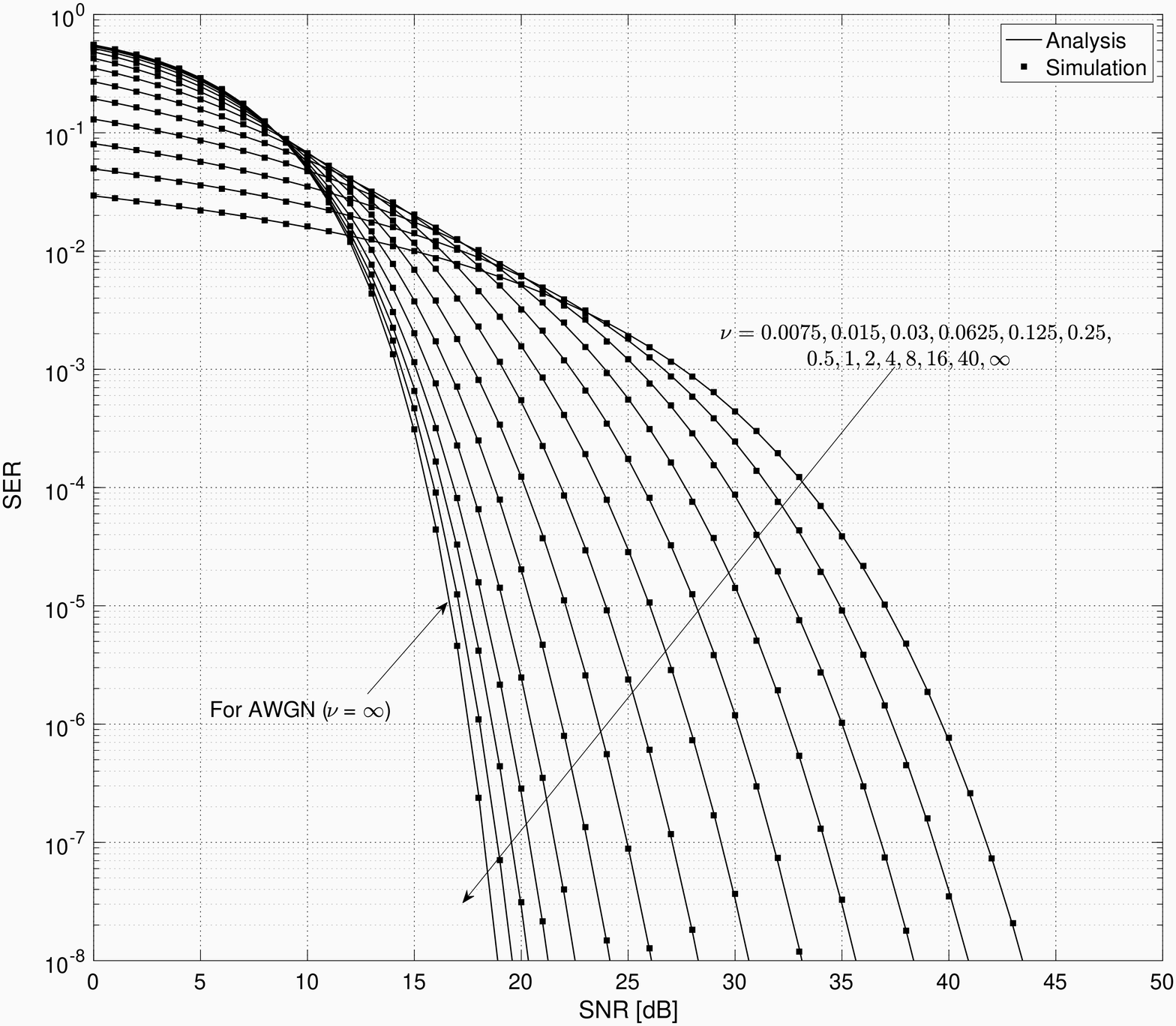}
    \caption{Modulation level $M=8$.}
    \vspace{5mm}
    \label{Figure:ConditionalSEPForMQAMB}
\end{subfigure}\\
\begin{subfigure}{0.47\columnwidth}
    \centering
    \includegraphics[clip=true, trim=0mm 0mm 0mm 0mm, width=1.0\columnwidth,height=0.85\columnwidth]{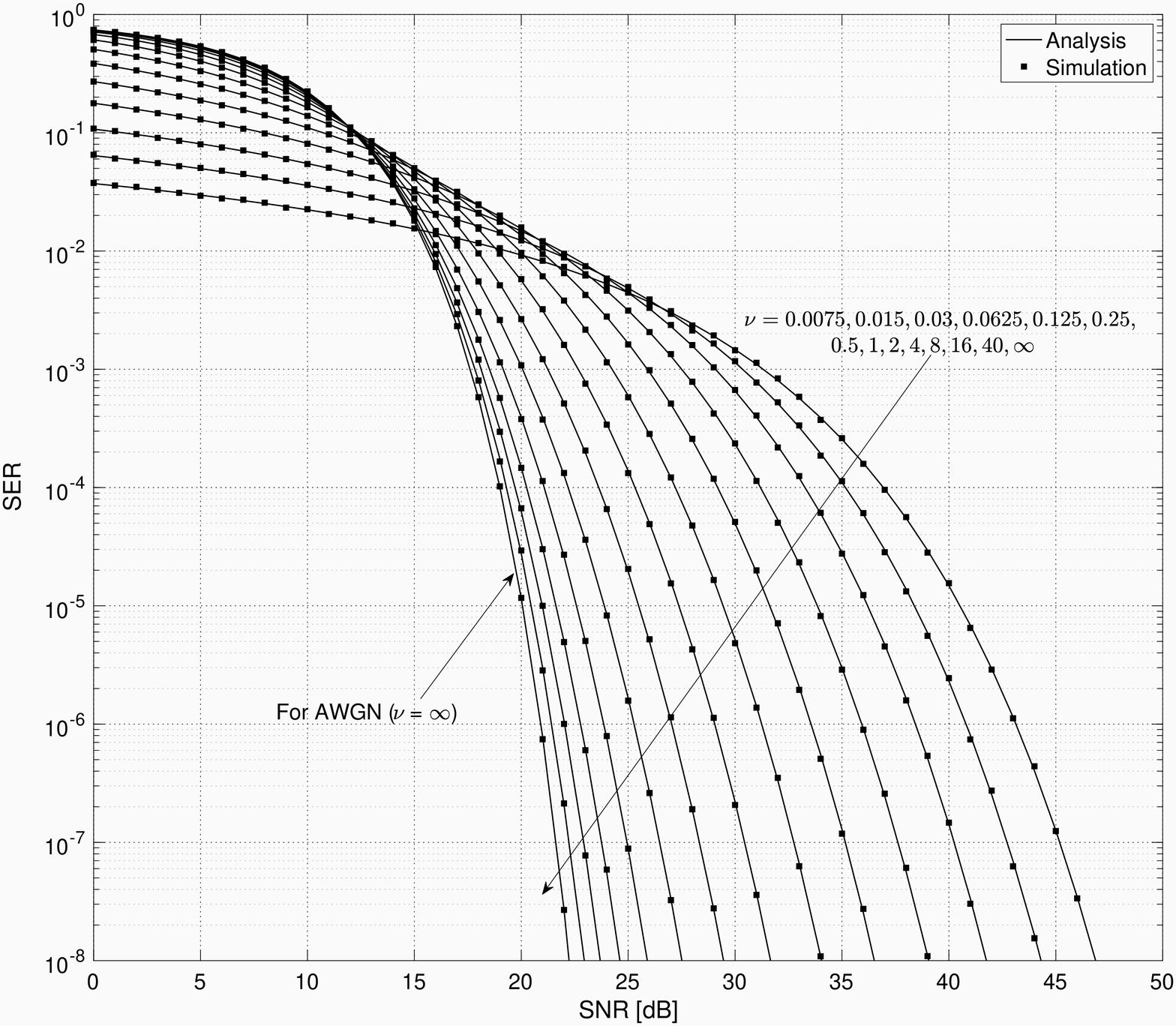}
    \caption{Modulation level $M=16$.}
    \label{Figure:ConditionalSEPForMQAMC}
\end{subfigure}
{~~~}
\begin{subfigure}{0.47\columnwidth}
    \centering
    \includegraphics[clip=true, trim=0mm 0mm 0mm 0mm, width=1.0\columnwidth,height=0.85\columnwidth]{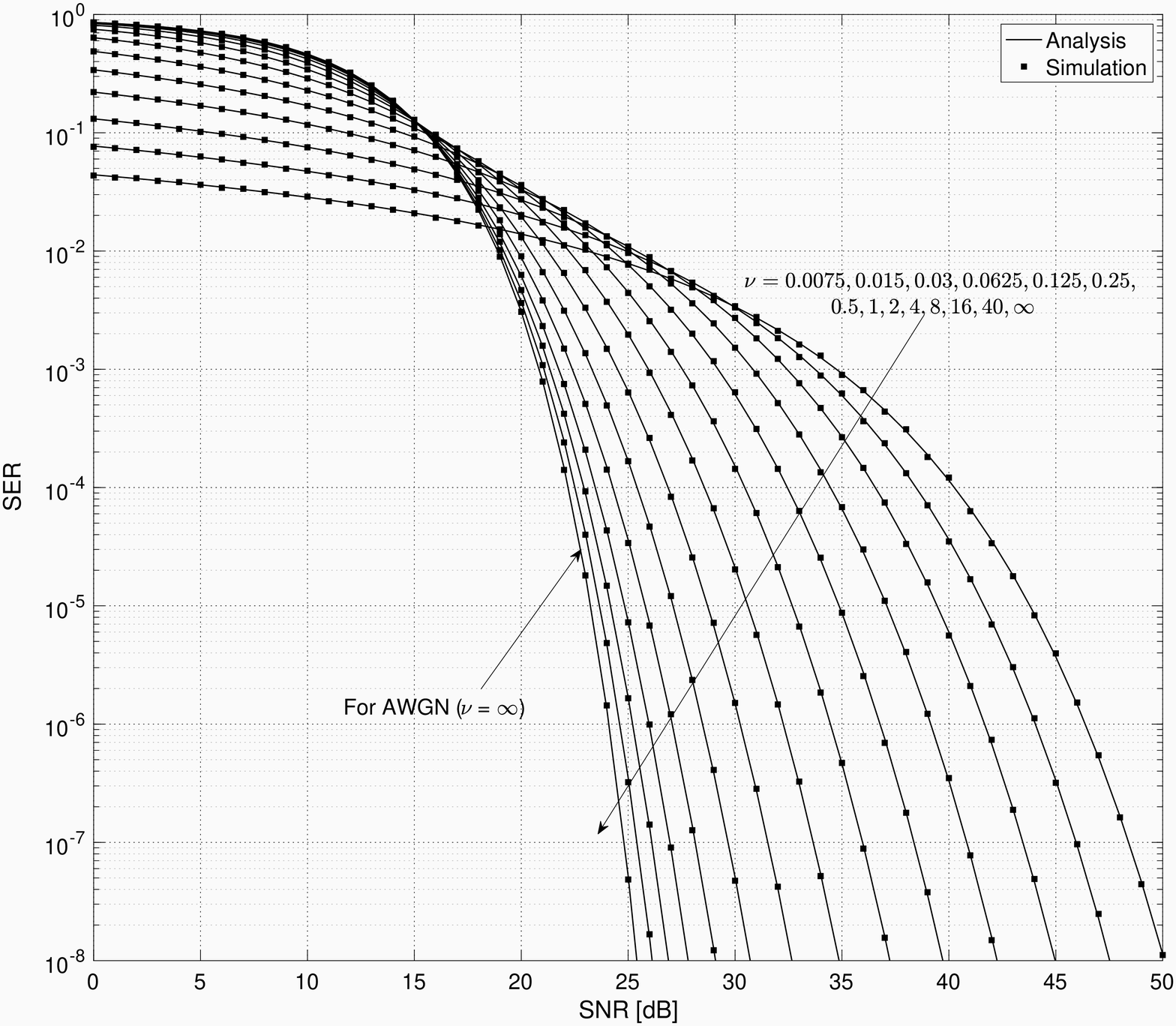}
    \caption{Modulation level $M=32$.}
    \label{Figure:ConditionalSEPForMQAMD}
\end{subfigure}
\caption{The \ac{SER} of \ac{M-QAM} signaling over \ac{AWMN} channels.}
\label{Figure:ConditionalSEPForMQAM}
\vspace{-2mm} 
\end{figure*} 

Let us check some special cases for completeness. For 4-QAM, \eqref{Eq:MLDecisionErrorProbabilityForSquareQAMCoherentSignallingOverAWMNChannels} reduces to
\begin{equation}\label{Eq:MLDecisionErrorProbabilityFor4QAMCoherentSignallingOverAWMNChannels}
    \Pr\{\bigl.e\,|\,H\}=2Q_{\nu}\bigl(\sqrt{\gamma}\bigr)-Q_{\nu}\bigl(\sqrt{\gamma},{\pi}/{4}\bigr),
\end{equation} 
where referring \eqref{Eq:McLeishPartialQFunction} with the total integration angle, i.e., $\pi/2+\pi/2-\pi/4=\pi-\pi/4$, we can reduce \eqref{Eq:MLDecisionErrorProbabilityFor4QAMCoherentSignallingOverAWMNChannels} more to 
\begin{equation}\label{Eq:MLDecisionErrorProbabilityFor4QAMCoherentSignallingOverAWMNChannelsII}
    \Pr\{\bigl.e\,|\,H\}=Q_{\nu}\bigl(\sqrt{\gamma},3{\pi}/{4}\bigr),
\end{equation} 
In addition, note that we have  
\begin{equation}
    \lim_{\nu\rightarrow\infty}Q_{\nu}\Bigl(x,\frac{\pi}{4}\Bigr)=Q(x)^2.
\end{equation}
Thus, when the normality factor $\nu\!\rightarrow\!\infty$, \eqref{Eq:MLDecisionErrorProbabilityForSquareQAMCoherentSignallingOverAWMNChannels} reduces to \cite[Eq. (4.3-30)]{BibProakisBook}, \cite[Eq. (8.10)]{BibAlouiniBook} as expected. 

For analytical accuracy and numerical completeness and correctness, in \figref{Figure:ConditionalSEPForMQAM}, we show the \ac{SER} of \ac{M-QAM} 
signaling over~\ac{AWMN}~channels by using \theoremref{Theorem:MLDecisionErrorProbabilityForSquareQAMCoherentSignallingOverAWMNChannels} for analytical accuracy and performing simulations for numerical correctness. We also therein observe that, for  $\nu\!\rightarrow\!{0}$, the system performance deteriorates in the high-\ac{SNR} regime. When we compare the performance of \ac{M-QAM} to that of \ac{M-ASK} (i.e., namely comparing \figref{Figure:ConditionalSEPForMQAMB} to \figref{Figure:ConditionalSEPForMASKC} for $M\!=\!8$), we notice that \ac{M-QAM} gives better performance.

\vspace{1mm}
\paragraph{Conditional \ac{SER} of \ac{M-PSK} Modulation}
\label{Section:SignallingOverAWMNChannels:CoherentSignalling:SymbolErrorProbability:MPSKModulation}
Considering the \ac{M-PSK} constellation as the rotational extension of the \ac{BPSK} constellation to the phase shift keying, let us denote its modulation symbols by $\{\defvec{s}_{1},\defvec{s}_{2},\ldots,\defvec{s}_{M}\}$, where $\defvec{s}_{m}\!=\!\alpha{e}^{\imaginary\theta_{m}}\defvec{s}$ such that $\defvec{s}$ denotes an arbitrary unit vector (i.e., $\lVert\defvec{s}\rVert\!=\!1$), the amplitude $\alpha\!\in\!\mathbb{R}_{+}$ determines the power per modulation symbol such that we can readily express the power of $\defvec{s}_{m}$ as $E_{m}\!=\!\lVert{\defvec{s}_{m}}\rVert^2\!=\!\alpha^2$. Further, the phase rotations ${\theta}_{m}$, ${1}\!\leq\!{m}\!\leq\!{M}$ encode information within the \ac{M-PSK} modulation symbols and are uniformly chosen for a modulation level $M$, that is
\begin{equation}
 {\theta}_{m}\!=\!{2\pi}(m-1)/{M},\quad{1}\leq{m}\leq{M}.
\end{equation}
Accordingly, we can rewrite the \ac{M-PSK} modulation symbols as $\defvec{s}_{m}\!=\!\alpha\exp\bigl(\imaginary{2\pi}(m-1)/{M}\bigr)\defvec{s}$, ${1}\!\leq\!{m}\!\leq\!{M}$ and therein making use of $E_{m}\!=\!\alpha^2$, ${1}\!\leq\!{m}\!\leq\!{M}$, we obtain the average power $E_{\defrvec{S}}$ as follows
\begin{equation}
    E_{\defrvec{S}}
        =\sum_{m=1}^{M}\Pr\{\defvec{s}_{m}\}{E}_{m}=\alpha^2. 
\end{equation}
Therefore, we have $\alpha\!=\!\sqrt{E_{\defrvec{S}}}$. Let us now find the conditional \ac{SER} for the \ac{M-PSK} modulation. Assuming $\defvec{s}_{m}$ is transmitted, we can write the received vector $\defrvec{R}_{c}$ using the mathematical model given by \eqref{Eq:PrecodedComplexAWMNVectorChannel} as follows 
\begin{equation}\label{Eq:AWMNPSKSignallingReceivedVector}
    \defrvec{R}_{c}=\alpha{H}{e}^{\imaginary\theta_{m}}\defvec{s}+\defrvec{Z}_{c}
\end{equation}
where $\defrvec{Z}_{c}\!\sim\!\mathcal{CM}^{L}_{\nu}(\defvec{0},\frac{N_{0}}{2}\defmat{I})$ and $\defrvec{R}_{c}\!\sim\!\mathcal{CM}^{L}_{\nu}(\alpha{H}{e}^{\imaginary\theta_{m}}\defvec{s},\frac{N_{0}}{2}\defmat{I})$. Since the information is carried by means of phase shift keying in form of $2{\pi}/{M}$ multiplies (i.e., the angle difference between the adjacent symbols is $2{\pi}/{M}$), a decision error occurs when the additive noise $\defrvec{Z}_{c}$ causes an enough rotational shift more than ${\pi}/{M}$ in clockwise or counterclockwise direction in $\defrvec{R}_{c}$. We give the projection of $\defrvec{R}_{c}$ on $\defvec{s}_{m}$ as
\begin{equation}\label{Eq:AWMNPSKSignallingProjection}
    {P}_{c}=\defvec{s}^{H}_{m}\defrvec{R}_{c}=\alpha{H}+{Z}_{c}
\end{equation}
where ${Z}_{c}\!\sim\!\mathcal{CM}_{\nu}(0,{N_{0}}/{2})$ follows the \ac{PDF} that we write with the aid of \theoremref{Theorem:CCSMcLeishPDF} as 
\begin{equation}\label{Eq:AWMNPSKSignallingAdditiveNoisePDF}
	f_{Z_{c}}(z)=\frac{2}{\pi}
		\frac{\abs{z}^{\nu-1}}{\Gamma(\nu)\,\Lambda_{0}^{\nu+1}}
			\BesselK[\nu-1]{\frac{2\abs{z}}{\Lambda_{0}}},
\end{equation}
defined over $z\!\in\!\mathbb{C}$ with the normality factor $\Lambda_{0}\!=\!\sqrt{N_{0}/\nu}$.
Therefore, ${P}_{c}\!\sim\!\mathcal{CM}_{\nu}(\alpha{H},{N_{0}}/{2})$ is decomposed as 
\begin{equation}\label{Eq:AWMNPSKSignallingProjectionII}
    {P}_{c}=I_{c}+\imaginary{Q}_{c},
\end{equation}
where ${I}_{c}\!\sim\!\mathcal{M}_{\nu}(\alpha{H},{N_{0}}/{2})$ and ${Q}_{c}\!\sim\!\mathcal{M}_{\nu}(0,{N_{0}}/{2})$. Hence, the amplitude fluctuation caused by the additive complex noise $\defrvec{Z}_{c}$ is apparently written as ${A}_{c}\!=\!\sqrt{I^2_{c}+{Q}^2_{c}}$. The rotational shift, which is another effect caused by the additive complex noise $\defrvec{Z}_{c}$, is written as 
$\Theta_{c}\!=\!\arctan\bigl({{Q}_{c}}/{{I}_{c}}\bigr)$, which follows such a random distribution that a decision error occurs when $\abs{\Theta_{c}}\!>\!{\pi}/{M}$ (i.e., a correct decision occurs when $\abs{\Theta_{c}}\!<\!{\pi}/{M}$). In other words, the error probability when $\defvec{s}_{m}$ was transmitted is readily written as
\begin{subequations}\label{Eq:AWMNPSKSignallingConditionalSEPExpression}
\begin{eqnarray}
    \label{Eq:AWMNPSKSignallingConditionalSEPExpressionA}
    \Pr\bigl\{\bigl.e\,|\,H,\defvec{s}_{m}\bigr\}
        &=&\Pr\bigl\{\abs{\Theta_{c}}>{\pi}/{M}\bigr\},\\
    \label{Eq:AWMNPSKSignallingConditionalSEPExpressionB}
        &=&1-\Pr\bigl\{-{\pi}/{M}<\Theta_{c}<{\pi}/{M}\bigr\}.{~~~~~~~~}
\end{eqnarray}
\end{subequations}
Since assuming that all modulation symbols are equiprobable, we perceive that, due to the rotational symmetry of the \ac{M-PSK} constellation, $\Pr\bigl\{\bigl.e\,|\,H,\defvec{s}_{m}\bigr\}\!=\!\Pr\bigl\{\bigl.e\,|\,H,\defvec{s}_{\widehat{m}}\bigr\}$ for all ${m}\!\neq\!\widehat{m}$. The conditional \ac{SER} of the \ac{M-PSK} is therefore equal to the probability of making a decision error when $\defvec{s}_{m}$ is transmitted, and accordingly we write 
\begin{subequations}
    \begin{eqnarray}
        \Pr\{\bigl.e\,|\,H\}
            &=&\sum_{m=1}^{M}\Pr\{\bigl.e\,|\,H,\defvec{s}_{m}\}\Pr\{\defvec{s}_{m}\},\\
            &=&\Pr\{\bigl.e\,|\,H,\defvec{s}_{m}\},\\
            &=&1-\Pr\{-{\pi}/{M}<\Theta_{c}<{\pi}/{M}\}.{~~}
    \end{eqnarray}
\end{subequations}
Referring to \eqref{Eq:AWMNPSKSignallingProjectionII}, and therefrom having both the amplitude ${A}_{c}\!=\!\sqrt{I^2_{c}+{Q}^2_{c}}$ and the phase $\Theta_{c}\!=\!\arctan\bigl({{Q}_{c}}/{{I}_{c}}\bigr)$, we can deduce the inphase and quadrature of the projection $P_{c}$ as ${I}_{c}={A}_{c}\cos\bigl(\Theta_{c}\bigr)$ and ${Q}_{c}={A}_{c}\sin\bigl(\Theta_{c}\bigr)$, from which we derive the joint \ac{PDF} of ${A}_{c}$ and $\Theta_{c}$ by utilizing \eqref{Eq:AWMNPSKSignallingAdditiveNoisePDF}, that is  
\begin{equation}\label{Eq:AWMNPSKSignallingJointAmplitudePhasePDF}
    f_{{A}_{c},{\Theta}_{c}}(a,\theta)=
        \frac{2\,\Omega(a,\theta)^{\nu-1}}{\pi\Gamma(\nu)\,\Lambda_{0}^{\nu+1}}
	        \BesselK[\nu-1]{
		        \frac{2}{\Lambda_{0}}\Omega(a,\theta)},\!\!
\end{equation}
where $\Omega(a,\theta)$ is given by 
\begin{equation}
\Omega(a,\theta)=\sqrt{{a^2-2a\sqrt{H^2E_{\defrvec{S}}}\cos(\theta)+H^2E_{\defrvec{S}}}}.   
\end{equation}
Accordingly, when we integrate \eqref{Eq:AWMNPSKSignallingJointAmplitudePhasePDF} over $a\!\in\!\mathbb{R}_{+}$, we obtain the marginal \ac{PDF} of ${\Theta}_{c}$, that is $f_{{\Theta}_{c}}(\theta)\!=\!\int_{0}^{\infty}f_{{A}_{c},{\Theta}_{c}}(a,\theta)da$, where substituting \eqref{Eq:AWMNPSKSignallingJointAmplitudePhasePDF} yields 
\begin{equation}\label{Eq:AWMNPSKSignallingPhasePDF}
\!\!\!\!f_{{\Theta}_{c}}(\theta)=
        \int_{0}^{\infty}
        \frac{2\,\Omega(a,\theta)^{\nu-1}}{\pi\Gamma(\nu)\,\Lambda_{0}^{\nu+1}}
	        \BesselK[\nu-1]{
		        \frac{2}{\Lambda_{0}}\Omega(a,\theta)}da,\!\!
\end{equation}
which does not simplify to a simple closed form and thus must be evaluated numerically. Nevertheless, making~use~of~$f_{{\Theta}_{c}}(\theta)$, we calculate the probability $\Pr\bigl\{\theta_{0}\!<\!\Theta_{c}\!<\!\theta_{1}\bigr\}\!=\!\int_{\theta_{0}}^{\theta_{1}}f_{{\Theta}_{c}}(\theta)d\theta$ and thereby derive the conditional \ac{SER} of \ac{M-PSK} constellation in the following.

\begin{figure*}[tp] 
\centering
\begin{subfigure}{0.47\columnwidth}
    \centering
    \includegraphics[clip=true, trim=0mm 0mm 0mm 0mm, width=1.0\columnwidth,height=0.85\columnwidth]{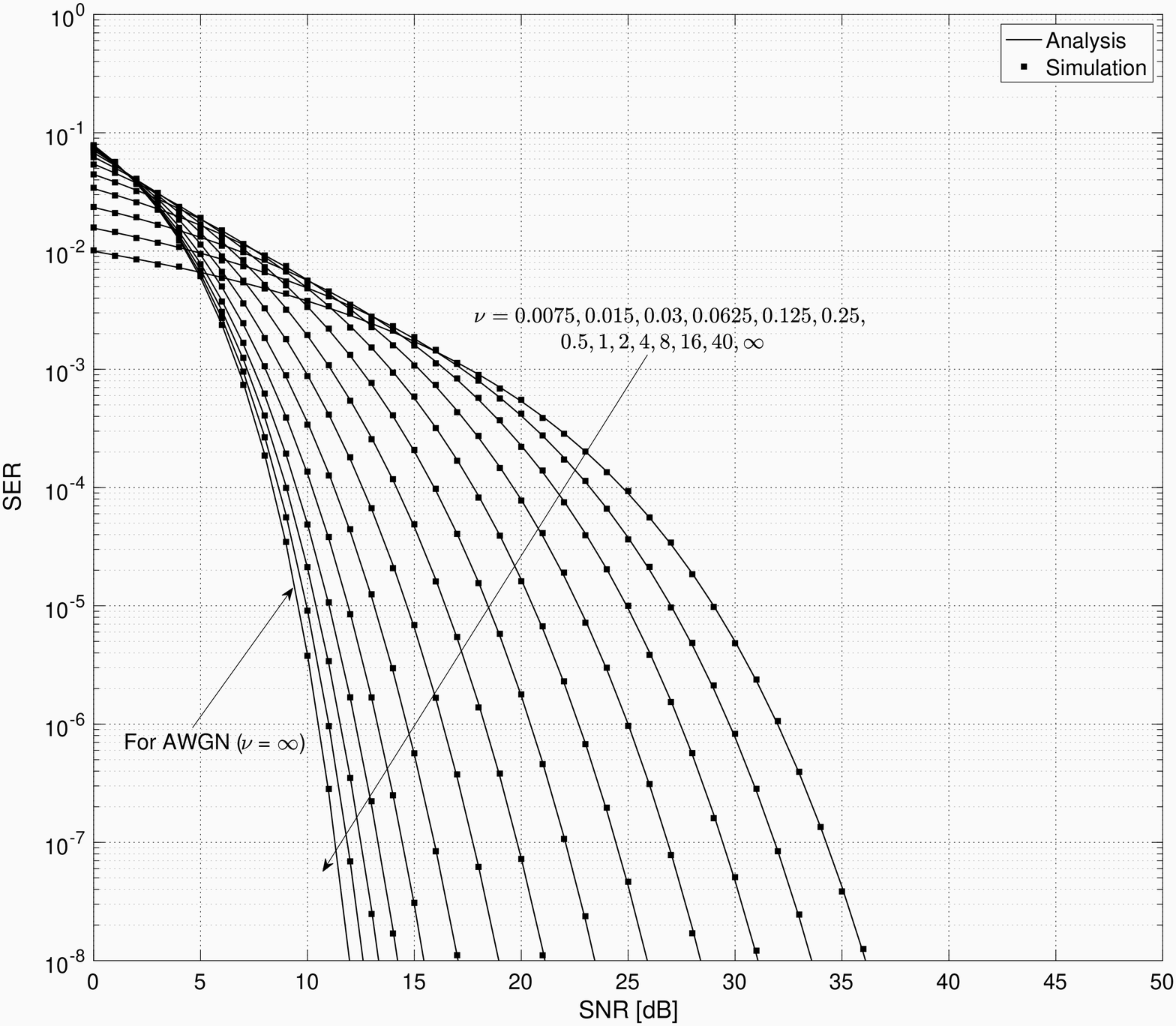}
    \caption{Modulation level $M=2$.}
    \vspace{5mm}
    \label{Figure:ConditionalSEPForMPSKA}
\end{subfigure}
{~~~}
\begin{subfigure}{0.47\columnwidth}
    \centering
    \includegraphics[clip=true, trim=0mm 0mm 0mm 0mm, width=1.0\columnwidth,height=0.85\columnwidth]{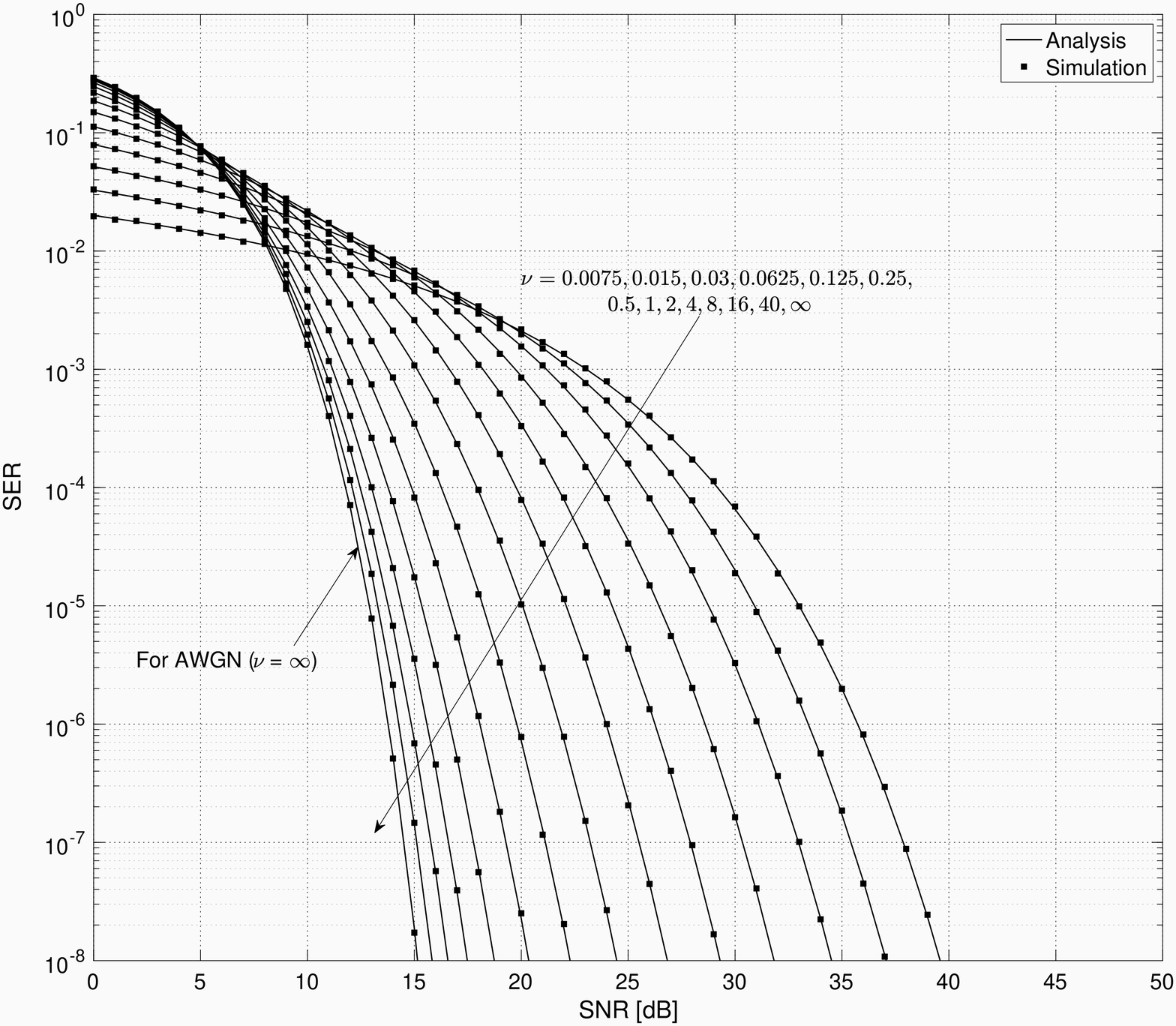}
    \caption{Modulation level $M=4$.}
    \vspace{5mm}
    \label{Figure:ConditionalSEPForMPSKB}
\end{subfigure}\\
\begin{subfigure}{0.47\columnwidth}
    \centering
    \includegraphics[clip=true, trim=0mm 0mm 0mm 0mm, width=1.0\columnwidth,height=0.85\columnwidth]{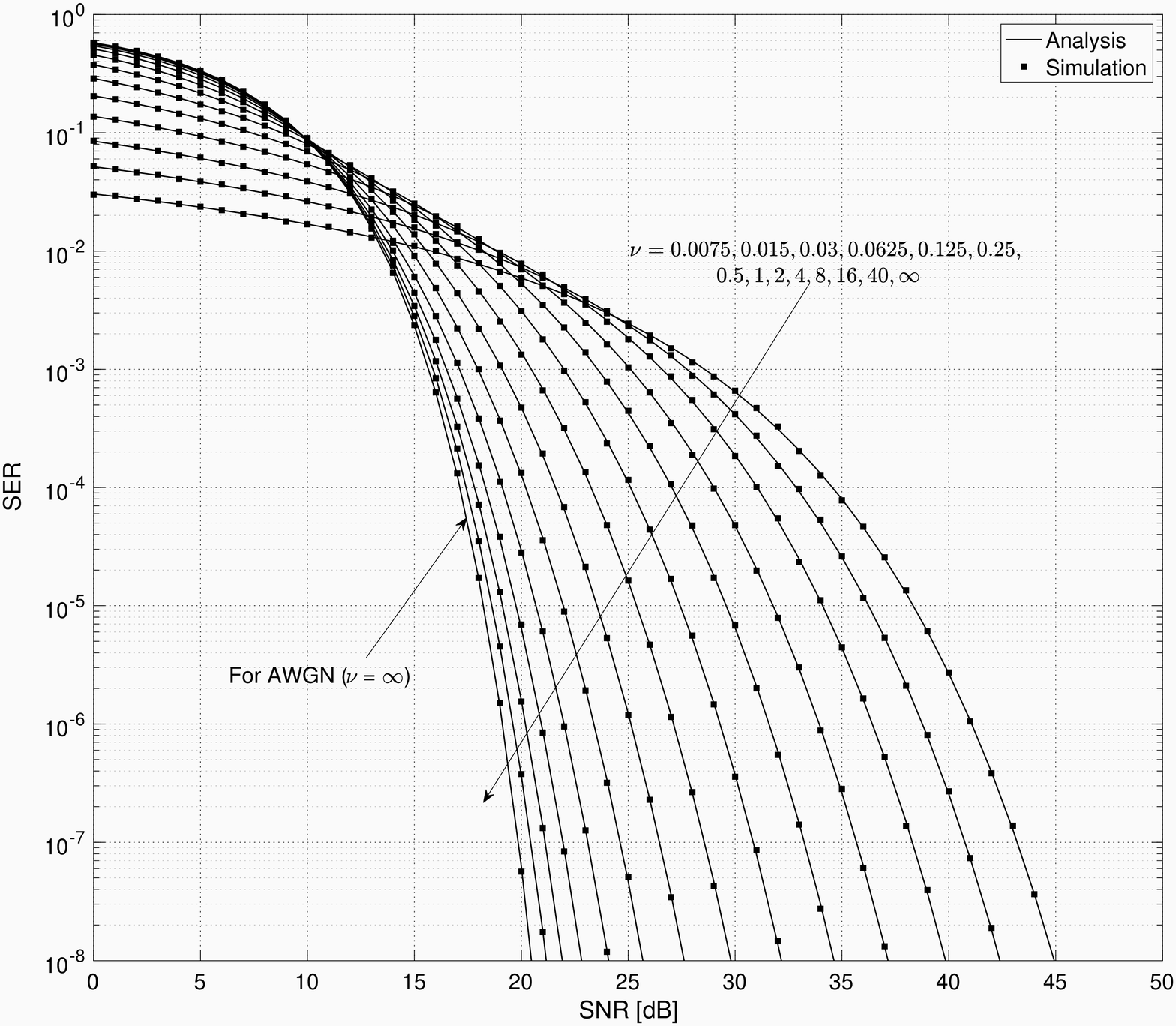}
    \caption{Modulation level $M=8$.}
    \label{Figure:ConditionalSEPForMPSKC}
\end{subfigure}
{~~~}
\begin{subfigure}{0.47\columnwidth}
    \centering
    \includegraphics[clip=true, trim=0mm 0mm 0mm 0mm, width=1.0\columnwidth,height=0.85\columnwidth]{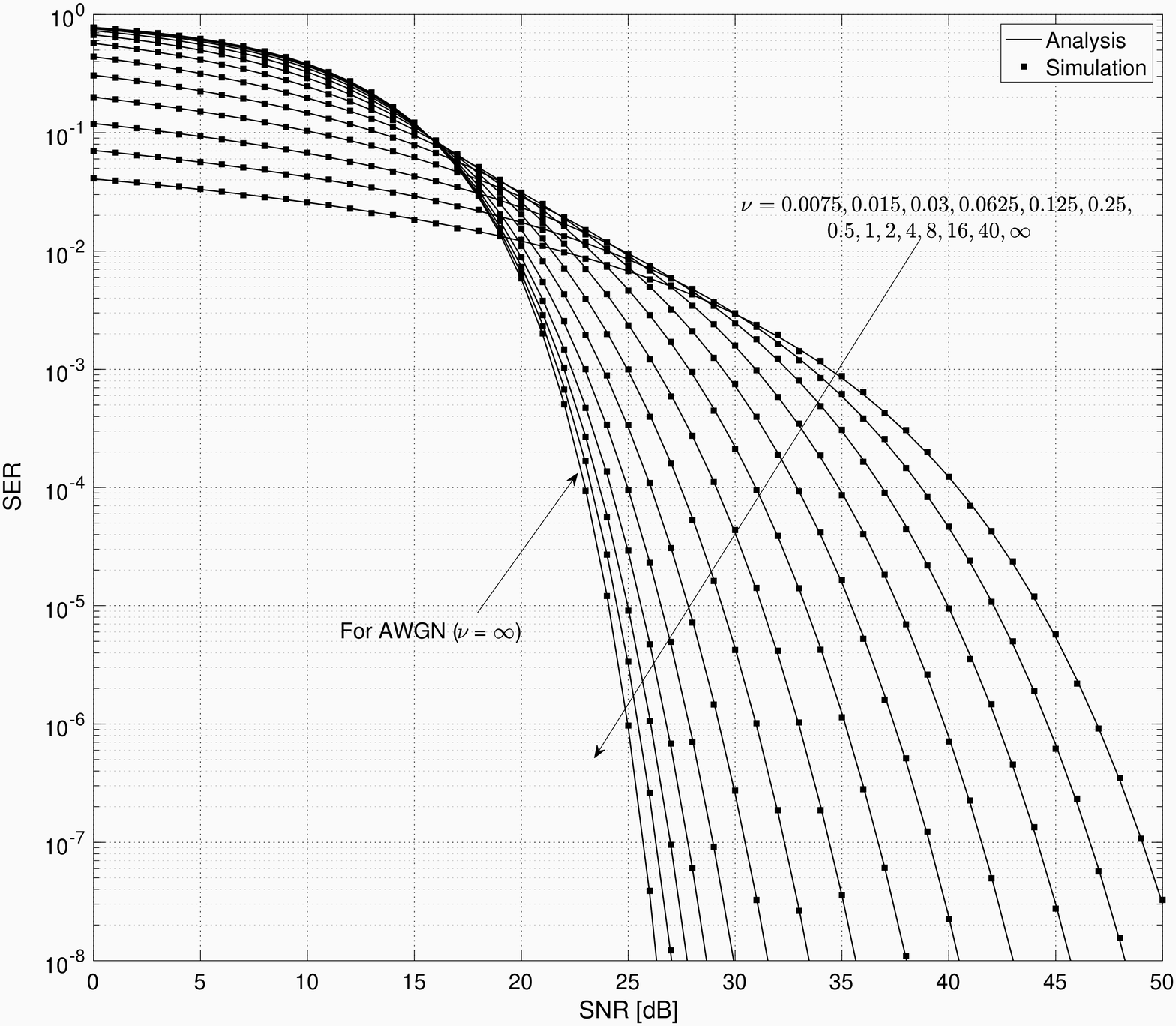}
    \caption{Modulation level $M=16$.}
    \label{Figure:ConditionalSEPForMPSKD}
\end{subfigure}
\caption{The \ac{SER} of \ac{M-PSK} modulation over \ac{AWMN} channels.}
\label{Figure:ConditionalSEPForMPSK}
\vspace{-2mm} 
\end{figure*} 

\begin{theorem}\label{Theorem:MLDecisionErrorProbabilityForMPSKCoherentSignallingOverAWMNChannels}
For the \ac{ML} decision rule, the conditional \ac{SER} of the rectangular \ac{M-PSK} signaling is given by
\begin{equation}\label{Eq:AWMNPSKSignallingConditionalSEP}
    \Pr\{\bigl.e\,|\,H\}=1-\int_{-\pi/M}^{\pi/M}f_{{\Theta}_{c}}(\theta)d\theta.
\end{equation}
\end{theorem}

\begin{proof}
The proof is obvious using \eqref{Eq:AWMNPSKSignallingConditionalSEPExpressionB} with the marginal \ac{PDF} of $\Theta_{c}$ given in \eqref{Eq:AWMNPSKSignallingPhasePDF} above. 
\end{proof}

\begin{figure}[tp]
    \centering
    \includegraphics[clip=true, trim=0mm 0mm 0mm 0mm,width=0.7\columnwidth,keepaspectratio=true]{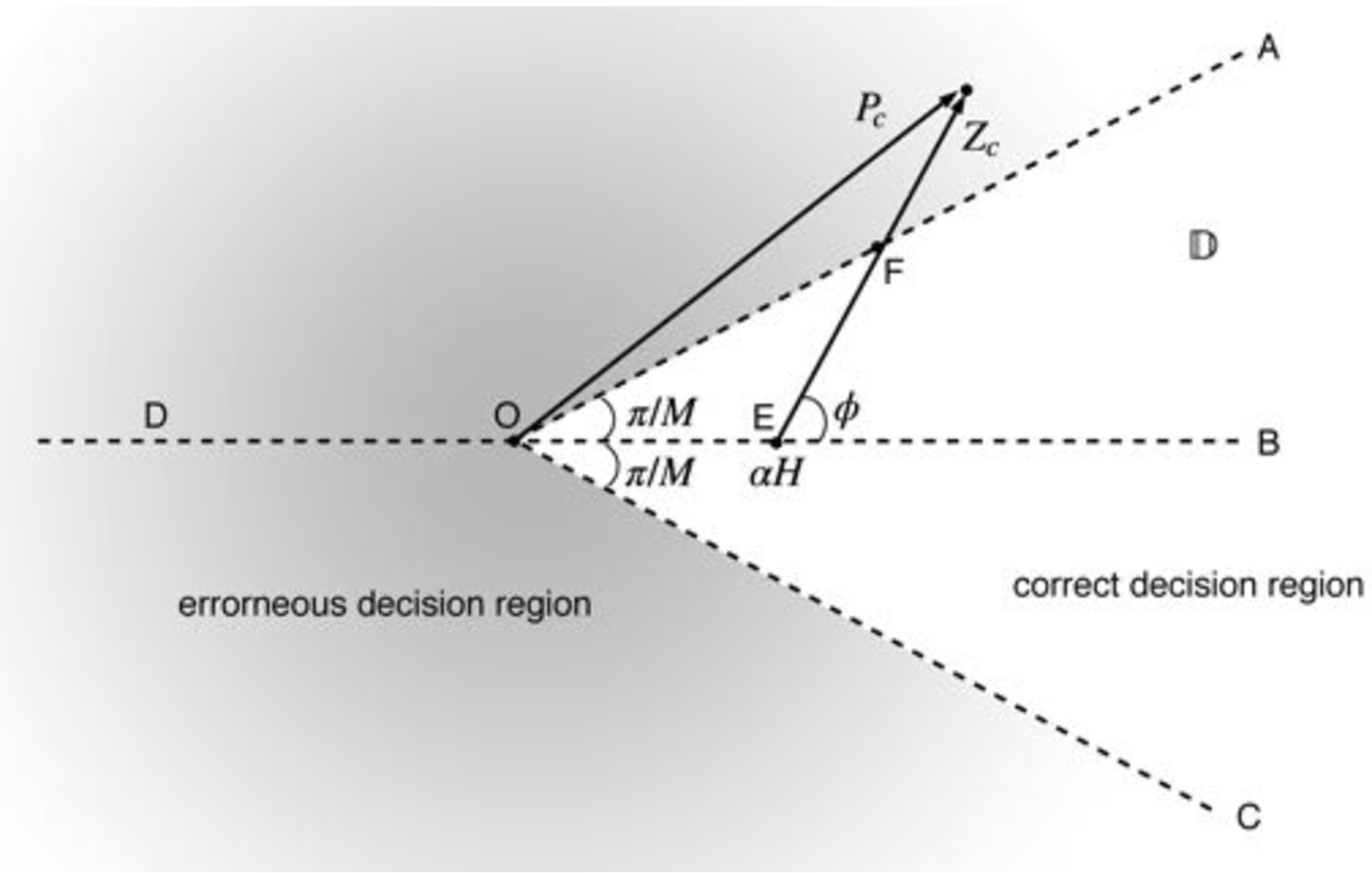}
    \caption{Received vector representation of the \ac{M-PSK} signaling whose projection model is given by \eqref{Eq:AWMNPSKSignallingProjection} with the decision region $\mathbb{D}\!=\!\{z\in\mathbb{C}|-\pi/M<\arg({z})<\pi/M\}$.}
    \label{Figure:DecisionRegionsForMPSKEquiprobableSignals}
    \vspace{-2mm}  
\end{figure} 

A closed-form expression to \eqref{Eq:AWMNPSKSignallingConditionalSEP} does not exist for $M\!>\!4$, and therefore the exact value of $\Pr\{\bigl.e\,|\,H\}$ must be calculated numerically and of course can be accurately approximated using Chebyshev-Gauss quadrature formula\cite[Eq. (25.4.39)]{BibAbramowitzStegunBook}. The other approach, which is similar to the one followed in \cite{BibCraigMILCOM1991}, to find the conditional \ac{SER} of \ac{M-PSK} constellation is to integrate the \ac{PDF} of ${Z}_{c}\!\sim\!\mathcal{CM}_{\nu}(0,{N_{0}}/{2})$ over the region of $\mathbb{D}\!=\!\{z\in\mathbb{C}\,|-\pi/M\!<\!\arg({z})\!<\!\pi/M\}$ and as presented in the following.

\begin{theorem}\label{Theorem:MLDecisionErrorProbabilityForMPSKCoherentSignallingOverAWMNChannelsII}
For the \ac{ML} decision rule, the conditional \ac{SER} of the \ac{M-PSK} signaling is given by
\begin{equation}\label{Eq:MLDecisionErrorProbabilityForMPSKCoherentSignallingOverAWMNChannelsII}
    \Pr\{\bigl.e\,|\,H\}=
        Q_{\nu}\!\left(\sqrt{2\gamma}\sin\Bigl(\frac{\pi}{M}\Bigr),\pi-\frac{\pi}{M}\right),
\end{equation}
where $\gamma$ is the instantaneous \ac{SNR} defined in \eqref{Eq:SignalToNoiseDefinition}. 
\end{theorem}
\begin{proof}
Referring to \eqref{Eq:AWMNPSKSignallingProjection}, we have ${Z}_{c}\!\sim\!\mathcal{CM}_{\nu}(0,{N_{0}}/{2})$~with the decomposition ${Z}_{c}\!=\!\allowbreak{X}_{c}+\imaginary{Y}_{c}$ in Cartesian~form,~where 
${X}_{c}\!\sim\!\mathcal{M}_{\nu}(0,{N_{0}}/{2})$ and ${Y}_{c}\!\sim\!\mathcal{M}_{\nu}(0,{N_{0}}/{2})$. Further,~we also have the Euler's form  ${Z}_{c}\!=\!{A}_{c}\exp\bigl(\imaginary{\Phi}_{c}\bigr)$ in polar form, where we express  ${A}_{c}\!=\!\sqrt{{X}^2_{c}+{Y}^2_{c}}$ and $\Phi_{c}\!=\!\arctan({Y}_{c}/{X}_{c})$. Using  \eqref{Eq:AWMNPSKSignallingAdditiveNoisePDF}, we obtain the joint \ac{PDF} of ${A}_{c}$ and $\Phi_{c}$ as  
\begin{subequations}\label{Eq:AWMNPSKSignallingAdditiveNoisePolarExpressionPDF}
\begin{eqnarray}
    \label{Eq:AWMNPSKSignallingAdditiveNoisePolarExpressionPDFA}
    f_{{A}_{c},{\Phi}_{c}}(a,\phi)
    &=&f_{{Z}_{c}}(z)\,J_{{Z}_{c}|{A}_{c},{\Phi}_{c}}, \\
    \label{Eq:AWMNPSKSignallingAdditiveNoisePolarExpressionPDFB}
    &=&f_{{Z}_{c}}(a\exp(\imaginary\phi))\,J_{{Z}_{c}|{A}_{c},{\Phi}_{c}},
\end{eqnarray}
\end{subequations}
where $f_{{Z}_{c}}(z)$ denotes the \ac{PDF} of ${Z}_{c}$, given in \eqref{Eq:AWMNPSKSignallingAdditiveNoisePDF}, and  $J_{{Z}_{c}|{A}_{c},{\Phi}_{c}}$ denotes the Jacobian of $z\!=\!a\exp\bigl(\imaginary\phi\bigr)$ and is derived~as
$J_{{Z}_{c}|{A}_{c},{\Phi}_{c}}\!=\!a$, whose replacement in \eqref{Eq:AWMNPSKSignallingAdditiveNoisePolarExpressionPDF} yields
\begin{equation}\label{Eq:AWMNPSKSignallingAdditiveNoisePolarPDF}
    f_{{A}_{c},{\Phi}_{c}}(a,\phi)=\frac{2}{\pi}
		\frac{a^{\nu}}{\Gamma(\nu)\,\Lambda_{0}^{\nu+1}}
			\BesselK[\nu-1]{\frac{2a}{\Lambda_{0}}}, 
\end{equation}
which is defined over $a\!\in\!\mathbb{R}_{+}$ and $\theta\!\in\![-\pi,\pi)$.
We notice that, as depicted in \figref{Figure:DecisionRegionsForMPSKEquiprobableSignals}, a decision error occurs if ${Z}_{c}$ falls into the erroneous decision region. Then, we write the conditional \ac{SER} of \ac{M-PSK} constellation as 
\begin{equation}\label{Eq:AWMNPSKSignallingConditionalSEPIntegral}
\Pr\bigl\{\bigl.e\,|\,H\bigr\}=
    2\int_{{\pi}/{M}}^{\pi}
        \int_{\abs{\text{EF}}}^{\infty}
            f_{{A}_{c},{\Phi}_{c}}(a,\phi)        
                \,{d}\phi\,{d}a,
\end{equation}
where $\abs{\text{EF}}$ is the distance between the modulation symbol (i.e., point E)  and the boundary~point~(i.e.,~point~F).~The~length $\abs{\text{EF}}$ is written from $H\alpha\sin(\pi/M)\!=\!\abs{\text{EF}}\sin(\phi-\pi/M)$ as 
\begin{equation}\label{Eq:AWMNPSKSignallingBoundaryDistance}
    \abs{\text{EF}}=
        2\gamma\biggl(\frac{\Lambda_{0}}{\lambda_{0}}\biggr)^{2}
            \frac{\sin(\pi/M)}{\sin(\phi-\pi/M)},
\end{equation}
with $\Lambda_{0}\!=\!\sqrt{N_{0}/\nu}$,  $\lambda_{0}\!=\!\sqrt{2/\nu}$, and  $\gamma\!=\!H^2E_{\defrvec{S}}/N_{0}$.~Conse\-quently, substituting \eqref{Eq:AWMNPSKSignallingBoundaryDistance} into \eqref{Eq:AWMNPSKSignallingConditionalSEPIntegral} and sequentially using \cite[Eqs. (8.2.2/8), (2.24.2/3) and (8.4.23/1)]{BibPrudnikovBookVol3}, we obtain 
\ifCLASSOPTIONtwocolumn
\begin{multline}\label{Eq:AWMNPSKSignallingConditionalSEPIntegralII}
\Pr\bigl\{\bigl.e\,|\,H\bigr\}=
    \frac{2^{1-\nu}}{\pi\Gamma(\nu)}
    \int_{{\pi}/{M}}^{\pi}
            {\biggl(\frac{2\sqrt{2\gamma}\sin(\pi/M)}{\lambda_{0}\sin(\phi-\pi/M)}\biggr)}^{\nu}\\
            \times
            \BesselK[\nu]{\frac{2\sqrt{2\gamma}\sin(\pi/M)}{\lambda_{0}\sin(\phi-\pi/M)}}
  		        {d}\phi,
\end{multline}
\else
\begin{equation}\label{Eq:AWMNPSKSignallingConditionalSEPIntegralII}
\Pr\bigl\{\bigl.e\,|\,H\bigr\}=
    \frac{2^{1-\nu}}{\pi\Gamma(\nu)}
    \int_{{\pi}/{M}}^{\pi}
            {\biggl(\frac{2\sqrt{2\gamma}\sin(\pi/M)}{\lambda_{0}\sin(\phi-\pi/M)}\biggr)}^{\nu}
            \BesselK[\nu]{\frac{2\sqrt{2\gamma}\sin(\pi/M)}{\lambda_{0}\sin(\phi-\pi/M)}}
  		        {d}\phi,
\end{equation}
\fi
where applying the change of variable $\theta\!=\!\phi-\pi/M$ and using \defref{Definition:McLeishPartialQFunction} yields \eqref{Eq:MLDecisionErrorProbabilityForMPSKCoherentSignallingOverAWMNChannelsII}, which proves \theoremref{Theorem:MLDecisionErrorProbabilityForMPSKCoherentSignallingOverAWMNChannelsII}. 
\end{proof}

Let us now consider some special cases for the closed-form conditional \ac{SER} of the \ac{M-PSK} signaling. The \ac{BPSK} constellation is the most reliable modulation as a special case of the \ac{M-PSK} constellation. Accordingly, setting $M\!=\!2$ in \eqref{Eq:MLDecisionErrorProbabilityForMPSKCoherentSignallingOverAWMNChannelsII} and utilizing the property $Q_{\nu}\bigl(x\bigr)=Q_{\nu}\bigl(x,{\pi}/{2}\bigr)$, we obtain the conditional \ac{SER} of \ac{BPSK} constellation as follows
\begin{subequations}
\begin{eqnarray}
    \Pr\{\bigl.e\,|\,H\}
        &=&Q_{\nu}\bigl(\sqrt{2\gamma},{\pi}/{2}\bigr),\\
        &=&Q_{\nu}\bigl(\sqrt{2\gamma}\bigr),  
\end{eqnarray}
\end{subequations}
which is perfect agreement with \eqref{Eq:ConditionalBEPForBPSK}. Further, setting $M\!=\!4$ in \eqref{Eq:MLDecisionErrorProbabilityForMPSKCoherentSignallingOverAWMNChannelsII}, we obtain the conditional \ac{SER} of \ac{QPSK} (i.e,. 4-QAM) constellation, that is 
\begin{subequations}
\begin{eqnarray}
\!\!\Pr\{\bigl.e\,|\,H\}
    &=&\Bigl.Q_{\nu}\bigl(\sqrt{2\gamma}\sin\bigl({\pi}/{M}\bigr),\pi-{\pi}/{M}\bigr)\Bigl|_{M=4},{~~~~~~~~~~}\\
    &=&Q_{\nu}\bigl(\sqrt{\gamma},3{\pi}/{4}\bigr),
\end{eqnarray}
\end{subequations}
which is in agreement with \eqref{Eq:MLDecisionErrorProbabilityFor4QAMCoherentSignallingOverAWMNChannelsII} as expected.

For the analysis of impulsive noise effects on the performance, we demonstrate in \figref{Figure:ConditionalSEPForMPSK} how the conditional \ac{SER} of \ac{M-PSK} signaling over complex  \ac{AWMN} channels varies with respect to the \ac{SNR}, the normality $\nu$ and the modulation level $M$, and notice that  numerical and simulation-based results are in perfect agreement. Further, we have observed previously obtained results. As such, the impulsive nature of the additive noise increases (i.e., the normality $\nu$ decreases), the performance deteriorates in high-\ac{SNR} regime while negligibly improves in low-\ac{SNR} regime. 

\subsection{Non-Coherent Signalling} 
\label{Section:SignallingOverAWMNChannels:NonCoherentSignalling}
In the previous subsection, we have investigated the coherent signaling in which the receiver has perfect knowledge about the received carrier phase. Detection techniques based on the absence of any knowledge about the received carrier phase are referred to as non-coherent detection techniques \cite{BibAlouiniBook,BibProakisBook,BibGoldsmithBook}. In the following, we consider the \ac{MAP} and \ac{ML} detection rules for non-coherent signaling in which the receiver does not have any information about both the transmitted modulation symbols and the carrier phase, and we obtain the \ac{SER} performance of non-coherent signaling. With the aid of the mathematical model, which is given in \eqref{Eq:ComplexAWMNVectorChannel}, we can re-express the received vector as $\defrmat{R}={H}{e}^{\imaginary\Theta}\defmat{F}\defvec{S}+\defrvec{Z}$, where the variables are well-explained immediately after \eqref{Eq:ComplexAWMNVectorChannel}. In needing to re-explain these variables, ${H}$ denotes the fading envelope following a non-negative random distribution, $\Theta$ denotes the fading phase uniformly distributed over $[0,2\pi]$. Further, both ${H}$ and $\Theta$ are assumed constant due to channel coherence \cite{BibProakisBook,BibGoldsmithBook,BibAlouiniBook}. Further, $\defvec{S}$ denotes the modulation symbol vector randomly chosen from the fixed set of modulation~symbols $\{\defvec{s}_{1},\defvec{s}_{2},\ldots,\defvec{s}_{M}\}$ according to the probabilities given by
\begin{equation}
    p_{m}=\Pr\{\defrvec{S}=\defvec{s}_{m}\},\text{~for~all~}{1}\leq{m}\leq{M},
\end{equation}
 such that $\sum^{M}_{m=1}p_{m}\!=\!1$. For non-coherent \ac{DPSK} signaling \cite{BibProakisBook,BibGoldsmithBook}, the modulation symbols $\defvec{s}_{1}$, $\defvec{s}_{2}$, $\ldots$, $\defvec{s}_{M}$ are not required to be orthogonal with each other, i.e., 
\begin{subequations}
    \setlength\arraycolsep{1.4pt}
    \begin{eqnarray}
        \defvec{s}^H_{m}\defvec{s}_{n}&\neq&{0},\text{~for~}{m}\neq{n}\\
        \defvec{s}^H_{m}\defvec{s}_{m}&=&E{m},
    \end{eqnarray}
\end{subequations}
where $E_{m}$ is the energy of the modulation symbol $m$. Without loss of generality, we assume that the energy of the modulation symbols are ordered, i.e., ${E}_{1}\leq{E}_{2}\leq\ldots\leq{E}_{M}$.
During each modulation symbol, the received vector $\defrvec{R}$ depends statistically on $\defrvec{S}$ and $\Theta$ with the conditional \ac{PDF} $f_{\defrvec{R}|\defrvec{S},\Theta}(\defvec{r}|\defvec{s},\theta)$, that is
\begin{equation}\label{Eq:NonCoherentAWMNReceivedVectorThetaConditionalPDF}
f_{\defrvec{R}|\defrvec{S},\Theta}(\defvec{r}|\defvec{s},\theta)=
    \frac{2}{\pi^L}
	\frac{{\lVert{\defvec{r}-{H}{e}^{\imaginary\Theta}\defmat{F}\defvec{s}}\rVert}_{\defmat{\Sigma}}^{\nu-L}}
		{\Gamma(\nu)\det(\defmat{\Sigma})\,\lambda_{0}^{\nu+L}}
		    {K}_{\nu-L}
			    \Bigl(
					\frac{2}{\lambda_{0}}
				{\bigl\lVert{\defvec{r}-{H}{e}^{\imaginary\Theta}\defmat{F}\defvec{s}}\bigr\rVert}_{\defmat{\Sigma}}
	        	\Bigr).\!\!
\end{equation}
The \ac{PDF} of the received vector $\defrvec{R}$ conditioned on the modulation symbols $\defrvec{S}$, i.e., $f_{\defrvec{R}|\defrvec{S}}(\defvec{r}|\defvec{s})$ is written as 
\begin{equation}
f_{\defrvec{R}|\defrvec{S}}(\defvec{r}|\defvec{s})=
\int_{0}^{2\pi}\!f_{\defrvec{R}|\defrvec{S},\Theta}(\defvec{r}|\defvec{s},\theta)f_{\Theta}(\theta)d\theta.    
\end{equation}
Since $\Theta$ is, without loss of generality, assumed uniformly distributed, $f_{\defrvec{R}|\defrvec{S}}(\defvec{r}|\defvec{s})$ is rewritten as 
\begin{equation}
f_{\defrvec{R}|\defrvec{S}}(\defvec{r}|\defvec{s})=
\frac{1}{2\pi}\int_{0}^{2\pi}\!f_{\defrvec{R}|\defrvec{S},\Theta}(\defvec{r}|\defvec{s},\theta)d\theta.    
\end{equation}
Thus, the joint \ac{PDF} of $\defrvec{R}$ and $\defrvec{S}$ is written as $f_{\defrvec{R},\defrvec{S}}(\defvec{r},\defvec{s})\!=\!f_{\defrvec{R}|\defrvec{S}}(\defvec{r}|\defvec{s})f_{\defrvec{S}}(\defvec{s})$, where $f_{\defrvec{S}}(\defvec{s})$ is given by \eqref{Eq:AWMNModulationSymbolVectorPMF}. In the receiver, the optimal detector without knowledge of the fading phase $\Theta$ observes the received vector $\defrvec{R}$ and produces the index of the most probable transmitted modulation symbol that maximizes $f_{\defrvec{R},\defrvec{S}}(\defvec{r},\defvec{s})$, that is 
\begin{subequations}\label{Eq:NoncoherentAWMNMAPRule}
\setlength\arraycolsep{1.4pt}
\begin{eqnarray}
    \label{Eq:NoncoherentAWMNMAPRuleA}
    \widehat{m}&=&
        \argmax_{{1}\leq{m}\leq{M}}\,
            f_{\defrvec{R},\defrvec{S}}(\defrvec{R},\defvec{s}_{m}),\\
    \label{Eq:NoncoherentAWMNMAPRuleB}
        &=&
        \argmax_{{1}\leq{m}\leq{M}}\,
        f_{\defrvec{R}|\defrvec{S}}(\defrvec{R}|\defvec{s}_{m})
            \Pr\{\defrvec{S}=\defvec{s}_{m}\},\\
    \label{Eq:NoncoherentAWMNMAPRuleC}
        &=&
        \argmax_{{1}\leq{m}\leq{M}}\,
        \frac{p_{m}}{2\pi}\int_{0}^{2\pi}\!\!\frac{2}{\pi^L}
	    \frac{{\lVert{\defvec{r}-{H}{e}^{\imaginary\Theta}\defmat{F}\defvec{s}_m}\rVert}_{\defmat{\Sigma}}^{\nu-L}}
		    {\Gamma(\nu)\det(\defmat{\Sigma})\,\lambda_{0}^{\nu+L}}
			  {K}_{\nu-L}
			    \Bigl(
					\frac{2}{\lambda_{0}}
				        {\bigl\lVert{\defvec{r}-{H}{e}^{\imaginary\Theta}\defmat{F}\defvec{s}_m}\bigr\rVert}_{\defmat{\Sigma}}
	        	\Bigr)d\theta,{~~~~~~~~~}
\end{eqnarray}
\end{subequations}
which means that if the transmitted symbol $m$ and the optimally detected symbol $\widehat{m}$ are not the same, a decision error occurs with the probability $\Pr\{e\}\!=\!\Pr\{\widehat{m}\neq{m}\}$. We can even simplify \eqref{Eq:NoncoherentAWMNMAPRuleC} more as shown in the following. 

\begin{theorem}\label{Theorem:NoncoherentMAPDecisionRuleForComplexAWMNVectorChannel}
For the complex vector channel introduced in \eqref{Eq:ComplexAWMNVectorChannel}, the non-coherent \ac{MAP} detection rule is given by 
\begin{equation}\label{Eq:NoncoherentMAPDecisionRuleForComplexAWMNVectorChannel}
\!\!\widehat{m}=\argmax_{{1}\leq{m}\leq{M}}\,
        {p}_{m}
        \exp\Bigl(\frac{1}{2}{H}^{2}\defvec{s}_{m}^H\defmat{F}^{H}\defmat{\Sigma}^{-1}\defmat{F}\defvec{s}_{m}\Bigr)
        \BesselI[0]{H\bigl|\defvec{s}_{m}^H\defmat{F}^{H}\defmat{\Sigma}^{-1}\defmat{F}\defrvec{R}\bigr|},\!\!
\end{equation}
where $\BesselI[0]{\cdot}$ is the modified Bessel function of the first kind of zero order \emph{\cite[Eq. (8.406/3)]{BibGradshteynRyzhikBook}},\emph{\cite[Eq. ( 03.02.02.0001.01)]{BibWolfram2010Book}}.
\end{theorem}

\begin{proof}
In the mathematical channel model given by \eqref{Eq:ComplexAWMNVectorChannel}, the vector $\defrmat{R}$ received during the transmission of the modulation symbol $\defvec{s}_{m}$ will have a multivariate \ac{CES} McLeish distribution, i.e., 
$\defrmat{R}\!\sim\!\mathcal{CM}_{\nu}^L\!\bigl({H}{e}^{\imaginary\Theta}\defvec{s}_{m},\defmat{\Sigma}\bigr)$. Using both \eqref{Eq:MultivariateComplexMcLeishComponentDecomposition} and \eqref{Eq:MultivariateCESMcLeishDecomposition}, we decompose the vector $\defrmat{R}$ given the symbol $\defrmat{S}$ as follows 
\begin{equation}
(\defrvec{R}|\defrvec{S})=
    {H}{e}^{\imaginary\Theta}\defmat{F}\defvec{s}_{m}+
        \sqrt{G}\,\defmat{D}\,(\defrvec{N}_1+\imaginary\defrvec{N}_2),
\end{equation}
where $\defmat{\Sigma}\!=\!\defmat{D}\defmat{D}^{H}$, $\defrvec{N}_1\!\sim\!\mathcal{N}^L(0,\defmat{I})$, $\defrvec{N}_2\!\sim\!\mathcal{N}^L(0,\defmat{I})$
and $G\!\sim\!\mathcal{G}(\nu,1)$. Further, $\defrvec{N}_1$ and $\defrvec{N}_2$ are mutually independent. Accordingly, the \ac{PDF} of $\defrmat{R}$ conditioned on both $\defrvec{S}$ and $G$, i.e., $f_{\defrmat{R}|\defrmat{S},G}(\defvec{z}|\defvec{s},g)$ can be written as
\begin{equation}\nonumber
f_{\defrmat{R}|\defrmat{S},G}(\defvec{r}|\defvec{s},g)=
    \frac{1}{2\pi}\int_{0}^{2\pi}
	\frac{\exp\bigl(-\frac{1}{2g}{\lVert\defvec{r}-{H}{e}^{\imaginary\theta}\defmat{F}\defvec{s}\rVert}^{2}_\defmat{\Sigma}\bigr)}
	    {(2\pi)^{L}g^{L}\det(\defmat{\Sigma})}
	        d\theta,
\end{equation} 
with the aid of which the conditional \ac{PDF} $f_{\defrvec{R}|\defrvec{S}}(\defrvec{R}|\defvec{s})$ is easily obtained by $f_{\defrvec{R}|\defrvec{S}}(\defrvec{R}|\defvec{s})=\allowbreak\int_{0}^{\infty}f_{\defrvec{R}|\defrvec{S},G}(\defrvec{R}|\defvec{s},g)\,\allowbreak{}f_{G}(g)\,dg$, where $f_{G}(g)$ is the \ac{PDF} of $G\!\sim\!\mathcal{G}(\nu,1)$, and given in \eqref{Eq:ProportionPDF}. Upon substituting $f_{\defrvec{R}|\defrvec{S}}(\defrvec{R}|\defvec{s}_{m})$ into \eqref{Eq:AWMNMAPRuleII}, the rule is rewritten as
\begin{subequations}\label{Eq:NoncoherentMAPDecisionDerivation}
\setlength\arraycolsep{1.4pt}
\begin{eqnarray}
    \label{Eq:NoncoherentMAPDecisionDerivationA}
    \!\!\widehat{m}
        &\overset{(a)}{=}&\argmax_{{1}\leq{m}\leq{M}}\,
            {p}_{m}\int_{0}^{\infty}               
                f_{\defrvec{R}|\defrvec{S},G}
                    (\defrvec{R}|\defvec{s}_{m},g)f_{G}(g)\,dg,\quad\quad\\
    \label{Eq:NoncoherentMAPDecisionDerivationB}                
        &\overset{(b)}{=}&\argmax_{{1}\leq{m}\leq{M}}\,
            {p}_{m}f_{\defrvec{R}|\defrvec{S},G}
                (\defrvec{R}|\defvec{s}_{m},\mathbb{E}[G]),   
\end{eqnarray}
\end{subequations}
where the following steps are used. In step $(a)$, we observe that \eqref{Eq:MAPDecisionConditionalPDF} is being averaged by the \ac{PDF} $f_{G}(g)$, and notice that $f_{G}(g)\!\geq\!{0}$ for all $g\in\mathbb{R}_{+}$, which simplifies \eqref{Eq:MAPDecisionDerivationA} to \eqref{Eq:MAPDecisionDerivationB} with $\mathbb{E}[G]\!=\!{1}$. In step $(b)$, we insert \eqref{Eq:MAPDecisionConditionalPDF} into \eqref{Eq:MAPDecisionDerivationB} and drop all the positive constant terms. Then, we obtain 
\begin{equation}\label{Eq:NoncoherentMAPDecisionDerivationII}
\!\!\!\widehat{m}=\argmax_{{1}\leq{m}\leq{M}}
        \frac{{p}_{m}}{2\pi}\!\int_{0}^{2\pi}\!\!\!\exp\Bigl(
                -\frac{1}{2}{\bigl\lVert
                    \defrvec{R}-{H}{e}^{\imaginary\theta}\defmat{F}\defvec{s}_{m}
                        \bigr\rVert}^{2}_\defmat{\Sigma}\Bigr)d\theta,\!
\end{equation}
where ${\lVert{\defrvec{R}-{H}{e}^{\imaginary\theta}\defmat{F}\defvec{s}_{m}}\rVert}^{2}_{\defmat{\Sigma}}$ can be decomposed as 
\ifCLASSOPTIONtwocolumn
\begin{multline}\label{Eq:NoncoherentMAPMahalanobisDistanceExpansion}
\!\!\!\!\!\!{\bigl\lVert{
        \defrvec{R}-{H}{e}^{\imaginary\theta}\defmat{F}\defvec{s}_{m}
    }\bigr\rVert}^{2}_{\defmat{\Sigma}}
    ={H}^{2}\defvec{s}_{m}^H\defmat{F}^{H}\defmat{\Sigma}^{-1}\defmat{F}\defvec{s}_{m}\\
        -2H\Re\bigl\{{e}^{-\imaginary\theta}\defvec{s}_{m}^H\defmat{F}^{H}\defmat{\Sigma}^{-1}\defmat{F}\defrvec{R}\bigr\}
        +\defrvec{R}^H\defmat{\Sigma}^{-1}\defrvec{R}.\!\!
\end{multline}
\else
\begin{equation}\label{Eq:NoncoherentMAPMahalanobisDistanceExpansion}
\!\!\!\!\!\!{\bigl\lVert{
        \defrvec{R}-{H}{e}^{\imaginary\theta}\defmat{F}\defvec{s}_{m}
    }\bigr\rVert}^{2}_{\defmat{\Sigma}}
    ={H}^{2}\defvec{s}_{m}^H\defmat{F}^{H}\defmat{\Sigma}^{-1}\defmat{F}\defvec{s}_{m}
        -2H\Re\bigl\{{e}^{-\imaginary\theta}\defvec{s}_{m}^H\defmat{F}^{H}\defmat{\Sigma}^{-1}\defmat{F}\defrvec{R}\bigr\}
        +\defrvec{R}^H\defmat{\Sigma}^{-1}\defrvec{R}.\!\!
\end{equation}
\fi
Putting \eqref{Eq:NoncoherentMAPMahalanobisDistanceExpansion} into \eqref{Eq:NoncoherentMAPDecisionDerivationII} and ignoring the term $\defrvec{R}^H\defmat{\Sigma}^{-1}\defrvec{R}$ since not depending on the modulation index $m$ yields 
\ifCLASSOPTIONtwocolumn
\begin{multline}\label{Eq:NoncoherentMAPDecisionDerivationIII}
\!\!\!\!\widehat{m}=\argmax_{{1}\leq{m}\leq{M}}\,
            \frac{{p}_{m}}{2\pi}
            \exp\Bigl(\frac{1}{2}{H}^{2}\defvec{s}_{m}^H\defmat{F}^{H}\defmat{\Sigma}^{-1}\defmat{F}\defvec{s}_{m}\Bigr)\\
        {~~~~~}\times\int_{0}^{2\pi}\!\!
            \exp\Bigl(H
            \bigl|\defvec{s}_{m}^H\defmat{F}^{H}\defmat{\Sigma}^{-1}\defmat{F}\defrvec{R}\bigr|\cos(\phi-\theta)
            \Bigr)\,d\theta,\!\!\!\!
\end{multline}
\else
\begin{equation}\label{Eq:NoncoherentMAPDecisionDerivationIII}
    \widehat{m}=\argmax_{{1}\leq{m}\leq{M}}\,
            \frac{{p}_{m}}{2\pi}
            \exp\Bigl(\frac{1}{2}{H}^{2}\defvec{s}_{m}^H\defmat{F}^{H}\defmat{\Sigma}^{-1}\defmat{F}\defvec{s}_{m}\Bigr)
        \int_{0}^{2\pi}\!\!
            \exp\Bigl(H
            \bigl|\defvec{s}_{m}^H\defmat{F}^{H}\defmat{\Sigma}^{-1}\defmat{F}\defrvec{R}\bigr|\cos(\phi-\theta)
            \Bigr)\,d\theta,
\end{equation}
\fi
where $\phi$ denotes the phase of $\defvec{s}_{m}^H\defmat{F}^{H}\defmat{\Sigma}^{-1}\defmat{F}\defrvec{R}$. Notice that the integration in \eqref{Eq:NoncoherentMAPDecisionDerivationIII} is certainly a periodic function of $\phi$ with period $2\pi$, and thus $\phi$ has no effect on the result. Utilizing the equality  $\BesselI[0]{x}\!=\!\frac{1}{2\pi}\allowbreak\int_{0}^{2\pi}\exp\bigl(x\cos(\theta)\bigr)d\theta$ \cite[Eq. (8.431/3)]{BibGradshteynRyzhikBook},\cite[Eq. (03.02.07.0001.01)]{BibWolfram2010Book}, we readily obtain \eqref{Eq:NoncoherentMAPDecisionRuleForComplexAWMNVectorChannel}, which proves  \theoremref{Theorem:NoncoherentMAPDecisionRuleForComplexAWMNVectorChannel}.
\end{proof}
\ifCLASSOPTIONtwocolumn
\fi

Note that the decision rule given in \eqref{Eq:NoncoherentMAPDecisionRuleForComplexAWMNVectorChannel} cannot be made simpler. However, in the case of equiprobable modulation symbols, the non-coherent \ac{ML} rule is given in the following. 
\ifCLASSOPTIONtwocolumn
\fi

\begin{theorem}\label{Theorem:NoncoherentMLDecisionRuleForComplexAWMNVectorChannel}
For the complex vector channel introduced in \eqref{Eq:ComplexAWMNVectorChannel}, the non-coherent \ac{ML} detection rule is given by 
\ifCLASSOPTIONtwocolumn
\begin{multline}\label{Eq:NoncoherentMLDecisionRuleForComplexAWMNVectorChannel}
\!\!\widehat{m}=\argmax_{{1}\leq{m}\leq{M}}\,
        \exp\Bigl(\frac{1}{2}{H}^{2}\defvec{s}_{m}^H\defmat{F}^{H}\defmat{\Sigma}^{-1}\defmat{F}\defvec{s}_{m}\Bigr)\\\times
        \BesselI[0]{H\bigl|\defvec{s}_{m}^H\defmat{F}^{H}\defmat{\Sigma}^{-1}\defmat{F}\defrvec{R}\bigr|}.\!\!
\end{multline}
\else
\begin{equation}\label{Eq:NoncoherentMLDecisionRuleForComplexAWMNVectorChannel}
\!\!\widehat{m}=\argmax_{{1}\leq{m}\leq{M}}\,
        \exp\Bigl(\frac{1}{2}{H}^{2}\defvec{s}_{m}^H\defmat{F}^{H}\defmat{\Sigma}^{-1}\defmat{F}\defvec{s}_{m}\Bigr)
        \BesselI[0]{H\bigl|\defvec{s}_{m}^H\defmat{F}^{H}\defmat{\Sigma}^{-1}\defmat{F}\defrvec{R}\bigr|}.\!\!
\end{equation}
\fi
\end{theorem}

\begin{proof}
The proof is obvious using \theoremref{Theorem:NoncoherentMAPDecisionRuleForComplexAWMNVectorChannel}.
\end{proof}
\ifCLASSOPTIONtwocolumn
\fi

In order to avoid non-zero cross correlation between channels, we should choose the precoding matrix filter $\defmat{F}$ to maximize the power of the received signal. Then, referring to the mathematical model given by \eqref{Eq:ComplexAWMNVectorChannel}, the precoding matrix filter $\defmat{F}$ meets $\defmat{\Sigma}\!=\!\frac{N_{0}}{2}\defmat{F}\defmat{F}^{H}\!$, and the received vector equalized by $\defmat{F}$ before being fed to the optimal detection is given by
\begin{subequations}\label{Eq:NoncoherentComplexAWMNVectorChannelEqualization}
\begin{eqnarray}
    \label{Eq:NoncoherentComplexAWMNVectorChannelEqualizationA}
    \defrvec{R}_{nc}&=&\defmat{F}^{-1}\defrvec{R},\\
    \label{Eq:NoncoherentComplexAWMNVectorChannelEqualizationB}
    &=&\defmat{F}^{-1}
            \bigl({H}{e}^{\imaginary\Theta}\defmat{F}\defvec{S}+\defrmat{Z}\bigr),\\
    \label{Eq:NoncoherentComplexAWMNVectorChannelEqualizationC}
    &\equiv&H{e}^{\imaginary\Theta}\defvec{S}+\defmat{F}^{-1}\defrmat{Z},\\
    \label{Eq:NoncoherentComplexAWMNVectorChannelEqualizationD}
    &=&H{e}^{\imaginary\Theta}\defvec{S}+\defrvec{Z}_{nc},
\end{eqnarray}
\end{subequations}
where $\defrvec{Z}\!\sim\!\mathcal{CM}_{\nu}^{L}(\defmat{0},\defmat{\Sigma})$ whose \ac{PDF} is already given by \eqref{Eq:AWMNAdditiveNoiseVectorPDF}, and  $\defrvec{Z}_{nc}\!\sim\!\mathcal{CM}_{\nu}^{L}(\defmat{0},\frac{N_{0}}{2}\defmat{I})$ follows the \ac{PDF} obtained with the aid of both \theoremref{Theorem:INIDMultivariateCESMcLeishPDF} and the special case \eqref{Eq:INIDMultivariateCESMcLeishPDFWithUniformVariances}, that is
\begin{equation}
\!\!\!\!f_{\defrvec{Z}_{c}}(\defvec{z})=
		\frac{2}{\pi^{L}}
		\frac{{\bigl\lVert{\defvec{z}}\bigr\rVert}^{\nu-{L}}}
			{\Gamma(\nu)\Lambda_{0}^{\nu+{L}}}
				{K}_{\nu-{L}}\Bigl(\frac{2}{\Lambda_{0}}{\bigl\lVert{\defvec{z}}\bigr\rVert}\Bigr)
\end{equation}
with the component deviation factor $\Lambda_{0}\!=\!\sqrt{N_{0}/\nu}$ (i.e., $N_{0}/\nu$ variance per each \ac{CCS} Laplacian noise component). Further, the equalization, which is presented above in \eqref{Eq:NoncoherentComplexAWMNVectorChannelEqualization}, simplifies the complex correlated \ac{AWMN} vector channel the uncorrelated complex \ac{AWMN} vector channels, whose mathematical model is typically given by
\begin{equation}\label{Eq:NoncoherentComplexAWMNVectorChannel}
	\defrvec{R}_{nc}={H}{e}^{\imaginary\Theta}\defvec{S}+\defrvec{Z}_{nc}.
\end{equation}
where the knowledge of $\Theta$ is as mentioned above not available at the receiver. The power of the modulation symbol $m$, which is denoted by $E_{m}$, is given by $E_{m}\!=\!{\lVert\defvec{s}_{m}\rVert}^{2}\!=\!\defvec{s}^{H}_{m}\defvec{s}_{m}$ for all ${1}\!\leq\!{m}\!\leq\!{M}$. Thus, we write the average power of $\defrvec{S}$ as 
\begin{equation}
    E_{\defrvec{S}}=\sum_{m=1}^{M}\Pr\{\defrvec{S}=\defvec{s}_{m}\}{E}_{m}=\sum_{m=1}^{M}{p}_{m}{E}_{m}.
\end{equation}
Therefore, considering the all modulation symbols, the total SNR is written as 
\begin{equation}
    \gamma=\frac{H^2E_{\defrvec{S}}}{N_{0}}=\sum_{m=1}^{M}{p}_{m}\gamma_{m},
\end{equation}
where $\gamma_{m}$ is the instantaneous \ac{SNR} for the transmission of the modulation symbol $m$ and written as $\gamma_{m}={H^2E_{m}}/{N_0}$. In addition, note that, during each modulation symbol, the received vector $\defrvec{R}_{c}$ statistically depends on both $\defrvec{S}$ and $\Theta$ with the conditional \ac{PDF} $f_{\defrvec{R}_{nc}|\defrvec{S},\Theta}(\defvec{r}|\defvec{s},\theta)$, that is
\ifCLASSOPTIONtwocolumn
\begin{multline}\label{Eq:NoncoherentComplexAWMNVectorChannelConditionalPDF}
    \!\!\!\!f_{\defrvec{R}_{nc}|\defrvec{S},\Theta}(\defvec{r}|\defvec{s},\theta)=
		\frac{2}{\pi^{L}}
		\frac{{\bigl\lVert{\defvec{r}-H{e}^{\imaginary\Theta}\defvec{s}}\bigr\rVert}^{\nu-{L}\!\!\!}}
			{\Gamma(\nu)\Lambda_{0}^{\nu+{L}}}\\\times
				{K}_{\nu-{L}}\Bigl(\frac{2}{\Lambda_{0}}{\bigl\lVert{\defvec{r}-H{e}^{\imaginary\Theta}\defvec{s}}\bigr\rVert}\Bigr).\!\!
\end{multline}
\else
\begin{equation}\label{Eq:NoncoherentComplexAWMNVectorChannelConditionalPDF}
    \!\!\!\!f_{\defrvec{R}_{nc}|\defrvec{S},\Theta}(\defvec{r}|\defvec{s},\theta)=
		\frac{2}{\pi^{L}}
		\frac{{\bigl\lVert{\defvec{r}-H{e}^{\imaginary\Theta}\defvec{s}}\bigr\rVert}^{\nu-{L}\!\!\!}}
			{\Gamma(\nu)\Lambda_{0}^{\nu+{L}}}
				{K}_{\nu-{L}}\Bigl(\frac{2}{\Lambda_{0}}{\bigl\lVert{\defvec{r}-H{e}^{\imaginary\Theta}\defvec{s}}\bigr\rVert}\Bigr).\!\!
\end{equation}
\fi
Accordingly and correspondingly, the non-coherent \ac{MAP} decision rule is obtained for the uncorrelated complex \ac{AWMN} vector channels in the following. 

\begin{theorem}\label{Theorem:NoncoherentMAPDecisionRuleForPrecodedComplexAWMNVectorChannel}
For complex uncorrelated \ac{AWMN} vector channels, defined in \eqref{Eq:NoncoherentComplexAWMNVectorChannel}, the non-coherent~\ac{MAP}~rule~is~given~by
\ifCLASSOPTIONtwocolumn
\vspace{-3mm}
\fi
\begin{subequations}
\label{Eq:NoncoherentMAPDecisionRuleForPrecodedComplexAWMNVectorChannel}
\begin{eqnarray}
\label{Eq:NoncoherentMAPDecisionRuleForPrecodedComplexAWMNVectorChannelA}
\!\!\!\!\widehat{m}
    &=&\argmax_{{1}\leq{m}\leq{M}}\,
        {p}_{m}
        \exp\bigl(\gamma_{m}\bigr)
        \BesselI[0]{2\frac{H}{N_{0}}\bigl|\defvec{s}_{m}^H\defrvec{R}_{nc}\bigr|},~~~~\\
\label{Eq:NoncoherentMAPDecisionRuleForPrecodedComplexAWMNVectorChannelB}        
    &\overset{(a)}{=}&\argmax_{{1}\leq{m}\leq{M}}\,
        {p}_{m}\gamma_{m}\,
        \BesselI[0]{2\frac{H}{N_{0}}\bigl|\defvec{s}_{m}^H\defrvec{R}_{nc}\bigr|},\\
\label{Eq:NoncoherentMAPDecisionRuleForPrecodedComplexAWMNVectorChannelC}        
    &\overset{(b)}{=}&\argmax_{{1}\leq{m}\leq{M}}\,
        2\,{p}_{m}\gamma_{m}\frac{H}{N_{0}}
        \bigl|\defvec{s}_{m}^H\defrvec{R}_{nc}\bigr|,\\
\label{Eq:NoncoherentMAPDecisionRuleForPrecodedComplexAWMNVectorChannelD}
    &=&\argmax_{{1}\leq{m}\leq{M}}\,
        {p}_{m}\gamma_{m}
            \bigl|\defvec{s}_{m}^H\defrvec{R}_{nc}\bigr|,\\
\label{Eq:NoncoherentMAPDecisionRuleForPrecodedComplexAWMNVectorChannelE}            
    &=&\argmax_{{1}\leq{m}\leq{M}}\,
        {p}_{m}\gamma_{m}R_{m}, 
\end{eqnarray}
\end{subequations}
where the decision variable $R_{m}\!=\!\bigl|\defvec{s}_{m}^H\defrvec{R}_{nc}\bigr|$, $1\!\leq\!{m}\!\leq\!{M}$.
\end{theorem}

\begin{proof}
It is obvious to obtain \eqref{Eq:NoncoherentMAPDecisionRuleForPrecodedComplexAWMNVectorChannelA} by using \theoremref{Theorem:NoncoherentMAPDecisionRuleForComplexAWMNVectorChannel} and then selecting both $\defmat{\Sigma}\!=\!\frac{{N}_{0}}{2}\defmat{I}$ and $\defmat{F}\!=\!\defmat{I}$. Subsequently, the following steps are performed. In step $(a)$ of \eqref{Eq:NoncoherentMAPDecisionRuleForPrecodedComplexAWMNVectorChannel}, The fact that $\exp({x})$ is monotonically increasing simplifies \eqref{Eq:NoncoherentMAPDecisionRuleForPrecodedComplexAWMNVectorChannelA} to \eqref{Eq:NoncoherentMAPDecisionRuleForPrecodedComplexAWMNVectorChannelB}. In step $(b)$, we notice that  $I_{0}(x)$ is also a monotonically increasing function for all $x\!\in\!\mathbb{R}_{+}$. Therefore, we can reduce \eqref{Eq:NoncoherentMAPDecisionRuleForPrecodedComplexAWMNVectorChannelB} to \eqref{Eq:NoncoherentMAPDecisionRuleForPrecodedComplexAWMNVectorChannelC}. Eventually, ignoring the constant terms $2$, ${H}$ and $N_{0}$ and denoting $R_{m}\!=\!\bigl|\defvec{s}_{m}^H\defrvec{R}_{nc}\bigr|$, we obtain \eqref{Eq:NoncoherentMAPDecisionRuleForPrecodedComplexAWMNVectorChannelE}, which completes the proof of \theoremref{Theorem:NoncoherentMAPDecisionRuleForPrecodedComplexAWMNVectorChannel}.
\end{proof}

From \theoremref{Theorem:NoncoherentMAPDecisionRuleForPrecodedComplexAWMNVectorChannel} above, we conclude that a non-coherent optimal detection correlates $\defrvec{R}_{nc}$ with all modulation symbols $\{\defvec{s}_{1},\defvec{s}_{2},\ldots,\defvec{s}_{M}\}$ and chooses the one that yields the maximum envelope. However, the probabilities of the modulation symbols must be available. Otherwise, the \ac{MAP} detection reduces to the \ac{ML} detection given in the following theorem.

\begin{theorem}\label{Theorem:NoncoherentMLDecisionRuleForPrecodedComplexAWMNVectorChannel}
For complex uncorrelated \ac{AWMN} vector channels, defined in \eqref{Eq:NoncoherentComplexAWMNVectorChannel}, the non-coherent~ML~rule~is~given~by
\begin{subequations}
\label{Eq:NoncoherentMLDecisionRuleForPrecodedComplexAWMNVectorChannel}
\ifCLASSOPTIONtwocolumn
\vspace{-3mm}
\fi
\begin{eqnarray}
\label{Eq:NoncoherentMLDecisionRuleForPrecodedComplexAWMNVectorChannelA}
\!\!\widehat{m}
    &=&\argmax_{{1}\leq{m}\leq{M}}\,
        \exp\bigl(\gamma_{m}\bigr)
        \BesselI[0]{2\frac{H}{N_{0}}\bigl|\defvec{s}_{m}^H\defrvec{R}_{nc}\bigr|},\!\!\\
\label{Eq:NoncoherentMLDecisionRuleForPrecodedComplexAWMNVectorChannelB}        
    &=&\argmax_{{1}\leq{m}\leq{M}}\,
        \gamma_{m}\,
        \BesselI[0]{2\frac{H}{N_{0}}\bigl|\defvec{s}_{m}^H\defrvec{R}_{nc}\bigr|},\!\!\\
\label{Eq:NoncoherentMLDecisionRuleForPrecodedComplexAWMNVectorChannelC}        
    &=&\argmax_{{1}\leq{m}\leq{M}}\,
        \gamma_{m}\,
        \bigl|\defvec{s}_{m}^H\defrvec{R}_{nc}\bigr|,\\
\label{Eq:NoncoherentMLDecisionRuleForPrecodedComplexAWMNVectorChannelD}                
    &=&\argmax_{{1}\leq{m}\leq{M}}\,
        \gamma_{m}R_{m},        
\end{eqnarray}
\end{subequations}
\end{theorem}

\begin{proof}
The proof is obvious using \theoremref{Theorem:NoncoherentMLDecisionRuleForComplexAWMNVectorChannel} and following the same steps in the proof of \theoremref{Theorem:NoncoherentMAPDecisionRuleForPrecodedComplexAWMNVectorChannel}. 
\end{proof}

Note that the non-coherent \ac{MAP} and \ac{ML} decision rules, given in \eqref{Eq:NoncoherentMAPDecisionRuleForPrecodedComplexAWMNVectorChannel} and \eqref{Eq:NoncoherentMLDecisionRuleForPrecodedComplexAWMNVectorChannel}, respectively, cannot be made much much simpler. However, in case of that the modulation symbols are equiprobable and have equal-energy, we can ignore the scales $p_{m}$ and $\gamma_{m}$, and the \ac{ML} detection rule becomes
\begin{subequations}
\label{Eq:NoncoherentOptimalDecisionRuleForPrecodedComplexAWMNVectorChannel}
\begin{eqnarray}
\label{Eq:NoncoherentOptimalDecisionRuleForPrecodedComplexAWMNVectorChannelA}
\widehat{m}
    &=&\argmax_{{1}\leq{m}\leq{M}}\,\bigl|\defvec{s}_{m}^H\defrvec{R}_{nc}\bigr|,\\
\label{Eq:NoncoherentOptimalDecisionRuleForPrecodedComplexAWMNVectorChannelB}
    &=&\argmax_{{1}\leq{m}\leq{M}}\,R_{m}.
\end{eqnarray}
\end{subequations}

\paragraph{\!Conditional \ac{SER} of Non-coherent Orthogonal Signalling}
\label{Section:SignallingOverAWMNChannels:NonCoherentSignalling:NonCoherentOrthogonalSignalling} 
To improve the performance of non-coherent receivers \cite{BibProakisBook,BibGoldsmithBook} 
(i.e., to increase the separability of the modulation symbols while using non-coherent detection rules), 
we assume that the modulation symbols $\defvec{s}_{1}$, $\defvec{s}_{2}$, $\ldots$, $\defvec{s}_{M}$ are orthogonal with each other, i.e., 
\begin{equation} 
    \defvec{s}^H_{m}\defvec{s}_{n}=\begin{cases}
    0&\text{if~}{m}\neq{n},\\
    E_{m}&\text{otherwise}.
    \end{cases}
\end{equation}
As we observe in both \eqref{Eq:NoncoherentMAPDecisionRuleForPrecodedComplexAWMNVectorChannelE} and \eqref{Eq:NoncoherentMLDecisionRuleForPrecodedComplexAWMNVectorChannelD}, a non-coherent \ac{MAP}\,/\,\ac{ML} detection computes and compares the scaled versions of $R_{m}\!=\!|\defvec{s}_{m}^H\defrvec{R}_{nc}|$ for all ${1}\!\leq\!{m}\!\leq\!{M}$, and subsequently chooses the modulation symbol that produces the maximum envelope. With the aid of  \theoremref{Theorem:INIDMultivariateCESMcLeishPDF}, we know that $R_{m}$ follows a \ac{CCS} McLeish distribution, and thus its inphase and quadrature components follow McLeish distribution. In more details, if the transmitted symbol is not the modulation symbol $m$ (i.e., $\defrvec{S}\!\neq\!\defvec{s}_{m}$), we notice
\begin{subequations}
\label{Eq:NoncoherentComponentsI}
\begin{eqnarray}
\label{Eq:NoncoherentRealComponentI}
\RealPart{\defvec{s}_{m}^H\defrvec{R}_{nc}}&\sim&\mathcal{M}_{\nu}(0,E_{m}N_{0}/2),\\
\label{Eq:NoncoherentImagComponentI}
\ImagPart{\defvec{s}_{m}^H\defrvec{R}_{nc}}&\sim&\mathcal{M}_{\nu}(0,E_{m}N_{0}/2).
\ifCLASSOPTIONtwocolumn
{~~~~~~~~~~~~~~~~~}
\fi
\end{eqnarray}
\end{subequations}
Moreover, if the transmitted symbol is the modulation symbol $m$ (i.e., $\defrvec{S}\!=\!\defvec{s}_{m}$), we notice 
\begin{subequations}
\label{Eq:NoncoherentComponentsII}
\begin{eqnarray}
\label{Eq:NoncoherentRealComponentII}
\RealPart{\defvec{s}_{m}^H\defrvec{R}_{nc}}&\sim&\mathcal{M}_{\nu}(HE_{m}\cos(\Theta),E_{m}N_{0}/2),\\
\label{Eq:NoncoherentImagComponentII}
\ImagPart{\defvec{s}_{m}^H\defrvec{R}_{nc}}&\sim&\mathcal{M}_{\nu}(HE_{m}\sin(\Theta),E_{m}N_{0}/2).
\ifCLASSOPTIONtwocolumn
{~~~~}
\fi
\end{eqnarray}
\end{subequations}
It is accordingly worth mentioning that, in both \eqref{Eq:NoncoherentComponentsI} and \eqref{Eq:NoncoherentComponentsII}, the components $\RealPart{\defvec{s}_{m}^H\defrvec{R}_{nc}}$ and $\ImagPart{\defvec{s}_{m}^H\defrvec{R}_{nc}}$ are uncorrelated but statistically not independent.  

\begin{theorem}\label{Theorem:NoncoherentInaccurateComponentPDF}
When~$\defrvec{S}\!\neq\!\defvec{s}_{m}$,~the~envelope $R_{m}\!=\!|\defvec{s}_{m}^H\defrvec{R}_{nc}|$ conditioned on the impulsive noise effects $G$ follows Rayleigh distribution with the \ac{PDF} given by 
\begin{equation}\label{Eq:NoncoherentInaccurateComponentConditionedPDF}
    f_{R_{m}|G}(r|g)=
        \frac{2r}{gE_{m}N_{0}}
            \exp\Bigl(-\frac{r^2}{gE_{m}N_{0}}\Bigr),
\end{equation}
defined over $r\in\mathbb{R}^{+}$. Further, the envelope $R_{m}\!=\!|\defvec{s}_{m}^H\defrvec{R}_{nc}|$ has a non-negative random distribution, which is modeled by $K$-distribution, whose \ac{PDF} is given by
\begin{equation}\label{Eq:NoncoherentInaccurateComponentPDF}
f_{R_{m}}(r)=
    \frac{4r^{\nu}}{\Gamma(\nu)\,\Lambda_{m}^{\nu+1}}
        \BesselK[\nu-1]{\frac{2r}{\Lambda_{m}}},
\end{equation}
defined in $r\in\mathbb{R}^{+}$, where the component deviation factor is given by $\Lambda_{m}\!=\!\sqrt{E_{m}}\Lambda_{0}\!=\!\sqrt{E_{m}N_{0}/\nu}$ (i.e., $\Lambda_{0}\!=\!\sqrt{N_{0}/\nu}$).
\end{theorem}

\begin{proof}
Defining $I_{m}\!=\!\Re\{{s}_{m}^H\defrvec{R}_{nc}\}$ and $Q_{m}\!=\!\Im\{{s}_{m}^H\defrvec{R}_{nc}\}$, we notice that  
 $I_{m}$ and $Q_{m}$ are uncorrelated but statistically not independent. Further, with the aid of \theoremref{Theorem:CCSMcLeishDefinition}, we have $I_{m}\!=\!\sqrt{G}X_{m}$ and $Q_{m}\!=\!\sqrt{G}Y_{m}$. Thus, we can write
 \begin{equation}\label{Eq:NoncoherentComponent}
    R_{m}=\sqrt{G}\sqrt{X^2_{m}+Y^2_{m}}=\sqrt{G}V_{m},
\end{equation}
with the distributions $G\!\sim\!\mathcal{G}(\nu,1)$, $X_{m}\!\sim\!\mathcal{N}(0,E_{m}N_{0}/2)$ and $Y_{m}\!\sim\!\mathcal{N}(0,E_{m}N_{0}/2)$. Using \cite[Eq. (2.3-42)]{BibProakisBook}, the component $V_{m}\!=\!\sqrt{X^2_{m}+Y^2_{m}}$ follows a Rayleigh distribution whose \ac{PDF} is given by \cite[Eq. (2.3-43)]{BibProakisBook}. Thus, the \ac{PDF} of $R_{m}$ conditioned on $G$ is written as \eqref{Eq:NoncoherentInaccurateComponentConditionedPDF}, which completes the first step of the proof. We obtain the \ac{PDF} of $R_{m}$ as 
\begin{subequations}
\label{Eq:NoncoherentInaccurateComponentPDFIntegral}
\begin{eqnarray}
\label{Eq:NoncoherentInaccurateComponentPDFIntegralA}
f_{R_{m}}(r)&=&\int_{0}^{\infty}f_{R_{m}|G}(r|g)f_{G}(g)dg,\\
\label{Eq:NoncoherentInaccurateComponentPDFIntegralB}
    &=&\int_{0}^{\infty}\!\!\frac{2r}{gE_{m}N_{0}}
            \exp\Bigl(-\frac{r^2}{gE_{m}N_{0}}\Bigr)f_{G}(g)dg,
    \ifCLASSOPTIONtwocolumn
        {~~~~~~}
    \fi
\end{eqnarray}
\end{subequations}
where the \ac{PDF} of $G\!\sim\!\mathcal{G}(\nu,1)$ is given in \eqref{Eq:ProportionPDF}. Finally, using \cite[Eq. (3.478/4)]{BibGradshteynRyzhikBook} in \eqref{Eq:NoncoherentInaccurateComponentPDFIntegralB} yields \eqref{Eq:NoncoherentInaccurateComponentPDF}, which completes the proof of \theoremref{Theorem:NoncoherentInaccurateComponentPDF}.
\end{proof}

\begin{theorem}\label{Theorem:NoncoherentAccurateComponentPDF}
When~$\defrvec{S}\!=\!\defvec{s}_{m}$,~the envelope $R_{m}\!=\!|\defvec{s}_{m}^H\defrvec{R}_{nc}|$ conditioned on the impulsive noise effects $G$ follows Ricean distribution with the \ac{PDF} given by 
\begin{equation}
\label{Eq:NoncoherentAccurateComponentConditionedPDF}
 \!\!\!\!f_{R_{m}|G}(r|g)=
        \frac{2r}{gE_{m}N_{0}}
            I_{0}\Bigl(\frac{2\kappa_{m}{r}}{gE_{m}N_{0}}\Bigr)
                \exp\Bigl(-\frac{r^2+\kappa_{m}^2}{gE_{m}N_{0}}\Bigr),\!\!\!
\end{equation}
where the Ricean parameter $\kappa_{m}\!=\!HE_{m}$. Furthermore, the envelope $R_{m}\!=\!|\defvec{s}_{m}^H\defrvec{R}_{nc}|$
has a non-negative distribution whose \ac{PDF} is 
\begin{equation}
\label{Eq:NoncoherentAccurateComponentPDF}
\!\!\!f_{R_{m}}(r)=
    \frac{r}{\pi}
    \int_{0}^{2\pi}
    \!\!\frac{q_{m}(r,\theta)^{\nu-1}}{\Gamma(\nu)\,\Lambda^{\nu+1}}
        \BesselK[\nu-1]{\frac{2}{\Lambda}q_{m}(r,\theta)}d\theta,\!\!
\end{equation}
defined over $r\in\mathbb{R}^{+}$, where the deviation factor is given by $\Lambda\!=\!\sqrt{E_{m}}\Lambda_{0}$, and $q_{m}(r,\theta)$ is defined as 
\begin{equation}\label{Eq:NoncoherentComponentDistance}
    q_{m}(r,\theta)=\sqrt{r^2+2\,r\kappa_{m}\cos(\theta)+\kappa^2_{m}}.
\end{equation}
\end{theorem}

\begin{proof}
When $\defrvec{S}\!=\!\defvec{s}_{m}$, the envelope $R_{m}\!=\!|\defvec{s}_{m}^H\defrvec{R}_{nc}|$ is also decomposed as \eqref{Eq:NoncoherentComponent} by following the same steps in the proof of \theoremref{Theorem:NoncoherentInaccurateComponentPDF}. Referring to both  \eqref{Eq:NoncoherentRealComponentII} and \eqref{Eq:NoncoherentImagComponentII},
we notice that $G\!\sim\!\mathcal{G}(\nu,1)$, and $X_{m}\!\sim\!\mathcal{N}(HE_{m}\cos(\Theta),E_{m}N_{0}/2)$ with $Y_{m}\!\sim\!\mathcal{N}(HE_{m}\sin(\Theta),E_{m}N_{0}/2)$. Further, utilizing \cite[Eq. (2.3-55)]{BibProakisBook}, we notice that $V_{m}$ follows the Ricean distribution with the \ac{PDF} given by \cite[Eq. (2.3-56)]{BibProakisBook}. Therefore, the \ac{PDF} of $R_{m}$ conditioned on $G$ is written as \eqref{Eq:NoncoherentAccurateComponentConditionedPDF} in which we obtain $\kappa^2\!=\!\mathbb{E}[I_{m}|G]^2+\mathbb{E}[Q_{m}|G]^2\!=\!H^2E^2_{m}$ in accordance with \theoremref{Theorem:CCSMcLeishDefinition}. Herewith, by means of using \cite[Eq. (3.339)]{BibGradshteynRyzhikBook}, we can write 
\begin{equation}
 \!\!f_{R_{m}|G}(r|g)=
    \frac{2r}{g\pi{E}_{m}N_{0}}
        \int_{0}^{\pi}
            \exp\Bigl(-\frac{q_{m}^2(r,\theta)}{gE_{m}N_{0}}\Bigr)d\theta,\!\!
\end{equation}
where $q_{m}(r,\theta)$ is defined above in \eqref{Eq:NoncoherentComponentDistance}. The \ac{PDF} of $R_{m}$ can be obtained by $f_{R_{m}}(r)\!=\!\int_{0}^{\infty}f_{R_{m}|G}(r|g)f_{G}(g)dg$,~that~is
\begin{equation}
f_{R_{m}}(r)=
    \frac{2r}{g\pi{E}_{m}N_{0}}
    \int_{0}^{\pi}
    \int_{0}^{\infty}
        \exp\Bigl(-\frac{q_{m}^2(r,\theta)}{gE_{m}N_{0}}\Bigr)
            f_{G}(g){dg}{d\theta},
\end{equation}
where $f_{G}(g)$ is given in \eqref{Eq:ProportionPDF}. Finally, using \cite[Eq. (3.478/4)]{BibGradshteynRyzhikBook}, we can readily rewrite the \ac{PDF} of $R_{m}$ as in \eqref{Eq:NoncoherentAccurateComponentPDF}, which completes the proof of \theoremref{Theorem:NoncoherentAccurateComponentPDF}.
\end{proof}

Let us now consider the conditional \ac{SER} of non-coherent \ac{MAP} detection for orthogonal modulations. We can write The probability of erroneous decision as 
\begin{equation}
\label{Eq:NoncoherentOrthogonalModulationErrorProbability}
    \Pr\bigl\{\bigl.e\,\bigr|\,H\bigr\}=1-\Pr\bigl\{\bigl.c\,\bigr|\,H\bigr\},
\end{equation}
where $\Pr\bigl\{\bigl.c\,\bigr|\,H\bigr\}$ is the the probability of correct decision, and can be readily rewritten as  
\begin{equation}
\label{Eq:NoncoherentOrthogonalModulationConditionedErrorProbability}
    \Pr\bigl\{\bigl.c\,\bigr|\,H\bigr\}=
        \sum_{m=1}^{M}
            \Pr\bigl\{\bigl.c\,\bigr|\,H,\defvec{s}_{m}\bigr\}
                \Pr\bigl\{\defrvec{S}=\defvec{s}_{m}\bigr\},
\end{equation}
where $\Pr\bigl\{\bigl.c\,\bigr|\,H,\defvec{s}_{m}\bigr\}$ denotes the probability of correct decision. Referring to \theoremref{Theorem:NoncoherentMAPDecisionRuleForPrecodedComplexAWMNVectorChannel}, when the modulation symbol $m$ is transmitted, a correct decision is made iff  
$p_{n}\gamma_{n}R_{n}\!<\!p_{m}\gamma_{m}R_{m}$ for all $1\!\leq\!{n}\!\leq\!{M}$ and ${m}\!\neq\!{n}$. Therefore, the probability of correct decision can be readily written as 
\begin{equation}
\nonumber
\Pr\bigl\{\bigl.c\,\bigr|\,H,\defvec{s}_{m}\bigr\}=\Pr\Bigl\{\bigcap_{n\neq{m}}p_{n}\gamma_{n}R_{n}<p_{m}\gamma_{m}\gamma_{m}\Bigl|\,H,\defvec{s}_{m}\Bigr.\Bigr\},
\end{equation}
where the envelopes $R_{1},R_{2},\ldots,R_{M}$ are certainly uncorrelated as a result of that modulation symbols are orthogonal (i.e., $\defvec{s}^T_{m}\defvec{s}_{n}\!=\!0$ for all ${m}\neq{n}$). They will however be entirely independent when conditioned on impulsive noise effects (i.e., conditioned on $G$). Then, we rewrite $\Pr\{c\,|\,H,\defvec{s}_{m}\}$ as
\begin{equation}
\label{Eq:NoncoherentModulationCorrectProbability}
\Pr\bigl\{\bigl.c\,\bigr|\,H,\defvec{s}_{m}\bigr\}=
        \int_{0}^{\infty}
            \Pr\bigl\{\bigl.c\,\bigr|\,H,\defvec{s}_{m},g\bigr\}
                f_{G}(g)dg,
\end{equation}
where $\Pr\bigl\{\bigl.c\,\bigr|\,H,\defvec{s}_{m},g\bigr\}$ is given by 
\begin{equation}
\label{Eq:NoncoherentModulationConditionedCorrectProbability}
\Pr\bigl\{\bigl.c\,\bigr|\,H,\defvec{s}_{m},g\bigr\}
    =\prod_{n\neq{m}}^{M}\Pr\left\{R_{n}<\frac{p_{m}E_{m}}{p_{n}E_{n}}R_{m}\right\},
\end{equation}
where $R_{m}$ follows a Ricean distribution whose \ac{PDF} is given by \eqref{Eq:NoncoherentAccurateComponentConditionedPDF}. For $1\!\leq\!{n}\!\neq\!{m}\!\leq\!{M}$, $R_{n}$ has Rayleigh distribution whose \ac{PDF} is given by \eqref{Eq:NoncoherentInaccurateComponentConditionedPDF}. From this point on, we rewrite 
\begin{equation}
\label{Eq:NoncoherentModulationConditionedCorrectProbabilityII}
\Pr\bigl\{\bigl.c\,\bigr|\,H,\defvec{s}_{m},g\bigr\}
    =\mathbb{E}\biggl[\prod_{n\neq{m}}^{M}F_{R_{n}}\Bigl(\frac{p_{m}E_{m}}{p_{n}E_{n}}R_{m}\Bigr)\biggr],
\end{equation}
where $F_{R_{n}}(r)$ is the \acp{CDF} of $V_{n}$ for all $1\!\leq\!{n}\!\neq\!{m}\!\leq\!{M}$.
With the aid of the equations from \eqref{Eq:NoncoherentOrthogonalModulationErrorProbability} to \eqref{Eq:NoncoherentModulationConditionedCorrectProbabilityII}, the conditional \ac{SER} of non-coherent orthogonal signaling is given in the following.

\begin{theorem}
\label{Theorem:MAPDecisionErrorProbabilityForNoncoherentOrthogonalSignallingOverAWMNChannels}
For the \ac{MAP} decision rule given by \theoremref{Theorem:NoncoherentMAPDecisionRuleForPrecodedComplexAWMNVectorChannel}, the conditional~\ac{SER}~of~non-coherent~orthogonal~signaling is given by 
\begin{subequations}
\label{Eq:MAPDecisionErrorProbabilityForNoncoherentOrthogonalSignallingOverAWMNChannels}
\begin{eqnarray}
    \label{Eq:MAPDecisionErrorProbabilityForNoncoherentOrthogonalSignallingOverAWMNChannelsA}
    \Pr\{\bigl.e\,|\,H\}
    &=&\frac{1}{\Gamma(\nu)}
        \sum_{k=1}^{2^M-1}\sum_{m=1}^{M}
            \frac{(-1)^{1+\sum_{n=1}^{M}k_{n}}p_{m}}{1+\Phi_{k,m}}
            \MeijerG[right]{2,0}{0,2}
                {\frac{\nu\Phi_{k,m}\gamma_{m}}{1+\Phi_{k,m}}}
                    {\emptycoefficient}
                        {0,\nu}
                    \KroneckerDelta{k_{m}}{0},\\[1mm]
    \label{Eq:MAPDecisionErrorProbabilityForNoncoherentOrthogonalSignallingOverAWMNChannelsB}
    &=&\frac{1}{\Gamma(\nu)}
        \sum_{k=1}^{2^M-1}\sum_{m=1}^{M}
            \frac{(-1)^{1+\sum_{n=1}^{M}k_{n}}p_{m}}
                {(1+\Phi_{k,m})\Lambda^{\nu}_{0}}
        \biggl(\frac{2\Phi_{k,m}\gamma_{m}}{1+\Phi_{k,m}}\biggr)^{\frac{\nu}{2}}
        K_{\nu}\Biggl(\frac{2}{\Lambda_{0}}
            \sqrt{\frac{2\Phi_{k,m}\gamma_{m}}{1+\Phi_{k,m}}}\Biggr)
                \KroneckerDelta{k_{m}}{0},{~~~~~~~~}
\end{eqnarray}
\end{subequations}
where the indexing $k_{n}$ is defined by $k_{n}\!=\!\lfloor{2k}/{2^n}\rfloor-2\lfloor{k}/{2^n}\rfloor$.
Further, $\Phi_{k,m}$ is the normalized \ac{SNR} for the modulation symbol $m$ and defined by 
\begin{equation}
\label{Eq:NormalizedSNRForNoncoherentOrthogonalSignalling}
    \Phi_{k,m}=\sum_{n=1}^{M}
            \biggl(\!\frac{p_{m}}{p_{n}}\!\biggr)^{\!2}
                \biggl(\!\frac{\gamma_{m}}{\gamma_{n}}\!\biggr)^{\!3}
                    {k}_{n},
\end{equation}
Further, for all ${1}\!\leq\!{m}\!\leq\!{M}$, $p_{m}$ is the probability of the modulation symbol $m$, and $\gamma_{m}$ is the instantaneous \ac{SNR} for the transmission of the modulation symbol $m$.
\end{theorem}

\begin{proof}
Note that, with the aid of \cite[Eq. (4.24)]{BibPapoulisBook}, \eqref{Eq:NoncoherentModulationConditionedCorrectProbability} can be shown to be \eqref{Eq:NoncoherentModulationConditionedCorrectProbabilityII}, in which the expectation is achieved with respect to the distribution $V_{m}$, and where $F_{V_{n}}$ is the \ac{CDF} of the distribution $V_{n}$ and easily found as \cite[Eq. (2.3-50)]{BibProakisBook},
\begin{equation}
\label{Eq:NoncoherentInaccurateComponentConditionedCDF}
    F_{R_{n}}(r)=
        1-\exp\Bigl(-\frac{r^2}{gE_{n}N_{0}}\Bigr),\quad{r}\in\mathbb{R}^{+}.
\end{equation}
For non-zero distinct $x_1,x_2,\ldots,x_N$, we can show that 
\begin{equation}
\label{Eq:SeriesExpansionForNewtonPolynomial}
    \prod_{{n}\neq{m}}^{N}(1+x_n)=
        1+\sum_{k=1}^{2^N-1}
            \prod_{n=1}^{N}x_{n}^{k_{n}}
                \KroneckerDelta{k_{m}}{0},
\end{equation}
where $k_{n}\!=\!\lfloor{2k}/{2^n}\rfloor-2\lfloor{k}/{2^n}\rfloor$, and therein $\lfloor{x}\rfloor$ is the floor function that returns the greatest integer less than or equal to $x$. Further, $\KroneckerDelta{x}{y}$ is the Kronecker's delta function that returns $1$ iff $x\!=\!y$ and $0$ otherwise. Putting \eqref{Eq:NoncoherentInaccurateComponentConditionedCDF} into \eqref{Eq:NoncoherentModulationConditionedCorrectProbabilityII} and using \eqref{Eq:SeriesExpansionForNewtonPolynomial}, we can rewrite \eqref{Eq:NoncoherentModulationConditionedCorrectProbabilityII} as follows
\begin{equation}
    \label{Eq:NoncoherentModulationConditionedCorrectProbabilityIII}
    {~~}\Pr\bigl\{\bigl.c\,\bigr|\,H,\defvec{s}_{m},g\bigr\}
        =1+\sum_{k=1}^{2^{M}-1}
            (-1)^{\sum_{n=1}^{M}{k}_{n}}
        \mathbb{E}\biggl[
                \exp\biggl(-\frac{\Phi_{k,m}}{gE_{0}N_{0}}
                    R^2_{m}\biggr)\biggr]
                        \KroneckerDelta{k_{m}}{0},{~~~~}
\end{equation}
where $\Phi_{k,m}$ is defined in \eqref{Eq:NormalizedSNRForNoncoherentOrthogonalSignalling}. As mentioned before, $R_{m}$ follows a Ricean distribution whose \ac{PDF} is given by
\begin{equation}
\!\!\!f_{R_{m}}(r)=
    \frac{2v}{gE_{m}N_{0}}
        I_{0}\Bigl(\frac{2\kappa_{m}{r}}{gE_{m}N_{0}}\Bigr) 
            \exp\Bigl(-\frac{r^2+\kappa_{m}^2}{gE_{m}N_{0}}\Bigr), 
\end{equation}
where $\kappa_{m}$ is a constant defined as $\kappa_{m}\!=\!HE_{m}$. Further, note that 
$\mathbb{E}[\exp(-sR^2_{m})]$, where $s\!=\!{\Phi_{k,m}}/{(gE_{0}N_{0})}$, is specifically required in \eqref{Eq:NoncoherentModulationConditionedCorrectProbabilityIII}. Thanks to $\int_{0}^{\infty}x\exp(-x^2/a)\BesselI[0]{bx}\allowbreak{d}x\!=\!a\exp(ab^2)/2$\cite[Eq. (2.15.20/8)]{BibGradshteynRyzhikBook}, we derive
\begin{equation}\label{Eq:SquaredRiceanDistributionMGF}
\!\!\!\!\mathbb{E}\bigl[\exp\bigl(-sR^2_{m}\bigr)\bigr]=
        \frac{\exp\bigl(-\frac{{s}\kappa_{m}^2}{1+sgE_{m}N_{0}}\bigr)}
            {1+sgE_{m}N_{0}}.
\end{equation}
Eventually, inserting both \eqref{Eq:NoncoherentModulationConditionedCorrectProbabilityIII} and \eqref{Eq:SquaredRiceanDistributionMGF} into \eqref{Eq:NoncoherentModulationCorrectProbability} yields 
\begin{equation}
\label{Eq:NoncoherentModulationCorrectProbabilityII}    
\Pr\bigl\{\bigl.c\,\bigr|\,H,\defvec{s}_{m}\bigr\}
    =1+\frac{1}{\Gamma(\nu)}\sum_{k=1}^{2^M-1}
        \frac{(-1)^{\sum_{n=1}^{M}k_{n}}}{1+\Phi_{k,m}}
        M_{{1}/{G}}\Bigl(\frac{\Phi_{k,m}\gamma_{m}}{1+\Phi_{k,m}}\Bigr)\KroneckerDelta{k_{m}}{0}.
        {~~}
\end{equation}
where $M_{{1}/{G}}(s)$, $s\in\mathbb{R}^{+}$ is the reciprocal \ac{MGF} and defined as $M_{{1}/{G}}(s)\!=\!\int_{0}^{\infty}\exp(-s/g)f_{G}(g)dg$, in which putting \eqref{Eq:ProportionPDF} and using both \cite[Eqs. (8.4.3/1) and (8.4.3/2)]{BibPrudnikovBookVol3} within \cite[Eq. (2.8.4)]{BibKilbasSaigoBook}, we obtain  
\begin{equation}\label{Eq:ProportionInverseMGF}
    M_{{1}/{G}}(s)=\frac{1}{\Gamma(\nu)}\MeijerG[right]{2,0}{0,2}{s\nu}{\emptycoefficient}{0,\nu}.
\end{equation}
Putting both  \eqref{Eq:ProportionInverseMGF} and \eqref{Eq:NoncoherentModulationCorrectProbabilityII} into \eqref{Eq:NoncoherentOrthogonalModulationConditionedErrorProbability} and using \eqref{Eq:NoncoherentOrthogonalModulationErrorProbability}, we obtain \eqref{Eq:MAPDecisionErrorProbabilityForNoncoherentOrthogonalSignallingOverAWMNChannelsA}, in which using \cite[Eqs. (8.2.2/15) and (8.4.23/1)]{BibPrudnikovBookVol3} results in 
\eqref{Eq:MAPDecisionErrorProbabilityForNoncoherentOrthogonalSignallingOverAWMNChannelsB}, which proves \theoremref{Theorem:MAPDecisionErrorProbabilityForNoncoherentOrthogonalSignallingOverAWMNChannels}.
\end{proof}

\begin{theorem}
\label{Theorem:MLDecisionErrorProbabilityForNoncoherentOrthogonalSignallingOverAWMNChannels}
For the \ac{ML} decision rule given by \theoremref{Theorem:NoncoherentMLDecisionRuleForPrecodedComplexAWMNVectorChannel}, the conditional \ac{SER} of non-coherent orthogonal signaling is given by 
\begin{subequations}
\label{Eq:MLDecisionErrorProbabilityForNoncoherentOrthogonalSignallingOverAWMNChannels}
\vspace{-1mm}
\begin{eqnarray}
    \label{Eq:MLDecisionErrorProbabilityForNoncoherentOrthogonalSignallingOverAWMNChannelsA}
    \Pr\bigl\{\bigl.e\,\bigr|\,H\bigr\}
    &=&\frac{1}{M\Gamma(\nu)}
        \sum_{k=1}^{2^M-1}\sum_{m=1}^{M}
            \frac{(-1)^{1+\sum_{n=1}^{M}k_{n}}}{1+\Phi_{k,m}}
            \MeijerG[right]{2,0}{0,2}
                {\frac{\nu\Phi_{k,m}\gamma_{m}}{1+\Phi_{k,m}}}
                    {\emptycoefficient}
                        {0,\nu}
                    \KroneckerDelta{k_{m}}{0},\\[1mm]
    \label{Eq:MLDecisionErrorProbabilityForNoncoherentOrthogonalSignallingOverAWMNChannelsB}
    &=&\frac{1}{M\Gamma(\nu)}
        \sum_{k=1}^{2^M-1}\sum_{m=1}^{M}
            \frac{(-1)^{1+\sum_{n=1}^{M}k_{n}}}
                {(1+\Phi_{k,m})\Lambda^{\nu}_{0}}
        \biggl(\frac{2\Phi_{k,m}\gamma_{m}}{1+\Phi_{k,m}}\biggr)^{{\nu}/{2}}
        K_{\nu}\Biggl(\frac{2}{\Lambda_{0}}
            \sqrt{\frac{2\Phi_{k,m}\gamma_{m}}{1+\Phi_{k,m}}}\Biggr)
                \KroneckerDelta{k_{m}}{0},{~~~~~~~~~~}    
\end{eqnarray}
\end{subequations}
where $\Phi_{k,m}\!=\!\sum_{n=1}^{M}({\gamma_{m}}/{\gamma_{n}})^{3}{k}_{n}$.
\end{theorem}

\begin{proof}
The proof is obvious setting $p_{m}\!=\!1/M$ for ${1}\!\leq\!{m}\!\leq\!{M}$ in \theoremref{Theorem:MAPDecisionErrorProbabilityForNoncoherentOrthogonalSignallingOverAWMNChannels}. 
\end{proof}

\begin{theorem}
\label{Theorem:MLDecisionErrorProbabilityForNoncoherentEqualEnergyOrthogonalSignallingOverAWMNChannels}
When the modulation symbols are equiprobable and have equal-energy, and referring to \eqref{Eq:NoncoherentOptimalDecisionRuleForPrecodedComplexAWMNVectorChannel}, the conditional \ac{SER} of non-coherent orthogonal signaling is given by 
\begin{subequations}
\label{Eq:MLDecisionErrorProbabilityForNoncoherentEqualEnergyOrthogonalSignallingOverAWMNChannels}
\begin{eqnarray}
    \label{Eq:MLDecisionErrorProbabilityForNoncoherentEqualEnergyOrthogonalSignallingOverAWMNChannelsA}
    \Pr\bigl\{\bigl.e\,\bigr|\,H\bigr\}
    &=&\frac{1}{\Gamma(\nu)}
        \sum_{k=1}^{M-1}
            \frac{(-1)^{1+k}}{1+k}\Binomial{M-1}{k}
            \MeijerG[right]{2,0}{0,2}
                {\frac{\nu{k}\gamma}{1+k}}
                    {\emptycoefficient}
                        {0,\nu},\\[1mm]
    \label{Eq:MLDecisionErrorProbabilityForNoncoherentEqualEnergyOrthogonalSignallingOverAWMNChannelsB}
    &=&\frac{2}{\Gamma(\nu)}
        \sum_{k=1}^{M-1}
            \frac{(-1)^{1+k}}
                {(1+k)\Lambda^{\nu}_{0}}\Binomial{M-1}{k}
        \biggl(\frac{2{k}\gamma}{1+k}\biggr)^{\frac{\nu}{2}}
            K_{\nu}\Biggl(\frac{2}{\Lambda_{0}}
                \sqrt{\frac{2{k}\gamma}{1+k}}\Biggr),{~~~}    
\end{eqnarray}
\end{subequations}
where $\gamma\!=\!H^2E_{\defrvec{S}}/N_{0}$ denotes the instantaneous \ac{SNR}.  
\end{theorem}

\begin{proof}
In case of that the modulation symbols $\defvec{s}_{m}$, ${1}\!\leq\!{m}\!\leq\!{M}$ are equiprobable and have equal energy (i.e., when $E_{m}\!=\!E_{\defrvec{S}}$ and $\Pr\{\defrvec{S}\!=\!\defvec{s}_{m}\}\!=\!1/M$ for all ${1}\!\leq\!{m}\!\leq\!{M}$), \eqref{Eq:NoncoherentModulationConditionedCorrectProbability} can be shown to be 
\begin{equation}
\label{Eq:NoncoherentIIDModulationConditionedCorrectProbability}
\Pr\bigl\{\bigl.c\,\bigr|\,H,\defvec{s}_{m},g\bigr\}
    =\mathbb{E}\bigl[F_{R_{n}}(R_{m})^{M-1}\bigr],
\end{equation}
where substituting \eqref{Eq:NoncoherentInaccurateComponentConditionedCDF} and then utilizing binomial expansion \cite[Eq. (3.1.1)]{BibAbramowitzStegunBook} results in
\begin{equation}
\label{Eq:NoncoherentIIDModulationConditionedCorrectProbabilityII}
\Pr\bigl\{\bigl.c\,\bigr|\,H,\defvec{s}_{m},g\bigr\}
    =1+\sum_{k=1}^{M-1}
        (-1)^{k}\Binomial{M-1}{k}
            \mathbb{E}\biggl[\exp\Bigl(-\frac{kR_{m}^2}{gE_{\defrvec{S}}N_{0}}\Bigr)\biggr],
\end{equation}
where the expectation is achieved with respect to the distribution $R_{m}$ and can be readily derived by setting $s\!=\!k/g/E_{\defrvec{S}}/N_{0}$ in \eqref{Eq:SquaredRiceanDistributionMGF}. From this point, we derive the closed-form expression of $\Pr\{e|H,\defvec{s}_{m}\}$, from which we can obtain $\Pr\{e|H,\defvec{s}_{m}\}\!=\allowbreak\!\int_{0}^{\infty}\Pr\{e|H,\defvec{s}_{m},g\}f_{G}(g)\,dg$. Accordingly, the proof is obvious performing almost the same steps in the proof of \theoremref{Theorem:MAPDecisionErrorProbabilityForNoncoherentOrthogonalSignallingOverAWMNChannels}.
\end{proof}

\begin{figure*}[tp] 
\centering
\begin{subfigure}{0.47\columnwidth}
    \centering
    \includegraphics[clip=true, trim=0mm 0mm 0mm 0mm, width=1.0\columnwidth,height=0.85\columnwidth]{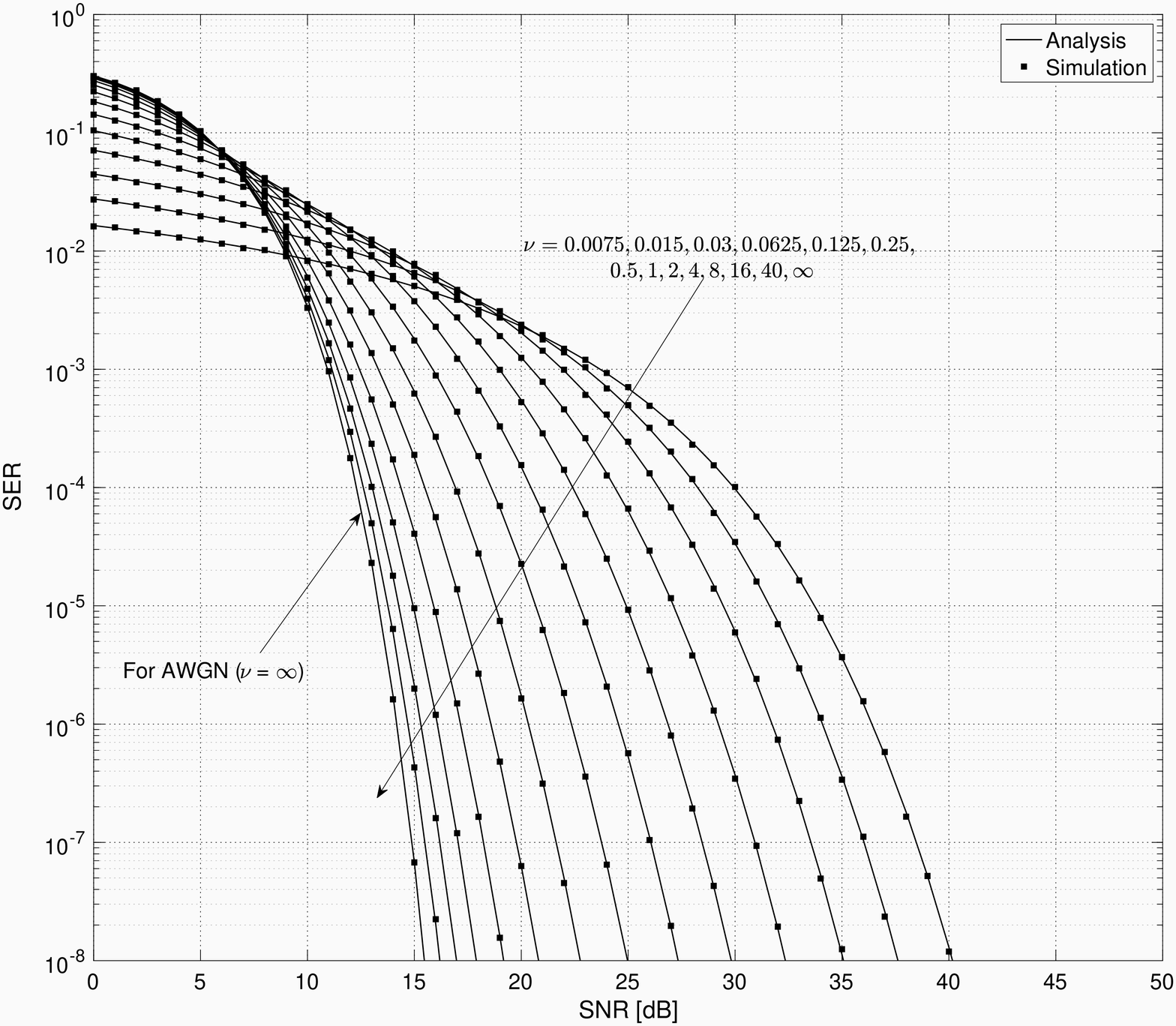}
    \caption{Modulation level $M=2$.} 
    \vspace{5mm}
    \label{Figure:ConditionalSEPForNoncoherentOrthogonalSignalingA}
\end{subfigure}
{~~~}
\begin{subfigure}{0.47\columnwidth}
    \centering
    \includegraphics[clip=true, trim=0mm 0mm 0mm 0mm, width=1.0\columnwidth,height=0.85\columnwidth]{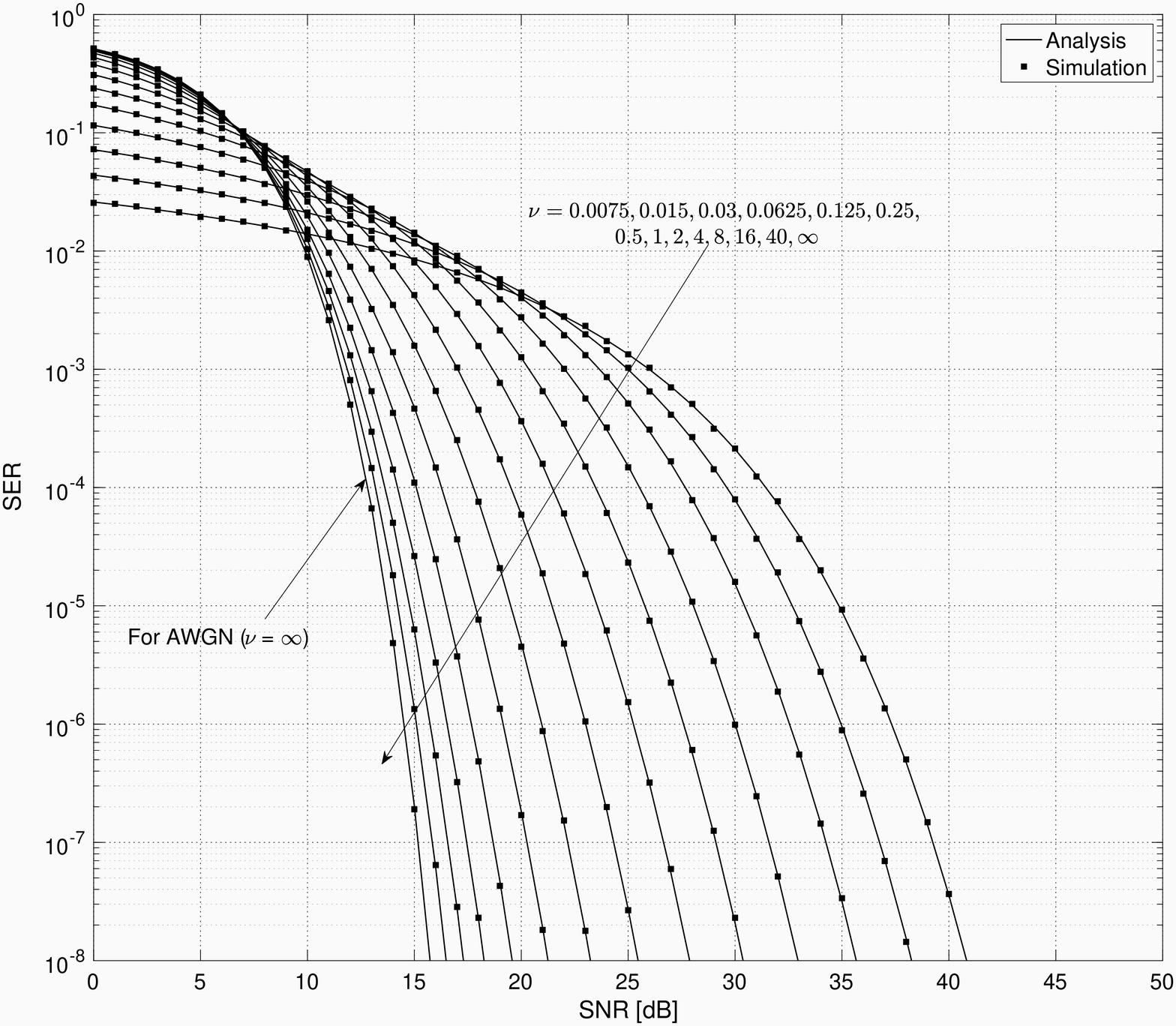}
    \caption{Modulation level $M=4$.}
    \vspace{5mm}
    \label{Figure:ConditionalSEPForNoncoherentOrthogonalSignalingB}
\end{subfigure}\\
\begin{subfigure}{0.47\columnwidth}
    \centering
    \includegraphics[clip=true, trim=0mm 0mm 0mm 0mm, width=1.0\columnwidth,height=0.85\columnwidth]{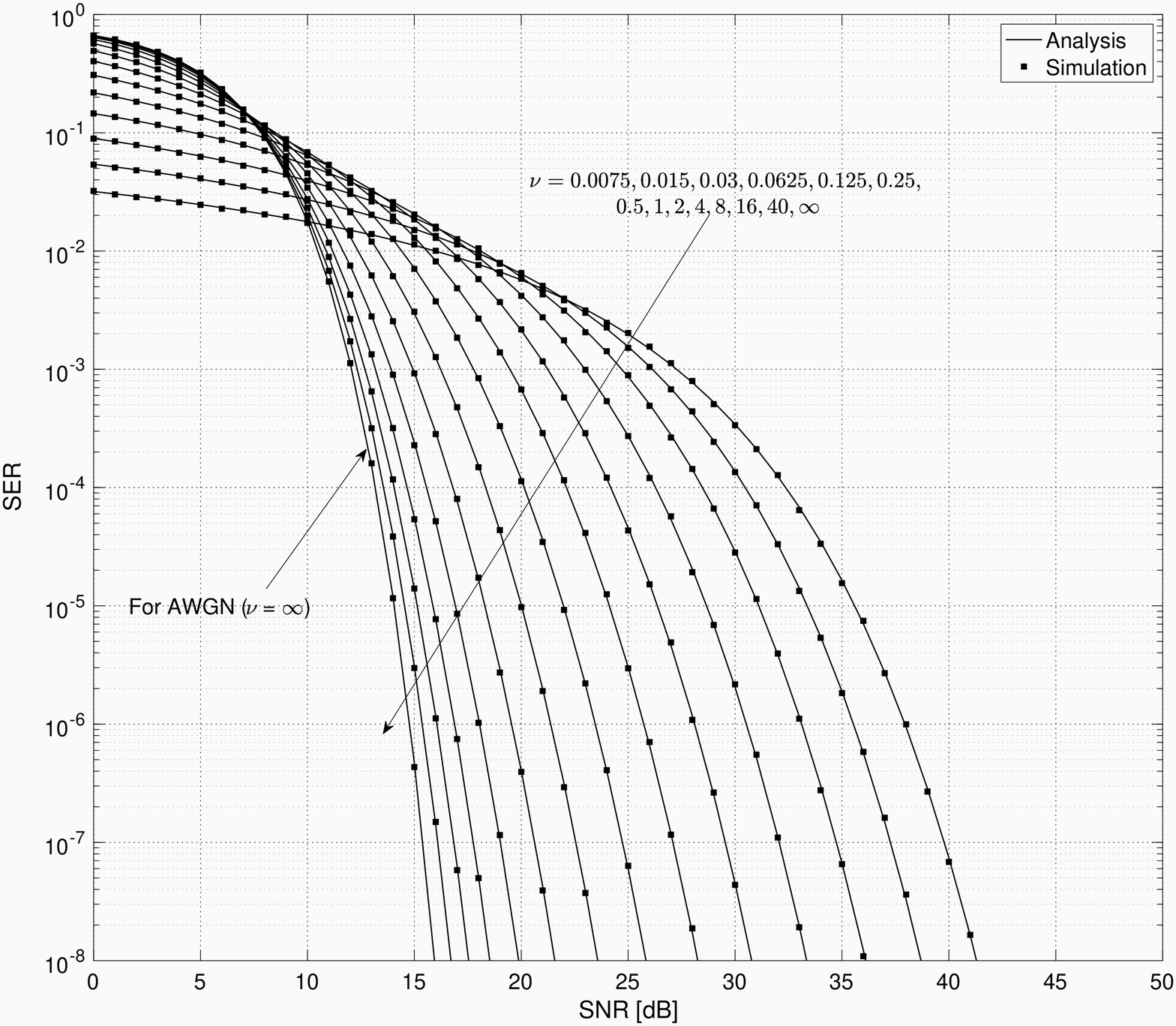}
    \caption{Modulation level $M=8$.}
    \label{Figure:ConditionalSEPForNoncoherentOrthogonalSignalingC}
\end{subfigure}
{~~~}
\begin{subfigure}{0.47\columnwidth}
    \centering
    \includegraphics[clip=true, trim=0mm 0mm 0mm 0mm, width=1.0\columnwidth,height=0.85\columnwidth]{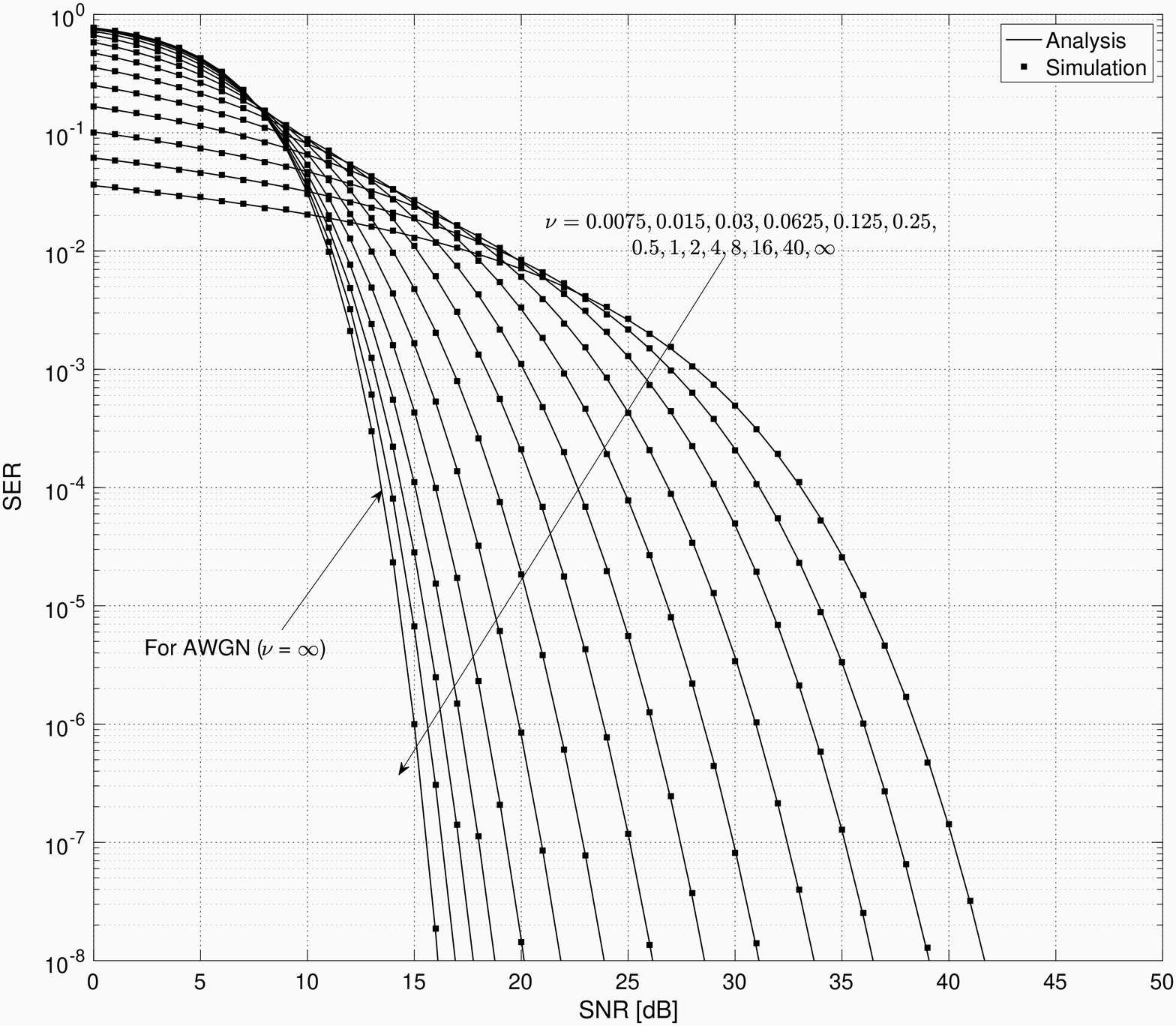}
    \caption{Modulation level $M=16$.}
    \label{Figure:ConditionalSEPForNoncoherentOrthogonalSignalingD}
\end{subfigure}
\caption{The \ac{SER} of non-coherent orthogonal signaling over \ac{AWMN} channels.}
\label{Figure:ConditionalSEPForNoncoherentOrthogonalSignaling}
\vspace{-2mm} 
\end{figure*} 

\begin{theorem}\label{Theorem:MLDecisionErrorProbabilityForNoncoherentEqualEnergyBinaryOrthogonalSignallingOverAWMNChannels}
The conditional \ac{BER} of orthogonal signaling, including \ac{BFSK}, with non-coherent \ac{ML} detection, where the binary modulation symbols are equiprobable and have equal-energy, is given by
\begin{eqnarray}
\Pr\bigl\{\bigl.e\,\bigr|\,H\bigr\}
    &=&\frac{1}{\Gamma(\nu)}\MeijerG[right]{2,0}{0,2}
                {\frac{\nu{k}\gamma}{1+k}}
                    {\emptycoefficient}
                        {0,\nu},\\
    &=&\frac{1}{\Gamma(\nu)}
        \biggl(\frac{\gamma}{\Lambda^2_{0}}\biggr)^{\!\frac{\nu}{2}}
            K_{\nu}\biggl(2\sqrt{\frac{\gamma}{\Lambda^2_{0}}}\biggr).
\end{eqnarray}
\end{theorem}

\begin{proof}
The proof is obvious setting $M\!=\!2$ in \theoremref{Theorem:MAPDecisionErrorProbabilityForNoncoherentOrthogonalSignallingOverAWMNChannels} and  performing simple algebraic manipulations.
\end{proof}

Let us now consider the special cases in order to check the numerical validity of the results presented above. It is worth noticing that, when the normality gets close to zero (i.e., while $\nu\!\rightarrow\!{0}^{+}$), the complex \ac{AWMN} channel turns into the noiseless channel and accordingly the conditional \ac{SER} approaches to zero  (i.e., $\Pr\{e\,|\,H\}\!\rightarrow\!{0}^{+}$) as expected. 
Further, in case of $\nu\!=\!1$, we simplify  \eqref{Eq:MLDecisionErrorProbabilityForNoncoherentEqualEnergyOrthogonalSignallingOverAWMNChannels} to 
\begin{subequations}
\label{Eq:MLDecisionErrorProbabilityForNoncoherentEqualEnergyOrthogonalSignallingOverAWLNChannels}
\begin{eqnarray}
    \label{Eq:MLDecisionErrorProbabilityForNoncoherentEqualEnergyOrthogonalSignallingOverAWLNChannelsA}
    \Pr\bigl\{\bigl.e\,\bigr|\,H\bigr\}
    &=&\sum_{k=1}^{M-1}
            \frac{(-1)^{1+k}}{1+k}\Binomial{M-1}{k}
            \MeijerG[right]{2,0}{0,2}
                {\frac{{k}\gamma}{1+k}}
                    {\emptycoefficient}{0,1},\\[1mm]
    \label{Eq:MLDecisionErrorProbabilityForNoncoherentEqualEnergyOrthogonalSignallingOverAWLNChannelsB}
    &=&2\sum_{k=1}^{M-1}
            \frac{(-1)^{1+k}}
                {1+k}\Binomial{M-1}{k}
        \sqrt{\frac{{k}\gamma}{1+k}}
            K_{1}\Biggl(2
                \sqrt{\frac{{k}\gamma}{1+k}}\Biggr),{~~~~~}    
\end{eqnarray}
\end{subequations}
which is the conditional \ac{SER} of non-coherent signaling over complex \ac{AWLN} channels. Setting $M\!=\!2$ in \eqref{Eq:MLDecisionErrorProbabilityForNoncoherentEqualEnergyOrthogonalSignallingOverAWLNChannels} results in the error probability for binary orthogonal signaling, including binary orthogonal FSK, with non-coherent detection in complex \ac{AWLN} channels, that is 
\begin{equation}
    \Pr\bigl\{\bigl.e\,\bigr|\,H\bigr\}=\sqrt{{\gamma}/{2}}K_{1}\bigl(\sqrt{2\gamma}\bigr).
\end{equation}
When the normality factor $\nu$ gets larger (i.e., $\nu\!\rightarrow\!\infty$), the additive white noise turns into \ac{AWGN} noise, and accordingly utilizing \cite[Eqs. (8.2.2/12) and (8.4.3/1)]{BibPrudnikovBookVol3} within 
\begin{equation}
\lim_{\nu\rightarrow\infty}
    \frac{1}{\Gamma(\nu)}
    \MeijerG[right]{2,0}{0,2}{\frac{\nu{k}\gamma}{1+k}}{\emptycoefficient}{0,\nu}=
        \exp\biggl(-\frac{{k}\gamma}{1+k}\biggr),
\end{equation}
the symbol error probability \eqref{Eq:MLDecisionErrorProbabilityForNoncoherentEqualEnergyOrthogonalSignallingOverAWMNChannels} readily simplifies more to 
\begin{equation}
\label{Eq:MLDecisionErrorProbabilityForNoncoherentEqualEnergyOrthogonalSignallingOverAWGNChannels}
\!\!\!\!\Pr\bigl\{\bigl.e\,\bigr|\,H\bigr\}=
        \sum_{k=1}^{M-1}
            \frac{(-1)^{1+k}}{1+k}\Binomial{M-1}{k}
                \exp\biggl(-\frac{{k}\gamma}{1+k}\biggr),\!\!
\end{equation}
which is in perfect agreement with the~conditional~\ac{SER}~performance of non-coherent \ac{ML} detection of equal-power orthogonal symbols\cite[Eq. (8.67)]{BibAlouiniBook}, \cite[Eq. (4.5-43)]{BibProakisBook}. For binary orthogonal signaling, including binary orthogonal FSK with non-coherent detection over complex \ac{AWGN} channels, \eqref{Eq:MLDecisionErrorProbabilityForNoncoherentEqualEnergyOrthogonalSignallingOverAWGNChannels} reduces to \cite[Eq. (4.5-45)]{BibProakisBook}, \cite[Eq. (8.69)]{BibAlouiniBook}, that is
\begin{equation}
    \Pr\bigl\{\bigl.e\,\bigr|\,H\bigr\}=\frac{1}{2}\exp\biggl(-\frac{\gamma}{2}\biggr).
\end{equation}
For numerical accuracy and convenience,  in \figref{Figure:ConditionalSEPForNoncoherentOrthogonalSignaling}, which is given at the top of the next page, we give the conditional \ac{SER} of non-coherent orthogonal signaling over complex \ac{AWMN} channels.  

\paragraph{\!Conditional \ac{SER} of Non-coherent Differential PSK} 
\label{Section:SignallingOverAWMNChannels:NonCoherentSignalling:NonCoherentDifferentialPSK}
The other type of non-coherent signaling is the \ac{DPSK} (i.e., the differentially encoded PSK) in which the information is encoded within the phase transition between two consecutive symbols and its demodulation\,/\,detection does not require the estimation of the carrier phase. In accordance with the channel model given by \eqref{Eq:ComplexAWMNVectorChannel}, the two consecutive received signal vectors can be readily written as 
\begin{eqnarray}
\defrvec{R}_{1}&=&{H}{e}^{\imaginary\Theta}
            \defmat{F}\defrvec{S}_{1}+
                \defrvec{Z}_{1},\\
\defrvec{R}_{2}&=&{H}{e}^{\imaginary\Theta}
            \defmat{F}\defrvec{S}_{2}+
                \defrvec{Z}_{2},
\end{eqnarray}
where~$\defrmat{Z}_{1}\!\sim\!\mathcal{CM}^{L}_{\nu}(\defvec{0},\defmat{\Sigma})$~and~$\defrmat{Z}_{2}\!\sim\!\mathcal{CM}^{L}_{\nu}(\defvec{0},\defmat{\Sigma})$~are~uncorrelated but certainly not independent, and $\defrvec{S}_{1}$ and $\defrvec{S}_{2}$ are two consecutive symbols. Accordingly, the vector representation of the lowpass equivalent of the received signal over a period of two symbol intervals is formally written as
\begin{equation}
\label{Eq:ComplexAWMNVectorChannelForDifferentialPSK}
    \underbrace{
    \begin{bmatrix}
        \defrvec{R}_{1}\\
        \defrvec{R}_{2}
    \end{bmatrix}}_{\defrvec{R}_{s}}=
    {H}{e}^{\imaginary\Theta}
    \underbrace{
    \begin{bmatrix}
        \defmat{F} & \defmat{0}\\
        \defmat{0} & \defmat{F}
    \end{bmatrix}}_{\defmat{F}_{s}}
    \underbrace{
    \begin{bmatrix}
        \defrvec{S}_{1}\\
        \defrvec{S}_{2}
    \end{bmatrix}}_{\defrvec{S}}
    +
    \underbrace{
    \begin{bmatrix}
        \defrvec{Z}_{1}\\
        \defrvec{Z}_{2}
    \end{bmatrix}}_{\defrvec{Z}_{s}},
\end{equation}
where $\defrmat{Z}_{s}\!\sim\!\mathcal{CM}^{2L}_{\nu}(\defvec{0},\defmat{\Sigma}_{s})$ is a \ac{CES} multivariate McLeish distribution whose the covariance matrix can be readily obtained as 
\begin{equation}
    \defmat{\Sigma}_{s}=
    \begin{bmatrix}
        \defmat{\Sigma} & \defmat{0}\\
        \defmat{0} & \defmat{\Sigma}
    \end{bmatrix},
\end{equation}
since the inphase and quadrature vectors of $\defrmat{Z}_{s}$ are mutually uncorrelated.
Moreover, in \eqref{Eq:ComplexAWMNVectorChannelForDifferentialPSK}, $\defvec{S}$ denotes the modulation symbol vector randomly chosen from the set of possible fixed modulation symbols $\{\defvec{s}_{1},\allowbreak\defvec{s}_{2},\ldots,\defvec{s}_{M}\}$.
As such, the $m$th message over a period of two modulation symbols can be written as 
\begin{equation}
\label{Eq:DifferentialPSKMessage}
    \defvec{s}_{m}=
        \begin{bmatrix}
            \defvec{s}\exp(\imaginary\phi_{\Sigma})\\
            \defvec{s}\exp(\imaginary(\phi_{m}+\phi_{\Sigma}))
        \end{bmatrix},
        \quad{1}\!\leq\!{m}\!\leq\!{M}
\end{equation}
where $\defvec{s}$ is such a signal that the power of the $m$th message, i.e., $E_{m}\!=\!\defvec{s}^{H}_{m}\defvec{s}_{m}$ is derived as $E_{m}\!=\!2\defvec{s}^{H}\defvec{s}$. Accordingly, the average power of signaling $E_{\defrvec{S}}$ is given by
\begin{equation}
    E_{\defrvec{S}}=\sum_{m=1}^{M}E_{m}\Pr\{\defrvec{S}=\defvec{s}_{m}\}=2\defvec{s}^{H}\defvec{s}.
\end{equation}
Further, in \eqref{Eq:DifferentialPSKMessage}, $\phi_{\Sigma}$ is the random phase due to non-coherent detection, and $\phi_{m}\!=\!2\pi(m-1)/M$ is the phase transition that encodes the information into the $m$th message. Since the information is entirely encoded in the phase transition between two consecutive symbols, the detection has to be carried over a period of two consecutive symbols. Referring to the \emph{slow variance uncertainty}, explained in \secref{Section:AWMNChannels:RandomNoiseVarianceFluctuations}, \textit{the variance fluctuation during two consecutive symbols is therefore assumed approximately constant}. With respect to \eqref{Eq:ComplexAWMNVectorChannelForDifferentialPSK}, the non-coherent \ac{MAP} receiver is given in the following theorem. 

\begin{theorem}
\label{Theorem:NoncoherentMAPDecisionRuleForDifferentialPSK}
For the complex vector channel given in \eqref{Eq:ComplexAWMNVectorChannelForDifferentialPSK}, the non-coherent \ac{MAP} detection rule of \ac{DPSK} is given by 
\begin{equation}
\label{Eq:NoncoherentMAPDecisionRuleForDifferentialPSK}
\widehat{m}
    =\argmax_{{1}\leq{m}\leq{M}}\,{p}_{m}
        I_{0}\Bigl(H\bigl|
            \defvec{s}^H\defmat{F}^{H}\defmat{\Sigma}^{-1}\defmat{F}\defrvec{R}_{1}\Bigr.+
        \Bigl.\exp(-\imaginary\phi_{m})\defvec{s}^H\defmat{F}^{H}\defmat{\Sigma}^{-1}\defmat{F}\defrvec{R}_{2}
            \bigr|\Bigr).{~~~~}
\end{equation}
\end{theorem}

\begin{proof}
Note that the \ac{MAP} detection of \ac{DPSK} uses \eqref{Eq:NoncoherentMAPDecisionRuleForComplexAWMNVectorChannel} for optimal detection. Accordingly, we have
\begin{equation}
\!\!\widehat{m}=\argmax_{{1}\leq{m}\leq{M}}\,
        {p}_{m}
        \exp\Bigl(\frac{1}{2}{H}^{2}\defvec{s}_{m}^H\defmat{F}_{s}^{H}\defmat{\Sigma}_{s}^{-1}\defmat{F}_{s}\defvec{s}_{m}\Bigr)
        \BesselI[0]{H\bigl|\defvec{s}_{m}^H\defmat{F}_{s}^{H}\defmat{\Sigma}_{s}^{-1}\defmat{F}_{s}\defrvec{R}_{s}\bigr|},\!\!
\end{equation}
which can be rewritten in terms of $\defrvec{R}_{1}$, $\defrvec{R}_{2}$, $\defmat{F}$, and $\defmat{\Sigma}$, that is
\begin{eqnarray}
\nonumber
\widehat{m}&=&\argmax_{{1}\leq{m}\leq{M}}\,
        {p}_{m}
        \exp\Bigl({H}^{2}\defvec{s}^H\defmat{F}^{H}\defmat{\Sigma}^{-1}\defmat{F}\defvec{s}\Bigr)\times\\[-2mm]
        \nonumber
        &~&{~~~~~~~~~~}I_{0}\Bigl(H\bigl|
            \exp(-\imaginary\phi_{\Sigma})\defvec{s}^H\defmat{F}^{H}\defmat{\Sigma}^{-1}\defmat{F}\defrvec{R}_{1}\Bigr.+\\[-2mm]
        \label{Eq:NoncoherentMAPDecisionRuleForDifferentialPSKII}
        &~&{~~~~~~~~~~~~}    
        \Bigl.\exp(-\imaginary(\phi_{\Sigma}+\phi_{m}))\defvec{s}^H\defmat{F}^{H}\defmat{\Sigma}^{-1}\defmat{F}\defrvec{R}_{2}
            \bigr|\Bigr),{~~~~~~~}
\end{eqnarray}
where $\exp(-\imaginary\phi_{\Sigma})$ can be ignored due to $\abs{e^{-\imaginary\phi_{\Sigma}}x}\!=\!\abs{x}$. In addition, since the term $\exp\bigl({H}^{2}\defvec{s}^H\defmat{F}^{H}\defmat{\Sigma}^{-1}\defmat{F}\defvec{s}\bigr)$ in \eqref{Eq:NoncoherentMAPDecisionRuleForDifferentialPSKII} is independent of index $m$, we can readily ignore it, which results in \eqref{Eq:NoncoherentMAPDecisionRuleForDifferentialPSK} and completes the proof of \theoremref{Theorem:NoncoherentMAPDecisionRuleForDifferentialPSK}.
\end{proof}

\begin{theorem}
\label{Theorem:NoncoherentMLDecisionRuleForDifferentialPSK}
For the complex vector channel given in \eqref{Eq:ComplexAWMNVectorChannelForDifferentialPSK}, the non-coherent \ac{ML} detection rule of \ac{DPSK} is given by 
\begin{equation}
\label{Eq:NoncoherentMLDecisionRuleForDifferentialPSK}
\widehat{m}
    =\argmax_{{1}\leq{m}\leq{M}}\,
        I_{0}\Bigl(H\bigl|
            \defvec{s}^H\defmat{F}^{H}\defmat{\Sigma}^{-1}\defmat{F}\defrvec{R}_{1}\Bigr.+
        \Bigl.\exp(-\imaginary\phi_{m})\defvec{s}^H\defmat{F}^{H}\defmat{\Sigma}^{-1}\defmat{F}\defrvec{R}_{2}
            \bigr|\Bigr).{~~~~}
\end{equation}
\end{theorem}

\begin{proof}
The proof is obvious using \theoremref{Theorem:NoncoherentMAPDecisionRuleForDifferentialPSK}.
\end{proof}

In order to avoid non-zero cross correlation between channels, we can equalize the channel by the precoding filter matrix $\defmat{F}_{s}$ whose diagonal matrix $\defmat{F}\!\in\!\mathbb{C}^{{2L}\times{2L}}$ supports $\defmat{\Sigma}\!=\!\frac{N_0}{2}\defmat{F}\defmat{F}^H$ for optimal reception, and then we can obtain  
\begin{subequations}
\label{Eq:NoncoherentComplexAWMNVectorChannelForDifferentialPSK}
\begin{eqnarray}
    \label{Eq:NoncoherentComplexAWMNVectorChannelForDifferentialPSKA}
    \defrvec{R}_{nc}
        &=&\defmat{F}_{s}^{-1}\defrvec{R}_{s},\\
    \label{Eq:NoncoherentComplexAWMNVectorChannelForDifferentialPSKB}
    &=&\defmat{F}_{s}^{-1}
            \bigl({H}{e}^{\imaginary\Theta}\defmat{F}_{s}\defvec{S}+\defrmat{Z}_{s}\bigr),\\
    \label{Eq:NoncoherentComplexAWMNVectorChannelForDifferentialPSKC}
    &\equiv&H{e}^{\imaginary\Theta}\defvec{S}+\defmat{F}_{s}^{-1}\defrmat{Z}_{s},\\
    \label{Eq:NoncoherentComplexAWMNVectorChannelForDifferentialPSKD}
    &=&H{e}^{\imaginary\Theta}\defvec{S}+\defrvec{Z}_{nc}.
\end{eqnarray}
\end{subequations}
where  $\defrvec{R}_{nc}\!=\![\defrvec{R}^{T}_{1,nc}\,\,\defrvec{R}^{T}_{2,nc}]^{T}$ is the received random vector in which 
$\defrvec{R}_{1,nc}\!=\!\defmat{F}^{-1}\defrvec{R}_{1}$ and $\defrvec{R}_{2,nc}\!=\!\defmat{F}^{-1}\defrvec{R}_{2}$~are~two~random vectors non-coherently recovered over a period of two modulation symbols. Moreover,  $\defrvec{Z}_{nc}\!\sim\!\mathcal{CM}_{\nu}^{2L}(\defmat{0},\frac{N_{0}}{2}\defmat{I})$ such that $\defrvec{Z}_{nc}\!=\![\defrvec{Z}^{T}_{1,nc} \defrvec{Z}^{T}_{2,nc}]^{T}$, where $\defrvec{Z}_{1,nc}\!\sim\!\mathcal{CM}_{\nu}^{L}(\defmat{0},\frac{N_{0}}{2}\defmat{I})$ and 
$\defrvec{Z}_{2,nc}\!\sim\!\mathcal{CM}_{\nu}^{L}(\defmat{0},\frac{N_{0}}{2}\defmat{I})$. 
Consequently, the non-coherent \ac{MAP} receiver is given in the following theorem. 

\begin{theorem}
\label{Theorem:NoncoherentUncorrelatedMAPDecisionRuleForDifferentialPSK}
For complex uncorrelated \ac{AWMN} vector channels, defined in \eqref{Eq:NoncoherentComplexAWMNVectorChannelForDifferentialPSK}, the non-coherent~\ac{MAP}~rule~is~given~by
\begin{equation}\label{Eq:NoncoherentUncorrelatedMAPDecisionRuleForDifferentialPSK}
    \widehat{m}=
        \argmax_{{1}\leq{m}\leq{M}}\,
            {p}_{m}\cos(\Phi-\phi_{m}),
\end{equation}
where the decision variable $\Phi$ is defined as the phase difference of the received signal in two adjacent intervals, that is
\begin{equation}\label{Eq:NoncoherentUncorrelatedMAPDecisionVariableForDifferentialPSK}
 \Phi=
    \arg\bigl(\defvec{s}^H\defrvec{R}_{2,nc}\bigr)-
        \arg\bigl(\defvec{s}^H\defrvec{R}_{1,nc}\bigr),
\end{equation}
where $\arg(z)$ gives the argument of the complex number $z$\emph{\cite[Eq. (12.02.02.0001.01)]{BibWolfram2010Book}}.
\end{theorem}

\begin{proof}
Using \theoremref{Theorem:NoncoherentMLDecisionRuleForDifferentialPSK} and then selecting both $\defmat{\Sigma}\!=\!\frac{{N}_{0}}{2}\defmat{I}$ and $\defmat{F}\!=\!\defmat{I}$, we have  
\begin{equation}\label{Eq:NoncoherentMLDecisionRuleForAWMNChannelDifferentialPSK}
\!\!\!\!\widehat{m}
    =\argmax_{{1}\leq{m}\leq{M}}\,
        {p}_{m}
        I_{0}\biggl(
            \scalemath{0.95}{0.95}{
            \frac{2H}{N_{0}}\Bigl|
                \defvec{s}^H\defrvec{R}_{1,nc}+
            {e}^{-\imaginary\phi_{m}}\,\defvec{s}^H\defrvec{R}_{2,nc}
                        \Bigr|}\biggr).\!\!
\end{equation}
Noticing that $I_{0}(x)$ is a monotonic increasing function for all $x\!\in\!\mathbb{R}_{+}$, we have $\argmax_{x}\,I_{0}\bigl(f(x)\bigr)=\argmax_{x}\,{f^{2}(x)}$,
for any monotonic increasing function $f\colon\mathbb{R}\!\rightarrow\!\mathbb{R}$. In consequence, ignoring $2{H}/N_{0}$, we can reduce \eqref{Eq:NoncoherentMLDecisionRuleForAWMNChannelDifferentialPSK}, that is 
\begin{equation}
\label{Eq:NoncoherentMLDecisionRuleForAWMNChannelDifferentialPSKII}
    \widehat{m}=\argmax_{{1}\leq{m}\leq{M}}\,
        {p}_{m}\bigl|
            \defvec{s}^H\defrvec{R}_{1,nc}+
                \exp(-\imaginary\phi_{m})\defvec{s}^H\defrvec{R}_{2,nc}
                        \bigr|^{2}.
\end{equation}
Using $|x+y|^2\!=\!|x|^2+|y|^2+2\RealPart{x^{*}y}$ and noticing that both $|\defvec{s}^H\defrvec{R}_{1,nc}|^2$ and $|\defvec{s}^H\defrvec{R}_{2,nc}|^2$ are independent of index $m$, we can simplify \eqref{Eq:NoncoherentMLDecisionRuleForAWMNChannelDifferentialPSKII} into
\begin{subequations}\label{Eq:NoncoherentMLDecisionRuleForAWMNChannelDifferentialPSKIII}
\begin{eqnarray}
\label{Eq:NoncoherentMLDecisionRuleForAWMNChannelDifferentialPSKIIIA}
\widehat{m}&=&\argmax_{{1}\leq{m}\leq{M}}\,
        {p}_{m}
            \Re\Bigl\{
                {C}_1^{*}{C}_2
                    \exp(-\imaginary\phi_{m})\Bigr\}.\\
    \label{Eq:NoncoherentMLDecisionRuleForAWMNChannelDifferentialPSKIIIB}
    &=&\argmax_{{1}\leq{m}\leq{M}}\,
        {p}_{m}
            \Re\Bigl\{
                \abs{{C}_1}\exp(-\imaginary\arg({C}_1))
                \abs{{C}_2}\exp(\imaginary\arg({C}_2))
                    \exp(-\imaginary\phi_{m})\Bigr\},{~~~~~~~}
\end{eqnarray}
\end{subequations}
where ${C}_1\!=\!\defvec{s}^H\defrvec{R}_{1,nc}$ and ${C}_2\!=\!\defvec{s}^H\defrvec{R}_{2,nc}$ are two complex envelopes recovered from two consecutive symbols, respectively, such that ${C}_1\!\sim\!\mathcal{CM}_{\nu}(H{e}^{\imaginary\phi_{\Sigma}},{E_{\defrvec{S}}N_{0}}/{4})$ and ${C}_2\!\sim\!\mathcal{CM}_{\nu}(H{e}^{\imaginary(\phi_{\Sigma}+\phi_{m})},\allowbreak{}{E_{\defrvec{S}}N_{0}}/{4})$.
Further, $\arg(z)$ is the argument of the complex number $z$, such that $z\!=\!\abs{z}{e}^{\imaginary\arg(z)}$\cite[Eq. (12.02.16.0029.01)]{BibWolfram2010Book}.
In addition, worth noting that $\abs{{C}_1}$ and $\abs{{C}_2}$ are independent of index $m$. Accordingly, \eqref{Eq:NoncoherentMLDecisionRuleForAWMNChannelDifferentialPSKIIIB} is reformulated as 
\begin{subequations}\label{Eq:NoncoherentMLDecisionRuleForAWMNChannelDifferentialPSKIV}
\begin{eqnarray}
    \label{Eq:NoncoherentMLDecisionRuleForAWMNChannelDifferentialPSKIVA}
    \widehat{m}
    &=&\argmax_{{1}\leq{m}\leq{M}}\,
        {p}_{m}\Re\Bigl\{
            {e}^{\imaginary(\arg({C}_2)-\arg({C}_1)-\phi_{m})}\Bigr\},\\
    \label{Eq:NoncoherentMLDecisionRuleForAWMNChannelDifferentialPSKIVB}
    &=&\argmax_{{1}\leq{m}\leq{M}}\,
        {p}_{m}\Re\Bigl\{
            {e}^{\imaginary(\Phi-\phi_{m})}\Bigr\},
\end{eqnarray}
\end{subequations}
where $\Phi$ denotes the phase difference of the received signal in two adjacent
intervals, simply defined as $\Phi\!=\!\arg({C}_{2}{C}^{*}_{1})\!=\!\arg({C}_{2})-\arg({C}_{1})$ and given by \eqref{Eq:NoncoherentUncorrelatedMAPDecisionVariableForDifferentialPSK}. Using Euler's formula\cite[Eq. (4.3.2)]{BibAbramowitzStegunBook}, \eqref{Eq:NoncoherentMLDecisionRuleForAWMNChannelDifferentialPSKIIIB} readily simplifies to \eqref{Eq:NoncoherentMAPDecisionRuleForDifferentialPSK}, which completes the proof of  \theoremref{Theorem:NoncoherentUncorrelatedMAPDecisionRuleForDifferentialPSK}.
\end{proof}

\begin{theorem}
\label{Theorem:NoncoherentUncorrelatedMLDecisionRuleForDifferentialPSK}
For complex uncorrelated \ac{AWMN} vector channels, defined in \eqref{Eq:NoncoherentComplexAWMNVectorChannelForDifferentialPSK}, the non-coherent~\ac{ML}~rule~is~given~by
\begin{equation}\label{Eq:NoncoherentUncorrelatedMLDecisionRuleForDifferentialPSK}
    \widehat{m}=
        \argmax_{{1}\leq{m}\leq{M}}\,
            \cos(\Phi-\phi_{m}).
\end{equation}
\end{theorem}

\begin{proof}
The proof is obvious using \theoremref{Theorem:NoncoherentUncorrelatedMAPDecisionRuleForDifferentialPSK}.
\end{proof}

As observed in both \theoremref{Theorem:NoncoherentUncorrelatedMAPDecisionRuleForDifferentialPSK} and \theoremref{Theorem:NoncoherentMLDecisionRuleForDifferentialPSK}, the receiver computes this phase difference $\Phi$ by using \eqref{Eq:NoncoherentUncorrelatedMAPDecisionVariableForDifferentialPSK} and compares it with all $\phi_{m}\!=\!2\pi(m-1)/M$, ${1}\!\leq\!{m}\!\leq\!{M}$ and selects the $m$ for which $\phi_{m}$ maximizes $\cos(\Phi-\phi_{m})$, thus resulting in minimum distance between $\Phi$ and $\phi_{m}$. In the following, we obtain the exact error probability of \ac{M-DPSK} signaling with non-coherent detection over complex \ac{AWMN} noise channels. 

\begin{figure*}[tp] 
\centering
\begin{subfigure}{0.47\columnwidth}
    \centering
    \includegraphics[clip=true, trim=0mm 0mm 0mm 0mm, width=1.0\columnwidth,height=0.85\columnwidth]{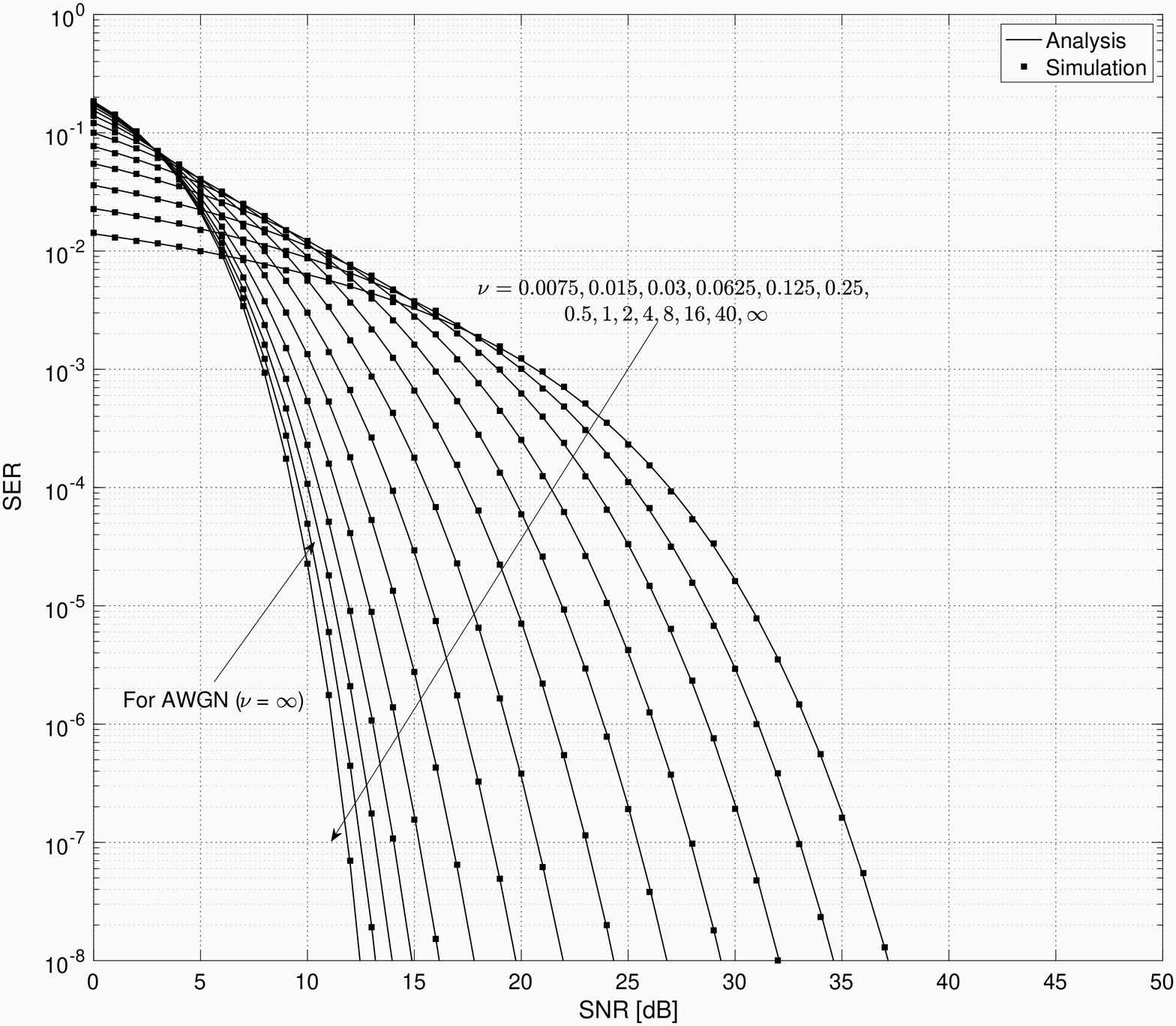}
    \caption{Modulation level $M=2$.}
    \vspace{5mm}
    \label{Figure:ConditionalSEPForNoncoherentMDPSKModulationA}
\end{subfigure}
{~~~}
\begin{subfigure}{0.47\columnwidth}
    \centering
    \includegraphics[clip=true, trim=0mm 0mm 0mm 0mm, width=1.0\columnwidth,height=0.85\columnwidth]{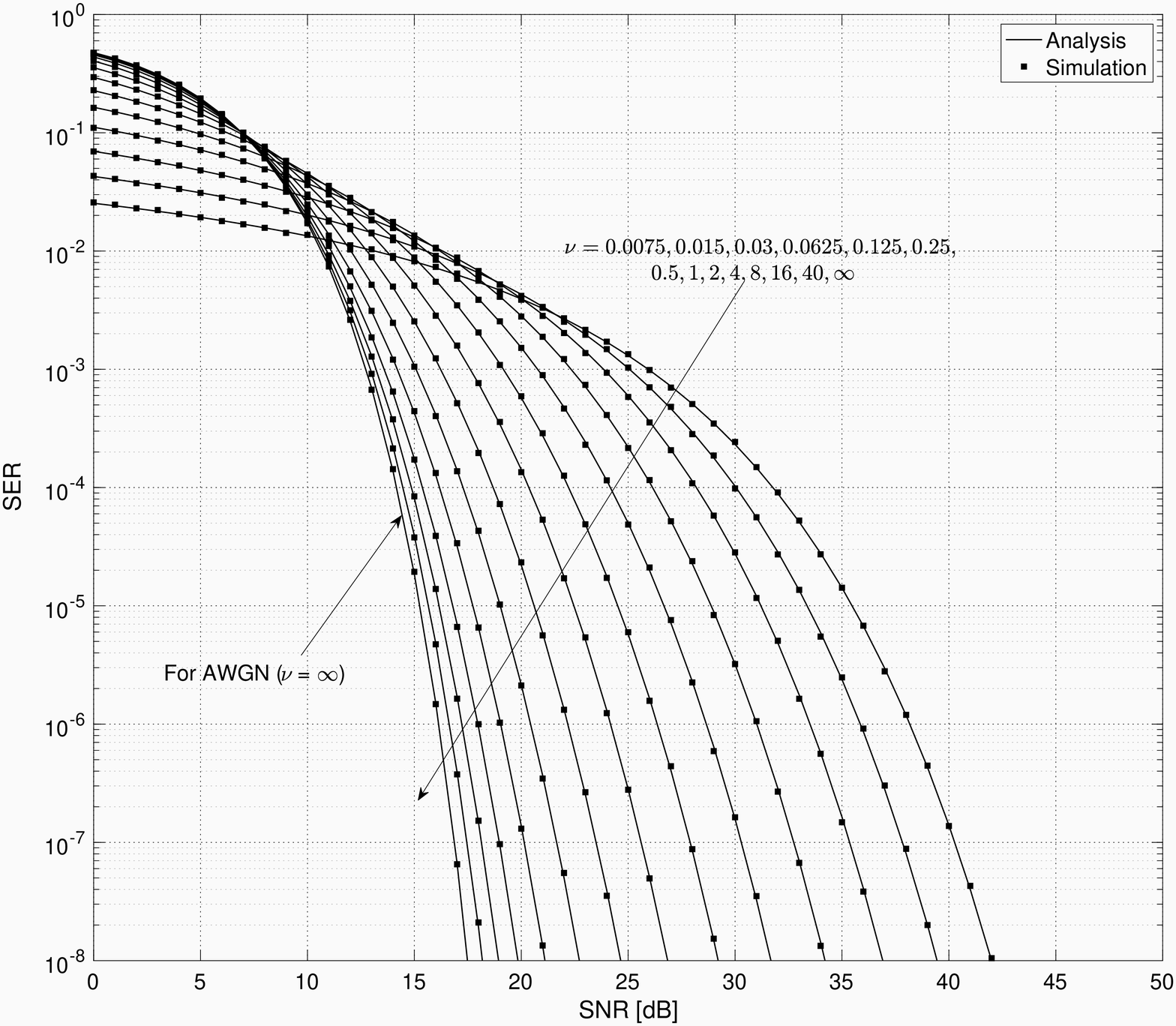}
    \caption{Modulation level $M=4$.}
    \vspace{5mm}
    \label{Figure:ConditionalSEPForNoncoherentMDPSKModulationB}
\end{subfigure}\\
\begin{subfigure}{0.47\columnwidth}
    \centering
    \includegraphics[clip=true, trim=0mm 0mm 0mm 0mm, width=1.0\columnwidth,height=0.85\columnwidth]{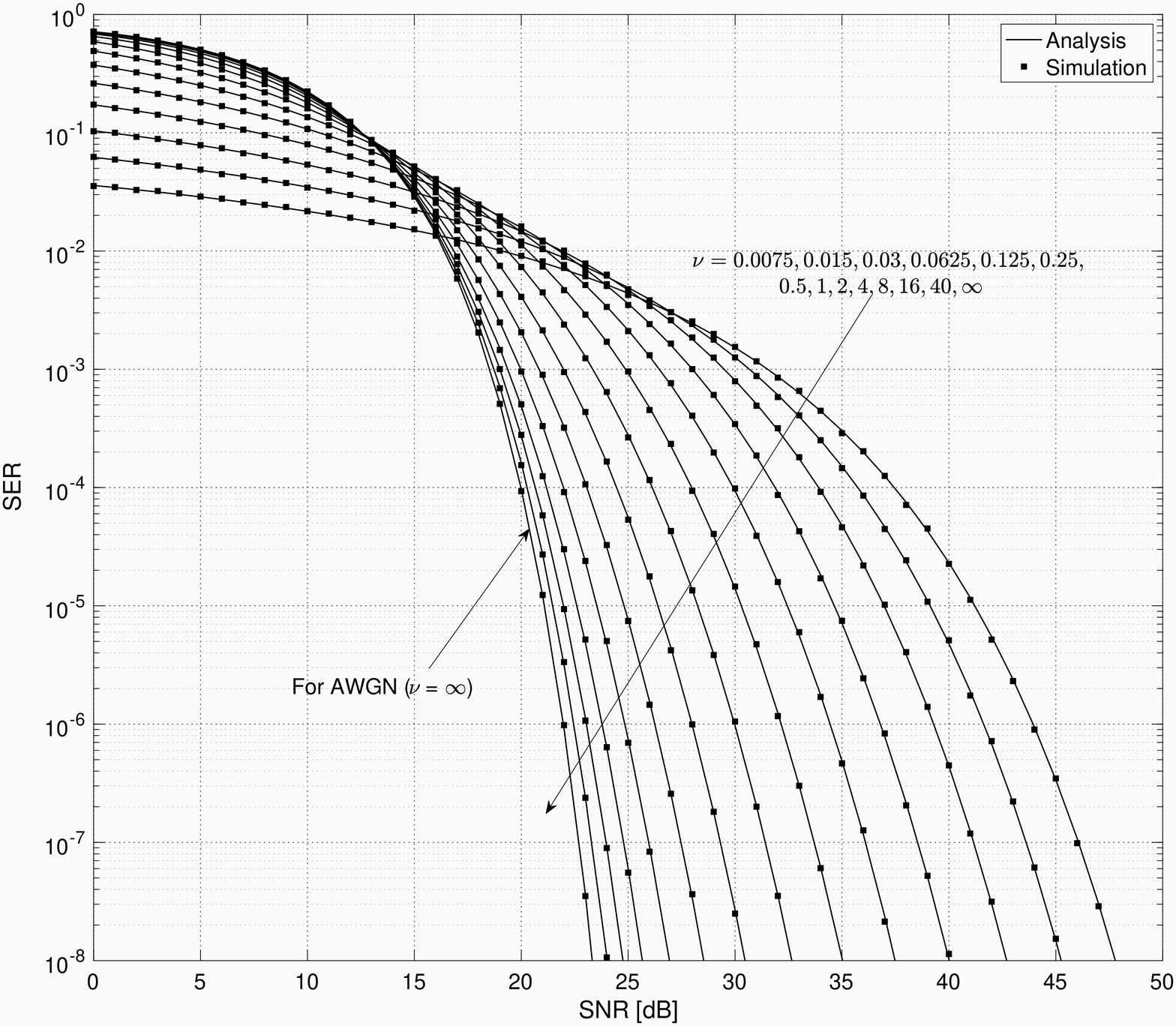}
    \caption{Modulation level $M=8$.}
    \label{Figure:ConditionalSEPForNoncoherentMDPSKModulationC}
\end{subfigure}
{~~~}
\begin{subfigure}{0.47\columnwidth}
    \centering
    \includegraphics[clip=true, trim=0mm 0mm 0mm 0mm, width=1.0\columnwidth,height=0.85\columnwidth]{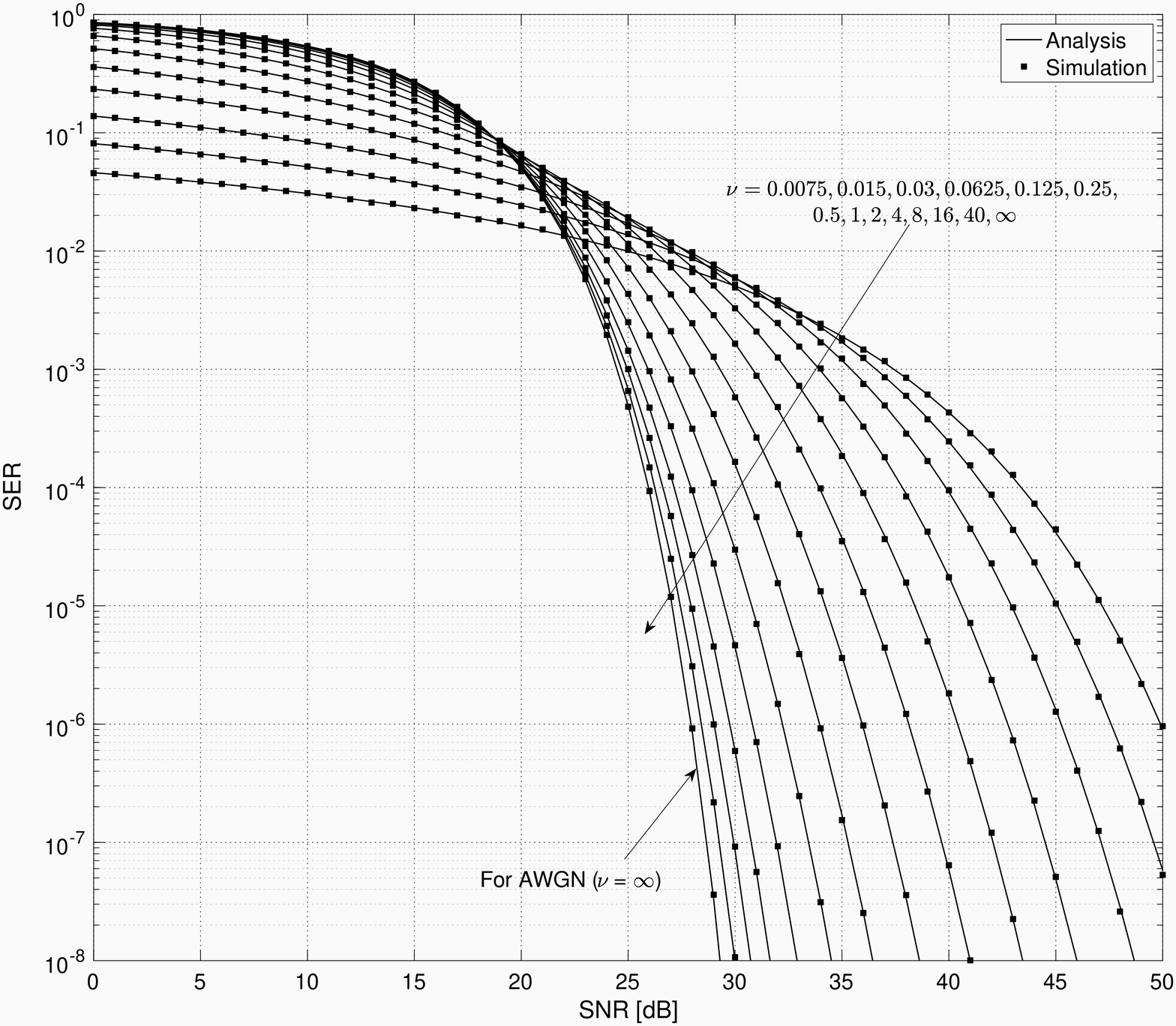}
    \caption{Modulation level $M=16$.}
    \label{Figure:ConditionalSEPForNoncoherentMDPSKModulationD}
\end{subfigure}
\caption{The \ac{SER} of non-coherent \ac{M-DPSK} modulation over \ac{AWMN} channels.}
\label{Figure:ConditionalSEPForNoncoherentMDPSKModulation}
\vspace{-2mm} 
\end{figure*} 

\begin{theorem}
\label{Theorem:MLDecisionErrorProbabilityForMDPSKCoherentSignallingOverAWMNChannels}
The conditional \ac{SER} of the \ac{M-DPSK} signaling with non-coherent \ac{ML} detection is given by
\begin{equation}
\label{Eq:MLDecisionErrorProbabilityForMDPSKCoherentSignallingOverAWMNChannels}
\Pr\{\bigl.e\,|\,H\}=\frac{2}{\pi\Gamma(\nu)\lambda^{\nu}_{0}}
        \int_{0}^{\pi-\frac{\pi}{M}}
        \biggl(
            \frac{2\gamma\sin^2(\frac{\pi}{M})}{1+\cos(\frac{\pi}{M})\cos(\theta)}
        \biggr)^{\!\frac{\nu}{2}}
            {K}_{\nu}\biggl(
                \frac{2}{\lambda_{0}}
                    \sqrt{
                    \frac{2\gamma\sin^2(\frac{\pi}{M})}{1+\cos(\frac{\pi}{M})\cos(\theta)}
                    }\biggr)d\theta,
\end{equation}
where $\gamma\!=\!{H^2E_{\defrvec{S}}}/{N_{0}}$ is the instantaneous \ac{SNR}. 
\end{theorem}

\begin{proof}
According to \eqref{Eq:NoncoherentUncorrelatedMAPDecisionVariableForDifferentialPSK}, the decision variable $\Phi$ is simply defined as the phase difference between ${C}_1\!=\!\defvec{s}^H\defrvec{R}_{1,nc}$ and ${C}_2\!=\!\defvec{s}^H\defrvec{R}_{2,nc}$,~where  
 ${C}_1\!\sim\!\mathcal{CM}_{\nu}(H{e}^{\imaginary\phi_{\Sigma}},{E_{\defrvec{S}}N_{0}}/{4})$ and ${C}_2\!\sim\!\mathcal{CM}_{\nu}(H{e}^{\imaginary(\phi_{\Sigma}+\phi_{m})},{E_{\defrvec{S}}N_{0}}/{4})$ are such two uncorrelated but not independent complex McLeish distributions that, using \theoremref{Theorem:CCSMcLeishDefinition}, their decomposition is written as 
\begin{eqnarray}
  {C}_1&=&\frac{1}{2}HE_{\defrvec{S}}
        \,{e}^{\imaginary\phi_{\Sigma}}+G({X}_{1}+\imaginary{Y}_{1})\\
  {C}_2&=&\frac{1}{2}HE_{\defrvec{S}}
        \,{e}^{\imaginary(\phi_{\Sigma}+\phi_{m})}+G({X}_{2}+\imaginary{Y}_{2}),
\end{eqnarray}
where ${X}_{1}\!\sim\!\mathcal{N}(0,{E_{\defrvec{S}}N_{0}}/{4})$, ${Y}_{1}\!\sim\!\mathcal{N}(0,{E_{\defrvec{S}}N_{0}}/{4})$,  ${X}_{2}\!\sim\!\mathcal{N}(0,{E_{\defrvec{S}}N_{0}}/{4})$ and ${Y}_{2}\!\sim\!\mathcal{N}(0,{E_{\defrvec{S}}N_{0}}/{4})$ are mutually \emph{i.i.d} Gaussian distributions. Further, $G\!\sim\!\mathcal{G}(\nu,1)$ follows the \ac{PDF} given in \eqref{Eq:ProportionPDF}. When conditioned on $G$, both ${C}_1$ and ${C}_2$ follow Gaussian distributions, and hence, $\Phi\!=\!\arg({C}_{2}{C}^{*}_{1})$ conditioned on $G$ is observed as the phase between two independent and identically distributed complex Gaussian distributions. Using \cite[Eq. (5)]{BibPawulaCOML1998}, we have 
\begin{equation}\label{Eq:ConditionalPhaseDistributionProbability}
\!\!\!\!\Pr\bigl\{-\phi<\Phi<\phi\,\bigl|\,G\bigr.\bigr\}=
        1-\frac{1}{2\pi}\int_{\phi-\pi}^{\pi-\phi}
            {e}^{-\frac{\gamma}{G}h(\phi,\theta)}
                d\theta,\!\!
\end{equation}
with $h(\phi,\theta)\!=\!\sin^2(\phi)/(1+\cos(\phi)\cos(\theta))$, where $\gamma$ denotes the instantaneous \ac{SNR} given by $\gamma\!=\!{H^2E_{\defrvec{S}}}/{N_{0}}$. When $\defvec{s}_{m}$ is transmitted, a correct decision is made iff $\phi_{m}-{\pi}/{M}\!<\!\Phi\!<\!\phi_{m}+{\pi}/{M}$ since $\arg(\defvec{s}_{m}\defvec{s}^{*}_{m\pm{1}})\!=\!\pi/M$. With the circularity of complex \ac{AWMN} noise, we notice that $\Pr\{\phi_{m}-{\pi}/{M}<\Phi<\phi_{m}+{\pi}/{M}\}$ and $\Pr\{-{\pi}/{M}<\Phi<{\pi}/{M}\}$ are the same. Hence, we can write the probability of making a correct decision as
\begin{equation}
    \Pr\bigl\{c\,\bigl|\,H,\defvec{s}_{m},G\bigr.\bigr\}=
        \Pr\bigl\{-{\pi}/{M}<\Phi<{\pi}/{M}\Bigr\}.
\end{equation}
Using $\Pr\{e\,|\,H,\defvec{s}_{m},G\}\!=\!1-\Pr\{c\,|\,H,\defvec{s}_{m},G\}$ and \eqref{Eq:ConditionalPhaseDistributionProbability}
and making allowance for the symmetry between the integration from $-(\pi-{\pi}/{M})$ to zero and the integration from zero to $(\pi-{\pi}/{M})$, we have
\begin{equation}\label{Eq:ConditionalMDPSKSymbolErrorProbability}
\!\!\!\!\Pr\bigl\{e\,\bigl|\,H,\defvec{s}_{m},G\bigr.\bigr\}=
        1-\frac{1}{2\pi}\int_{-(\pi-\phi)}^{\pi-\phi}
            {e}^{-\frac{\gamma}{G}h(\pi/M,\theta)}
                d\theta,\!\!
\end{equation}
Noticing $\Pr\bigl\{e\,\bigl|\,H,\defvec{s}_{m},G\bigr.\bigr\}\!=\!\Pr\bigl\{e\,\bigl|\,H,\defvec{s}_{n},G\bigr.\bigr\}$ for all ${m}\!\neq\!{n}$, we can obtain the probability $\Pr\bigl\{e\,\bigl|\,H,G\bigr.\bigr\}$ as follows
\vspace{-2mm}
\begin{subequations}\label{Eq:ConditionalMDPSKErrorProbability}
\begin{eqnarray}
    \label{Eq:ConditionalMDPSKErrorProbabilityA}
    \Pr\{e|H,G\}
        &=&\sum_{m=1}^{M}\Pr\{e|H,\defvec{s}_{m},G\}\Pr\{\defvec{s}_{m}\},\\
    \label{Eq:ConditionalMDPSKErrorProbabilityB}
        &=&\Pr\{e|H,\defvec{s}_{m},G\}.
\end{eqnarray}
\end{subequations}
Hence, the \ac{SER} $\Pr\{e|H\}$ of non-coherent \ac{M-DPSK} signaling over complex \ac{AWMN} channels can be  written as $\Pr\{e|H\}\!=\!\int_{0}^{\infty}\Pr\{e|H,g\}f_{G}(g)dg$, where substituting both \eqref{Eq:ProportionPDF} and \eqref{Eq:ConditionalMDPSKErrorProbability} results in
$\Pr\{e|H\}\!=\!\frac{1}{\pi}\int_{0}^{\pi-{\pi}/{M}}I_{M}(\gamma,\theta)\,d\theta$, where $I_{M}(\gamma,\theta)$ is obtained using \cite[Eq. (3.478/4)]{BibGradshteynRyzhikBook}, that is  
\begin{equation}\label{Eq:MDPSKAuxillaryIntegration}        
    I_{M}(\gamma,\theta)=\frac{2}{\Gamma(\nu)\lambda^{\nu}_{0}}
        \bigl(2\gamma\,{h}_{M}(\pi/M,\theta)\bigl)^{\nu/2}
            {K}_{\nu}\biggl(
                \frac{2}{\lambda_{0}}
                    \sqrt{2\gamma\,{h}(\pi/M,\theta)}\biggr).        
\end{equation}
Finally, using \eqref{Eq:MDPSKAuxillaryIntegration} in $\Pr\{e|H\}$ given above yields \eqref{Eq:MLDecisionErrorProbabilityForMDPSKCoherentSignallingOverAWMNChannels}, which completes the proof of \theoremref{Theorem:MLDecisionErrorProbabilityForMDPSKCoherentSignallingOverAWMNChannels}.
\end{proof}

\vspace{-1mm}
\begin{theorem}
\label{Theorem:MLDecisionErrorProbabilityForBDPSKCoherentSignallingOverAWMNChannels}
The conditional \ac{SER} of the \ac{BDPSK} signaling with non-coherent \ac{ML} detection is given by
\begin{equation}
\label{Eq:MLDecisionErrorProbabilityForBDPSKCoherentSignallingOverAWMNChannels}
\!\!\!\!\Pr\{\bigl.e\,|\,H\}=
    \frac{1}{\Gamma(\nu)\lambda^{\nu}_{0}}
        \bigl(2\gamma\bigr)^{\!\frac{\nu}{2}}
            {K}_{\nu}\biggl(
                \frac{2}{\lambda_{0}}
                    \sqrt{2\gamma}\biggr).
\end{equation}
\end{theorem}

\begin{proof}
The proof is obvious using \theoremref{Theorem:MLDecisionErrorProbabilityForMDPSKCoherentSignallingOverAWMNChannels}.
\end{proof}

For numerical accuracy and analytical validity with respect to \ac{SNR}, normality and modulation levels, we show  in \figref{Figure:ConditionalSEPForNoncoherentMDPSKModulation} the conditional \ac{SER} of non-coherent \ac{M-DPSK} signaling over complex \ac{AWMN} channels, where numerical and simulation-based results are in perfect agreement. We also therein acknowledge that the \ac{SER} performance deteriorates in high-\ac{SNR} regime while negligibly improves in low-\ac{SNR} regime when the impulsive nature of the additive noise increases (i.e., the normality $\nu$ decreases). 

\section{Summary and Conclusions}
\label{Section:SummaryAndConclusion}
In this article, we introduce and investigate a more general additive \mbox{non-Gaussian} distribution that we term as McLeish distribution. We study the basic statistical principles behind the laws of McLeish distribution, not only ranging from non-Gaussian distribution to Gaussian distribution but also starting with the univariate case and continuing through to the multivariate case either in real domains or complex domains. Notably, we propose the following distributions and obtain closed-form \ac{PDF}, \ac{CDF}, \ac{MGF}, and moment expressions for their statistical characterization:
\begin{itemize}
    \item McLeish distribution,
    \item The sum of McLeish distributions,
    \item \ac{CCS}\,/\,\ac{CES} McLeish distribution,
    \item Multivariate McLeish distribution,
    \item Multivariate McLeish distribution with real, symmetric and positive-definite covariance matrix,
    \item Multivariate \ac{CCS}\,/\,\ac{CES} McLeish distribution,
    \item Multivariate \ac{CCS}\,/\,\ac{CES} McLeish distribution with complex, Hermitian symmetric and positive-definite covariance matrix.
\end{itemize}
As a result of \textit{these closed-form expressions, each of which is mathematically tractable and practically (physically) understandable},
we propose the framework of the laws of McLeish distribution for the first time in the literature. Further, we offer solutions to the challenges and problems caused by impulsive effects that lead to the heavy-tailed distribution of non-Gaussian noise. So much so that with this framework, we can obtain \textit{mathematically tractable} results that facilitate the analytical and numerical solutions of many problems in science and engineering. 

Besides, aside from the statistical laws of McLeish distribution, we propose and demonstrate that \textit{the random nature of McLeish distribution can model different impulsive noise environments} commonly encountered in wireless communications. We theoretically justify the existence of McLeish noise distribution in communication systems in case of uncertainty due to that the additive noise distribution has impulsive effects causing the variance of additive noise varies over time. We analyze how these impulsive effects can be reasonably modeled as uncertainty in the variance of additive noise. For the first time in the literature, we use \textit{Allan's variance} to determine the coherence time at which the variance of the additive noise can be considered constant. Concerning this coherence time, we demonstrate how to classify the additive noise channels as i) constant variance, ii) slow-variance uncertainty, and iii) fast-variance uncertainty. Accordingly, we investigate and prove the existence of McLeish noise distribution and show that \textit{the thermal noise in electronic materials follows McLeish distribution rather than Gaussian distribution}. Also, we demonstrate that \textit{\ac{MAI}\,/\,\ac{MUI} follows McLeish distribution rather than Laplacian distribution}. To represent how McLeish distribution can model a wide range of realistic impulsive effects (uncertainty of noise variance), we find out that the \textit{McLeish distribution exhibits a superior fit to the different impulsive noise} from non-Gaussian to Gaussian distribution.

Consequently, as an outcome of modeling \mbox{additive noise} as McLeish distribution, we \textit{present \acf{AWMN} channels} for the first time in the literature. For coherent\,/\,non-coherent signaling over \ac{AWMN} channels, we propose analytical \ac{MAP} and \ac{ML} symbol decision rules for optimum receivers and thereby \textit{obtain closed-form expressions for both \ac{BER} of binary  modulation schemes and \ac{SER} of various M-ary modulation schemes}. We conclude and identify how the non-Gaussian nature of additive noise impacts on the performance of communications systems by using McLeish distribution. Furthermore, we verify the validity and accuracy of our novel \ac{BER}\,/\,\ac{SER} expressions by some selected numerical examples and some computer-based simulations.

\bibliography{IEEEfull,yilmaz_additive_white_mcleish_noise_distribution}
\bibliographystyle{IEEEtran}
\phantomsection
\end{document}